\documentclass[oneside,english]{amsart}
\usepackage[T1]{fontenc}
\usepackage[latin9]{inputenc}
\usepackage{geometry}
\usepackage{verbatim}
\usepackage{amsthm}
\usepackage{amstext}
\usepackage{amssymb}
\usepackage{graphicx}
\usepackage{esint}
\usepackage{url}
\usepackage{yhmath}
\usepackage[all]{xy}
\usepackage{color}

\usepackage[percent]{overpic}

\usepackage[space]{grffile} 
\makeatletter
\numberwithin{equation}{section}
\numberwithin{figure}{section}
\theoremstyle{plain}
\newtheorem*{thm*}{\protect\theoremname}
\theoremstyle{plain}
\newtheorem{thm}{\protect\theoremname}[section]
\theoremstyle{plain}
\newtheorem{lem}[thm]{\protect\lemmaname}
\theoremstyle{remark}
\newtheorem{rem}[thm]{\protect\remarkname}
\theoremstyle{plain}
\newtheorem{prop}[thm]{\protect\propositionname}
\theoremstyle{plain}
\newtheorem{cor}[thm]{\protect\corollaryname}
\theoremstyle{plain}
\newtheorem{conj}[thm]{\protect\conjecturename}
\newtheorem*{upd*}{Update}
\theoremstyle{plain}

\theoremstyle{plain}
\newtheorem{exa}[thm]{Example}
\theoremstyle{plain}
\newtheorem{defn}[thm]{Definition}
\theoremstyle{plain}

\newtheorem{algorithm}[thm]{Algorithm}
\theoremstyle{plain}

\usepackage{mathrsfs}
\usepackage{subfigure}

\newcommand{\SLE}{\mathrm{SLE}}
\newcommand{\SLEk}{\mathrm{SLE}_{\kappa}}

\newcommand{\SLEmeasure}{\mathsf{P}}

\newcommand{\PR}{\mathsf{P}}

\newcommand{\sZ}{\mathcal{Z}}
\newcommand{\sD}{\mathcal{D}}

\newcommand{\sL}{\mathcal{L}}

\newcommand{\bR}{\mathbb{R}}

\newcommand{\bRpos}{\mathbb{R}_{> 0}}

\newcommand{\bZ}{\mathbb{Z}}
\newcommand{\Z}{\bZ}
\newcommand{\bN}{\mathbb{N}}

\newcommand{\bZpos}{\mathbb{Z}_{> 0}}
\newcommand{\bZnn}{\mathbb{Z}_{\geq 0}}
\newcommand{\bQ}{\mathbb{Q}}

\newcommand{\bC}{\mathbb{C}}
\newcommand{\C}{\bC}

\newcommand{\bH}{\mathbb{H}}

\newcommand{\im}{\Im\mathfrak{m}}

\newcommand{\domain}{\Lambda}
\newcommand{\bdry}{\partial}
\newcommand{\cl}[1]{\overline{#1}}
\newcommand{\Mob}{\mu}
\newcommand{\confmap}{\phi}

\newcommand{\OO}{\mathcal{O}}
\newcommand{\oo}{\mathit{o}}

\newcommand{\ud}{\mathrm{d}}

\newcommand{\pder}[1]{\frac{\partial}{\partial#1}}
\newcommand{\pdder}[1]{\frac{\partial^{2}}{\partial#1^{2}}}
\newcommand{\pddder}[1]{\frac{\partial^{3}}{\partial#1^{3}}}
\newcommand{\pddmix}[2]{\frac{\partial^{2}}{\partial#1 \partial#2}}

\newcommand{\set}[1]{\left\{  #1\right\}  }

\newcommand{\Arch}{\mathrm{LP}}
\newcommand{\LP}{\Arch}

\newcommand{\link}[2]{\{ #1 , #2 \}}
\newcommand{\nested}{\boldsymbol{\underline{\Cap}}}
\newcommand{\unnested}{\boldsymbol{\underline{\cap\cap}}}

\newcommand{\Catalan}{\mathrm{C}}

\newcommand{\KWleq}{\stackrel{\scriptscriptstyle{()}}{\scriptstyle{\longleftarrow}}} 
\newcommand{\Mmat}{\mathscr{M}}
\newcommand{\Minv}{\mathscr{M}^{-1}}

\newcommand{\chamber}{\mathfrak{X}}
\newcommand{\PartF}{\sZ}

\newcommand{\HarmMeas}{\mathsf{H}}

\newcommand{\ExcK}{\mathsf{K}}
\newcommand{\ExcKH}{\mathcal{K}}
\newcommand{\ExcKdom}{\mathcal{K}_\domain}
\newcommand{\FominDet}{\mathbf{\Delta}}
\newcommand{\LPdet}[2]{\FominDet_{#1}^{#2}}
\newcommand{\Ampl}{\zeta}

\newcommand{\eps}{\varepsilon}

\newcommand{\Gr}{\mathcal{G}}
\renewcommand{\Vert}{\mathcal{V}}
\newcommand{\Edg}{\mathcal{E}}
\newcommand{\GrWired}{\Gr/\bdry}
\newcommand{\ein}{e_\mathrm{in}}
\newcommand{\eout}{e_\mathrm{out}}

\newcommand{\pin}{p_\mathrm{in}}
\newcommand{\pout}{p_\mathrm{out}}

\newcommand{\xin}{x_\mathrm{in}}
\newcommand{\xout}{x_\mathrm{out}}

\newcommand{\tree}{\mathcal{T}}

\newcommand{\LE}{\mathrm{LE}}

\newcommand{\SymmGrp}{\mathfrak{S}}
\newcommand{\sgn}{\mathrm{sgn}}

\newcommand{\walksfromto}[2]{\mathscr{W}(#1,#2)}
\newcommand{\walkfromto}[3]{#1 \in \walksfromto{#2}{#3}}
\newcommand{\pathfromto}[2]{#1 \rightsquigarrow #2}

\newcommand{\edgeof}[2]{{\langle #1 , #2 \rangle}}

\newcommand{\DP}{\mathrm{DP}}
\newcommand{\DPleq}{\preceq} 
\newcommand{\DPgeq}{\succeq} 
\newcommand{\BPE}{\mathrm{BPE}}
\newcommand{\BPEfont}[1]{\textup{\tt{#1}}}
\newcommand{\wedgeat}[1]{\lozenge_#1} 
\newcommand{\upwedgeat}[1]{\wedge^#1}
\newcommand{\downwedgeat}[1]{\vee_#1}
\newcommand{\wedgelift}[1]{\uparrow \wedgeat{#1}} 
\newcommand{\removewedge}[1]{\setminus \wedgeat{#1}}
\newcommand{\removeupwedge}[1]{\setminus \upwedgeat{#1}}
\newcommand{\removedownwedge}[1]{\setminus \downwedgeat{#1}}

\newcommand{\slopeat}[1]{\times_#1} 
\newcommand{\opar}{\rm{\url{(}}} 
\newcommand{\cpar}{\rm{\url{)}}}
\newcommand{\nestedtilingof}{T_0}

\newcommand{\CItilingsof}{\mathcal{C}}

\newcommand{\walk}{\alpha}

\makeatletter
\newcommand*{\centerfloat}{%
  \parindent \z@
  \leftskip \z@ \@plus 1fil \@minus \textwidth
  \rightskip\leftskip
  \parfillskip \z@skip}
\makeatother

\AtBeginDocument{

}

\makeatother

\usepackage{babel}
\providecommand{\corollaryname}{Corollary}
\providecommand{\lemmaname}{Lemma}
\providecommand{\propositionname}{Proposition}
\providecommand{\remarkname}{Remark}
\providecommand{\theoremname}{Theorem}
\providecommand{\conjecturename}{Conjecture}

\definecolor{kallecol}{rgb}{.75,.0,.55}

\begin{document}


\author{A.~Karrila, K.~Kytölä, and E.~Peltola}

\

\vspace{2.5cm}

\begin{center}
\LARGE \bf \scshape {Boundary correlations in planar LERW and UST 
}
\end{center}

\vspace{0.75cm}

\begin{center}
{\large \scshape Alex Karrila}\\
{\footnotesize{\tt alex.karrila@aalto.fi}}\\
{\small{Department of Mathematics and Systems Analysis}}\\
{\small{P.O. Box 11100, FI-00076 Aalto University, Finland}}\bigskip{}
\\
{\large \scshape Kalle Kyt\"ol\"a}\\
{\footnotesize{\tt kalle.kytola@aalto.fi}}\\
{\small{Department of Mathematics and Systems Analysis}}\\
{\small{P.O. Box 11100, FI-00076 Aalto University, Finland}\\
\url{https://math.aalto.fi/~kkytola/}}\bigskip{}
\\
{\large \scshape Eveliina Peltola}\\
{\footnotesize{\tt eveliina.peltola@unige.ch}}\\
{\small{Section de Math\'{e}matiques, Universit\'{e} de Gen\`{e}ve}}\\
{\small{2--4 rue du Li\`{e}vre, Case Postale 64, 1211 Gen\`{e}ve 4, Switzerland}}
\end{center}

\vspace{0.75cm}

\begin{center}
\begin{minipage}{0.85\textwidth} \footnotesize
{\scshape Abstract.}
We find explicit formulas for the probabilities of general boundary visit events for planar
loop-erased random walks, as well as 
connectivity events for branches in the uniform spanning tree. We show that both probabilities, when 
suitably renormalized, converge in the scaling
limit to conformally covariant functions which satisfy partial differential equations
of second and third order, as predicted by conformal field theory. The scaling limit
connectivity probabilities also provide formulas for the pure partition functions of
multiple $\SLE_\kappa$ at $\kappa=2$.
\end{minipage}
\end{center}



\tableofcontents

\bigskip{}

\section{\label{sec: intro}Introduction}

In this article, we consider the planar loop-erased random walk (LERW) and 
a closely related model, the planar uniform spanning tree (UST). We find
explicit expressions for probabilities of certain connectivity and boundary visit
events in these models, and prove that, in the scaling limit as the lattice spacing 
tends to zero, these observables converge to conformally covariant functions, 
which satisfy systems of second and third order partial differential equations (PDEs) 
predicted by conformal field theory (CFT)
\cite{BPZ-infinite_conformal_symmetry_of_critical_fluctiations,
BSA-degenerate_CFTs_and_explicit_expressions,
BB-SLE_CFT_and_zigzag_probabilities, BBK-multiple_SLEs,
KP-conformally_covariant_boundary_correlation_functions_with_a_quantum_group,
Dubedat-SLE_and_Virasoro_representations_localization,
Dubedat-SLE_and_Virasoro_representations_fusion, JJK-SLE_boundary_visits}. 
As a byproduct, we also obtain 
formulas for the pure partition functions of multiple 
Schramm-Loewner evolutions ($\SLE$s), and the existence of all extremal
local multiple $\SLE$s at the specific parameter value $\kappa=2$ corresponding
to these models \cite{BBK-multiple_SLEs, Dubedat-commutation, 
KP-pure_partition_functions_of_multiple_SLEs}.

We now give a brief overview of the main results of this article. The 
detailed formulations of the statements made here are given in the bulk of the 
article, as referred to below.

\subsection{Boundary visit probabilities of the loop-erased random walk}

Throughout, we let $\Gr = (\Vert,\Edg)$ be a finite connected graph with 
a distinguished non-empty subset $\bdry \Vert \subset \Vert$ of vertices, 
called boundary vertices. We assume that $\Gr$ is a planar graph embedded in a 
Jordan domain $\domain \subset \bC$ such that the boundary vertices $\bdry \Vert$ 
lie on the Jordan curve $\bdry \domain$. 
We denote by $\bdry \Edg \subset \Edg$ the set of edges 
$e = \edgeof{e^\bdry}{e^\circ}$ which connect a boundary vertex 
$e^\bdry \in \bdry \Vert$ to an interior vertex 
$e^\circ \in \Vert^\circ := \Vert \setminus \bdry \Vert$.

Our scaling limit results are valid in any setup where the random walk excursion 
kernels and their discrete derivatives on the graphs converge to the Brownian 
excursion kernels and their derivatives. To be specific, we present the results 
in the following square grid approximations (see Figure~\ref{fig: square grid approx}).
We fix a Jordan domain $\domain$ and a number of boundary points 
$p_1 , p_2 , \ldots$ and $\hat{p}_1 , \hat{p}_2 , \ldots$
on horizontal or vertical parts of the boundary $\bdry \domain$.
For any small $\delta>0$,
we let $\Gr^\delta = (\Vert^\delta , \Edg^\delta)$ be a subgraph of the square
lattice $\delta \bZ^2$ of mesh size $\delta$ naturally approximating the domain 
$\domain$ (as detailed in Section~\ref{sec: applications to UST}).
We also let $e_j^\delta \in \bdry \Edg^\delta$ be 
the boundary edge nearest to $p_j$, and let $\hat{e}_s^\delta \in \Edg^\delta$ 
be the edge at unit graph distance from the boundary nearest to $\hat{p}_s$.
By scaling limits we mean limits of our quantities of interest as $\delta \to 0$ 
in this setting.

\begin{figure}
\centering
\includegraphics[width=.55\textwidth]{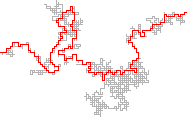}
\caption{
\label{fig: LERW}
A loop-erased random walk 
on the square grid and the underlying random walk (in grey).
}
\end{figure}

Generally, a loop-erased random walk (LERW) 
is a simple path on the graph obtained by erasing loops from 
a random walk
in their order of appearance ---
see Figure~\ref{fig: LERW} for illustration and Section~\ref{sec: application to UST} for the precise definitions.
Choose two boundary edges $\ein = \edgeof{\ein^\bdry}{\ein^\circ}$ and
$\eout = \edgeof{\eout^\bdry}{\eout^\circ}$. By the LERW $\lambda$ on $\Gr$
from $\ein$ to $\eout$ we mean the loop-erasure of a random walk
started from the vertex $\ein^\circ$ and conditioned to reach the boundary 
via the edge $\eout$. 
In the present article, we find an explicit formula for the probability of 
the event that the LERW $\lambda$ passes through given edges 
$\hat{e}_1 , \ldots , \hat{e}_{N'}$ at unit distance from the boundary,
as illustrated in Figure~\ref{fig: LERW boundary visit event}. Moreover, 
we prove that this probability converges in the scaling limit 
to a conformally covariant function $\Ampl$, which satisfies a system of 
second and third order partial differential equations, 
as has been previously predicted
based on conformal field theory. 

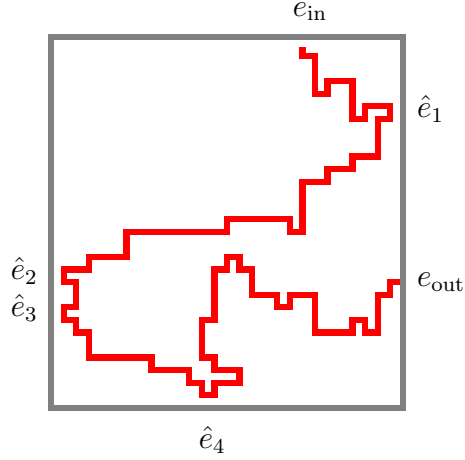
\begin{figure}
\centering
\begin{overpic}[width=.3\textwidth]{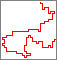}
 \put (65,105) {\Large $\ein$}
 \put (98,78) {\Large $\hat{e}_1$}
 \put (-10,35) {\Large $\hat{e}_2$}
 \put (-10,25) {\Large $\hat{e}_3$}
 \put (40,-10) {\Large $\hat{e}_4$}
 \put (98,33) {\Large $\eout$}
\end{overpic}
\bigskip
\caption{\label{fig: LERW boundary visit event}
Illustration of a loop-erased random walk passing through edges 
at unit distance from the boundary. 
We study the probabilities of such boundary visit events.
}
\end{figure}

\begin{thm*}
[Theorem~\ref{thm: scaling limit of LERW bdry visits}]
Let $\Gr^\delta$ with $\ein^\delta , \eout^\delta \in \bdry \Edg^\delta$ and
$\hat{e}_1^\delta , \ldots, \hat{e}_{N'}^\delta \in \Edg^\delta$
be a square grid approximation of a domain $\domain$ with marked boundary
points $\pin, \pout, \hat{p}_1 , \ldots, \hat{p}_{N'}$.
The probability that the LERW $\lambda^\delta$ on $\Gr^\delta$ from $\ein^\delta$
to $\eout^\delta$ 
passes through the edges $\hat{e}_1^\delta , \ldots , \hat{e}_{N'}^\delta$ 
has the following conformally covariant scaling limit:
\begin{align*}
\frac{1}{\delta^{3N'}} \, 
\PR_{\ein^\delta,\eout^\delta} \big[ \lambda^\delta 
        \text{ uses $\hat{e}_1^\delta , \ldots , \hat{e}_{N'}^\delta$} \big] 
\; \underset{\delta \to 0}{\longrightarrow} \;
\prod_{s=1}^{N'} |\confmap'(\hat{p}_s)|^3 \times
     \frac{\Ampl \big( \confmap(\pin); \confmap(\hat{p}_1), \ldots,
         \confmap(\hat{p}_{N'}); \confmap(\pout) \big)}{{\big( \confmap(\pout) - \confmap(\pin) \big)^{-2}}} ,
\end{align*}
where $\confmap \colon \domain \to \bH$ is a conformal map from the domain
$\domain$ to the upper half-plane 
$\bH = \set{z \in \bC \;|\; \im(z) > 0}$.
%
The function $\Ampl(\xin;\hat{x}_1,\ldots,\hat{x}_{N'};\xout)$  is a
positive  function of the real points 
$\xin, \hat{x}_1, \cdots, \hat{x}_{N'}, \xout$, 
satisfying the following system of linear partial differential equations: 
the $N'$ third order PDEs
\begin{align*}
\Bigg[ \, & \pddder{\hat{x}_s}
    + \frac{8}{\xin-\hat{x}_s} \pddmix{\xin\,}{\hat{x}_s} - \frac{8}{(\xin-\hat{x}_s)^2} \pder{\hat{x}_s}
    - \frac{12}{(\xin-\hat{x}_s)^2} \pder{\xin} + \frac{24}{(\xin-\hat{x}_s)^3} \\
&\quad   + \frac{8}{\xout-\hat{x}_s} \pddmix{\xout\,}{\hat{x}_s} - \frac{8}{(\xout-\hat{x}_s)^2} \pder{\hat{x}_s} 
    - \frac{12}{(\xout-\hat{x}_s)^2} \pder{\xout} + \frac{24}{(\xout-\hat{x}_s)^3} \\
&\quad   + \sum_{\substack{1 \leq t \leq N' \\ t \neq s}} \bigg( \frac{8}{\hat{x}_t-\hat{x}_s} \pddmix{\hat{x}_t}{\hat{x}_s}
            - \frac{24}{(\hat{x}_t-\hat{x}_s)^2} \pder{\hat{x}_s} 
    - \frac{12}{(\hat{x}_t-\hat{x}_s)^2} \pder{\hat{x}_t}
            + \frac{72}{(\hat{x}_t-\hat{x}_s)^3} \bigg)
\, \Bigg] \; \Ampl 
        = 0
\end{align*}
for $s=1,\ldots,N'$, as well as the second order PDE
\begin{align*}
\Bigg[ \pdder{\xin} + \frac{2}{\xout-\xin} \pder{\xout} - \frac{2}{(\xout-\xin)^2 }
    + \sum_{t=1}^{N'} \frac{2}{\hat{x}_t-\xin} \pder{\hat{x}_t} - \sum_{t=1}^{N'} \frac{6}{(\hat{x}_t-\xin)^2 }
\Bigg] \;\Ampl 
        = 0
\end{align*}
and another second order PDE obtained by interchanging the roles of $\xin$ and $\xout$ above.
\end{thm*}

We make some remarks concerning the above result.


The scaling limit of the LERW path $\lambda^\delta$ above is a conformally
invariant random curve called the chordal $\SLEk$ with parameter $\kappa=2$,
by some of the earliest of the celebrated works showing convergence of
lattice model interfaces to $\SLE$s \cite{LSW-LERW_and_UST, Zhan-scaling_limits_of_planar_LERW},
see also \cite{CS-discrete_complex_analysis_on_isoradial,YY-Loop-erased_random_walk_and_Poisson_kernel_on_planar_graphs}.

The scaling limit of the probability that a LERW uses one particular edge
in the interior of the domain has been studied in
\cite{Kenyon-the_asymptotic_determinant_of_the_discrete_Laplacian,
Lawler-the_probability_that_planar_LERW_uses_a_given_edge},
see also \cite{KW-spanning_trees_of_graphs_on_surfaces_and_the_intensity_of_LERW}.
The most accurate results currently known \cite{BLV-scaling_limit_of_LERW_Greens_function}
state that this probability divided by $\delta^{3/4}$ 
converges in the scaling limit to a conformally covariant expression known as the
$\SLE$ Green's function at $\kappa=2$.
This has been recently used to prove the convergence of the LERW to
the chordal $\SLEk$ with $\kappa=2$ in a sense that includes 
parametrization~\cite{LV-natural_parametrization_for_SLE}.
In contrast to the probability of visiting an interior edge,
our scaling limit results concern probabilities of visits to an arbitrary number
of edges at unit distance from the boundary.

The scaling exponent $3$ appearing in the renormalization by $\delta^{3 N'}$ and
the conformal covariance factors $|\confmap'(\hat{p}_j)|^3$ of the above theorem 
is the conformal weight $h_{1,3}$ in the Kac table
for a CFT with central charge $c=-2$, i.e., $h_{1,3}=3$ is a highest
weight for a degenerate highest weight representation of the Virasoro algebra
\cite{Kac-ICM_proceedings_Helsinki, FF-representations, IK-representation_theory_of_the_Virasoro_algebra}.
Among the most remarkable predictions of CFT are
partial differential equations of order $r s$ for
the correlation functions of any primary field whose
Virasoro representation has highest weight $h_{r,s}$ and an appropriate degeneracy.
The third order PDEs in our theorem are exactly the degeneracy equations predicted by CFT.
Conformal field theory PDEs 
of higher than second order have recently appeared in the context of $\SLE$s
\cite{Dubedat-SLE_and_Virasoro_representations_fusion,
LV-Coulomb_gas_for_commuting_SLEs}, but to our knowledge, the only earlier rigorous
scaling limit result for a lattice model establishing such equations is Watts' formula for
percolation \cite{Dubedat-excursion_decompositions_and_Watts_crossing_formula,
SW-Schramms_proof_of_Watts_formula}.
Conformal field theory PDEs of second order, on the other hand, 
are well known to arise from It\^o calculus in the growth process
description of $\SLE$s,
and they have
been established for the scaling limits of various lattice models --- often as a step
towards the proof of convergence of interfaces to
$\SLE$. The second order PDEs in our theorem are of this more familiar type.

In this article we in fact treat a refinement of the boundary visit probabilities, 
in which the contribution of each possible order of visits is isolated from the rest.
The scaling limit contribution
$\Ampl_\omega(\xin;\hat{x}_1,\ldots,\hat{x}_{N'};\xout)$
of a given visit order $\omega$
has a particular asymptotic behavior as its arguments approach each other,
as was argued by considerations of CFT fusion channels in 
\cite{BB-SLE_CFT_and_zigzag_probabilities,
JJK-SLE_boundary_visits}, and as we prove in Section~\ref{sub: asymptotics of boundary visits}.
Such asymptotic behavior is believed to specify the solution to the PDEs up to 
a multiplicative constant. Admitting this,
our formulas coincide with the 
$\SLE$ boundary zig-zag amplitude prediction of 
\cite{JJK-SLE_boundary_visits, KP-pure_partition_functions_of_multiple_SLEs}.


\subsection{Connectivity probabilities of boundary branches in the uniform spanning tree}

\begin{figure}
\includegraphics[width=.3\textwidth]{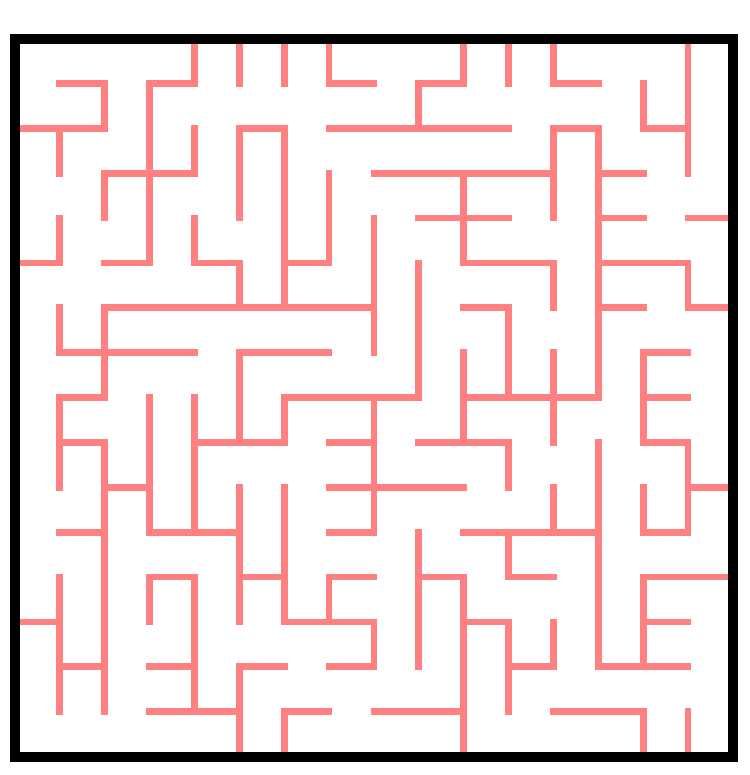} \hspace{2cm}
\includegraphics[width=.3\textwidth]{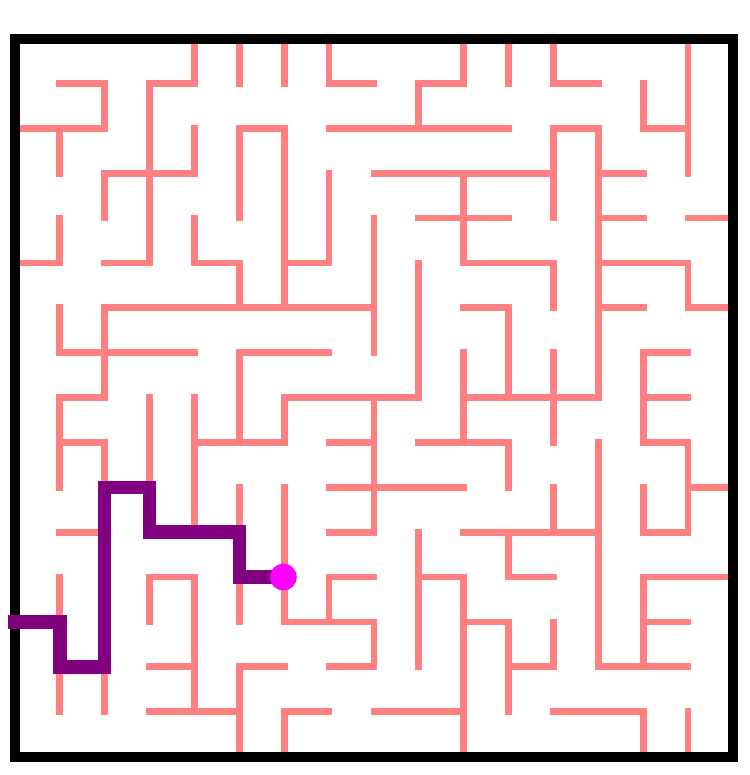}
\caption{\label{fig: UST and its boundary branch from bulk point}
A uniform spanning tree with wired boundary conditions in a $15 \times 15$ square 
grid graph, and its boundary branch from one given interior vertex.}
\end{figure}

The proof of the above theorem is based on
careful analysis of explicit formulas for probabilities of
certain connectivity events in another related model,
the uniform spanning tree on $\Gr$.
Formulas of this type 
were discovered
by Kenyon and Wilson in~\cite{KW-boundary_partitions_in_trees_and_dimers,
KW-double_dimer_pairings_and_skew_Young_diagrams}.
In this article, we also present a derivation of these formulas
based on Fomin's formulas~\cite{Fomin-LERW_and_total_positivity},
develop combinatorial tools for their detailed analysis, and
prove that the probabilities converge in the scaling limit 
to explicit functions, which satisfy the PDEs of second order predicted by CFT.

A uniform spanning tree (UST) $\tree$ on
$\Gr$ is a uniformly randomly chosen 
connected subgraph of $\Gr$ which contains all vertices $\Vert$ and has no cycles.
We impose wired boundary conditions, 
by which we mean that the boundary $\bdry \Vert$ is thought of as 
a single vertex. For any interior vertex $v \in \Vert^\circ$, there exists 
a unique path $\gamma_v$ in the tree $\tree$ from $v$ to $\bdry \Vert$,
see Figure~\ref{fig: UST and its boundary branch from bulk point}.
This path is called the boundary branch of $\tree$ from $v$, and
it is distributed like a loop-erased random
walk~\cite{Pemantle-choosing_a_spanning_tree_for_the_integer_lattice_uniformly,Wilson-generating_random_spanning_trees},
as we recall
in Section~\ref{sec: applications to UST}.
For two boundary edges 
$\ein, \eout \in \bdry \Edg$, 
we denote by $\set{\ein \rightsquigarrow \eout}$ the event 
that the boundary branch $\gamma_{\ein^\circ}$ of $\tree$ from $\ein^\circ$ 
reaches the boundary via the edge $\eout$, 
see Figure~\ref{fig: UST boundary branch}(left). 
Conditioned on this event, the boundary branch $\gamma_{\ein^\circ}$
has the distribution of a LERW from $\ein$ to $\eout$.

\begin{figure}
\centering
\begin{overpic}[width=.3\textwidth]{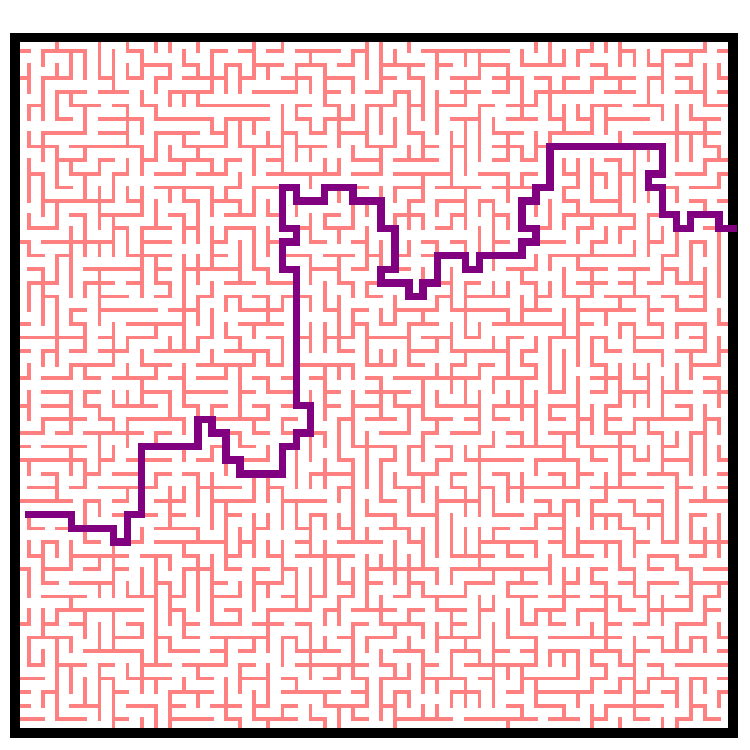}
 \put (-10,30) {\Large $\ein$}
 \put (101,67) {\Large $\eout$}
\end{overpic}
\hspace{2cm}
\centering
\begin{overpic}[width=.3\textwidth]{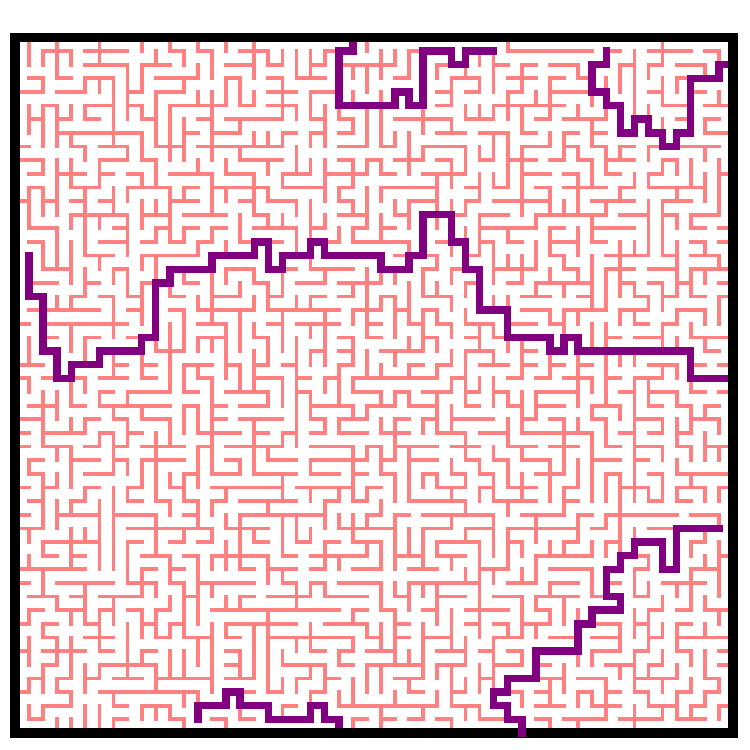}
 \put (25,-5) {\Large $e_1$}
 \put (45,-5) {\Large $e_2$}
 \put (70,-5) {\Large $e_3$}
 \put (101,27) {\Large $e_4$}
 \put (101,48) {\Large $e_5$}
 \put (101,90) {\Large $e_6$}
 \put (80,100) {\Large $e_7$}
 \put (65,100) {\Large $e_8$}
 \put (45,100) {\Large $e_9$}
 \put (-12,65) {\Large $e_{10}$}
\end{overpic}
\bigskip
\caption{\label{fig: UST boundary branch}
Uniform spanning trees with wired boundary conditions in a $50 \times 50$
square grid graph. The left figure depicts a boundary branch of the UST 
from $\ein$ to $\eout$. The right figure depicts a connectivity event 
containing several boundary branches, 
$\set{e_1 \rightsquigarrow e_2} \cap \set{e_4 \rightsquigarrow e_3} \cap
\set{e_{10} \rightsquigarrow e_5} \cap \set{e_7 \rightsquigarrow e_6}
\cap \set{e_8 \rightsquigarrow e_9}$.
}
\end{figure}
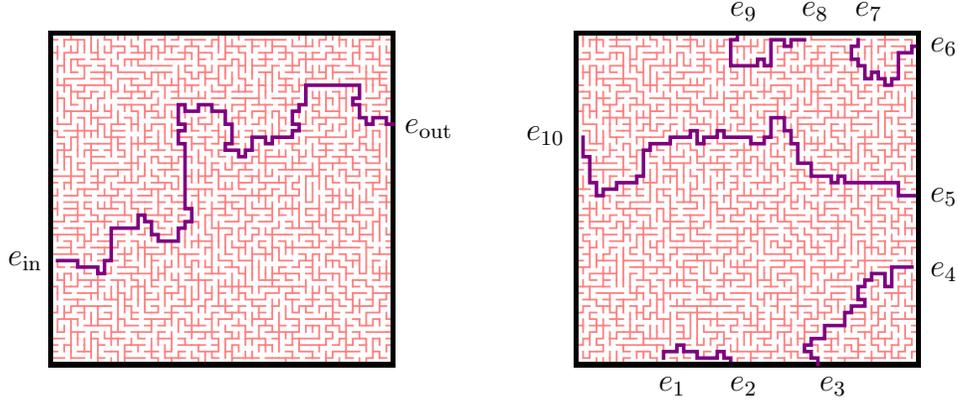

We will study probabilities of connectivity events
concerning several boundary branches, depicted in
Figure~\ref{fig: UST boundary branch}(right). More precisely, let 
$e_1 , \ldots, e_{2N} \in \bdry \Edg$ be boundary edges appearing in counterclockwise order
along the boundary.
We consider connectivity 
events $\bigcap_{\ell=1}^N \set{e_{a_\ell} \rightsquigarrow e_{b_\ell}}$,
where $a_1, \ldots, a_N, b_1, \ldots, b_N$ is some indexing of the boundary edges
by $1,2,\ldots,2N$. The possible topological connectivities of the branches are 
described by planar pair partitions 
$\alpha = \set{\set{a_1,b_1},\ldots,\set{a_N,b_N}}$ of $2N$ points, 
that tell which edges should
be connected by boundary branches in the UST, see Figure~\ref{fig: connectivity}.
The number of such possible connectivities is the Catalan number 
$\Catalan_N = \frac{1}{N+1} \binom{2N}{N}$.


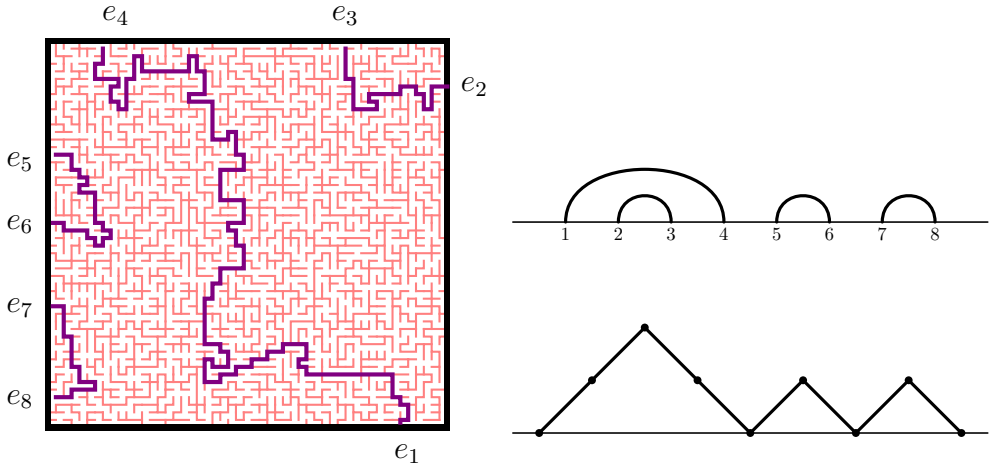
\begin{figure}
\centering
\begin{overpic}[width=.35\textwidth]{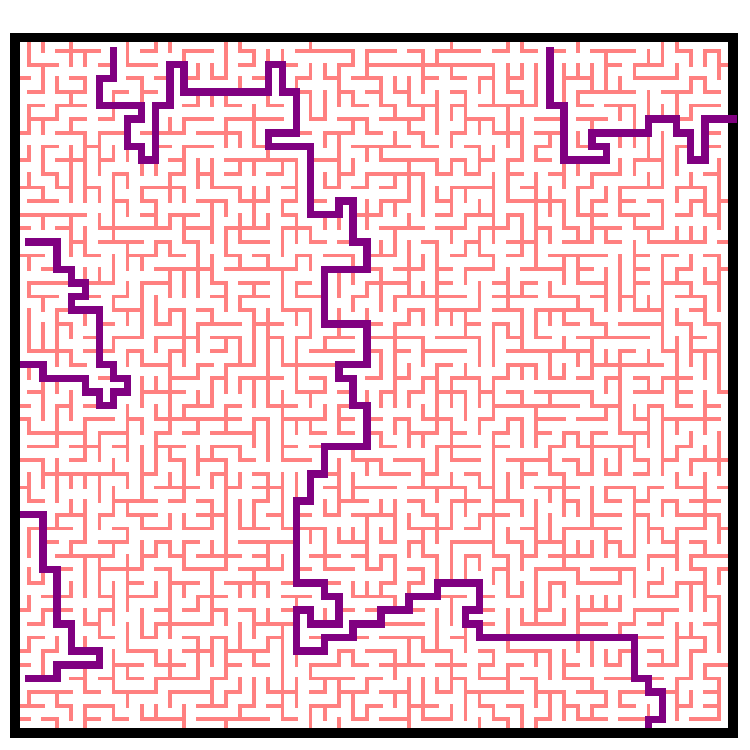}
 \put (85,-5) {\Large $e_1$}
 \put (101,83) {\Large $e_2$}
 \put (70,100) {\Large $e_3$}
 \put (15,100) {\Large $e_4$}
 \put (-8,65) {\Large $e_5$}
 \put (-8,50) {\Large $e_6$}
 \put (-8,30) {\Large $e_7$}
 \put (-8,8) {\Large $e_8$}
\end{overpic}
\hspace{.5cm}
\centering
\includegraphics[width=.4\textwidth]{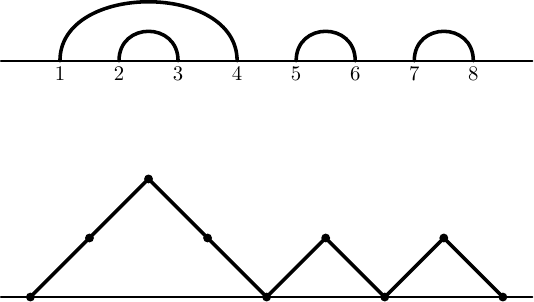}
\bigskip
\caption{\label{fig: connectivity}
The connectivity event 
$\set{e_4 \rightsquigarrow e_1} \cap \set{e_3 \rightsquigarrow e_2} \cap
\set{e_5 \rightsquigarrow e_6} \cap \set{e_8 \rightsquigarrow e_7}$
of the boundary branches in a UST, depicted in the left figure,
has the same probability that the event 
$\set{e_1 \rightsquigarrow e_4} \cap \set{e_2 \rightsquigarrow e_3} \cap
\set{e_5 \rightsquigarrow e_6} \cap \set{e_7 \rightsquigarrow e_8}$.
This connectivity event can be encoded in a planar pair partition 
$\alpha = \set{\set{1,4},\set{2,3},\set{5,6},\set{7,8}}$,
or, equivalently, a Dyck path also denoted by $\alpha$, 
depicted in the right figure.
}
\end{figure}

In the expressions for the UST connectivity probabilities, we use
determinants of matrices formed of random walk excursion kernels, 
like in Fomin's formulas for loop-erased random walks 
\cite{Fomin-LERW_and_total_positivity}.
Let $\ExcK(e_1,e_2)$ denote the excursion kernel of the symmetric random walk 
on $\Gr$ between the boundary edges $e_1$ and $e_2$, that is, the discrete 
harmonic measure of $e_2$ seen from $e_1^\circ$
(see Section~\ref{sec: applications to UST} for details).
With a suitable renormalization, the kernel $\ExcK$ converges in the scaling limit
to the Brownian excursion kernel $\ExcKdom$, in regular enough approximations 
$\Gr^\delta$ of $\domain$. 
In the upper half-plane $\bH$, 
the Brownian excursion kernel is simply a constant multiple of 
\begin{align*}
\ExcKH(x_1 , x_2) = \frac{1}{(x_2 - x_1)^2}  \qquad 
\text{for $x_1 , x_2 \in \bR$, $x_1 \neq x_2$}.
\end{align*}

For any planar pair partition $\beta$ of $2N$ points and marked boundary edges
$e_1 , \ldots, e_{2N} \in \bdry \Edg$, we define 
in Section~\ref{sec: application to UST}
a determinant $\LPdet{\beta}{\ExcK}(e_1 , \ldots, e_{2N})$ 
of an $N \times N$ matrix whose entries are the excursion kernels $\ExcK$ evaluated at pairs of marked edges determined
by $\beta$  --- see Equation~\eqref{eq: definition of LPdet in the discrete} for the precise definition.
We also define a similar determinant
$\LPdet{\beta}{\ExcKH}(x_1 , \ldots, x_{2N})$ of an $N \times N$ matrix
whose entries are instead the kernels $\ExcKH$ evaluated at pairs of points among
$x_1 < \cdots < x_{2N}$.

\begin{thm*}[Theorem~\ref{thm: the formula for UST partition functions}]
For the 
UST with wired boundary conditions, the probability of the connectivity 
described by the planar pair partition $\alpha = \set{\set{a_1,b_1},\ldots,\set{a_N,b_N}}$, 
as illustrated in Figure~\ref{fig: connectivity}, equals
\begin{align*}
\PR \Big[ \bigcap_{\ell=1}^N \set{\pathfromto{e_{a_\ell}}{e_{b_\ell}}} \Big]
= \; &    \sum_{\beta} \Minv_{\alpha,\beta} \; 
    \LPdet{\beta}{\ExcK} (e_1 , \ldots, e_{2N}) ,
\end{align*}
where the sum is over planar pair partitions $\beta$,
and, for each $\beta$, the coefficient $\Minv_{\alpha, \beta}$ is 
the number of cover-inclusive Dyck tilings of a skew shape 
between the Young diagrams associated to the two Dyck paths of $\alpha$ and $\beta$
(the precise definitions are given in Section~\ref{sec: combinatorics}).
\end{thm*}

In fact, the numbers $\Minv_{\alpha, \beta}$ are entries of the inverse 
matrix $\Minv$ of a signed incidence matrix $\Mmat$ of a binary relation 
on Dyck paths, introduced by Kenyon and Wilson~\cite{KW-boundary_partitions_in_trees_and_dimers,
KW-double_dimer_pairings_and_skew_Young_diagrams}
and Shigechi and Zinn-Justin~\cite{SZ-path_representations_of_maximal_paraboloc_KL_polynomials}.
These matrices are given explicitly in Example~\ref{ex: signed incidence matrix of KWleq}
in Section~\ref{sec: combinatorics}.
Combinatorics related to the binary relation were furthermore recently used by 
Poncelet~\cite{Poncelet} to generalize Schramm's passage probability 
formula to multiple LERWs.

\begin{thm*}[Theorems~\ref{thm: scaling limit of partition functions}~and~\ref{thm: pure partition functions at kappa equals 2}]
Let $\Gr^\delta$ with $e_1^\delta , \ldots, e_{2N}^\delta \in \bdry \Edg^\delta$
be a square grid approximation of a domain $\domain$ with marked boundary
points $p_1 , \ldots, p_{2N}$ appearing in counterclockwise order
along the boundary. Then, the UST connectivity probability has the following 
conformally covariant scaling limit:
\begin{align*}
\frac{1}{\delta^{2N}} \, 
\PR \Big[ \bigcap_{\ell=1}^N \set{\pathfromto{e^\delta_{a_\ell}}{e^\delta_{b_\ell}}} \Big]
\; \underset{\delta \to 0}{\longrightarrow} \; \frac{1}{\pi^{N}} \times \prod_{j=1}^{2N} |\confmap'(p_j)| \times
     \PartF_\alpha \big( \confmap(p_1) , \ldots, \confmap(p_{2N}) \big) ,
\end{align*}
where $\confmap \colon \domain \to \bH$ is a conformal map, and,
for $x_1 < \cdots < x_{2N}$, the function $\PartF_\alpha$ is given explicitly by
\begin{align*}
\PartF_\alpha (x_1 , \ldots , x_{2N})
= \; & \sum_{\beta} \Minv_{\alpha,\beta} \, \LPdet{\beta}{\ExcKH} (x_1 , \ldots, x_{2N}) .
\end{align*}
%
The function $\PartF_\alpha$ is positive 
and it 
satisfies the following system of $2N$ second order PDEs: 
\begin{align*}
& \Bigg[ \pdder{x_j}
    + \sum_{i \neq j} \Big( \frac{2}{x_i-x_j} \pder{x_i} - \frac{2}{(x_i-x_j)^2} \Big) \Bigg] \;
  \PartF_\alpha = 0  \qquad \text{for all } j=1,\ldots,2N .
\end{align*}
\end{thm*}

%




\subsection{Steps of the proof(s)}
The proof of Theorem~\ref{thm: the formula for UST partition functions},
concerning the connectivity probabilities in the discrete UST model,
consists of the relatively well known Wilson's algorithm and Fomin's formulas
combined with combinatorics of Dyck tilings.
Then, Theorem~\ref{thm: scaling limit of partition functions} 
about the scaling limits of the connectivity probabilities
follows from 
scaling limit results for random walk excursion kernels, and
the PDEs in Theorem~\ref{thm: pure partition functions at kappa equals 2}
are directly verified from the explicit formulas.

The proof of our main theorem about the boundary visit probabilities of LERW (Theorem~\ref{thm: scaling limit of LERW bdry visits}) 
uses the expressions from Theorems~\ref{thm: the formula for UST partition functions} 
and~\ref{thm: scaling limit of partition functions}.
Namely, we first observe that the (order-refined) boundary visit probability can be written as
\[ \PR_{\ein,\eout} \big[ \gamma \text{ uses $\hat{e}_1 , \ldots , \hat{e}_{N'}$ in this order} \big]
=  \sum_{\beta} \Minv_{\alpha, \beta} \frac{ \LPdet{\beta}{\ExcK}(e_1 , \ldots, e_{2N}) }{\ExcK(\ein,\eout)} 
, \]
for a suitably chosen connectivity $\alpha$ and boundary edges 
$e_1, \ldots, e_{2N}$, with $N = N'+1$.
The key steps then are to show cancellations in the leading terms of 
the explicit determinantal formulas in the $\delta \to 0$ limit, 
and an exchange of limits property for the subleading terms, 
detailed in Section~\ref{appendix: boundary visits}.
Our proofs of both of these properties rely heavily on combinatorics of
cover-inclusive Dyck tilings.
The cancellations and exchange of limits eventually allow us also to
establish the asserted third order partial 
differential equations by a fusion argument
similar to the recent work of 
Dub\'edat~\cite{Dubedat-SLE_and_Virasoro_representations_fusion}.


%

\subsection{Multiple $\SLE_2$}

Collections of several branches of the UST 
are expected to converge to multiple $\SLEk$ curves at $\kappa = 2$
\cite{BBK-multiple_SLEs, Dubedat-Euler_integrals, Dubedat-commutation, 
LK-configurational_measure, KW-boundary_partitions_in_trees_and_dimers,
KP-pure_partition_functions_of_multiple_SLEs}, as stated in more detail
in Conjecture~\ref{conj: multiple branches of UST converge to multiple SLE2}.%
\footnote{See also the update below the conjecture for recent progress.}
In general, a multiple $\SLEk$, for $\kappa > 0$, is a process of random 
conformally invariant curves in a simply connected planar domain, 
connecting marked boundary points $p_1,\ldots,p_{2N}$ 
pairwise without crossing. The possible topological connectivities of the curves can be 
described by planar pair partitions $\alpha$ similarly as in the discrete case
(see Figure~\ref{fig: connectivity}).
To construct a (local) multiple $\SLEk$, one uses as an input a positive function
$\PartF(x_1, \ldots, x_{2N})$, defined for $x_1 < \cdots < x_{2N}$, called 
a multiple $\SLE$ partition function. This function is subject to the
requirements of positivity, M\"obius covariance~\eqref{eq: COV for multiple SLEs},
and $2N$ 
partial differential equations
of second order~\eqref{eq: PDE for multiple SLEs},
recalled explicitly in Section~\ref{sec: applications to SLEs}.

Consider a chosen connectivity pattern $\alpha$. It has been argued in
\cite{BBK-multiple_SLEs, KP-pure_partition_functions_of_multiple_SLEs}
that the multiple $\SLEk$ in which the random curves form this deterministic 
connectivity $\alpha$ has a particular partition function 
$\PartF_\alpha^{(\kappa)}$ determined by certain asymptotic boundary conditions
\eqref{eq: ASY for multiple SLEs}, given in
Section~\ref{sec: applications to SLEs}.
In~\cite{KP-pure_partition_functions_of_multiple_SLEs}, solutions 
$\PartF_\alpha^{(\kappa)}$ 
satisfying the required covariance, PDEs and asymptotics
are constructed for the generic parameter values 
$\kappa \in (0,8) \setminus \bQ$. 
These functions are called the multiple $\SLEk$ pure partition functions.

In the present article, we show that the pure partition functions of the 
local multiple $\SLE_2$
can be obtained as the scaling limits of the probabilities
of the connectivity events 
$\bigcap_{\ell=1}^N \set{\pathfromto{e_{a_\ell}}{e_{b_\ell}}}$ 
of the UST boundary branches.\footnote{This is the first proof of the existence 
of \textit{positive} pure partition functions of multiple $\SLE$s,
implying in particular the existence of the extremal (local)
multiple $\SLE$s. 
After the present work, the existence of positive pure partition functions 
for~$\kappa \in (0,6]$
has been proven by $\SLE$ 
methods~\cite{PH-Global_multiple_SLEs_and_pure_partition_functions, 
Wu-HyperSLE}, which however do not yield explicit formulas.
}
Importantly, we can conclude the existence of local multiple $\SLEk$ at $\kappa = 2$,
with the pure partition functions~$\PartF_\alpha$.

\begin{thm*}[Theorems~\ref{thm: pure partition functions at kappa equals 2} and~\ref{thm: existence of local multiple SLE2s}]
For any planar pair partition $\alpha$, the function 
$\PartF_\alpha = \sum_{\beta} \Minv_{\alpha,\beta} \, \LPdet{\beta}{\ExcKH}$
is a positive, M\"obius covariant solution to the $2N$ second order PDEs
required from the multiple $\SLE_2$ partition functions.
In particular, there exists a local multiple $\SLE_2$ 
with partition function $\PartF_\alpha$.
\end{thm*}
In fact, 
for each $\alpha \in \LP_N$, 
the function $\PartF_\alpha$ is the unique solution which satisfies the asymptotic boundary 
conditions~\eqref{eq: ASY for multiple SLEs},
when solutions with at most power-law growth are considered, 
see Section~\ref{sec: applications to SLEs}.

\subsection{Beyond the present work: conformal blocks and $q$-deformations}

Many of the delicate properties
of the scaling limits of connectivity
probabilities of UST branches and boundary visit probabilities of
LERWs 
rely heavily on remarkable combinatorial structures.
Some key aspects of the combinatorics are captured by 
a certain
binary relation on planar pair partitions, 
and Fomin's formulas lie at the very heart of the applications to the
UST 
and LERW. 
Fomin's formulas certainly seem specific to the loop-erased random walks,
and one could thus be lead to suspect that the remarkable combinatorics in our scaling
limit results might also be specific to $\SLE$s at $\kappa=2$ and 
conformal field theories at the corresponding central charge $c=-2$. 
Somewhat surprisingly, however, it turns out that
versions of the combinatorial structures in fact persist at generic $\kappa$ and $c$.
The generic parameter analogues of the results relate
the multiple $\SLE$ pure partition functions $\PartF_\alpha^{(\kappa)}$
not to determinant functions $\LPdet{\alpha}{\ExcKH}$
as in Fomin's formulas, but rather to the conformal block
functions 
of conformal field theories. 
Some combinatorial enumerations will have to be replaced by appropriate
$q$-analogues, where the deformation parameter $q$ depends on $\kappa$.
This generalization is the topic of a companion paper~\cite{KKP-companion}.

\subsection{Organization of the article}
The roles of the remaining four sections can be roughly summarized as follows.
The main probabilistic content is in
Sections~\ref{sec: application to UST}--\ref{appendix: boundary visits}.
Section~\ref{sec: combinatorics} has a crucial but auxiliary role for the main content:
the other sections repeatedly rely on the combinatorial results there.
A~recommended approach is to read Sections~\ref{sec: application to UST}--\ref{appendix: boundary visits},
while consulting Section~\ref{sec: combinatorics} as it is needed.


More precisely, the combinatorics of Section~\ref{sec: combinatorics}
will be employed in the rest of the article as follows. The results from
Sections~\ref{subsec: Combinatorial objects and bijections}--\ref{subsec: Dyck tilings}
will be needed in Section~\ref{sec: partition functions}
and the follow-up work~\cite{KKP-companion}. 
The results in Sections~\ref{subsec: Wedges, slopes, and link removals}--\ref{subsec: cascades}
are mainly for the purpose of the follow-up \cite{KKP-companion}.
The results in Sections~\ref{subsec: Wedges, slopes, and link removals}
and~\ref{subsec: inverse Fomin sums} are used
in Sections~\ref{subsub: partition function asymptotics},
\ref{sec: determinant Taylor expansions}, and~\ref{sub: asymptotics of boundary visits}.

Section~\ref{sec: application to UST} constitutes the derivation 
of the explicit formulas for the connectivity and boundary visit probabilities,
as well as the convergence proof for 
the connectivity probabilities in the scaling limit.
Sections~\ref{sec: graphs in planar domain}--\ref{sec: UST definitions}
introduce the uniform spanning tree with wired boundary
conditions and its boundary branches, and contain elementary observations
about their connectivity and boundary visit events.
Sections~\ref{sec: LERW and UST} and~\ref{sec: Fomin} review two well-known
but essential tools for the uniform spanning tree and loop-erased random walk:
Wilson's algorithm and Fomin's formula.
In Section~\ref{sec: partition functions}, our solution to the 
connectivity and boundary visit probabilities is presented for the discrete models.
Scaling limits are addressed in Section~\ref{sec: scaling limits}.
Some generalizations of our main results, in particular pertaining to the
uniform spanning tree with free boundary conditions, are discussed in
Section~\ref{sec: generalizations}.

The topic of Section~\ref{sec: applications to SLEs} is the relation of the
scaling limit connectivity probabilities to multiple $\SLE$ processes.
Multiple~$\SLEk$ is introduced in Section~\ref{sec: multiple SLEs}.
In Section~\ref{sec: proof of pure partition function properties}, we prove the second order partial
differential equations, positivity, and asymptotics for the functions $\PartF_\alpha$, and
as a consequence get that each of these functions can be used to construct a local multiple~$\SLEk$
at $\kappa = 2$.

Section~\ref{appendix: boundary visits} contains the proof 
of the scaling limit results for the boundary visit probabilities.
In Section~\ref{sec: determinant Taylor expansions}, we prove
the existence of the scaling limit and obtain a formula for it
which allows us to interchange two limits.
Section~\ref{sec: third order PDEs} contains a fusion argument 
by which the third order partial differential equations can be proven
after the two limits have been exchanged.
Section~\ref{sub: proof of bdry visit thm} then summarizes the proof of
the main scaling limit result for the boundary visit probabilities.
We finish in Section~\ref{sub: asymptotics of boundary visits} by proving
asymptotics properties for the limit of boundary visit probabilities.

\bigskip
\textbf{Acknowledgments:}
We thank Christian Hagendorf for useful discussions,
and in particular for drawing our attention to the results of
\cite{KW-boundary_partitions_in_trees_and_dimers,
KW-double_dimer_pairings_and_skew_Young_diagrams,
KW-spanning_trees_of_graphs_on_surfaces_and_the_intensity_of_LERW}.
We also thank Dmitry Chelkak, Steven Flores, Christophe Garban, Konstantin Izyurov, Richard Kenyon,
Marcin Lis, Wei Qian, David Radnell, Fredrik Viklund, David Wilson, and Hao Wu
for interesting and helpful discussions.

A.K. and K.K. are supported by the Academy of Finland project
``Algebraic structures and random geometry of stochastic lattice models''.
During this work, E.P. was supported by Vilho, Yrj\"o and Kalle V\"ais\"al\"a Foundation
and later by the ERC AG COMPASP, the NCCR SwissMAP, and the Swiss NSF.


\bigskip{}

%
%
%

\section{\label{sec: combinatorics}Combinatorics of link patterns and the parenthesis reversal relation}

In this section, we introduce combinatorial objects and present needed results about them. 

The first sections \ref{subsec: Combinatorial objects and bijections} 
and \ref{subsec: Partial order and Kenyon-Wilson relation} 
introduce basic definitions and tools.
The first fundamental combinatorial result is 
Theorem~\ref{thm: weighted KW incidence matrix inversion}
in Section~\ref{subsec: Dyck tilings}, which gives a formula for the
inverse of a weighted incidence matrix of a certain binary relation. 
This formula will be instrumental in Section~\ref{sec: applications to UST}
when solving for 
the probabilities of connectivity events
for multiple branches of the uniform spanning tree, as well as the
probabilities of boundary visit events for the loop-erased random walk.
It will also be used in the companion paper~\cite{KKP-companion} 
to obtain a change of basis matrix from the pure partition functions of multiple $\SLE$s
to the basis of conformal blocks of CFT. For the latter purpose, we present the
theorem in a generality which allows for weighted incidence matrices.


In the rest of this section, we give some further combinatorial tools that are needed 
in analyzing the asymptotics and scaling limits of the probabilistic observables of interest.
These pertain to what we call cascade properties, which describe the behavior of
multiple $\SLE$s when reducing the number of curves.
Section~\ref{subsec: Wedges, slopes, and link removals} introduces the basic
notations and definitions.
In Section~\ref{subsec: cascades}, we show that cascade properties uniquely
determine certain weighted incidence matrices, such as those encountered
in~\cite{KKP-companion}.
In Section~\ref{subsec: inverse Fomin sums}, 
we develop concrete combinatorial tools which will be relied on in
all further analysis of the UST connectivity and LERW boundary visit pattern probabilities.

\subsection{\label{subsec: Combinatorial objects and bijections}Combinatorial objects and bijections}

The \textit{Catalan numbers} $\Catalan_N = \frac{1}{N + 1} \binom{2N}{N}$ are 
typically introduced as the number of different planar pair partitions of $2N$ points.
Planar pair partitions naturally appear as different connectivity 
patterns or configurations of interfaces in various planar random models,
and in this context they are usually called link patterns. 
It is convenient to identify planar pair partitions
with yet other families of combinatorial objects, since each interpretation makes some of our
combinatorial considerations easier or more transparent. We will
work interchangeably with the following three families,
illustrated in Figure~\ref{fig: bijections}.

\textit{Link patterns} (planar pair partitions) with $N$ links
are partitions
\begin{align}\label{eq: link pattern}
\alpha = \{ \link{a_1}{b_1},\ldots, \link{a_N}{b_N} \} 
\end{align}
of the set $\set{1, \ldots, 2N}$ into $N$ pairs, called \textit{links},
such that the points $a_\ell$ and $b_\ell$, for $\ell=1,\ldots,N$, can be connected by non-intersecting
paths in the upper half-plane, see Figure~\ref{fig: connectivity}. 
The family of all link patterns with
$N$ links is denoted by $\LP_N$. 
By convention, we set $\LP_0 := \set{\emptyset}$.

\textit{Dyck paths} are non-negative walks $w$ of $2N$ up- or down-steps 
starting and ending at zero,
\begin{align*}
\DP_N := \set{ w: \{ 0,...,2N \} \rightarrow \Z_{\ge 0} \; \big| \; w(0) = w(2N) = 0, 
\text{ and } | w(j) - w(j-1) | = 1 \text{ for all } j } .
\end{align*}

\textit{Balanced parenthesis expressions}
are sequences of $2N$ parentheses ``{\opar}'' and ``{\cpar}'', balanced in 
the conventional sense of parentheses. More precisely, given a sequence 
$\alpha$ of $2N$ parentheses, denote by $o_j(\alpha)$ the number of 
opening parentheses {\opar} among the $j$ first parentheses from the left, 
and by $c_j(\alpha)$ the number of closing parentheses {\cpar}. 
Then, $\alpha$ is a balanced parenthesis expression if and only if 
$c_j (\alpha) \le o_j (\alpha)$ for all $j$ and 
$c_{2N} (\alpha) = o_{2N} (\alpha) = N$.
This is equivalent to $j \mapsto o_j(\alpha) - c_j (\alpha)$ being a Dyck path.
The family of all balanced parenthesis expressions with
$N$ pairs of parentheses is denoted by $\BPE_N$,
and it is in bijective correspondence with $\DP_N$.
By a slight abuse of notation, we thus identify a balanced parenthesis
expression $\alpha \in \BPE_N$ with the Dyck path also denoted by $\alpha \in \DP_N$,
\begin{align*}
\alpha(j) = o_j(\alpha) - c_j (\alpha) .
\end{align*}

In a balanced parenthesis expression $\alpha$, an opening parenthesis~{\opar}
at position $j$ and a closing parenthesis~{\cpar} at position $i$, with $j < i$,
are said to be a \textit{matching pair} if the subexpression $\BPEfont{Y}$ consisting of the
parentheses at $j+1 , \ldots, i-1$ is also a balanced parenthesis expression,
so that $\alpha = \BPEfont{X(Y)Z}$, where 
$\BPEfont{X}$ and $\BPEfont{Z}$ are (not necessarily balanced) sequences of parentheses.
In terms of the Dyck path $\alpha$, 
the $j$:th and $i$:th steps of $\alpha$ are the
opposite slopes at equal height of a single mountain silhouette, as illustrated in
Figure~\ref{fig: bijections}. There are $N$ 
matching pairs in a balanced parenthesis expression $\alpha \in \BPE_N$, and these
determine a link pattern. Thus the sets $\BPE_N$ and $\LP_N$ are in bijection,
and by a slight abuse of notation we again interpret $\alpha$ interchangeably
as either a balanced parenthesis expression or a link pattern.

\begin{figure}
\includegraphics[width = 0.2\textwidth]{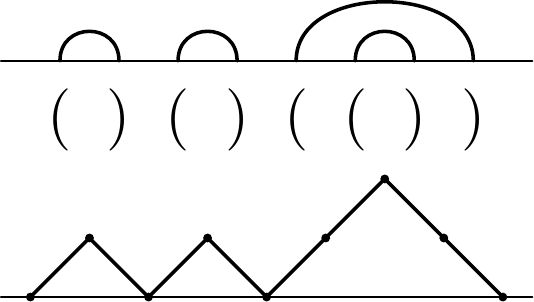} \qquad \qquad \qquad
\includegraphics[width = 0.2\textwidth]{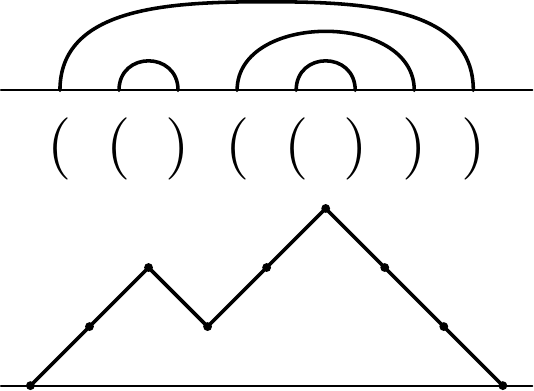} \qquad \qquad \qquad
\includegraphics[width = 0.2\textwidth]{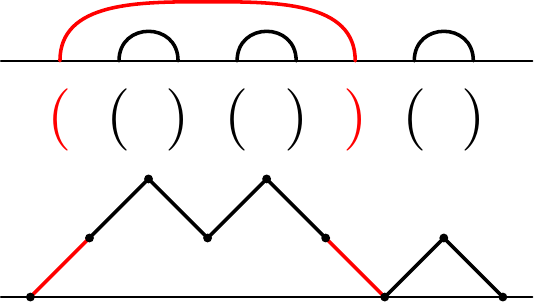}
\caption{\label{fig: bijections}
Illustration of the bijections between $\LP_4$, $\BPE_4$, and $\DP_4$,
and the correspondence of links, matching pairs of parentheses, and opposite slopes.}
\end{figure}

\begin{figure}
\includegraphics[width = 0.2\textwidth]{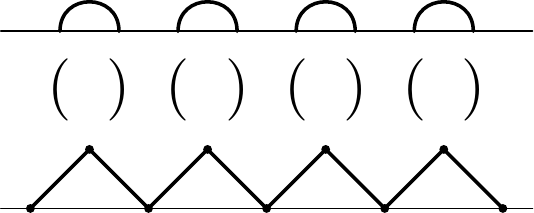} \qquad \qquad \qquad \qquad
\includegraphics[width = 0.2\textwidth]{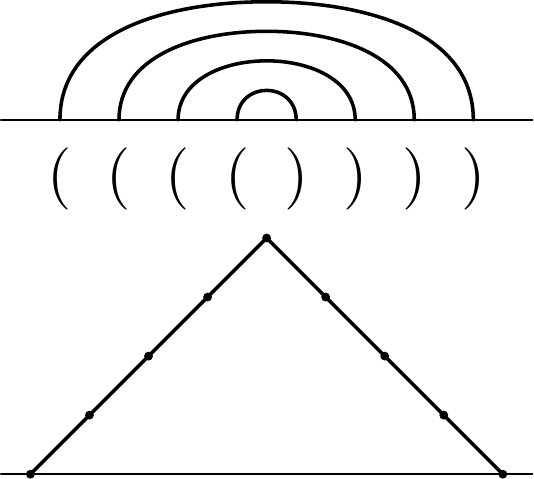}
\caption{\label{fig: minimal and maximal} 
The minimal (left) and maximal (right) elements {$\unnested_N$} and {$\nested_N$}. }
\end{figure}

\begin{figure}
\begin{displaymath}
\xymatrixcolsep{3.5pc}
\xymatrixrowsep{3.5pc}
\xymatrix{
        & \begin{minipage}{4cm} \begin{center} \includegraphics[width=0.6\textwidth]{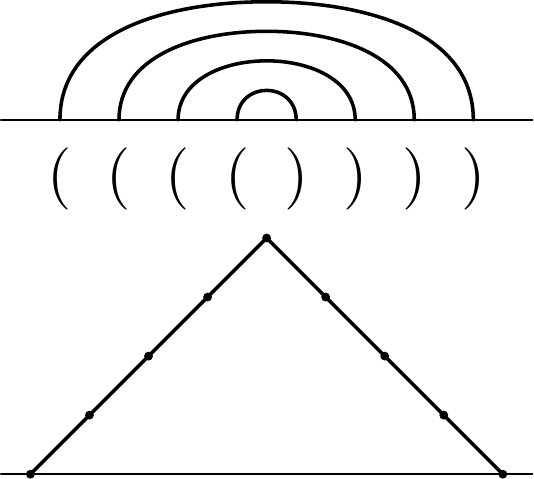} \end{center} \end{minipage}\ar[d] \\
        & \hspace{-0mm}\begin{minipage}{4cm} \begin{center}  \includegraphics[width=0.6\textwidth]{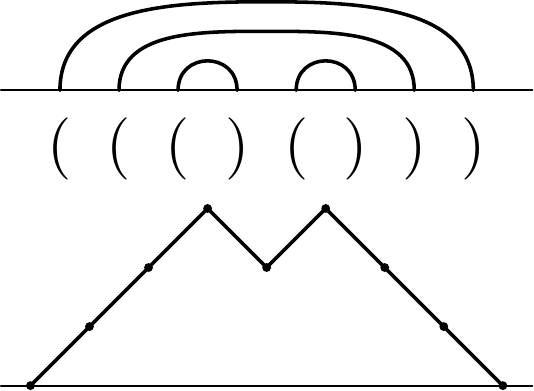} \end{center} \end{minipage} \ar[dl]\ar[dr] & \\
        \hspace{-0mm}\begin{minipage}{4cm} \begin{center}  \includegraphics[width=0.6\textwidth]{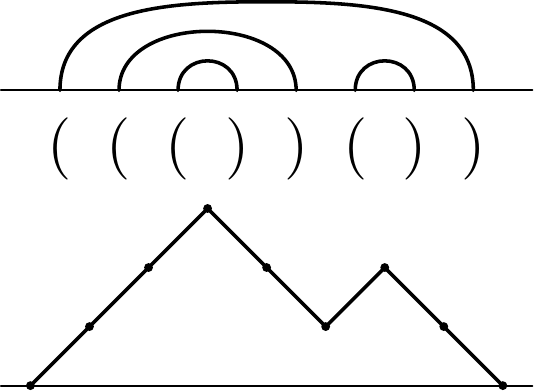} \end{center} \end{minipage}\ar[d]\ar[dr] & & 
        \hspace{-0mm}\begin{minipage}{4cm} \begin{center}  \includegraphics[width=0.6\textwidth]{bijections2.pdf} \end{center} \end{minipage}\ar[dl]\ar[d]\\
        \hspace{-0mm}\begin{minipage}{4cm} \begin{center}  \includegraphics[width=0.6\textwidth]{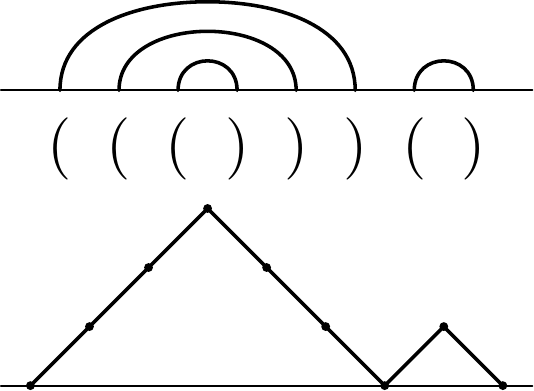} \end{center} \end{minipage}\ar[d] & 
        \hspace{-0mm}\begin{minipage}{4cm} \begin{center}  \includegraphics[width=0.6\textwidth]{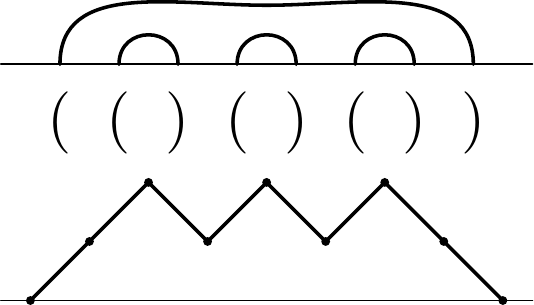} \end{center} \end{minipage}\ar[dl]\ar[d]\ar[dr] & 
        \hspace{-0mm}\begin{minipage}{4cm} \begin{center}  \includegraphics[width=0.6\textwidth]{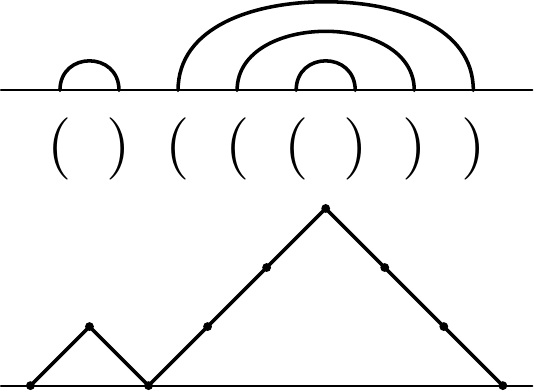} \end{center} \end{minipage}\ar[d] \\
        \hspace{-0mm}\begin{minipage}{4cm} \begin{center}  \includegraphics[width=0.6\textwidth]{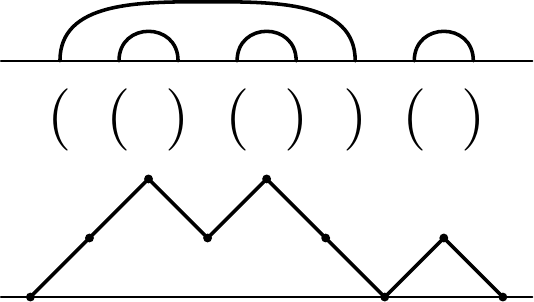}\end{center} \end{minipage}\ar[d]\ar[dr] & 
        \hspace{-0mm}\begin{minipage}{4cm} \begin{center}  \includegraphics[width=0.6\textwidth]{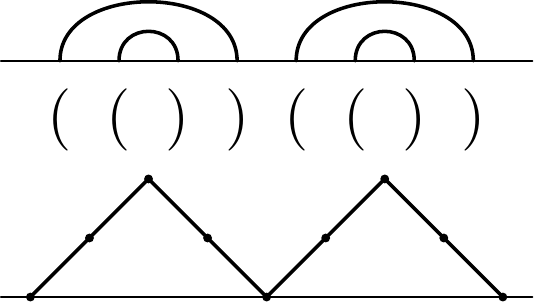} \end{center} \end{minipage}\ar[dl]\ar[dr] & 
        \hspace{-0mm}\begin{minipage}{4cm} \begin{center} \includegraphics[width=0.6\textwidth]{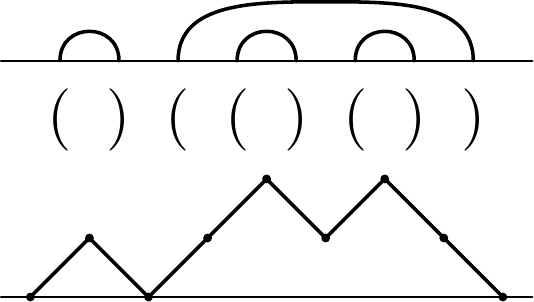} \end{center} \end{minipage}\ar[dl]\ar[d] \\
        \hspace{-0mm}\begin{minipage}{4cm} \begin{center}  \includegraphics[width=0.6\textwidth]{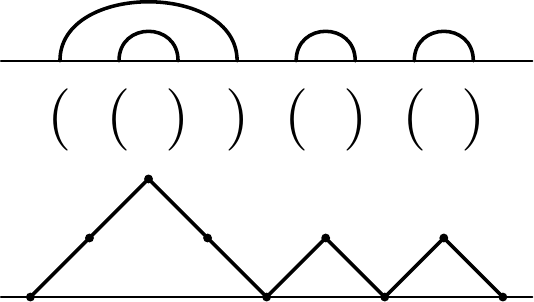} \end{center} \end{minipage}\ar[dr] & 
        \hspace{-0mm}\begin{minipage}{4cm} \begin{center}  \includegraphics[width=0.6\textwidth]{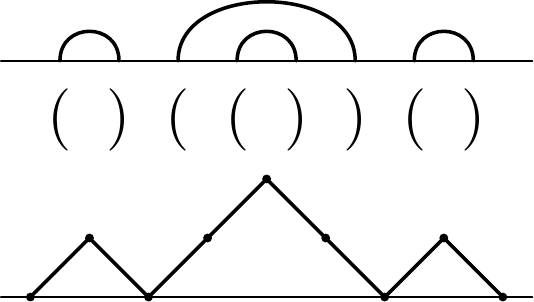} \end{center} \end{minipage}\ar[d] & 
        \hspace{-0mm}\begin{minipage}{4cm}\begin{center}  \includegraphics[width=0.6\textwidth]{bijections1.pdf} \end{center} \end{minipage}\ar[dl] \\
        & \hspace{-0mm}\begin{minipage}{4cm} \begin{center} \includegraphics[width=0.6\textwidth]{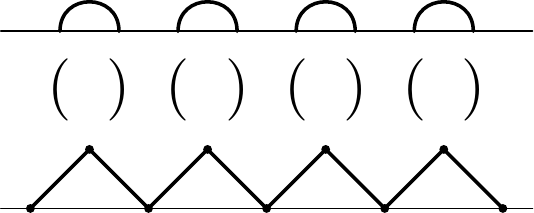} \end{center} \end{minipage} &
    }
\end{displaymath}
\caption{\label{fig: poset of link patterns}%
The partially ordered set $\LP_4$ of link patterns with four links, and
the corresponding Dyck paths and balanced parenthesis expressions.}
\end{figure}

Via the bijections above, we identify the three sets $\LP_N$, $\DP_N$, and $\BPE_N$,
and an element of them will be denoted by the same symbol.
With the identifications, the $j$:th index of $\alpha \in \LP_N$
is a left (resp. right) endpoint of a link, if and only if
the $j$:th parenthesis of $\alpha \in \BPE_N$ is 
an opening (resp. closing) parenthesis, if and only if
the $j$:th step of $\alpha \in \DP_N$ is an up-step (resp. down-step).
The Dyck path $\alpha(j) = o_j(\alpha) - c_j(\alpha)$ tells how many pairs of parentheses (or links) 
remain open after reading the first $j$ parentheses (or link endpoints) from the left 
in $\alpha$. This measures how nested the link pattern or parenthesis expression 
$\alpha$ is at $j$.

\subsection{\label{subsec: Partial order and Kenyon-Wilson relation}Partial order and the parenthesis reversal relation}

\begin{defn}
A partial order $\DPleq$ on the set $\DP_N$ of $2N$-step Dyck paths is defined 
by setting $\alpha \DPleq \beta$ if and only if $\alpha(j) \le \beta(j)$ 
for all $0 \le j \le 2N$.
\end{defn}
This also naturally defines a partial order on the sets of link patterns $\LP_N$ 
and parenthesis expressions $\BPE_N$, where we have $\alpha \preceq \beta$ 
if and only if $\beta$ is more nested than $\alpha$ at every position $j$.
This partial order has unique minimal and maximal elements, denoted by 
$\unnested_N$ and $\nested_N$, respectively, and illustrated in
Figure~\ref{fig: minimal and maximal}.
Figure~\ref{fig: poset of link patterns} illustrates the partial order $\DPleq$.


The following relation was introduced 
in~\cite{KW-double_dimer_pairings_and_skew_Young_diagrams}
and~\cite{SZ-path_representations_of_maximal_paraboloc_KL_polynomials}.

\begin{defn}
The \textit{parenthesis reversal relation} $\KWleq$ on the set $\BPE_N$
is defined by setting $\alpha \KWleq \beta$ if and only if 
$\alpha$ can be obtained from $\beta$ by choosing a subset $B$ of parentheses 
in $\beta$ such that the matching pair of each parenthesis in the set 
$B$ also belongs to $B$, and then reversing all the parentheses in $B$.
\end{defn}
Note that the relation $\alpha \KWleq \beta$ implies $\alpha \DPleq \beta$, 
since the reversal of a matching pair shifts the opening parenthesis to the right.
In fact, the relation $\DPleq$ is the transitive closure of the non-transitive
relation~$\KWleq$, see~\cite{KW-double_dimer_pairings_and_skew_Young_diagrams}.
We also use the binary relation $\KWleq$ on the sets
$\LP_N$ and $\DP_N$, and in Lemmas~\ref{lem: link patterns and KW relation}
and~\ref{lem: nested tilings and KW relation} 
below, we 
characterize the relation in terms of link patterns and Dyck paths, respectively.

It is often preferable to reverse matching pairs of parentheses
one at a time, and with the following convention for the order of reversals.
Suppose that $\alpha \KWleq \beta$, and let $B$ be the collection of the reversed
matching pairs of $\beta$. By the \textit{nested chain of reversals}
from $\beta$ to $\alpha$ we mean the sequence of intermediate steps
$\beta = \beta_0, \beta_1, \ldots, \beta_m = \alpha$, where $\beta_n$
is obtained from $\beta_{n-1}$ by reversing the matching pair in $B$
whose opening parenthesis is the $n$:th from the left.
The term nested refers to the following property of the sequence 
$\beta_0, \beta_1, \ldots, \beta_m$:
if any two matching pairs to be reversed are nested, one inside the other,
then the reversal of the outer is performed before the inner.
\begin{exa}
\label{ex: bracket reversals}
We have the following parenthesis reversal relation
\[ \BPEfont{(())(())()()()} \; \KWleq \; \BPEfont{(([[]]))([()])} , \]
where we have emphasized the subset $B$ of reversed matching pairs by square brackets.
Starting from the latter balanced parenthesis expression $\BPEfont{(((())))((()))}$
and reversing one pair at a time with the above convention,
each reversal turns out to yield a parenthesis reversal relation as follows
\[ \BPEfont{(())(())()()()} \; \KWleq \; \BPEfont{(())(())([()])}
    \; \KWleq \; \BPEfont{(()[]())((()))} \; \KWleq \; \BPEfont{(([()]))((()))} . \]
\end{exa}
The following lemma asserts that the example above featured a general phenomenon.

\begin{lem}
\label{lem: bracket reversals one by one}
Let $\alpha \KWleq \beta$, and let $\beta = \beta_0, \beta_1, \ldots, \beta_m = \alpha$ be 
the nested chain of reversals from $\beta$ to $\alpha$.
Then we have $\beta_n \in \BPE_N$ for all $n = 0,\ldots,m$,
and the following parenthesis reversal relations hold:
\[ \alpha = \beta_m \KWleq \beta_{m-1} \KWleq \; \cdots \; \KWleq \beta_1 \KWleq \beta_0 = \beta . \]
Moreover, for all $n$, we have the relation $\alpha \KWleq \beta_n$ with the
nested chain $\beta_n, \beta_{n+1}, \ldots, \beta_m = \alpha$.
\end{lem}
\begin{proof}
Since $\alpha$ is a balanced parenthesis expression, we have 
$o_j(\alpha) \ge c_j(\alpha)$ for all $j$.
Because a reversal of a matching pair always shifts the opening parenthesis to the right, 
we see that
\begin{align*}
o_j(\beta_n) \ge o_j(\alpha) \ge c_j(\alpha) \ge c_j(\beta_n)
\qquad \text{for all $j$ and $n$, and} \qquad 
c_{2N} (\beta_n) = o_{2N} (\beta_n) = N 
\qquad \text{for all $n$}.
\end{align*}
This shows that the intermediate steps $\beta_n$ are balanced parenthesis expressions.
By the chosen order of reversals in a nested chain, each matching pair of $\beta$ to be reversed
remains matching in the intermediate steps $\beta_n$ until that pair is reversed.
This implies the relations $\beta_{n+1} \KWleq \beta_n$ and $\alpha \KWleq \beta_n$.
The reversals in the subchain $\beta_n, \beta_{n+1}, \ldots, \beta_m = \alpha$ are still ordered
by their opening parentheses from the left.
\end{proof}


\begin{figure}
\includegraphics[width = 0.24\textwidth]{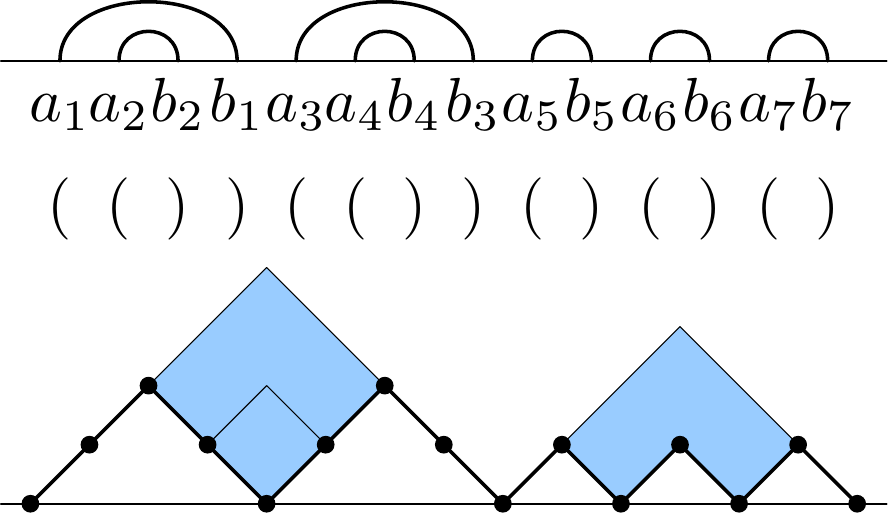}
\includegraphics[width = 0.24\textwidth]{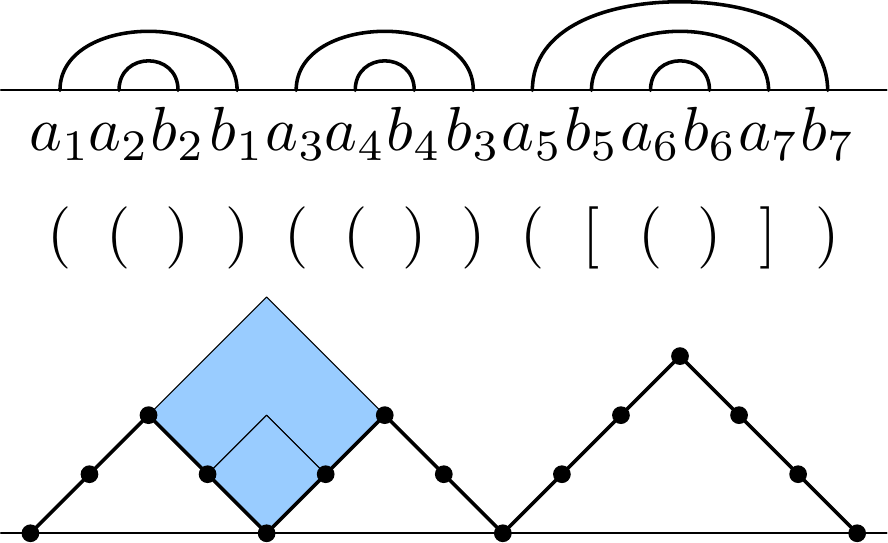}
\includegraphics[width = 0.24\textwidth]{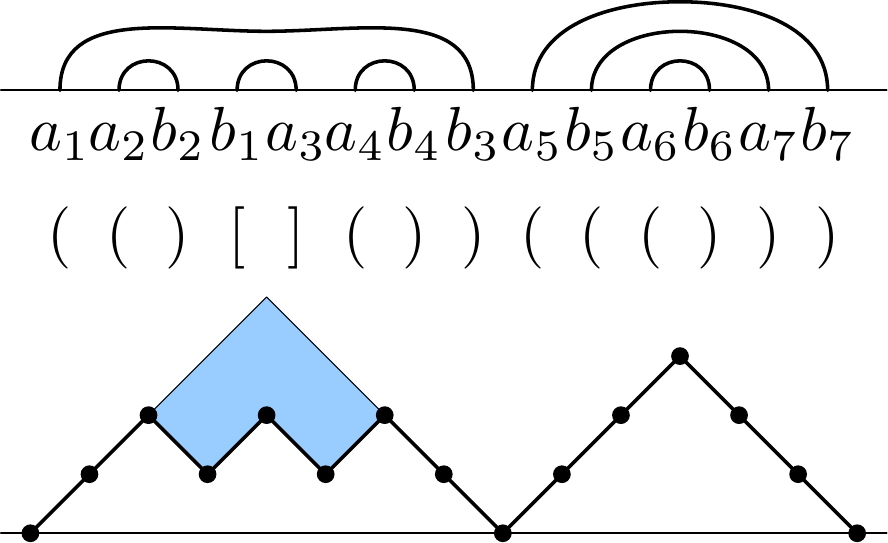}
\includegraphics[width = 0.24\textwidth]{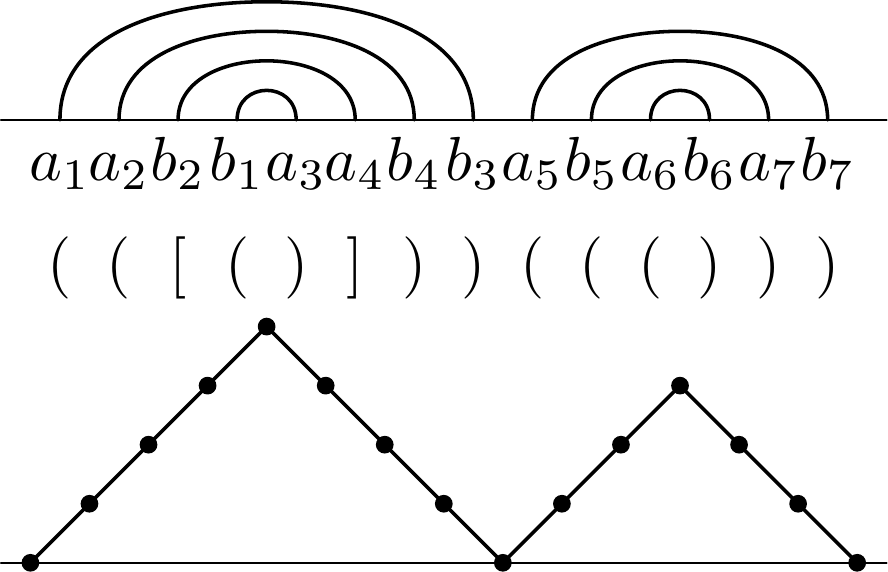}
\caption{\label{fig: KW} 
The nested chain of Example \ref{ex: bracket reversals}, with the interpretations 
in terms of link patterns and Dyck paths, given by Lemmas \ref{lem: link patterns 
and KW relation} and \ref{lem: nested tilings and KW relation}, respectively.
}
\end{figure}


We next characterize the  parenthesis reversal relation in terms of (oriented) link
patterns, depicted in Figure~\ref{fig: KW}. This characterization will be crucial
in Section~\ref{sec: applications to UST} for
recovering the uniform spanning tree connectivity probabilities from Fomin's formulas.

We frequently need to refine the link patterns with choices of orientation. 
Recall from~\eqref{eq: link pattern} that a link pattern $\alpha$ is an unordered
collection of unordered pairs $\alpha = \set{\link{a_1}{b_1} , \ldots , \link{a_N}{b_N}}$. 
An ordered collection of ordered pairs 
$\big( (a_\ell , b_\ell) \big)_{\ell=1}^N$ is called an \textit{orientation} of $\alpha$.
The points $a_\ell$ are then called \textit{entrances} and $b_\ell$ \textit{exits}. 
As a standard reference orientation of a link pattern $\alpha$, we will use
the \textit{left-to-right orientation},
defined
by the conditions $a_\ell < b_\ell$ for all $\ell$, and $a_1 < \ldots < a_N$.

\begin{lem}\label{lem: link patterns and KW relation}
Let $((a_\ell, b_\ell))_{\ell = 1}^N$ be the left-to-right orientation of a link 
pattern $\alpha \in \LP_N$. The following statements are equivalent.
\begin{description}
\item[(a)] A link pattern $ \beta\in \LP_N$ connects every entrance of
$((a_\ell, b_\ell))_{\ell = 1}^N$ 
to an exit, that is, there exists a permutation 
$\sigma \in \SymmGrp_N$ such that 
$\beta = \set{ \link{a_1}{b_{\sigma(1)}}, \ldots, \link{a_N}{b_{\sigma(N)}} }$.
\item[(b)] We have $\alpha \KWleq \beta$.
\end{description}
Moreover, we then have $\sgn(\sigma) = (-1)^m$, where $m$ is the number 
of matching pairs of parentheses reversed in the nested chain from $\beta$ 
to $\alpha$.
\end{lem}

\begin{proof}
To prove that (a) implies (b), let $\beta \in \LP_N$ be a link pattern connecting 
entrances of the left-to-right oriented $\alpha$ to exits.
Let $B$ consist of those matching pairs 
of parentheses in $\beta$ whose opening parenthesis {\opar} corresponds to a left link endpoint
labeled as an exit $b_{\sigma(\ell)}$. Since $\beta$ connects entrances to exits, all 
closing parentheses {\cpar} in $B$ correspond to entrances $a_\ell$.
Reversing the parentheses in $B$, we obtain the balanced parenthesis expression of 
$\alpha$, so $\alpha \KWleq \beta$. This shows that (a) implies (b).


To prove that (b) implies (a), we show that the links of $\beta_n$ connect entrances $(a_\ell)_{\ell=1}^N$ of 
$\alpha$ to exits $(b_{\ell})_{\ell = 1}^N$ of $\alpha$ in
any intermediate step $\beta_n$ of the nested chain $\beta = \beta_0, \beta_1, \ldots, \beta_m = \alpha$
of reversals from $\beta$ to $\alpha$.
Recall from Lemma~\ref{lem: bracket reversals one by one} that
the nested chain has subchains of the form 
$\alpha = \beta_m \KWleq \beta_{m-1} \KWleq \cdots \KWleq \beta_{m-k}$.
We perform an induction on the length $k$ of the subchain.
In the base case $k=0$, each entrance $a_\ell$ connects to the corresponding exit $b_\ell$,
since $\beta_m=\alpha$.
We then assume that in $\beta_{m-k}$, entrances of $\alpha$ connect to exits, and we
show that $\beta_{m-k-1}$ also satisfies this property.
Since $\beta_{m-k} \KWleq \beta_{m-k-1}$,
we can write $\beta_{m-k-1} = \cdots \BPEfont{(X(Y)Z)} \cdots$ and 
$\beta_{m-k} = \cdots \BPEfont{(X)Y(Z)} \cdots$, where the parentheses written out explicitly 
denote matching pairs, 
and $\BPEfont{X,Y,Z}$, and the ellipses denote parenthesis expressions that are identical in 
$\beta_{m-k-1}$ and $\beta_{m-k}$; see also Figure~\ref{fig: swap}.
By the induction assumption, it suffices to show that the matching parentheses written
explicitly in $\beta_{m-k-1} = \cdots \BPEfont{(X(Y)Z)} \cdots$ connect entrances of $\alpha$ to exits.
Lemma~\ref{lem: bracket reversals one by one} also guarantees that
\[ \alpha = \beta_m \KWleq \cdots \KWleq 
     \underbrace{\beta_{m-k}}_{\cdots \BPEfont{(X)Y(Z)} \cdots} \KWleq
     \underbrace{\beta_{m-k-1}}_{\cdots \BPEfont{(X(Y)Z)} \cdots}
\]
is a nested chain of reversals.
By the order of reversals in the nested chain from $\beta_{m-k-1} = \cdots \BPEfont{(X(Y)Z)} \cdots$
to $\alpha$,
the parentheses written out explicitly in $\beta_{m-k} = \cdots \BPEfont{(X)Y(Z)} \cdots$
are not reversed in the subchain\linebreak[4]%
$\beta_{m-k}, \beta_{m-k+1}, \ldots, \beta_m = \alpha$.
Therefore, these two matching pairs of parentheses of $\beta_{m-k}$ correspond to links
$\link{a_\ell}{b_\ell}$ and $\link{a_s}{b_r}$ as in Figure~\ref{fig: swap}(left), so the corresponding
matching pairs of $\beta_{m-k-1}$ 
also connect entrances of $\alpha$ to exits, as desired --- see Figure~\ref{fig: swap}(right).
This finishes the induction step, and proves that (b) implies (a).
\begin{figure}[h!]
\includegraphics[width = 0.3\textwidth]{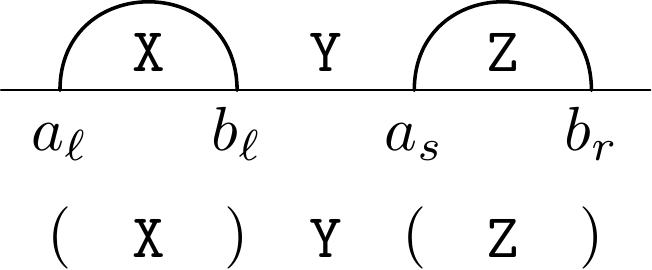} \hspace{2cm}
\includegraphics[width = 0.3\textwidth]{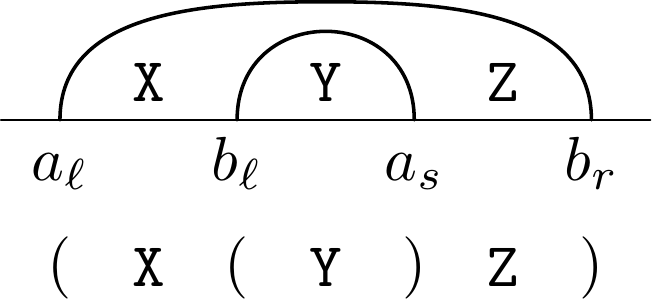}
\caption{\label{fig: swap} 
The balanced subexpressions and sub-link patterns of $\beta_{m-k}$ (left) 
and $\beta_{m-k-1}$ (right) of the proof of 
Lemma~\ref{lem: link patterns and KW relation}.}
\end{figure}

The last assertion follows by noticing that each reversal 
in the nested chain from $\beta$ to $\alpha$ corresponds to a transposition exchanging two exits of 
$\alpha$, see Figure~\ref{fig: swap}.
The permutation $\sigma$ is a composition of $m$ such transpositions and we have
$\sgn (\sigma) = (-1)^m$. This concludes the proof.

\end{proof}

\subsection{\label{subsec: Dyck tilings}Dyck tilings and inversion of weighted incidence matrices}

Dyck paths only take two kinds of steps, given by the 
vectors $(1,1)$  and $(1,-1)$ --- the paths live on the tilted square lattice generated
by these vectors. In particular, the area between any two Dyck paths $\alpha \DPleq \beta$ 
is a union of the atomic squares
of this lattice that lie between 
the highest and lowest Dyck paths $\nested_N$ and $\unnested_N$, 
illustrated in Figure~\ref{fig: DT}(left). 
The squares between $\alpha$ and $\beta$ with 
$\alpha \DPleq \beta$ form a \textit{skew Young diagram}, denoted by 
$\alpha / \beta$.

We consider tilings of skew Young diagrams by so called Dyck tiles. 
A \textit{Dyck tile} is a nonempty union of atomic squares, where the midpoints of the squares 
form a shifted Dyck path (possibly a zero-step path).
A \textit{Dyck tiling} $T$ of a skew Young diagram $\alpha / \beta$ is a
collection of non-overlapping tiles, whose union is the diagram: 
$\bigcup T = \alpha / \beta$. 
Figure~\ref{fig: DT} depicts some Dyck tiles and a Dyck tiling.


\begin{figure}[h!]
\includegraphics[width = 0.15\textwidth]{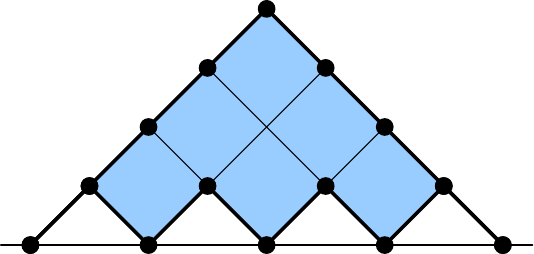} \qquad \qquad
\includegraphics[width = 0.35\textwidth]{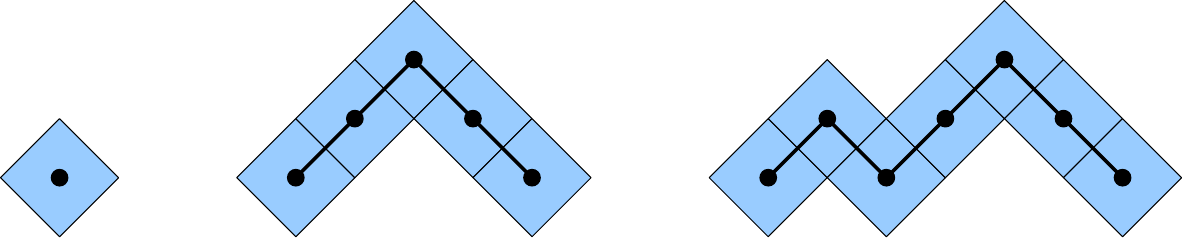} \qquad \qquad
\includegraphics[width = 0.25\textwidth]{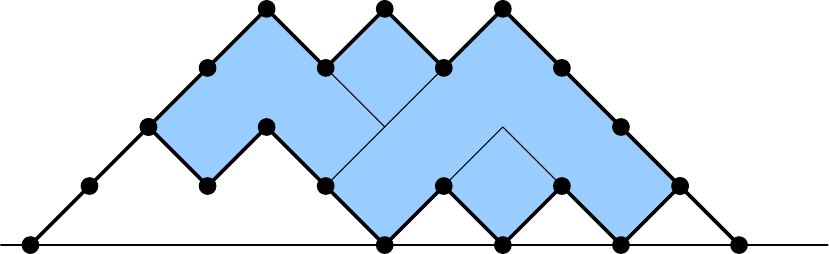}
\caption{\label{fig: DT} 
From the left: the atomic squares of $\DP_4$, three different Dyck tile shapes, 
and a Dyck tiling of a skew Young diagram.}
\end{figure}

In our applications, skew Young diagrams, Dyck tilings, and Dyck tiles have 
a shape and placement.
By skew shapes we mean the shift equivalence classes of skew Young diagrams.
Similarly, the shape of a Dyck tile $t$
is the underlying Dyck path whose bottom left position is at $(0,0)$.
The placement of $t$ is the applied shift,
i.e., the integer coordinates $(x_t,h_t)$ of the bottom left position of $t$.
We need four notions related to the horizontal and vertical placement
of Dyck tiles.
If the coordinates of the bottom left and bottom right positions of $t$
are $(x_t,h_t)$ and $(x'_t,h_t)$, then we say that the \textit{height}
of $t$ is $h_t \in \bZpos$, the \textit{horizontal extent} of $t$ is
the closed interval $[x_t , x'_t] \subset \bR$,
and the \textit{shadow} of $t$ is the open interval $(x_t-1 , x'_t +1) \subset \bR$,
see Figure~\ref{fig: extent shadow and height}.
A Dyck tile $t_1$ is said to \textit{cover} a Dyck tile $t_2$
if $t_1$ contains an atomic square which is an upward vertical
translation of some atomic square of $t_2$.
\begin{figure}[h!]
\includegraphics[width = 0.4\textwidth]{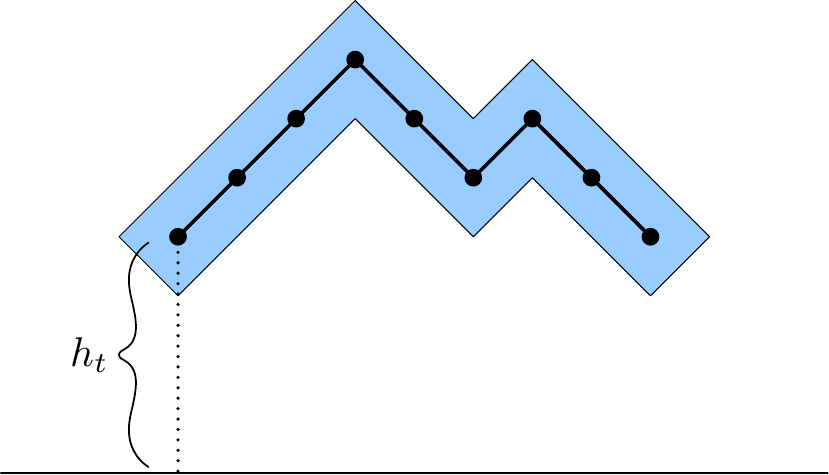} \hspace{1cm}
\includegraphics[width = 0.4\textwidth]{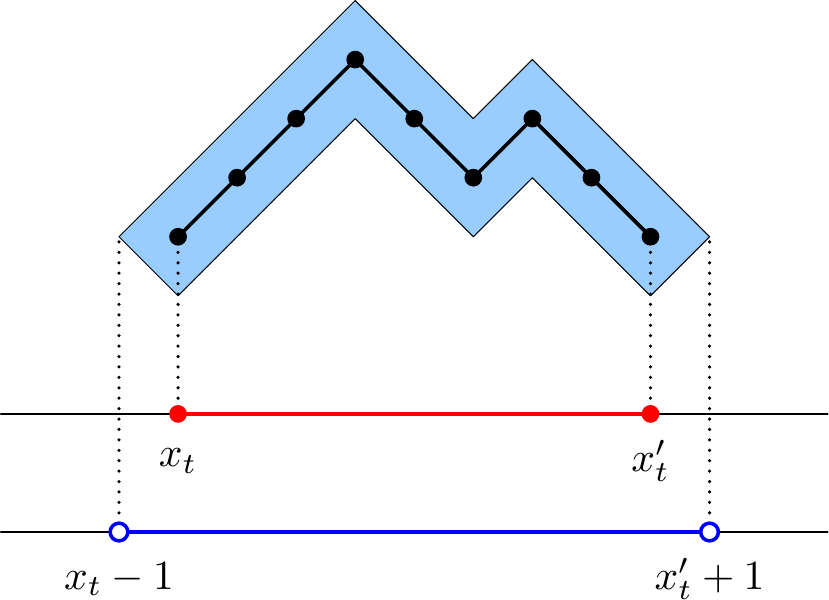}
\caption{\label{fig: extent shadow and height} The vertical position of a
Dyck tile $t$ is described by the integer height $h_t$.
The horizontal extent $[x_t , x'_t]$ (in red) and shadow $(x_t-1 , x'_t +1)$
(in blue) are intervals that describe the horizontal position.}
\end{figure}

We will use two special types of Dyck tilings,
\textit{nested Dyck tilings} (Definition~\ref{def: nested Dyck tiling})
and \textit{cover-inclusive Dyck tilings} (Definition~\ref{def: cover inclusive Dyck tiling}).

\begin{defn}\label{def: nested Dyck tiling}
A Dyck tiling $T$ 
is 
nested 
if the shadows of any two distinct tiles
of $T$ are either disjoint or one contained in the other,
and in the latter case the tile with the larger shadow covers the other.
\end{defn}
\begin{figure}
\includegraphics[width = 0.6\textwidth]{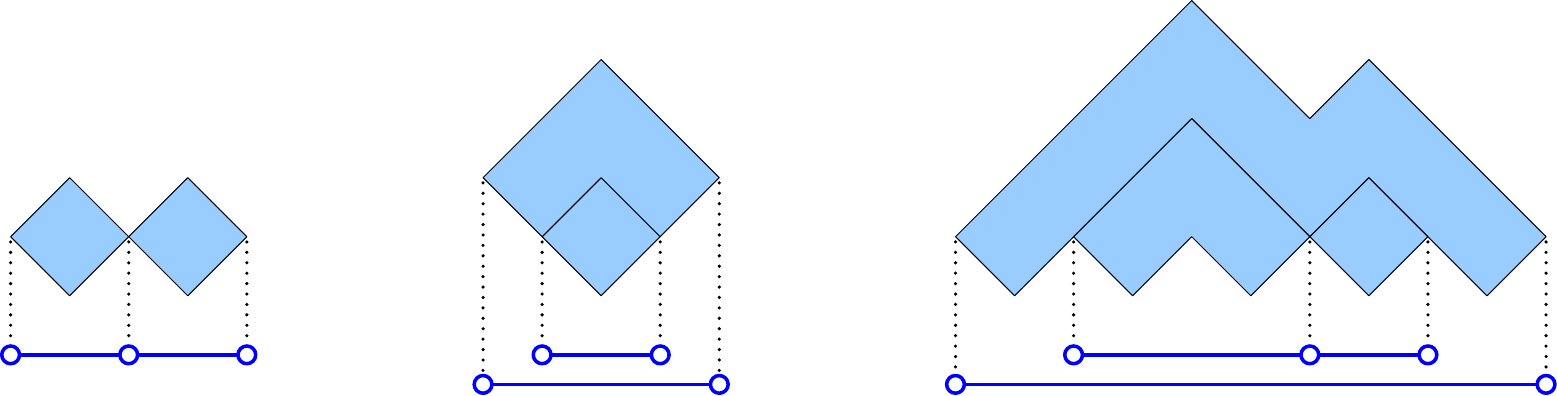}
\caption{\label{fig: NDT} Nested Dyck tilings of a skew Young diagram,
with shadows illustrated.}
\end{figure}
\begin{figure}
\includegraphics[width = 0.4\textwidth]{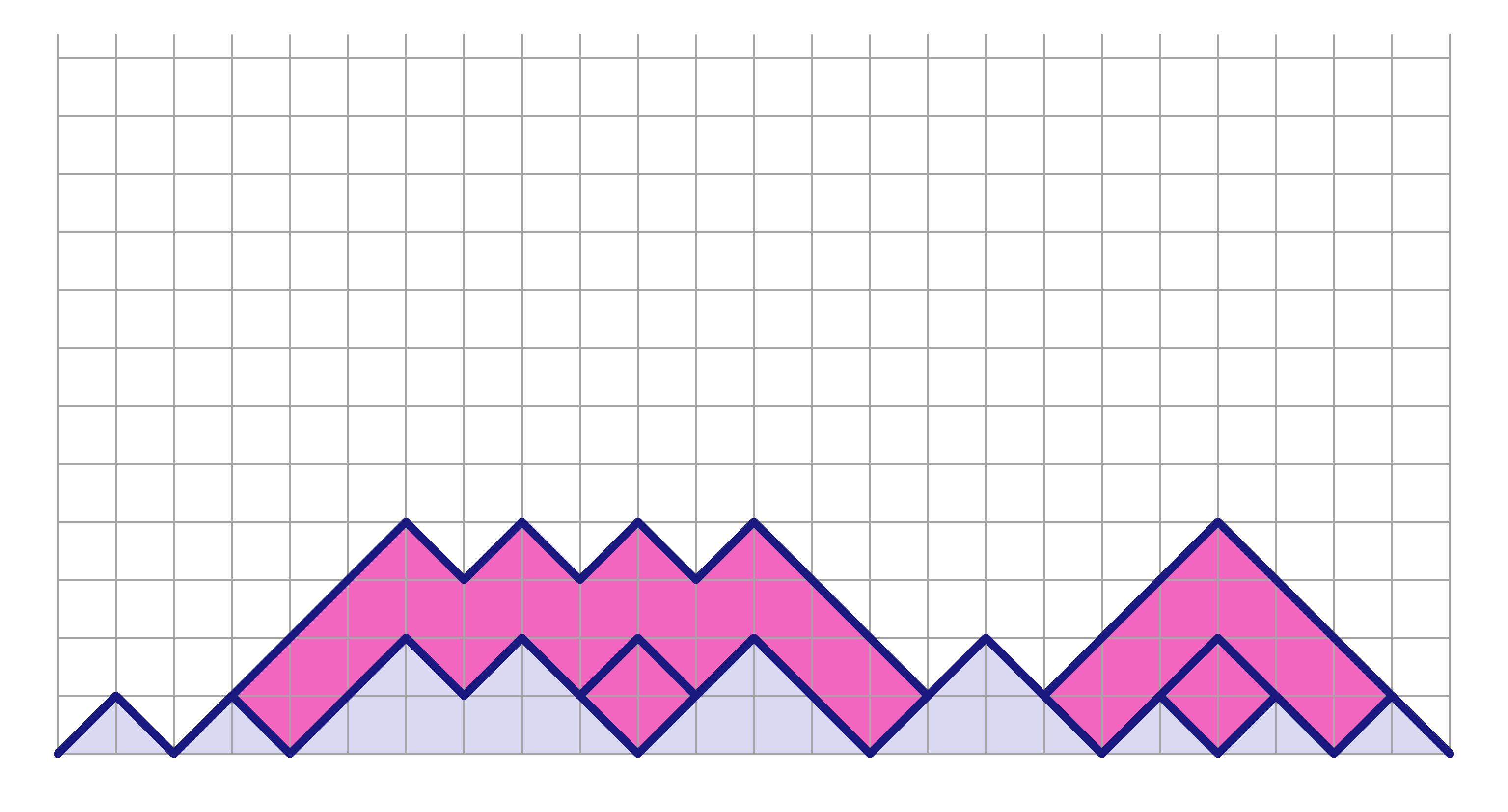} \quad
\includegraphics[width = 0.4\textwidth]{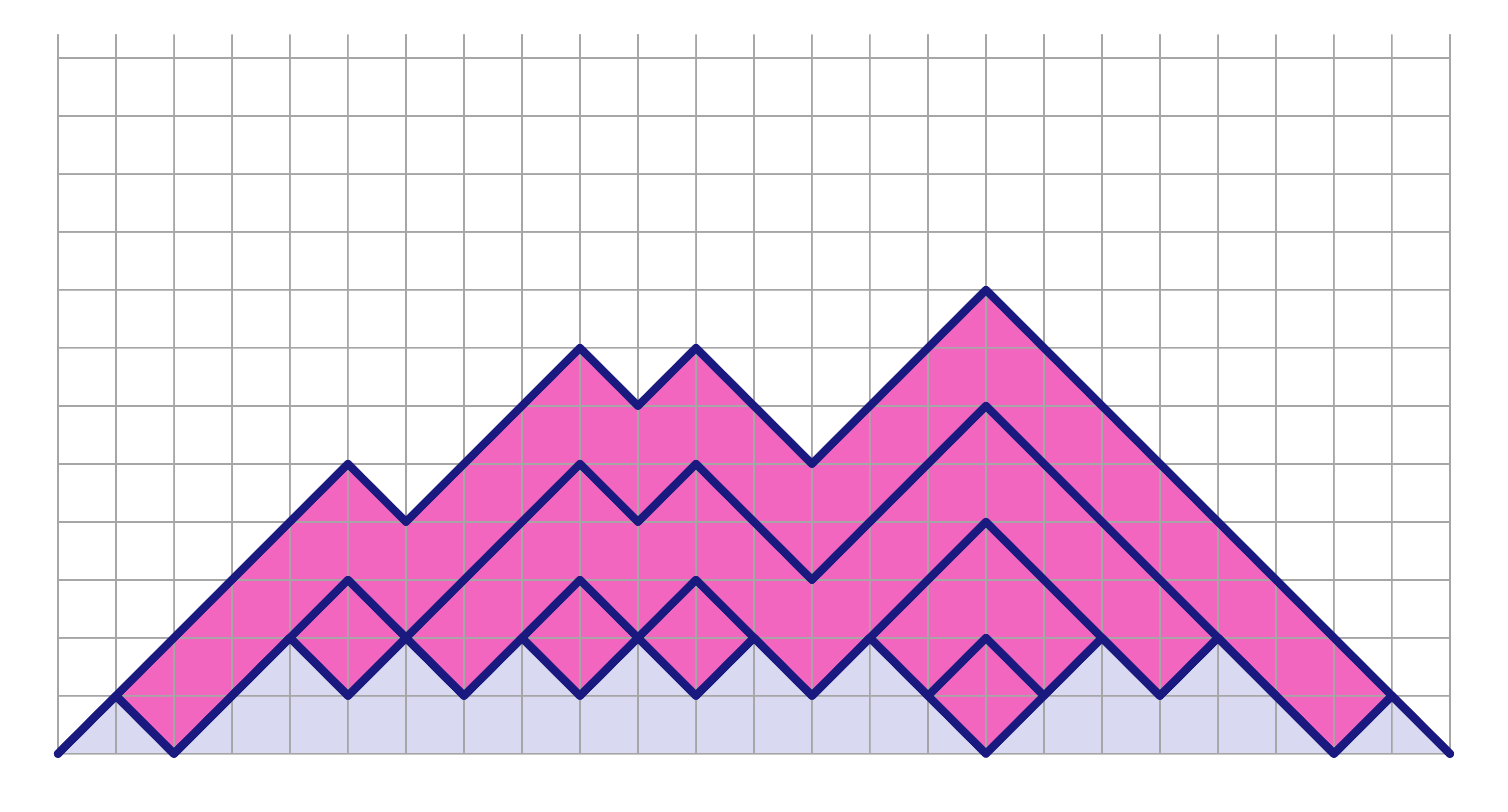}
\caption{\label{fig: NDT examples} Nested Dyck tilings of skew Young diagrams.}
\end{figure}

Nested Dyck tilings can always be described
as follows, see Figures~\ref{fig: NDT} and~\ref{fig: NDT examples} for illustration.
The top layer of each connected component of the
skew shape $\alpha / \beta$ must form a single tile in a nested tiling of $\alpha/\beta$
(if any exists), since breaking the top layer to more than one tile
would lead to non-disjoint shadows without the containment property.
Recursively, after removing these top layer tiles, the new top layers of the
remaining components form again single tiles. This shows first of all that there is
at most one nested Dyck tiling of any given skew Young diagram~$\alpha / \beta$,
which we then denote by $\nestedtilingof (\alpha / \beta)$.
Moreover, in a nested Dyck tiling, the containment of the shadows is always strict,
since the unique leftmost and rightmost atomic squares of a component are
contained in its top layer.
The following lemma shows that the existence of a nested tiling characterizes the
parenthesis reversal relation for Dyck paths. 
See also Figure~\ref{fig: KW}.
\begin{figure}
\includegraphics[width = 0.3\textwidth]{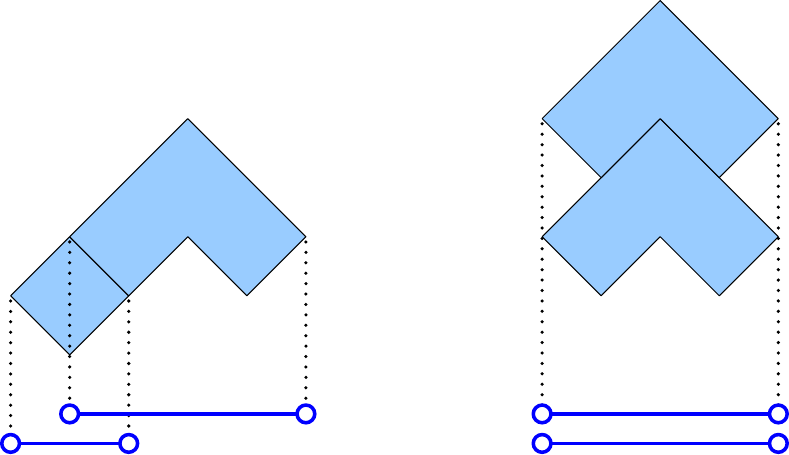}
\caption{\label{fig: NDT2} A non-nested tiling and a non-skew shape,
with shadows illustrated.}
\end{figure}
\begin{lem}\label{lem: nested tilings and KW relation}
Let $\alpha, \beta \in \DP_N$. The following statements are equivalent.
\begin{description}
\item[(a)] We have $\alpha \DPleq \beta$ and
the skew Young diagram $\alpha/\beta$ admits a 
nested Dyck tiling. 
\item[(b)] We have $\alpha \KWleq \beta$.
\end{description}
Moreover, in this case the number 
of matching pairs of parentheses reversed
in the nested chain from $\beta$ to $\alpha$ is the number of tiles in the
unique nested Dyck tiling $\nestedtilingof (\alpha / \beta)$.
\end{lem}
\begin{proof}
To prove that (b) implies (a), 
we assume that $\alpha \KWleq \beta$ and we
consider the nested chain of reversals $\beta = \beta_0 , \beta_1 , \ldots, \beta_m = \alpha$
of matching pairs of parentheses from $\beta$ to $\alpha$. The area between the
consecutive intermediate steps $\beta_{n} \KWleq \beta_{n-1}$ forms a Dyck tile, and these tiles form
a Dyck tiling of $\alpha/\beta$. The tiling is nested because 
if any two matching pairs of parentheses to be reversed are one inside the other,
then the reversal of the outer is performed before the inner.

To prove that (a) implies (b), consider the nested Dyck tiling $\nestedtilingof (\alpha / \beta)$
of the skew Young diagram $\alpha/\beta$. The endpoints of each tile $t \in \nestedtilingof (\alpha/\beta)$
correspond to a matching pair of parentheses in the balanced parenthesis expression~$\beta$.
The reversal of these matching pairs of $\beta$ produces~$\alpha$.

The remaining part of the statement is clear.
\end{proof}

In a nested Dyck tiling, wide tiles are on the top. 
Conversely, in a \textit{cover-inclusive Dyck tiling}, 
wide tiles are on the bottom.

\begin{defn}\label{def: cover inclusive Dyck tiling}
A Dyck tiling $T$ 
is cover-inclusive 
if for any two distinct tiles of $T$, 
either the horizontal extents 
are disjoint, or the tile that covers the other
has horizontal extent contained in the horizontal extent of the other.
\end{defn}
The property of being cover-inclusive is illustrated in 
Figures~\ref{fig: CIDT} and~\ref{fig: CIDT examples}. 
Any skew Young diagram admits a cover-inclusive Dyck tiling
at least with atomic square tiles. We denote by $\CItilingsof (\alpha / \beta)$ the family of all 
cover-inclusive tilings of $\alpha / \beta$.
Note that the disjointness of horizontal extents is less restrictive than
the disjointness of shadows, which is essentially why there are more
cover-inclusive Dyck tilings than nested Dyck tilings.

\begin{figure}
\includegraphics[width = 0.35\textwidth]{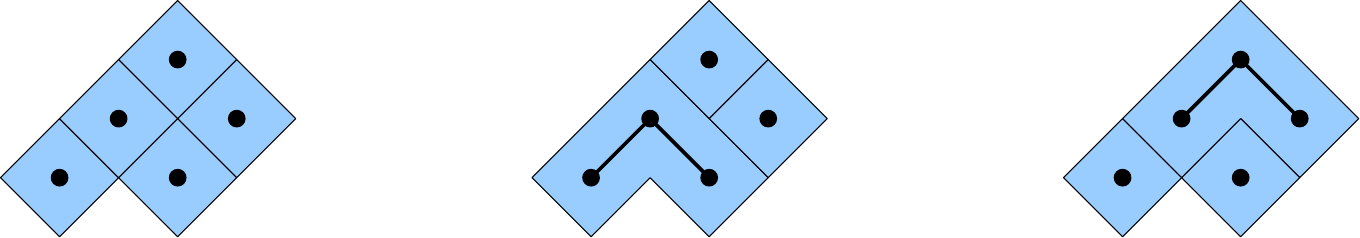}
\caption{\label{fig: CIDT} 
All the Dyck tilings of a small skew shape. The two first ones (from the left) 
are cover-inclusive. The third one is neither cover-inclusive nor nested.}
\end{figure}
\begin{figure}
\includegraphics[width = 0.4\textwidth]{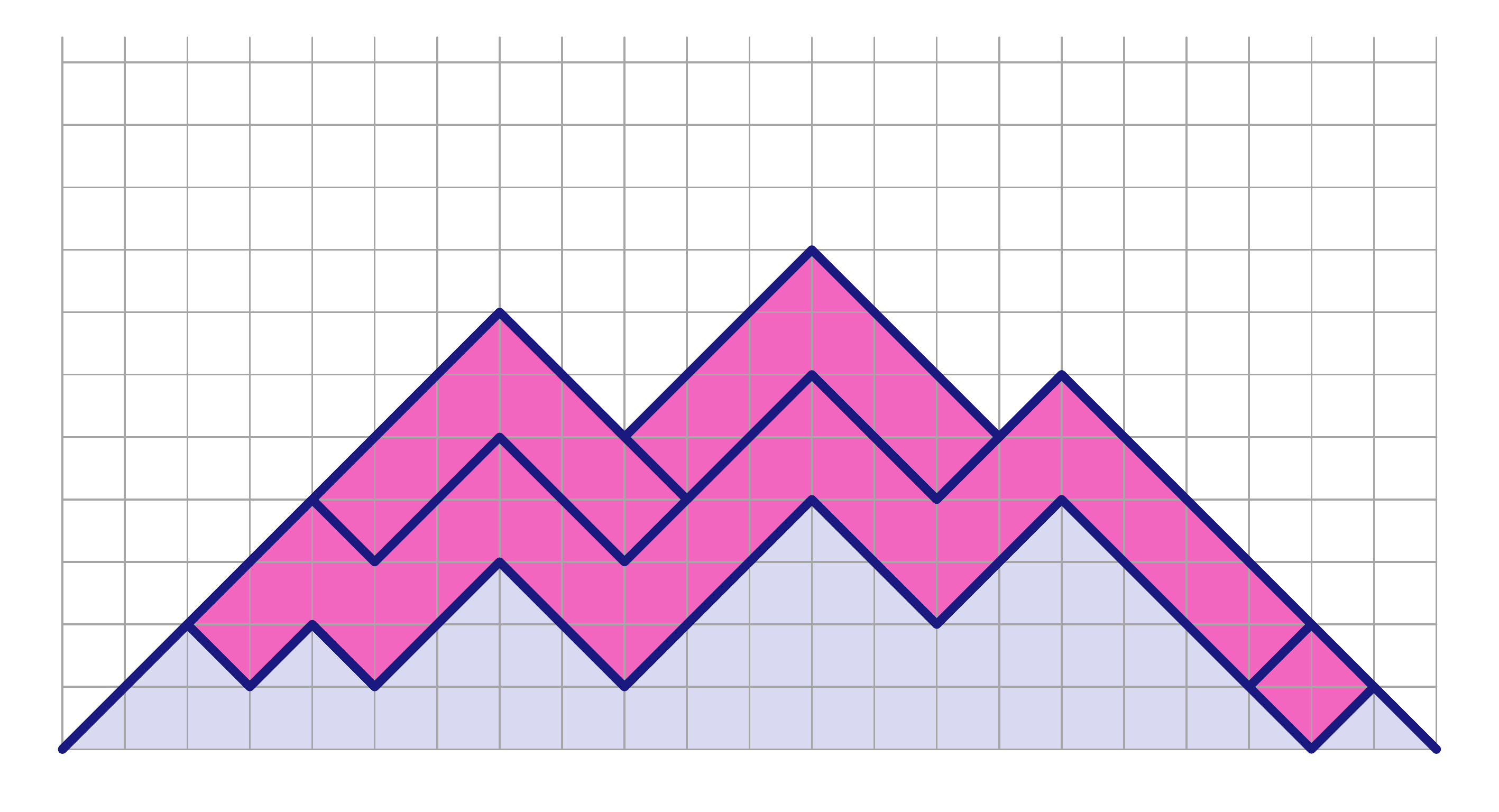} \quad
\includegraphics[width = 0.4\textwidth]{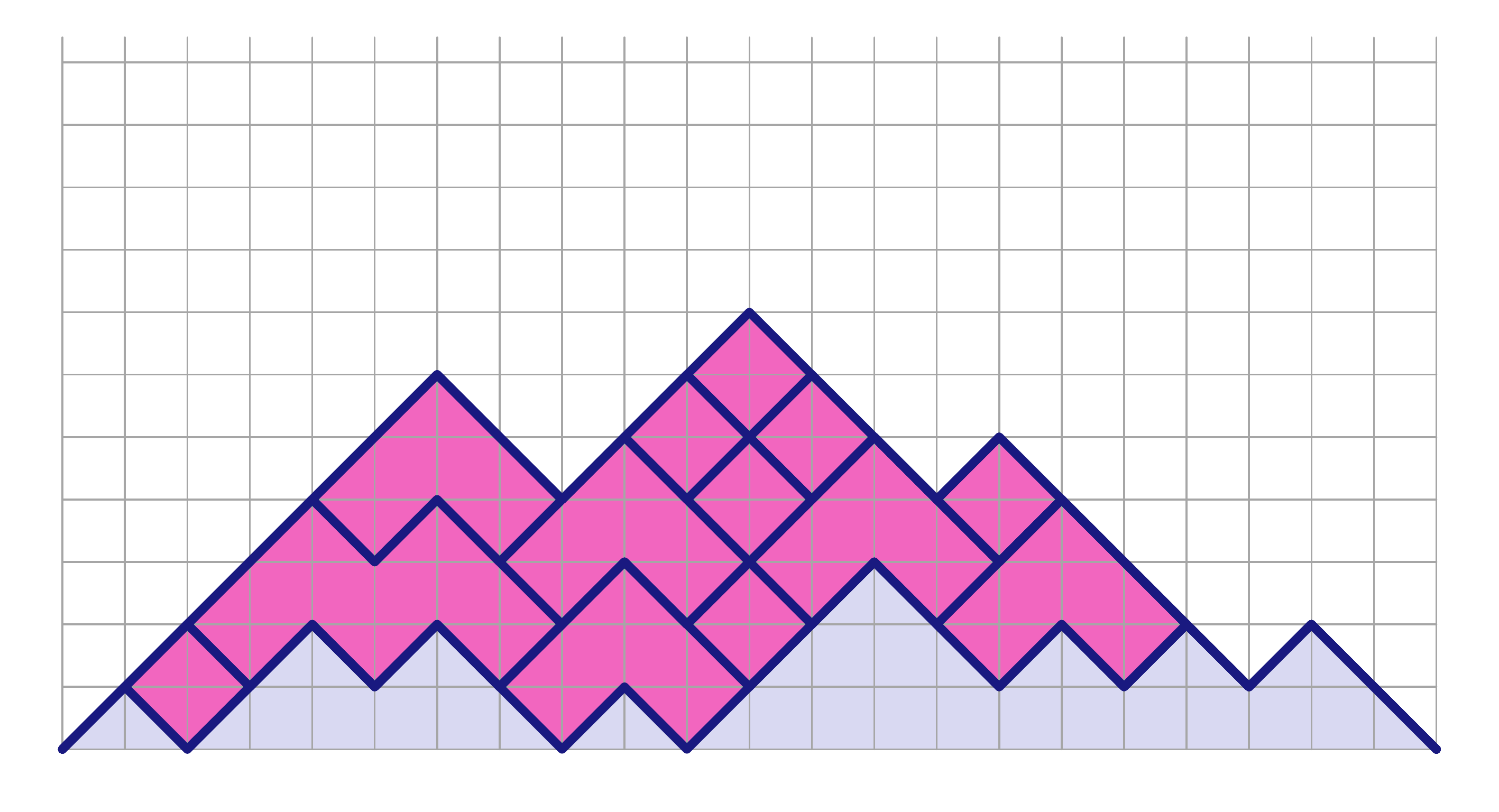} \\
\includegraphics[width = 0.4\textwidth]{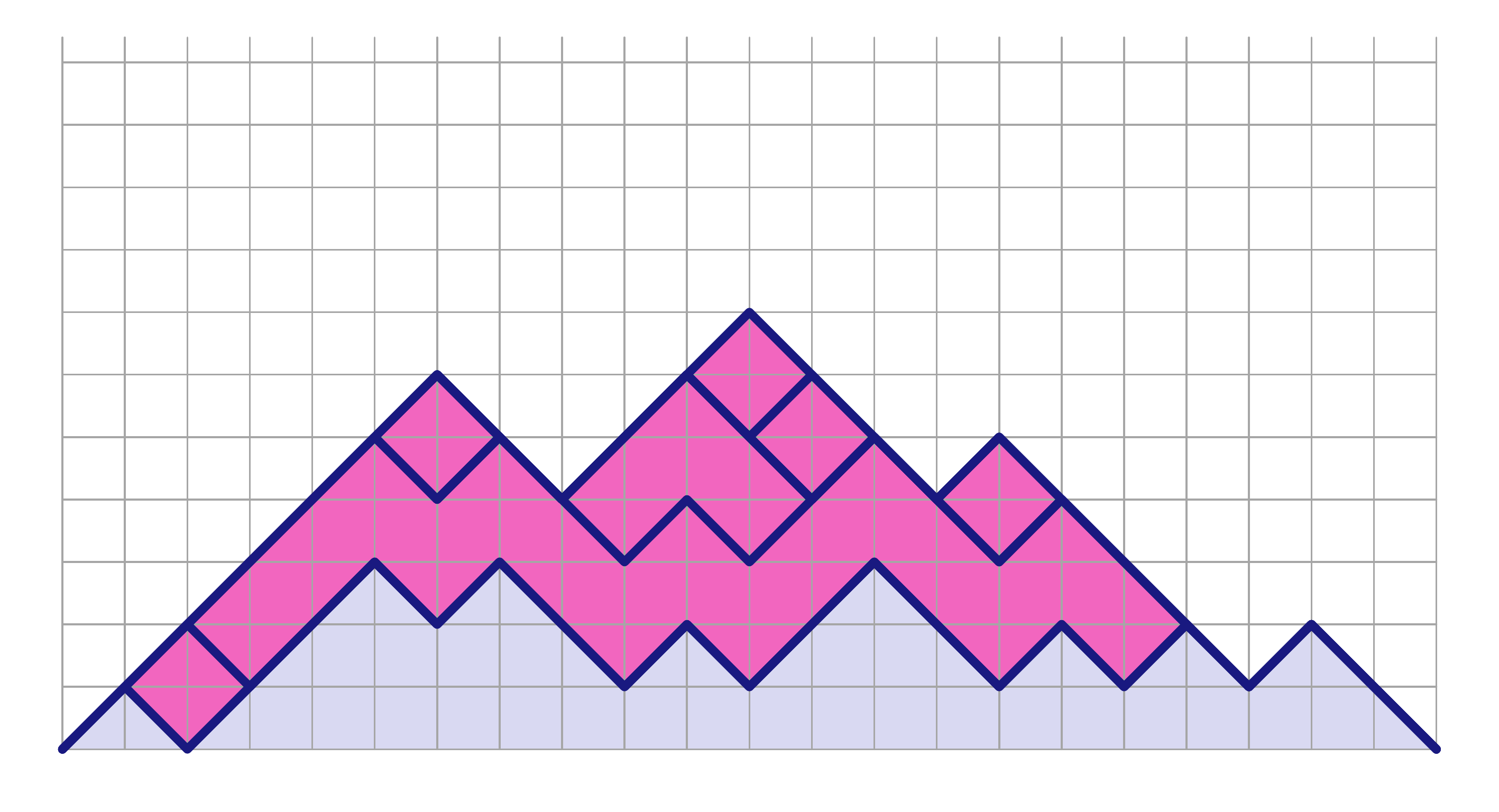} \quad
\includegraphics[width = 0.4\textwidth]{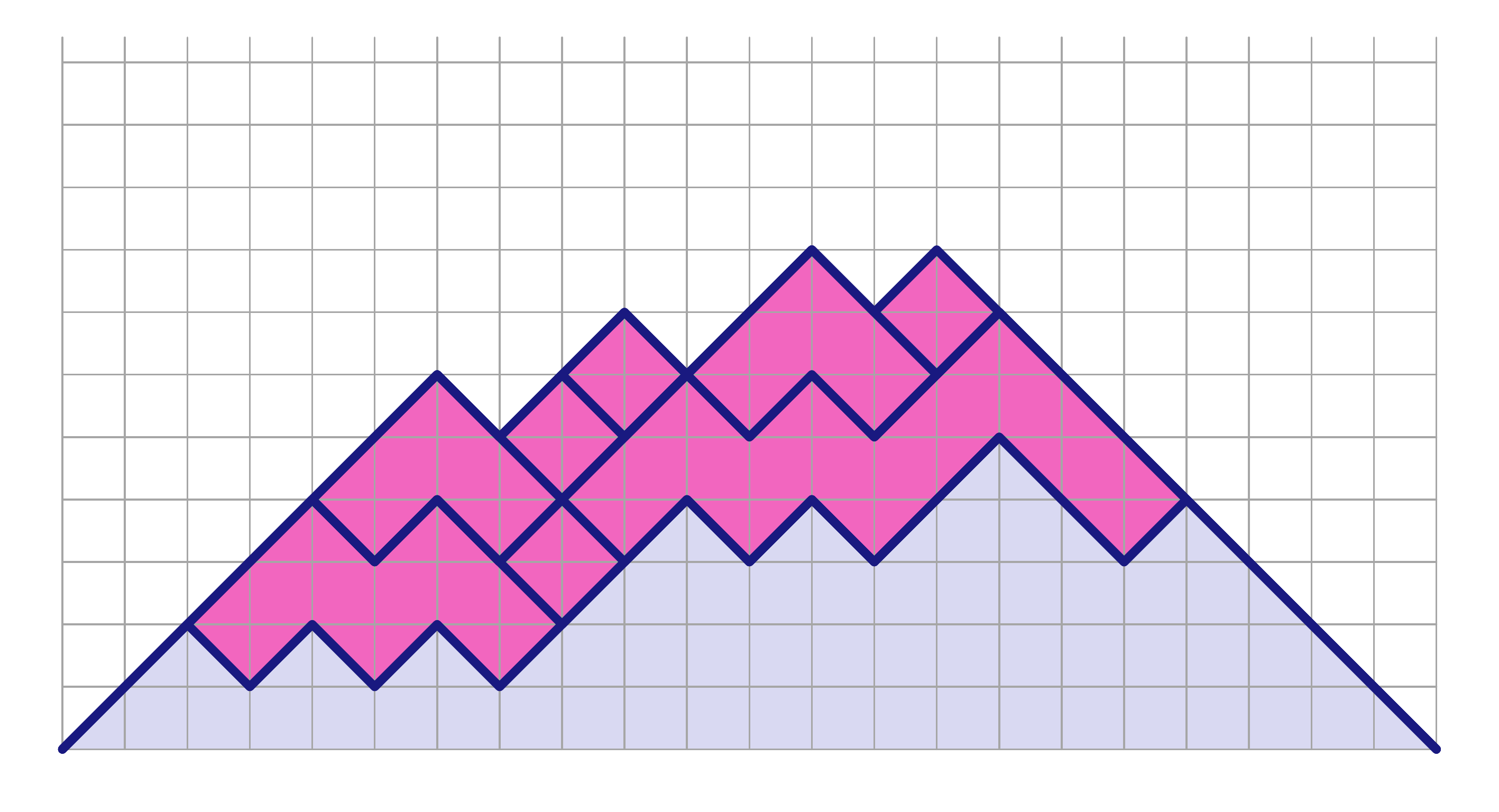} \\
\caption{\label{fig: CIDT examples} Cover-inclusive Dyck tilings of skew Young diagrams.}
\end{figure}

The cover-inclusive tilings are 
a key ingredient in the following theorem, which gives an explicit 
inversion formula for weighted incidence matrices.
Allowing for the weights makes this theorem a slight generalization of a
result of Kenyon and Wilson~\cite[Theorem~1.5]{KW-double_dimer_pairings_and_skew_Young_diagrams}.
The unit weight case will be used in Section~\ref{sec: applications to UST},
whereas a nontrivially weighted case will
be needed in~\cite{KKP-companion}.
We give the beginning of the proof
up to a point where the reduction to the results of Kenyon and Wilson is clear.

\begin{thm}
\label{thm: weighted KW incidence matrix inversion}
Assign weights $w(t) \in \bC$ to all Dyck tiles $t$ (with shape and placement).
Let $M \in \mathbb{C}^{\DP_N \times \DP_N}$ be the weighted incidence matrix
\begin{align}\label{eq: def of weighted incidence matrix}
M_{\alpha, \beta} := \begin{cases}
\prod_{t \in \nestedtilingof (\alpha / \beta)} (-w(t)) & \text{if }\alpha \KWleq \beta\\
0 & \text{otherwise}
\end{cases}
\end{align}
of the parenthesis reversal relation $\KWleq$, where $\nestedtilingof (\alpha / \beta)$ is
the unique nested tiling of the skew Young diagram $\alpha / \beta$.
Then $M$ is invertible, and the entries of the inverse matrix $M^{-1}$ are given by the weighted sums 
\begin{align}\label{eq: inverse of weighted incidence matrix}
M^{-1}_{\alpha, \beta} = \begin{cases}
\sum_{T \in \CItilingsof (\alpha/\beta)} \prod_{t \in T} w(t) & \text{if }\alpha \DPleq \beta \\
0 & \text{otherwise} 
\end{cases}
\end{align}
over the sets 
$\CItilingsof (\alpha/\beta)$ 
of cover-inclusive Dyck tilings of the skew Young diagrams $\alpha/\beta$.
\end{thm}

\begin{proof}
Since $\alpha \KWleq \beta$ implies $\alpha \DPleq \beta$,
the matrix $M$ is upper-triangular with respect to the partial order~$\DPleq$.
The diagonal entries are all ones, $M_{\alpha,\alpha} = 1$.
Thus $M$ is invertible, and 
also $M^{-1}$ is upper-triangular with ones on the diagonal:
$M^{-1}_{\alpha, \beta} = 0$
unless $\alpha \DPleq \beta$, and $M^{-1}_{\alpha, \alpha} = 1$.

It remains to compute the entries $M^{-1}_{\alpha, \beta}$ when $\alpha \DPleq \beta$ and $\beta \neq \alpha$. 
We then have
\begin{align*}
\sum_{\lambda \in \DP_N} M^{-1}_{\alpha, \lambda}M_{\lambda, \beta} = \delta_{\alpha, \beta} = 0,
\end{align*}
and with the knowledge of the zero entries of $M$ and $M^{-1}$, we can restrict the summation to obtain
\begin{align*}
\sum_{ \substack{ \lambda \in \DP_N \\ \lambda \KWleq \beta \  \& \  \alpha \DPleq \lambda } } M^{-1}_{\alpha, \lambda}M_{\lambda, \beta} = 0.
\end{align*}
By Lemma \ref{lem: nested tilings and KW relation},
instead of $\lambda$, the summation can 
be indexed by the nested Dyck tilings $S = \nestedtilingof(\lambda/\beta)$,
whose upper boundary is a subpath 
of $\beta$ and which are contained
in the skew Young diagram between $\alpha$ and $\beta$, i.e., 
$\bigcup S \subseteq \alpha / \beta$.
With the notation  $\lambda = \beta \downarrow S$, we then have
\begin{align*} 
\sum_{ \substack{ \text{nested Dyck tilings $S$ of } \bigcup S \subseteq \alpha / \beta \\  \text{ with upper boundary in $\beta$}}} M^{-1}_{\alpha, \beta \downarrow S} M_{\beta \downarrow S , \beta} = 0.
\end{align*}
Using the definition of $M$, from the equation above 
we solve $M^{-1}_{\alpha, \beta}$ in terms of $M^{-1}_{\alpha, \gamma}$,
where $\alpha \DPleq \gamma \DPleq \beta$:
\begin{align*}
M^{-1}_{\alpha, \beta} \; = \sum_{ \substack{ \text{nested Dyck tilings $S$ of } \bigcup S \subseteq \alpha / \beta, \\  S \ne \emptyset, \text{ with upper boundary in $\beta$}}} - (-1)^{\vert S \vert} \left( \prod_{t \in S} w(t) \right) M^{-1}_{\alpha, \beta \downarrow S},
\end{align*}
where $\vert S \vert$ denotes the number of Dyck tiles in $S$. This formula 
combined with the initial condition $M^{-1}_{\alpha, \alpha} = 1$ can be used 
to find $M^{-1}_{\alpha, \beta}$ recursively for all $\beta$.
Inductively,
we now first deduce that $M^{-1}_{\alpha, \beta}$ is a sum of weights of 
Dyck tilings of $\alpha / \beta$, that is, there exist coefficients $c_T$ that 
are independent of the weight function $w$, such that we have
\begin{align*}
M^{-1}_{\alpha, \beta} = \sum_{\text{Dyck tilings }T \text{ of }\alpha / \beta} c_T \left( \prod_{t \in T} w(t) \right).
\end{align*}
The remaining task is to find the coefficients $c_T$.
Here we rely on results of Kenyon and Wilson: it follows from 
\cite[Theorem~1.6.]{KW-double_dimer_pairings_and_skew_Young_diagrams} that
\[ c_T = \begin{cases}
    1 & \text{ if $T$ is cover-inclusive} \\
    0 & \text{ otherwise.}
\end{cases}
\]
\end{proof}
The next example gives the 
special case of the above theorem that will be employed in the present article,
with all tiles having weights $w(t)=1$.
This case will be needed in Section~\ref{sec: applications to UST} in order to solve
for the UST connectivity probabilities, and it is very closely
related to the original choice of Kenyon and Wilson~\cite{KW-double_dimer_pairings_and_skew_Young_diagrams},
which can be recovered by setting $w(t)=-1$ instead.
More general choices of weights will be needed in the follow-up work~\cite{KKP-companion}.
\begin{exa}
\label{ex: signed incidence matrix of KWleq}
Let $\Mmat \in \C^{\DP_N \times \DP_N}$ be the unit weight incidence matrix of the parenthesis reversal relation,
obtained by choosing the weight function $w(t) = 1$ for all tiles $t$, 
\begin{align*}
\Mmat_{\alpha,\beta} = \begin{cases}
    (-1)^{\vert \nestedtilingof (\alpha / \beta) \vert} & \text{ if } \alpha \KWleq \beta \\
    0 & \text{ otherwise},
\end{cases}
\end{align*}
where $\vert \nestedtilingof (\alpha / \beta) \vert$ is the number of Dyck tiles 
in the nested tiling $\nestedtilingof (\alpha / \beta)$.
Let $( (a_\ell,b_\ell) )_{\ell=1}^N$ denote the left-to-right orientation of
the link pattern $\alpha$. Lemmas~\ref{lem: link patterns and KW relation}
and~\ref{lem: nested tilings and KW relation} show that
the relation ${\alpha \KWleq \beta}$ is equivalent to the existence of a unique permutation
$\sigma \in \SymmGrp_N$ of the exits $(b_\ell)_{\ell=1}^N$ of $\alpha$
such that ${\beta = \{ \link{a_1}{b_{\sigma(1)}} \ldots \link{a_N}{b_{\sigma(N)}} \}}$, and
then $(-1)^{|\nestedtilingof(\alpha/\beta)|} = \sgn(\sigma)$ is the sign of that permutation.
Hence, 
we can equivalently write
\begin{align}\label{eq: basic KW matrix}
\Mmat_{\alpha,\beta} = \begin{cases}
    \sgn(\sigma) \quad & \text{ if } \beta = \{ \link{a_1}{b_{\sigma(1)}} \ldots \link{a_N}{b_{\sigma(N)}} \}
            \text{ for some } \sigma \in \SymmGrp_N \\ 
    0 & \text{ otherwise}.
\end{cases}
\end{align}
Theorem~\ref{thm: weighted KW incidence matrix inversion} shows that
$\Mmat$ has inverse with entries counting the cover-inclusive Dyck tilings,
\begin{align}\label{eq: basic KW inverse matrix}
\Mmat_{\alpha,\beta}^{-1} = \begin{cases}
    \# \CItilingsof (\alpha / \beta) & \text{ if } \alpha \DPleq \beta \\
    0 & \text{ otherwise} .
\end{cases}
\end{align}
In particular, the entries $\Mmat_{\alpha,\beta}^{-1}$
are non-negative integers, and positive precisely when $\alpha \DPleq \beta$.

For concreteness, the matrices $\Mmat$ and $\Minv$ as well as the illustrations of 
all cover-inclusive and nested Dyck tilings of each skew shape
are given explicitly for $N=2,3,4$ in 
Figures~\ref{fig: KW matrices}--\ref{fig: CIDT matrices 4}.
\end{exa}
\begin{figure}
\bigskip
\bigskip
\bigskip
\centerfloat
\includegraphics[width = 0.28 \textwidth]{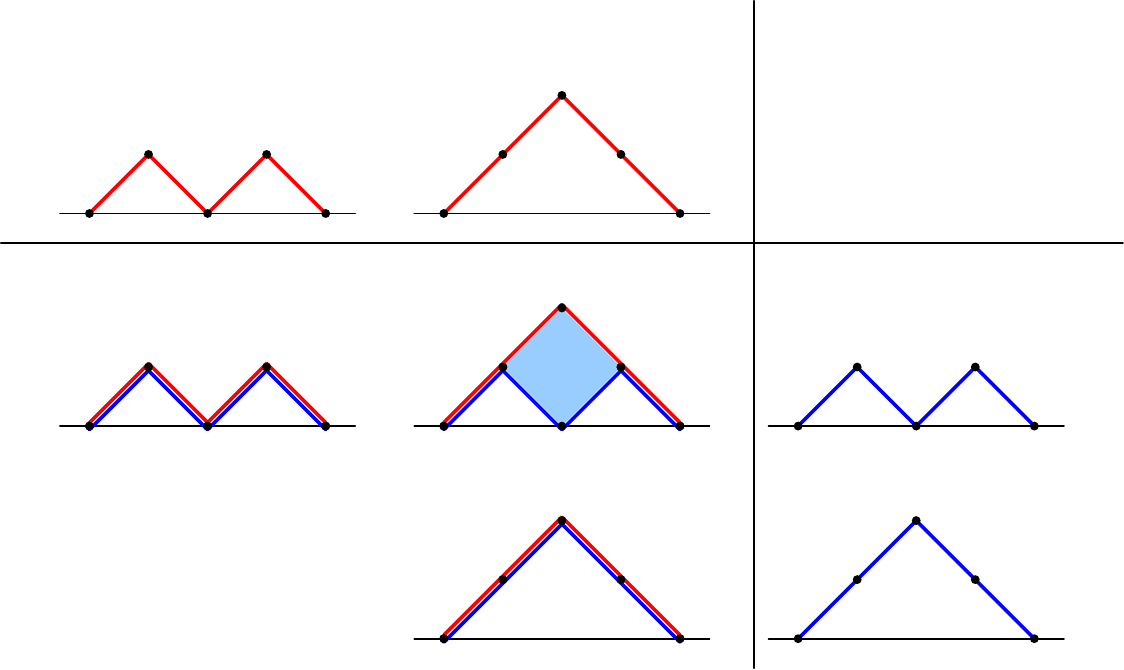} \qquad
\includegraphics[width = 0.28 \textwidth]{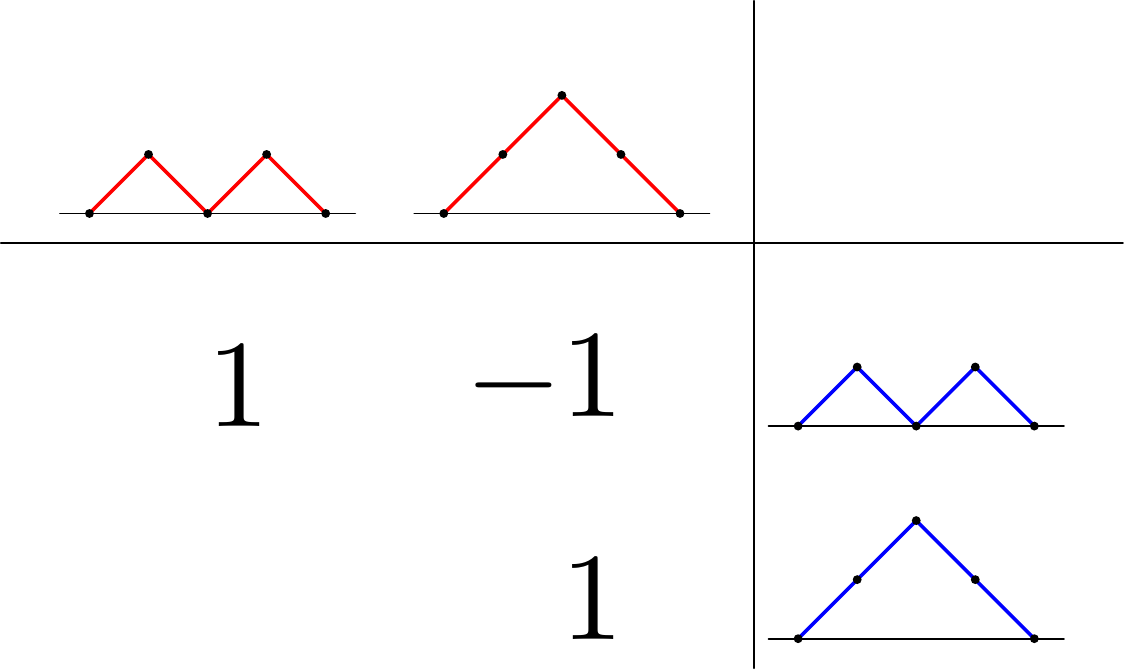} \\
\vspace{0.5 cm}
\includegraphics[width = 0.45 \textwidth]{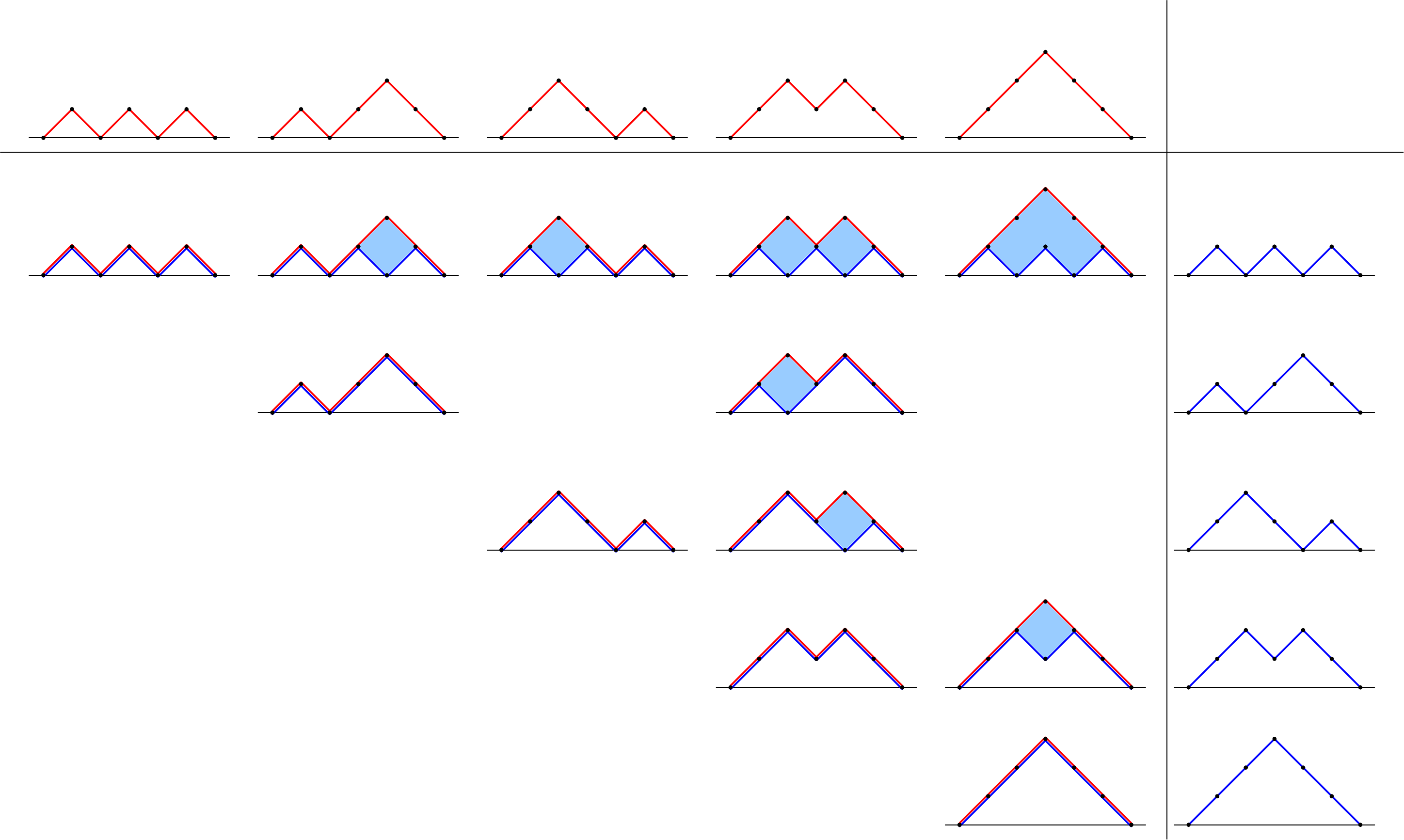} \qquad
\includegraphics[width = 0.45 \textwidth]{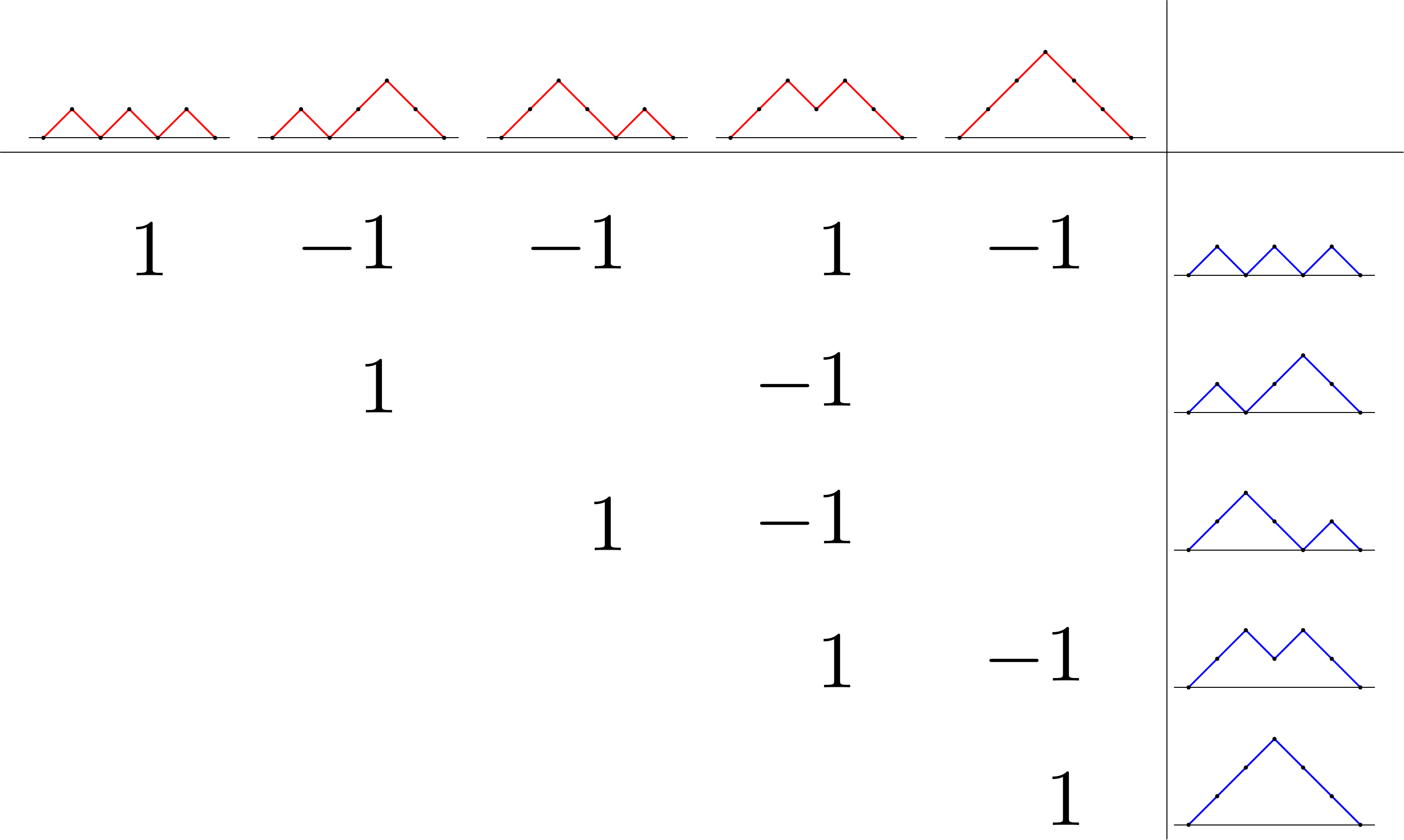}
\bigskip
\bigskip
\bigskip
\caption{\label{fig: KW matrices}
The nested Dyck tilings of all skew Young diagrams for $N=2$ and $N=3$
and the corresponding signed incidence matrices $\Mmat$ 
defined in Equation~\eqref{eq: basic KW matrix}.
}
\end{figure}
\begin{figure}
\bigskip
\bigskip
\bigskip
\centerfloat
\includegraphics[width = 0.28 \textwidth]{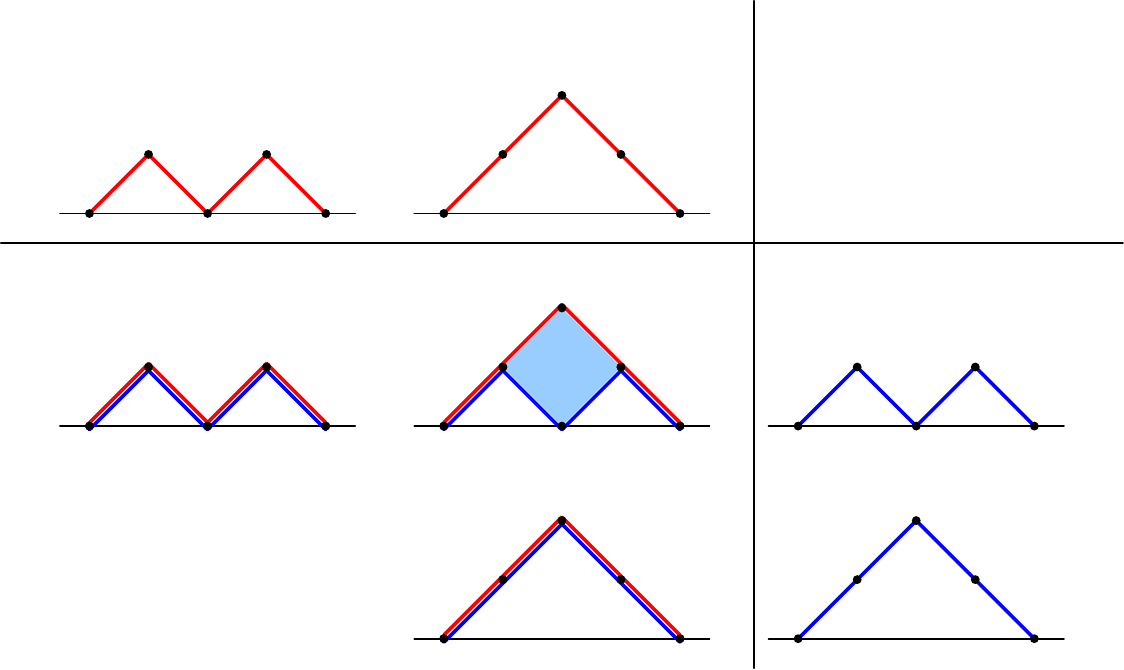} \qquad
\includegraphics[width = 0.28 \textwidth]{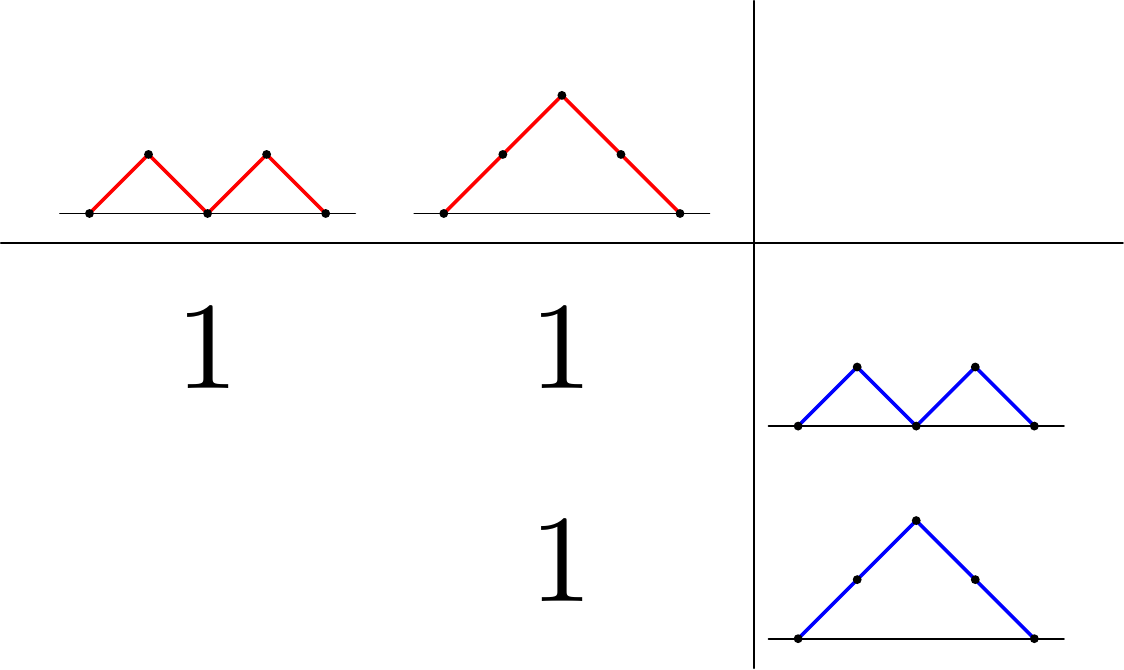} \\
\vspace{0.5 cm}
\includegraphics[width = 0.45\textwidth]{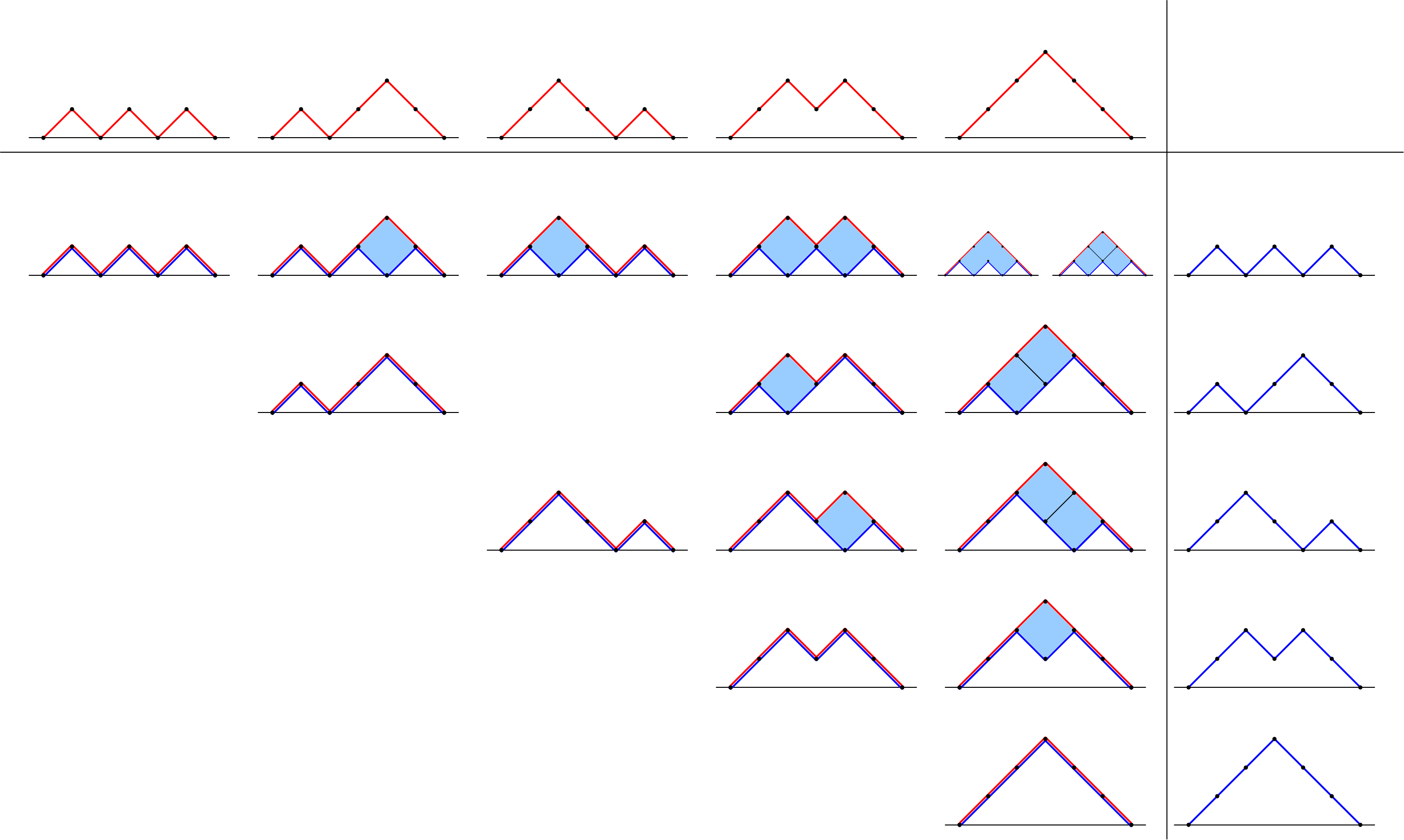} \qquad 
\includegraphics[width = 0.45\textwidth]{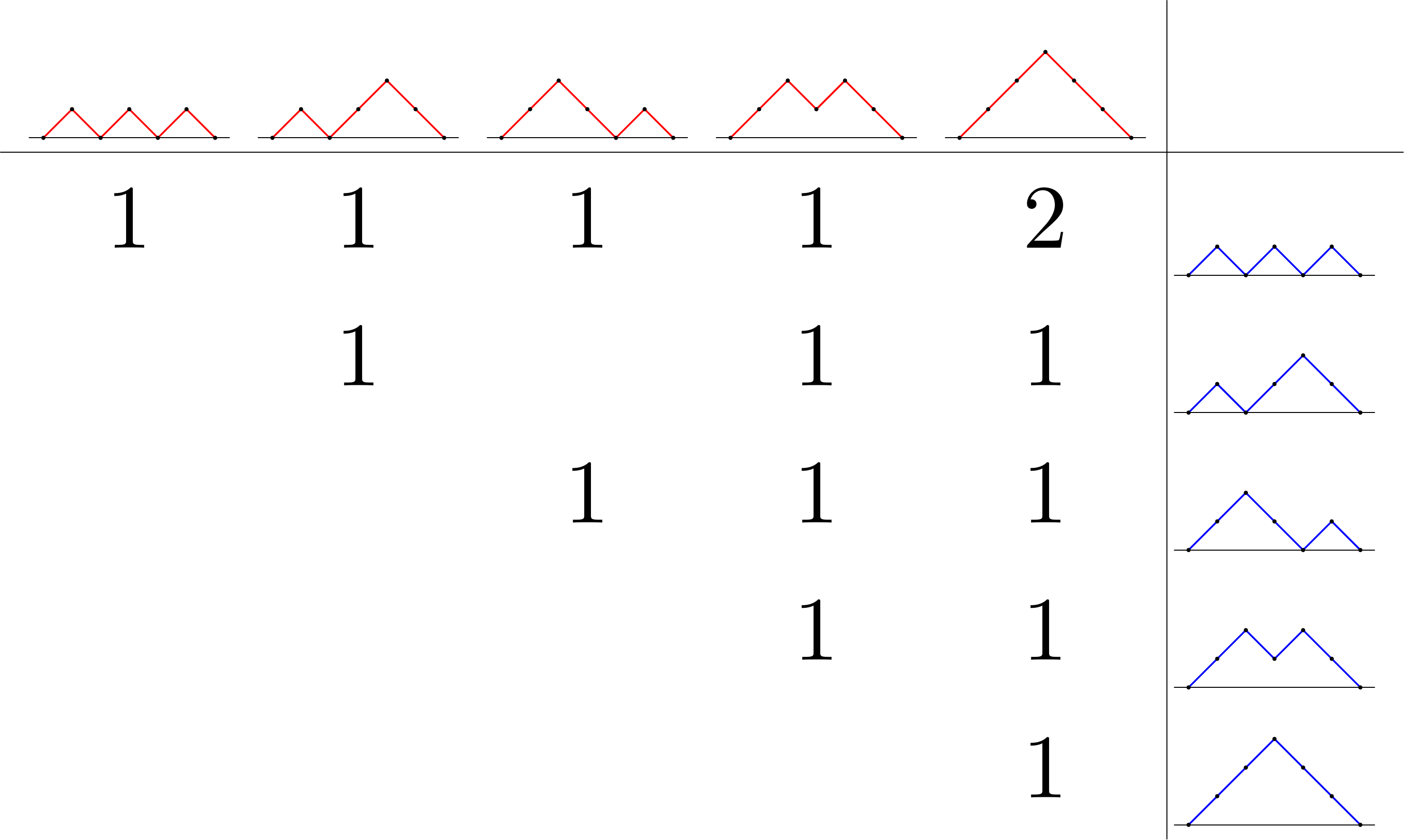}
\bigskip
\bigskip
\bigskip
\caption{\label{fig: CIDT matrices}
The cover-inclusive Dyck tilings of all skew Young diagrams for $N=2$ and $N=3$
and the corresponding matrices $\Minv$,
whose entries count the number of such tilings according to Equation~\eqref{eq: basic KW inverse matrix}.}
\end{figure}
\begin{figure}
\centerfloat
\includegraphics[width = 1.0\textwidth]{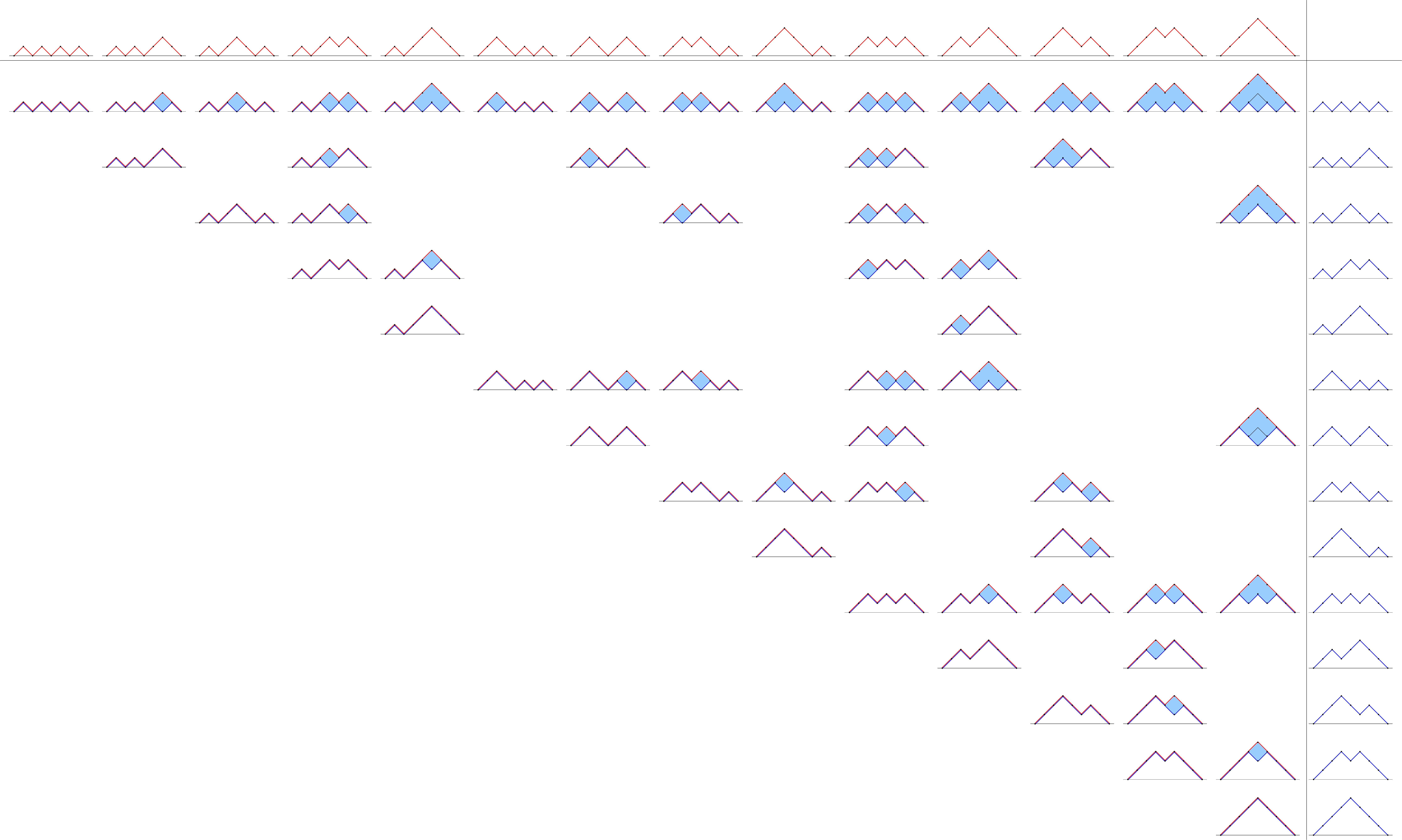} \\
\vspace{0.5 cm}
\includegraphics[width = 1.0\textwidth]{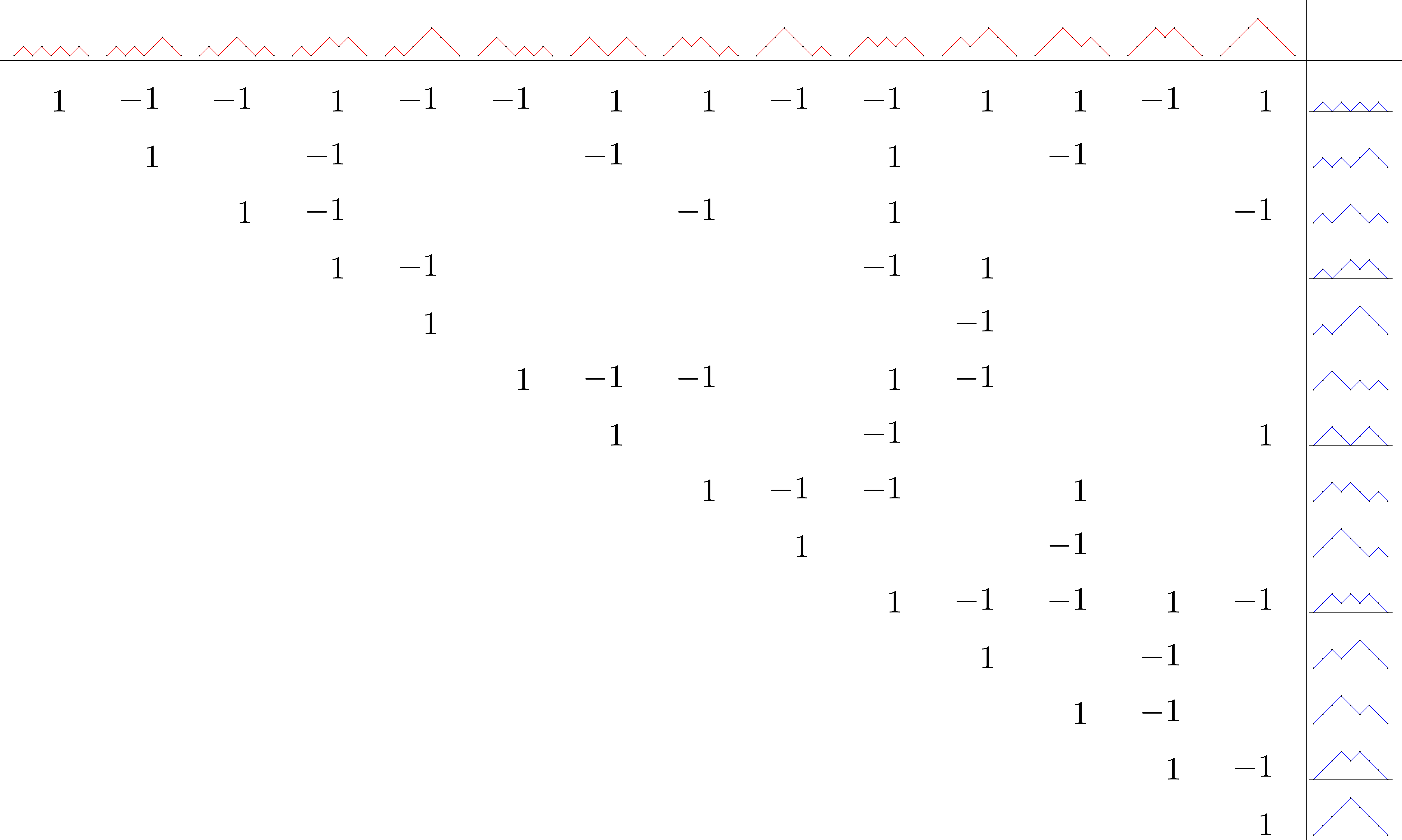}
\caption{\label{fig: KW matrices 4}
The nested Dyck tilings of all skew Young diagrams for $N=4$
and the corresponding signed incidence matrix $\Mmat$.}
\end{figure}
\begin{figure}
\centerfloat
\includegraphics[width=1.35\textwidth, angle =90]{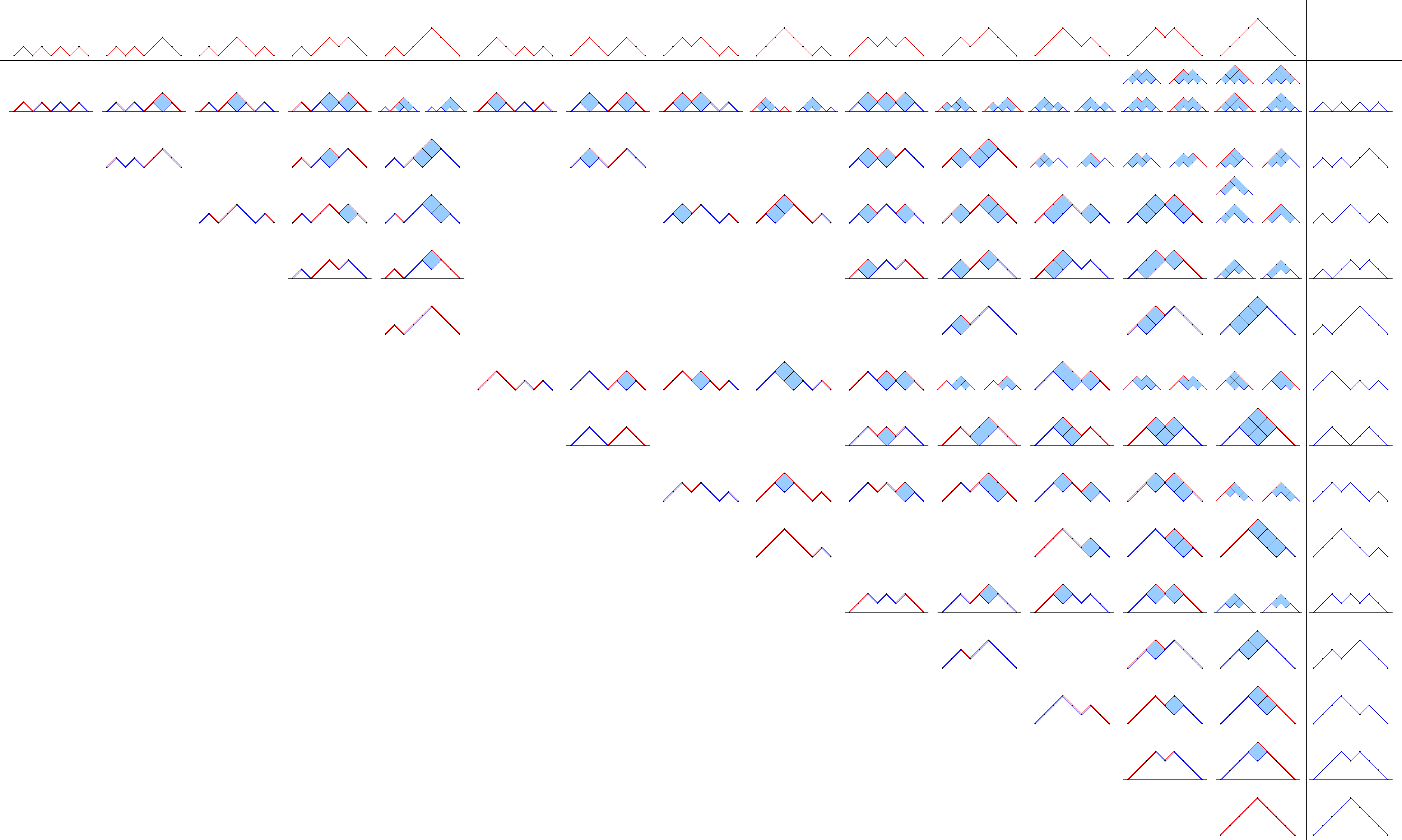}
\caption{\label{fig: CIDT matrices 4 pics}
The cover-inclusive Dyck tilings of all skew Young diagrams for $N=4$.
}
\end{figure}
\begin{figure}
\centerfloat
\includegraphics[width=1.35\textwidth, angle =90]{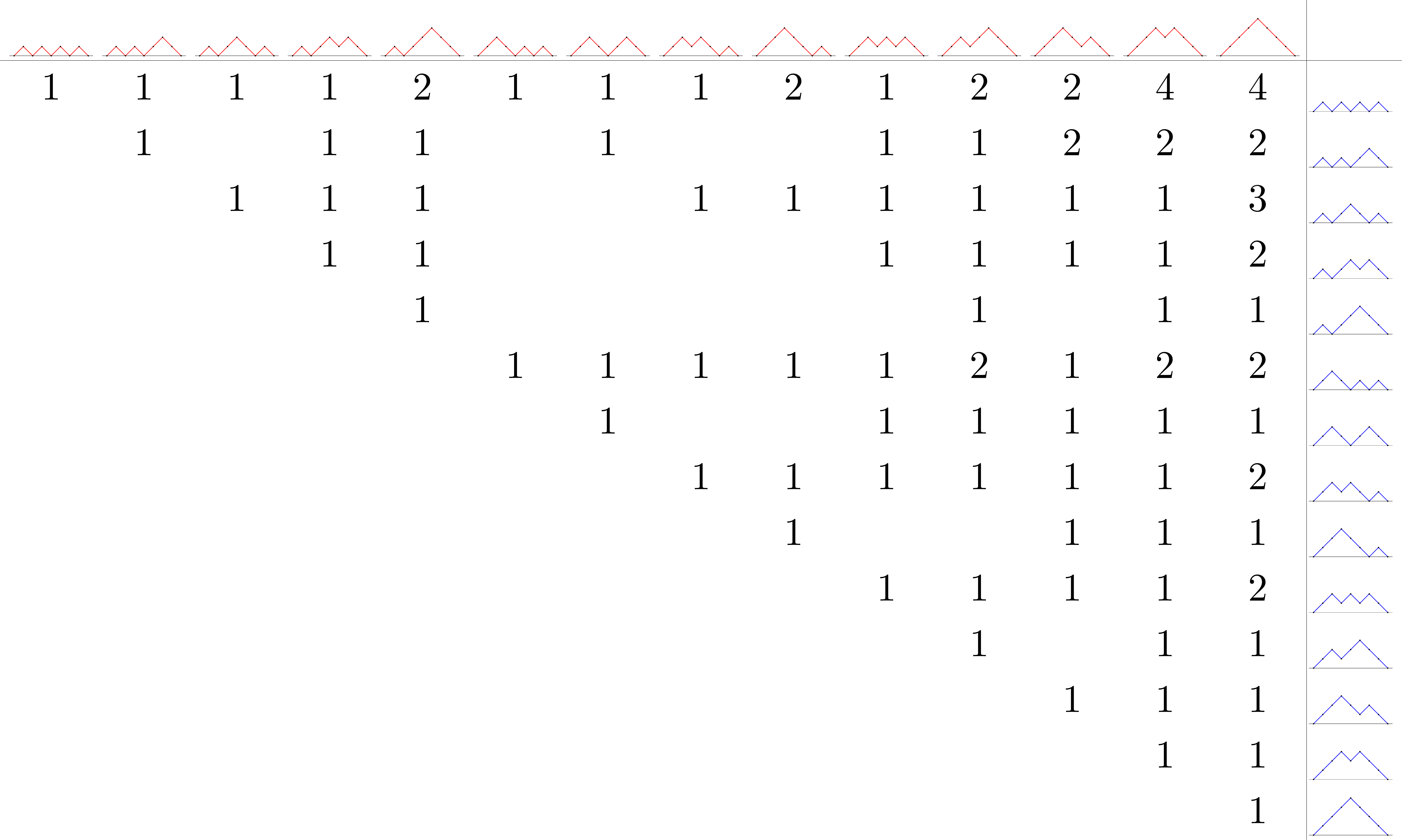}
\caption{\label{fig: CIDT matrices 4}
The matrix $\Minv$ counting the cover-inclusive Dyck tilings in
Figure~\ref{fig: CIDT matrices 4 pics}. 
}
\end{figure}

\subsection{\label{subsec: Wedges, slopes, and link removals}Wedges, slopes, and link removals}
A recurrent topic in this article is cascade properties of random interface models
(branches in the uniform spanning tree in Section~\ref{sec: applications to UST}
and multiple SLEs in Section~\ref{sec: applications to SLEs}),
with the following interpretation. Suppose that $N$ interfaces connect $2N$
boundary points $p_1 , \ldots, p_{2N}$ in a planar domain according to a link
pattern $\alpha \in \LP_N$, which contains a link $\link{j}{j+1} \in \alpha$ between
two consecutive points.
If we let the endpoints $p_j$ and $p_{j+1}$ of the corresponding curve approach
each other, then the other $N-1$ random curves are described by the random interface model 
in which the curves connect the remaining points
$p_1, \ldots, p_{j-1},p_{j+2},\ldots p_{2N}$ according to
a link pattern obtained from $\alpha$ by removing the link $\link{j}{j+1}$.
The rest of this section considers the combinatorics of such link removals.

If a link pattern $\alpha \in \LP_N$ has a link $\link{j}{j+1} \in \alpha$
of the above kind, then in the corresponding balanced parenthesis expression $\alpha \in \BPE_N$,
the $j$:th and $(j+1)$:st parentheses form a matching pair, i.e., we have
$\alpha = \BPEfont{X()Y}$ for some parenthesis sequences
$\BPEfont{X}$ and $\BPEfont{Y}$ of lengths $j-1$ and $2N-j-1$. 
Then, we denote by $\alpha \removeupwedge{j} = \BPEfont{XY} \in \BPE_{N-1}$ the
balanced parenthesis expression with this one matching pair 
removed. As usual, we use the same notation for link patterns and Dyck paths,
and call the operation $\alpha \mapsto \alpha \removeupwedge{j}$
the link removal of $\alpha$.
Figure~\ref{fig: link removal} illustrates this with all three equivalent
combinatorial objects.

\begin{figure}[h!]
\includegraphics[width = 0.35\textwidth]{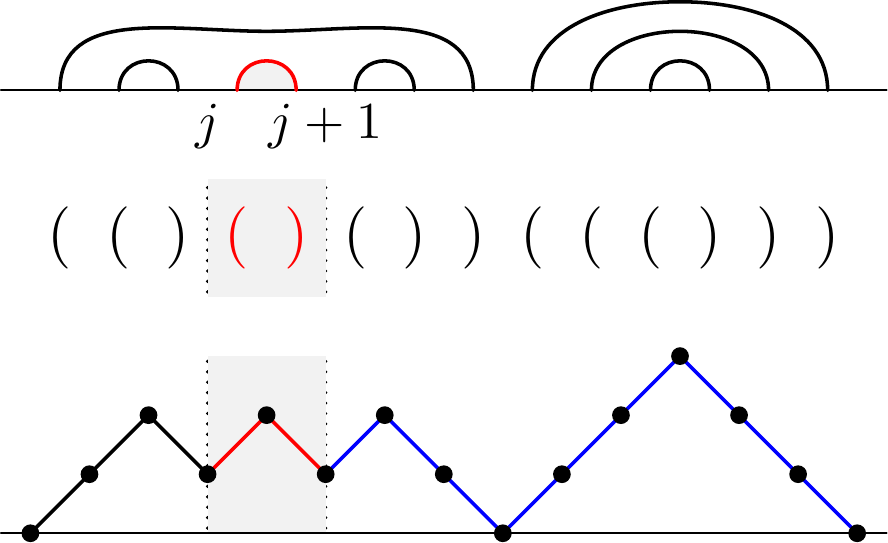} \hspace{2cm}
\includegraphics[width = 0.35\textwidth]{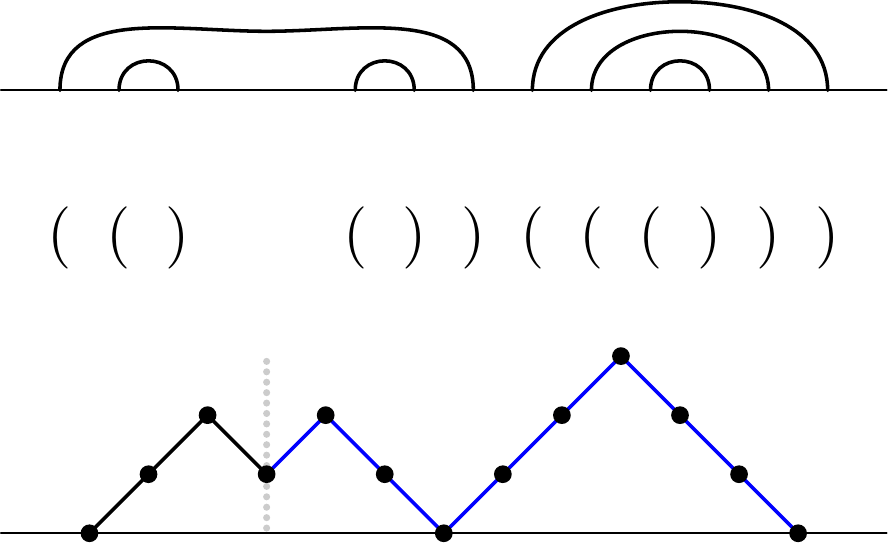}
\caption{\label{fig: link removal} 
Link removal and its interpretation in terms of parentheses and Dyck paths.}
\end{figure}

We usually formulate our combinatorial results in terms of Dyck paths. 
Then, a link between $j$ and $j+1$ corresponds with
an up-step followed by a down-step, so $\link{j}{j+1} \in \alpha \in \LP_N$
is equivalent to $j$ being a local maximum of the Dyck path $\alpha \in \DP_N$.
In this situation, we say that $\alpha$ has an \textit{up-wedge} at $j$ and denote 
$\upwedgeat{j} \in \alpha$. 
\textit{Down-wedges} $\downwedgeat{j}$ are defined analogously, and an unspecified 
local extremum is called a \textit{wedge} $\wedgeat{j}$. Otherwise, we say that $\alpha$ 
has a \textit{slope} at $j$, denoted by $\slopeat{j} \in \alpha$.

For Dyck paths, the link removal $\alpha \mapsto \alpha \removeupwedge{j}$
could alternatively be called an up-wedge removal, and 
one can define a completely analogous down-wedge removal $\alpha \mapsto \alpha \removedownwedge{j}$.
Occasionally, it is not important to specify the type of wedge that is removed, so
whenever $\alpha$ has either type of local extremum at $j$, we denote by $\alpha \removewedge{j} \in \DP_{N-1}$ the two steps
shorter Dyck path obtained by removing the two steps around that local extremum.
Wedge removals are depicted in Figure~\ref{fig: wedge removal}.

\begin{figure}[h!]
\includegraphics[width = 0.35\textwidth]{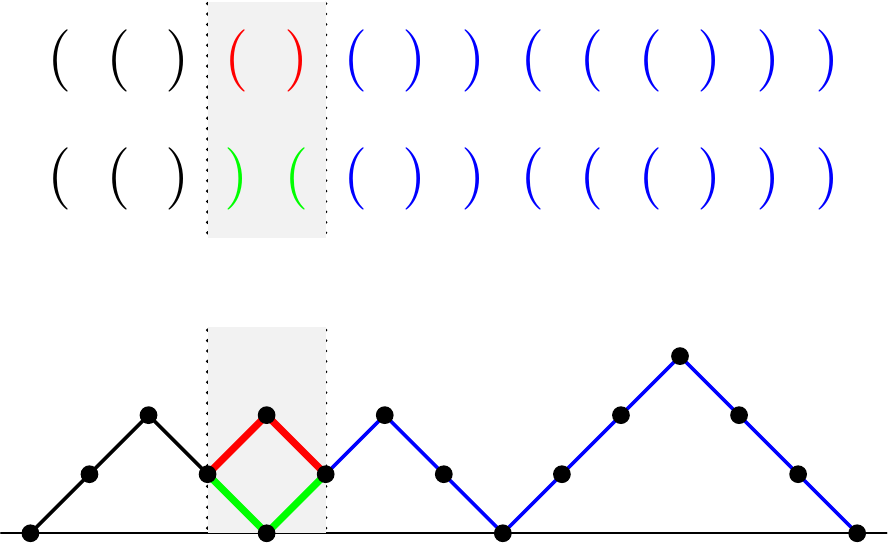}  \hspace{2cm}
\includegraphics[width = 0.35\textwidth]{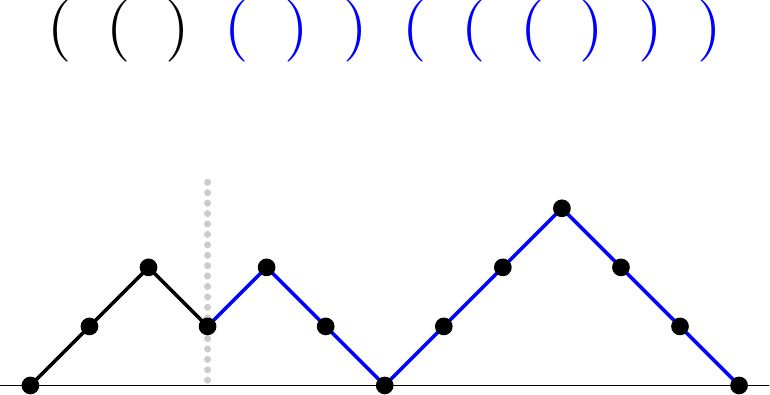}
\caption{\label{fig: wedge removal} Wedge removal.}
\end{figure}

Finally, when $\alpha$ has a down-wedge,
$\downwedgeat{j} \in \alpha$, we define the \textit{wedge-lifting operation} 
$\alpha \mapsto \alpha \wedgelift{j}$ by letting $\alpha \wedgelift{j}$ be the 
Dyck path obtained by converting the down-wedge $\downwedgeat{j}$ in $\alpha$ 
into an up-wedge $\upwedgeat{j}$. Note that for the  balanced
parenthesis expression  $\alpha$, the property $\downwedgeat{j} \in \alpha$ is equivalent to that 
the $j$:th and $(j+1)$:st parentheses are \BPEfont{)(}, and a wedge-lift converts 
them to a matching pair \BPEfont{()}.

The following two lemmas will be needed later in this section.
\begin{lem}
\label{lem: pairs}
Assume that $\upwedgeat{j} \not \in \alpha$ and $\downwedgeat{j} \in \beta$. 
Then, we have $\alpha \DPleq \beta$ if and only if $\alpha \DPleq \beta \wedgelift{j}$.
\end{lem}
\begin{proof}
Drawing the Dyck paths, the assertion is immediate. 
\end{proof}

A useful reinterpretation of the above lemma is that if 
we have $\upwedgeat{j} \not \in \alpha$, then the Dyck paths $\beta$ such that
$\beta \DPgeq \alpha$ and $\wedgeat{j} \in \beta$ come in pairs,
one containing an up-wedge and the other a down-wedge at $j$.

The next lemma states in what sense the parenthesis reversal relation is preserved under wedge removals.

\begin{lem}
\label{lem: main wedge lemma}
Assume that $\upwedgeat{j} \in \beta$. Then, we have
$\alpha \KWleq \beta$ if and only if $\wedgeat{j} \in \alpha$ and $\alpha \removewedge{j} \KWleq \beta \removeupwedge{j}$.
\end{lem}

\begin{proof}
The $j$:th and $(j+1)$:st parentheses are a matching pair \BPEfont{()} in $\beta$. 
When $\alpha \KWleq \beta$, then the reversal of matching pairs
either leaves the $j$:th and $(j+1)$:st parentheses unchanged as \BPEfont{()},
or reverses them to \BPEfont{)(}.
In both cases, $\alpha$ contains the wedge $\wedgeat{j}$ and $\alpha \removewedge{j}$ 
can be defined. All other matching pairs of parentheses in $\beta$ correspond 
bijectively with the matching pairs of $\beta \removeupwedge{j}$.
\end{proof}

\subsection{\label{subsec: cascades}Cascades of weighted incidence matrices}
In this section, we establish a 
characterization of weighted incidence matrices
by a recursion property under wedge removals, that will be needed in~\cite{KKP-companion}.
This property holds 
whenever the weights $w(t)$ of the Dyck tiles $t$ only depend on the height $h_t$ of the tile:
$w(t) = f(h_t)$ for some function $f \colon \Z_+ \to \C$.

Cascade Recursion~\eqref{eq: recurrence} captures how
weighted incidence matrix elements change under the removal
of a wedge.
Note that, for fixed $j \in \set{1,\ldots,2N-1}$,
the wedge removal gives rise to the natural bijection 
$\beta \mapsto \beta \removeupwedge{j}$ between 
the elements of $\DP_N$ containing the up-wedge $\upwedgeat{j}$, 
and the elements of $\DP_{N-1}$.
Furthermore, by Lemma \ref{lem: main wedge lemma},
the parenthesis reversal relation is preserved under wedge removals:
if $\hat{\alpha} = \alpha \removewedge{j}$ and
$\hat{\beta} = \beta \removeupwedge{j}$
are the wedge removals of $\alpha$ and $\beta$, then
$\alpha \KWleq \beta$ if and only if $\hat{\alpha} \KWleq \hat{\beta}$.
The Cascade Recursion
expresses the incidence matrix entry at $(\alpha,\beta)$ 
in terms of that at $(\hat{\alpha},\hat{\beta})$.



Consider a collection of matrices $( M^{(N)} )_{N \ge 1}$, with 
$M^{(N)} = \big( M^{(N)}_{\alpha, \beta} \big) \in \C^{\DP_N \times \DP_N}$.
This collection is said to satisfy the \textit{Cascade Recursion} if
for any $\alpha , \beta \in \DP_N$ and any $j \in \set{1,\ldots,2N-1}$ 
such that  $\upwedgeat{j} \in \beta$, we have
\begin{align}\label{eq: recurrence}
M_{\alpha,\beta}^{(N)} = \; & \begin{cases} 
0 & \text{if } \alpha \not \KWleq \beta \\ 
M_{\hat{\alpha},\hat{\beta}}^{(N-1)}
& \text{if } \alpha \KWleq \beta \text{ and } \upwedgeat{j} \in \alpha \\ 
- f \big( \walk(j)+1 \big) \times
M_{\hat{\alpha},\hat{\beta}}^{(N-1)} 
& \text{if } \alpha \KWleq \beta \text{ and } \downwedgeat{j} \in \alpha, 
\end{cases}
\end{align}
where we denote by $\hat{\alpha} = \alpha \removewedge{j} \in \DP_{N-1}$ and 
$\hat{\beta} = \beta \removeupwedge{j} \in \DP_{N-1}$.

\begin{lem}\label{lem: KW cascades}
The Cascade Recursion~\eqref{eq: recurrence} has a unique solution 
$ ( M^{(N)} )_{N \ge 1}$ with the initial condition $M^{(1)} = 1$,
given by a weighted incidence 
matrix with tile weights determined by heights:
\begin{align}\label{eq: matrix solving the recursion}
M^{(N)}_{\alpha, \beta} = \begin{cases}
\prod_{t \in \nestedtilingof (\alpha / \beta)} (-f(h_t)) & \text{if }\alpha \KWleq \beta\\
0 & \text{otherwise}.
\end{cases}
\end{align}
\end{lem}
\begin{proof}
Suppose first that $( M^{(N)} )_{N \ge 1}$ and $( \widetilde{M}^{(N)} )_{N \ge 1}$
are two solutions to the Cascade Recursion~\eqref{eq: recurrence}.
We show by induction on $N$ that $M^{(N)} = \widetilde{M}^{(N)}$.
The case $N=1$ is just the initial condition $M^{(1)} = 1 = \widetilde{M}^{(1)}$
of the recursion. Assume then that the matrices
$M^{(N-1)} = \widetilde{M}^{(N-1)}$ coincide. Then, for any $\alpha,\beta \in \DP_N$,
choosing $j$ such that $\upwedgeat{j} \in \beta$ (such $j$ always exists),
it follows from the recursion~\eqref{eq: recurrence} that 
$M^{(N)}_{\alpha, \beta} - \widetilde{M}^{(N)}_{\alpha, \beta} = 0$.
Thus, the solution to the recursion~\eqref{eq: recurrence} is necessarily unique.

It remains to prove that the matrix~\eqref{eq: matrix solving the recursion} satisfies 
the recursion~\eqref{eq: recurrence}. 
The initial condition $M^{(1)} = 1$ is obviously satisfied. 
Fix $\alpha, \beta \in \DP_N$, and let $j \in \set{1,\ldots,2N-1}$
be such that $\upwedgeat{j} \in \beta$. 
We may assume that $\alpha \KWleq \beta$. 
Then, by Lemma~\ref{lem: main wedge lemma}, we have 
$\wedgeat{j} \in \alpha$ and 
$\alpha \removewedge{j} \KWleq \beta \removeupwedge{j}$.
%

Suppose first that $\upwedgeat{j} \in \alpha$. 
Then, 
the nested tilings $\nestedtilingof (\alpha / \beta)$ and
$\nestedtilingof (\hat{\alpha} / \hat{\beta})$
of the skew shapes $\alpha / \beta$ and 
$\hat{\alpha} / \hat{\beta}$, respectively, contain equally many tiles 
and the heights of the tiles are equal as well, 
see Figure~\ref{fig: matrix recursion lemma}. 
From this, we immediately get the asserted recursion~\eqref{eq: recurrence}
in the case $\alpha \KWleq \beta$ and $\upwedgeat{j} \in \alpha$:
\begin{align}\label{eq: recursion with up wedge}
M^{(N)}_{\alpha, \beta} 
= \prod_{t \in \nestedtilingof (\alpha / \beta)} (-f(h_t))
= \prod_{t \in \nestedtilingof (\hat{\alpha} / \hat{\beta})} (-f(h_t))
= M^{(N-1)}_{\hat{\alpha}, \hat{\beta}} .
\end{align}

\begin{figure}
\includegraphics[width = 0.3\textwidth]{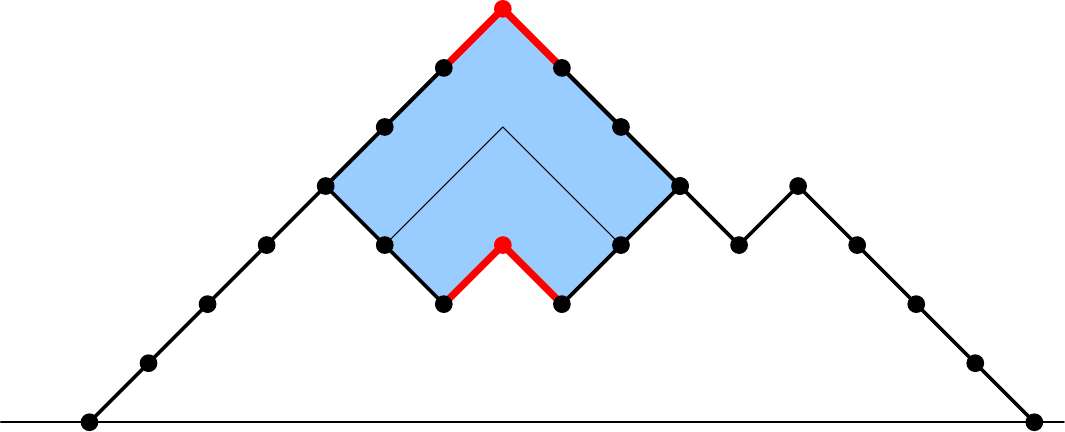}
\quad 
\includegraphics[width = 0.3\textwidth]{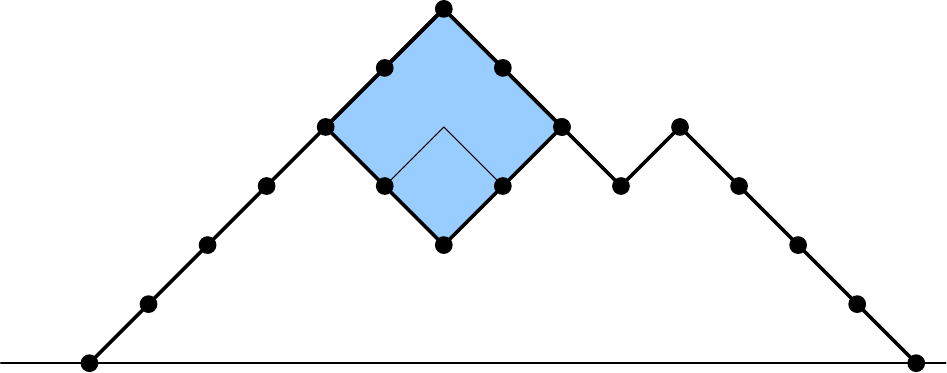}
\quad 
\includegraphics[width = 0.3\textwidth]{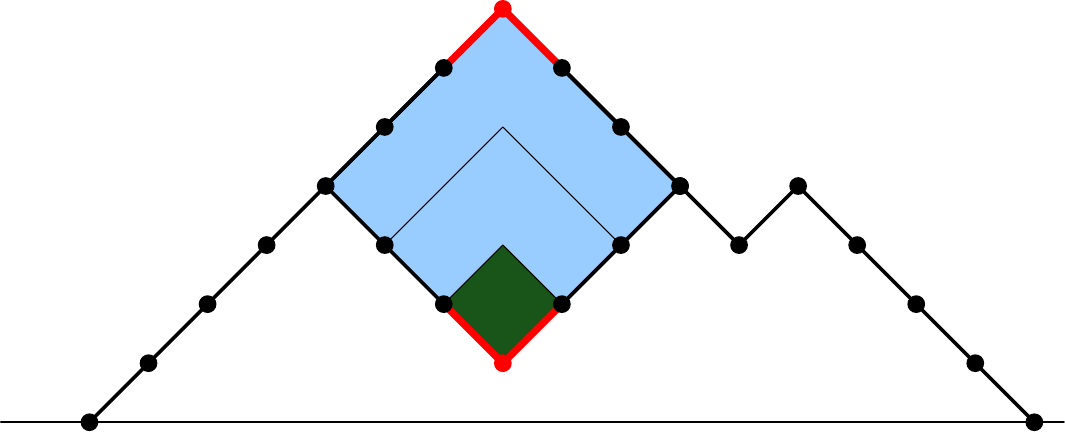}
\caption{\label{fig: matrix recursion lemma} 
The nested tilings {$\nestedtilingof (\alpha / \beta)$} 
and {$\nestedtilingof (\hat{\alpha} / \hat{\beta})$} 
of the skew shapes $\alpha / \beta$ and $\hat{\alpha} / \hat{\beta}$ 
in the cases when $\upwedgeat{j} \in \alpha$ 
(left and middle figure, respectively), and 
$\downwedgeat{j} \in \alpha$ (right and middle figure, respectively).}
\end{figure}

Suppose then that $\downwedgeat{j} \in \alpha$. In this case, the 
nested tiling $\nestedtilingof (\alpha / \beta)$ contains an atomic square 
tile $t_0$ at position $(x_{t_0},h_{t_0}) = (j,\alpha(j)+1)$ at the bottom of
$\nestedtilingof (\alpha / \beta )$,
as illustrated 
in Figure~\ref{fig: matrix recursion lemma}.
Removing the tile $t_0$ from $\nestedtilingof (\alpha / \beta)$,
we obtain the nested tiling 
$\nestedtilingof ((\alpha \wedgelift{j}) / \beta)$ of the skew shape 
$(\alpha \wedgelift{j}) / \beta$, that is,
$\nestedtilingof (\alpha / \beta) 
= \nestedtilingof ((\alpha \wedgelift{j}) / \beta) \cup \set{t_0}$.
Using this observation, the identity $h_{t_0} = \alpha(j) + 1$, and 
Equation~\eqref{eq: matrix solving the recursion}, we get 
\begin{align*}
M^{(N)}_{\alpha, \beta} 
= \prod_{t \in \nestedtilingof (\alpha / \beta)} (-f(h_t))
= -f( h_{t_0} ) \times  
\prod_{t \in \nestedtilingof ((\alpha \wedgelift{j}) / \beta)} (-f(h_t))
= -f( \alpha(j) + 1 ) \times M^{(N)}_{\alpha \wedgelift{j}, \beta} .
\end{align*}
To obtain the asserted recursion~\eqref{eq: recurrence}
in the case $\alpha \KWleq \beta$ and $\downwedgeat{j} \in \alpha$, it remains to note 
that Equation~\eqref{eq: recursion with up wedge} with $\upwedgeat{j} \in \alpha \wedgelift{j}$ gives
\begin{align*}
M^{(N)}_{\alpha \wedgelift{j}, \beta}
= M^{(N-1)}_{\widehat{\alpha \wedgelift{j}}, \hat{\beta}}
= M^{(N-1)}_{\hat{\alpha}, \hat{\beta}}.
\end{align*}
\end{proof}

The recursion~\eqref{eq: recurrence} can equivalently be cast in the following form, 
only referring to the local structure of the Dyck paths.
\begin{lem}\label{cor: recursion for matrix elements}
The Cascade Recursion relations~\eqref{eq: recurrence}
are equivalent to the following linear recursion relations: 
for any $N$, any $\alpha, \beta \in \DP_N$, and any
$j \in \set{1,\ldots,2N-1}$ such that 
$\upwedgeat{j} \in \beta$, we have
\begin{align}\label{eq: equivalent recursion}
M_{\walk,\beta}^{(N)} = \begin{cases} 
0 & \text{if } \slopeat{j} \in \walk \\
M_{\hat{\walk},\hat{\beta}}^{(N-1)}
& \text{if } \upwedgeat{j} \in \walk \\
- f(\walk(j)+1) \times
M_{\hat{\walk},\hat{\beta}}^{(N-1)} 
& \text{if } \downwedgeat{j} \in \walk,
\end{cases} 
\end{align}
where we denote by $\hat{\walk} = \walk \removewedge{j} \in \DP_{N-1}$ and 
$\hat{\beta} = \beta \removeupwedge{j} \in \DP_{N-1}$.
\end{lem}
\begin{proof}
If $\alpha \KWleq \beta$, then $\wedgeat{j} \in \alpha$ by Lemma~\ref{lem: main wedge lemma},
so the content of Equations~\eqref{eq: recurrence} and~\eqref{eq: equivalent recursion}
is the same. If $\alpha \not\KWleq \beta$, it suffices to show
that~\eqref{eq: equivalent recursion}
implies $M^{(N)}_{\alpha,\beta} = 0$. 
In that case,
by Lemma~\ref{lem: main wedge lemma}, we either have 
$\wedgeat{j} \in \alpha$ and 
$\walk \removewedge{j} \not \KWleq \beta \removeupwedge{j}$, 
or $\slopeat{j} \in \alpha$. In both cases, the relations 
\eqref{eq: equivalent recursion} imply $M^{(N)}_{\walk, \beta} = 0$
--- the latter case is a defining property, and the former follows by induction on~$N$.
\end{proof}

\subsection{\label{subsec: inverse Fomin sums}Inverse Fomin type sums}

Let $\mathfrak{K} \colon \set{1,\ldots,2N} \times \set{1,\ldots,2N} \to \bC$
be a symmetric kernel: $\mathfrak{K}(i,j) = \mathfrak{K}(j,i)$. Let 
$\beta \in \LP_N$ be a link pattern and
$((a_\ell, b_\ell))_{\ell = 1}^N$ its left-to-right orientation, i.e.,
$a_1 < a_2 < \cdots < a_N$ and $a_\ell < b_\ell$ for all $\ell$.
We define the determinant of $\beta$ with the kernel $\mathfrak{K}$ as
\begin{align}
\label{eq: LPdet with general kernel}
\Delta^{\mathfrak{K}}_\beta := 
\det \Big( \mathfrak{K}(a_k, b_\ell) \Big)_{k,\ell=1}^N.
\end{align}
This makes sense even if the diagonal entries $\mathfrak{K}(i,i)$
of the kernel are not defined, since they do not appear in the determinants.
For a link pattern $\alpha$, we set
\begin{align}\label{eq: inverse Fomin-type sum}
\mathfrak{Z}^{\mathfrak{K}}_\alpha 
:= \sum_{ \beta \DPgeq \alpha} \# \CItilingsof (\alpha / \beta) \, \Delta^{\mathfrak{K}}_\beta ,
\end{align}
where $\# \CItilingsof (\alpha / \beta)$ is the number of cover-inclusive Dyck tilings
of the skew Young diagram~$\alpha/\beta$.
We call $\mathfrak{Z}^{\mathfrak{K}}_\alpha$ the \textit{inverse Fomin type sum} associated to $\alpha$.
We will see in Section~\ref{sec: applications to UST} that sums of this type give
connectivity probabilities in the uniform spanning tree as well as boundary visit probabilities of
the loop-erased random walk,
ultimately by virtue of Fomin's formula~\cite{Fomin-LERW_and_total_positivity}.
In the rest of this section, we prove properties of the inverse Fomin type 
sums~\eqref{eq: inverse Fomin-type sum} for the later purpose of analyzing the asymptotics and 
scaling limits of these probabilities.


We first prove that the coefficients $\# \CItilingsof (\alpha / \beta) $ of 
the determinants $\Delta^{\mathfrak{K}}_\beta$ in the inverse Fomin type sum
$\mathfrak{Z}^{\mathfrak{K}}_\alpha $
have the same value for the pairs 
of Dyck paths described after Lemma~\ref{lem: pairs}.

\begin{lem}\label{lem: numofC}
Assume that $\upwedgeat{j} \not \in \alpha$, $\downwedgeat{j} \in \beta$, and $\alpha \DPleq \beta$. 
Then we have 
$\# \CItilingsof (\alpha / \beta) = 
\# \CItilingsof \big(\alpha / (\beta \wedgelift{j}) \big)$.
\end{lem}
\begin{proof}
The equality of the cardinalities is shown by giving a bijection
between the sets $\CItilingsof (\alpha / \beta )$
and $\CItilingsof \big(\alpha / (\beta \wedgelift{j}) \big)$ of 
cover-inclusive Dyck tilings of the two
skew Young diagrams. The only difference between the diagrams is that
$\alpha / (\beta \wedgelift{j})$ contains exactly one atomic square
more than $\alpha / \beta$.
The bijection is defined by adding to a tiling of $\alpha / \beta$
the tile formed by this atomic square.

Clearly such extensions of tilings in $\CItilingsof (\alpha / \beta )$
produce $\# \CItilingsof (\alpha / \beta )$ distinct elements of the set
$\CItilingsof \big( \alpha / (\beta \wedgelift{j}) \big)$, so it remains to
prove that all cover-inclusive tilings of $\alpha / (\beta \wedgelift{j})$
must have an atomic square tile at the lifted square. Consider a Dyck
tiling $S$ of  $\alpha / (\beta \wedgelift{j})$ not satisfying this property.
Then, the tile covering the lifted square contains at least a ``three-square
$\Lambda$-shape''  growing down from the lifted square. 
In order for $S$ to be cover-inclusive, also all tiles  below the
``three-square $\Lambda$-shape''
would have to contain a lowered ``three-square $\Lambda$-shape''. Stacking
such shapes until
they touch $\alpha$ as in Figure~\ref{fig: beta lifted},
we notice that $S$ could only be cover-inclusive if $\upwedgeat{j} \in \alpha$,
which is ruled out by our assumption $\upwedgeat{j} \notin \alpha$.
\end{proof}
\begin{figure}[h!]
\includegraphics[width = 0.35\textwidth]{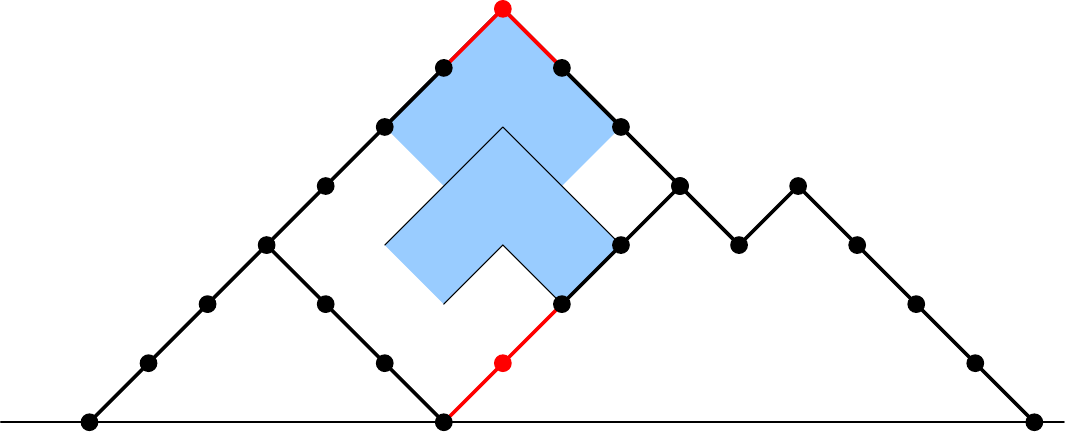}
\caption{\label{fig: beta lifted} 
Illustration of a situation in the proof of Lemma~\ref{lem: numofC}
where a Dyck tiling of the skew-shape {$\alpha / (\beta \wedgelift{j})$} would
not have an atomic square tile at the lifted square. 
The Dyck tiling cannot be cover-inclusive, because 
{$\upwedgeat{j} \not \in \alpha$}.}
\end{figure}


We then prove that the determinants $\Delta^{\mathfrak{K}}_\beta $ in 
the inverse Fomin type sum~\eqref{eq: inverse Fomin-type sum}
do not change too much in the wedge-lifting operation either.
\begin{lem}
\label{lem: wedge-lifts in link patterns}
Let $((a_\ell, b_\ell))_{\ell =1}^N$ be the left-to-right orientation of
the link pattern $\beta \in \LP_N$, and let $\downwedgeat{j} \in \beta$, 
so that $j = b_s$ and $j+1 = a_r$ for some 
$s < r$. Let $((a_\ell', b_\ell'))_{\ell=1}^N$ be the left-to-right orientation 
of $\beta \wedgelift{j}$. Then we have
\[
a_k' = \begin{cases} a_k & \text{ for } k \neq r 
\\ b_s & \text{ for } k = r \end{cases}
\qquad \text{ and } \qquad 
b_\ell' = \begin{cases} b_\ell & \text{ for } \ell \neq r, s \\ 
a_r & \text{ for } \ell = r \\ 
b_r & \text{ for } \ell = s \end{cases} .
\]
%
\end{lem}

\begin{proof}
In terms of parenthesis expressions, 
the $j$:th and $(j+1)$:st parentheses of $\beta$ read \BPEfont{)(}.
Writing out their matching pairs, we see that $\beta$ contains the balanced subexpression 
\BPEfont{(X)(Y)}, which converts to \BPEfont{(X()Y)}
in $\beta \wedgelift{j} $, while everything else remains unchanged.
Recalling that matching pairs of a parenthesis expression correspond to links of
a link pattern,
the assertion is immediate from Figure~\ref{fig: wedge-lift}.

\begin{figure}[h!]
\includegraphics[width = 0.27\textwidth]{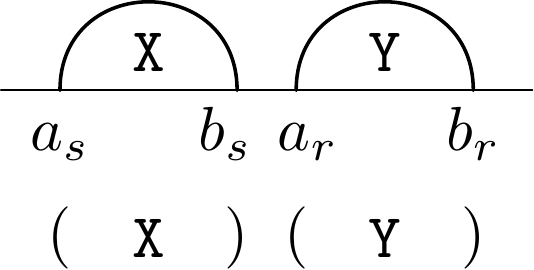} \hspace{2cm}
\includegraphics[width = 0.27\textwidth]{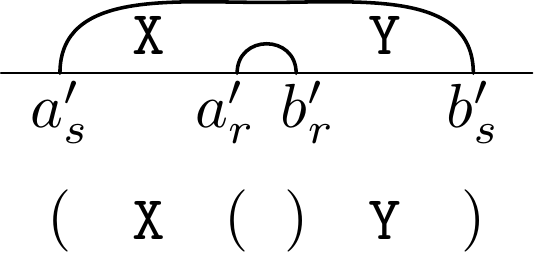}
\caption{\label{fig: wedge-lift} 
The $s$:th and $r$:th links in left-to-right orientations 
before and after a wedge-lift.}
\end{figure}
\end{proof}

\begin{rem}\label{rem: interpretation before Zero-replacing Rule}
\textit{
This lemma has an important interpretation in terms of the 
determinants with a symmetric kernel $\mathfrak{K}$.
Compare the two determinants $\Delta^{\mathfrak{K}}_\beta$ and 
$\Delta^{\mathfrak{K}}_{ \beta \wedgelift{j} }$, and interchange 
the $r$:th and $s$:th columns in the
matrix in the determinant $\Delta^{\mathfrak{K}}_{ \beta \wedgelift{j} }$.
Then, the resulting matrix and the matrix in $\Delta^{\mathfrak{K}}_\beta $
only differ in the $r$:th row and the $s$:th column, 
containing the kernel entries
depending on $j$ and $j+1$. Explicitly, the resulting determinants read
\begin{align*}
\LPdet{\beta}{\mathfrak{K}}
= \det \left( \begin{array}{c c c}
\ddots & \mathfrak{K}(a_k, j )_{k=1}^{r - 1} & \ddots \\[.225cm]
\mathfrak{K} ({j + 1}, b_\ell)_{\ell=1}^{s - 1} & \mathfrak{K}(j+1, j) & \mathfrak{K} (j + 1, b_\ell)_{\ell=s+1}^{N}  \\
\ddots & \mathfrak{K} (a_k, j)_{k=r + 1}^{N} & \ddots 
\end{array}
 \right)
\end{align*}
and
\begin{align*}
\LPdet{\beta \wedgelift{ { j }}}{\mathfrak{K}}   
=  - \det \left( \begin{array}{c c c}
\ddots & \mathfrak{K}(a_k, j+ 1)_{k=1}^{r - 1} & \ddots \\[.225cm]
\mathfrak{K}(j, b_\ell)_{\ell=1}^{s - 1} & \mathfrak{K}(j, j +1) & \mathfrak{K} (j, b_\ell)_{\ell=s+1}^{N}  \\
\ddots & \mathfrak{K}(a_k, j +1)_{k=r + 1}^{N} & \ddots 
\end{array}
 \right),
\end{align*}
where the 
ellipses stand for submatrices which are identical in both cases. 
Notice also that, by the symmetry of $\mathfrak{K}$,
the entries in the middle are equal: 
$\mathfrak{K}(j, j +1) = \mathfrak{K}(j +1, j)$.
}
\end{rem}




The determinants $\LPdet{\beta}{\mathfrak{K}}$ and the inverse Fomin type sums
$\mathfrak{Z}^{\mathfrak{K}}_\alpha$, defined in~\eqref{eq: LPdet with general kernel}
and~\eqref{eq: inverse Fomin-type sum}, are polynomials in the entries $\mathfrak{K}(i,j)$
of the kernel $\mathfrak{K}$. The kernels are symmetric and the diagonal
kernel entries $\mathfrak{K}(i,i)$ do not appear in $\LPdet{\beta}{\mathfrak{K}}$ and
$\mathfrak{Z}^{\mathfrak{K}}_\alpha$,
so we can view the entries $\mathfrak{K}(i,j)$, for $i<j$, as the independent variables of these polynomials.
For notational convenience, for any $j$, we let $\mathfrak{K}(\cdot,j) = \mathfrak{K}(j, \cdot)$
stand for the collection of independent variables
$\big( \mathfrak{K}(1,j), \mathfrak{K}(2,j), \ldots, \mathfrak{K}(j-1,j) , \mathfrak{K}(j,j+1), \ldots, \mathfrak{K}(j,2N) \big)$
that involve the index~$j$.
For any $j$, the determinants $\LPdet{\beta}{\mathfrak{K}}$
are linear in the collection
$\mathfrak{K}(\cdot,j) = \mathfrak{K}(j, \cdot)$, i.e., they are of the form
\[ \LPdet{\beta}{\mathfrak{K}} = \sum_{\substack{1 \leq i \leq 2N \\ i \neq j}} \big[ \LPdet{\beta}{\mathfrak{K}} \big]_{i,j} \, \mathfrak{K}(i,j) , \]
where $[ \LPdet{\beta}{\mathfrak{K}} ]_{i,j}$ is a polynomial in the variables other than $\mathfrak{K}(i,\cdot)$
and $\mathfrak{K}(\cdot,j)$.
Below, by the coefficient of $\mathfrak{K}(i,j)$ in $\LPdet{\beta}{\mathfrak{K}}$
we mean $[ \LPdet{\beta}{\mathfrak{K}} ]_{i,j}$. Note that we have
$[ \LPdet{\beta}{\mathfrak{K}} ]_{i,j} = [ \LPdet{\beta}{\mathfrak{K}} ]_{j,i}$.
Similarly, we may define the coefficient $[ \mathfrak{Z}^{\mathfrak{K}}_\alpha ]_{i,j}$ of 
$\mathfrak{K}(i,j)$ in the inverse Fomin type sum $\mathfrak{Z}^{\mathfrak{K}}_\alpha$,
by noting that $\mathfrak{Z}^{\mathfrak{K}}_\alpha$ is a linear combination of 
determinants $\LPdet{\beta}{\mathfrak{K}}$.

\begin{prop}
\label{prop: Zero-replacing Rule}
Let $\alpha \in \LP_N$ be a link pattern, and suppose that $j \in \set{1,\ldots,2N-1}$ is 
such that~$\upwedgeat{j} \not \in \alpha$. 
Then, the following statements hold.
\begin{description}
\item[(a)] The inverse Fomin type sum $\mathfrak{Z}_\alpha$ is antisymmetric
under interchanging 
the kernel entries at $j$ and $j+1$ in the following sense: if
\begin{align*}
\begin{cases}
\widetilde{\mathfrak{K}}(j, \cdot) = \mathfrak{K} (j + 1, \cdot) & \text{and} \quad \widetilde{\mathfrak{K}}(j +1, \cdot) = \mathfrak{K} (j, \cdot) \\
\widetilde{\mathfrak{K}}(i,k) = \mathfrak{K}(i,k), & \text{for all other indices $(i,k)$},
\end{cases}
\end{align*}
then for the inverse Fomin type sums $\mathfrak{Z}^{\widetilde{\mathfrak{K}}}_\alpha$ and
$\mathfrak{Z}^{\mathfrak{K}}_\alpha$ with kernels $\widetilde{\mathfrak{K}}$ and ${\mathfrak{K}}$, respectively, we have
\begin{align*}
\mathfrak{Z}^{\widetilde{\mathfrak{K}}}_\alpha = - \mathfrak{Z}^{\mathfrak{K}}_\alpha.
\end{align*}
\item[(b)] If the collections $\mathfrak{K}(j, \cdot)$ and $\mathfrak{K}(j+1, \cdot)$ are identical, then $\mathfrak{Z}^{\mathfrak{K}}_\alpha = 0$.
\item[(c)] The coefficient of $\mathfrak{K}(j,j+1)$ in $\mathfrak{Z}^{\mathfrak{K}}_\alpha$ is zero: 
$[ \mathfrak{Z}^{\mathfrak{K}}_\alpha ]_{j,j+1} = 0$.
In particular, the replacement
\begin{align*}
\begin{cases}
{\mathfrak{K}_0}(j, j+1) = {\mathfrak{K}_0}(j+1, j) = 0 \\
{\mathfrak{K}_0}(i,k) = \mathfrak{K}(i,k), \qquad \text{for all other indices $(i,k)$},
\end{cases}
\end{align*}
of the $(j,j+1)$ entries by zero does not affect the inverse Fomin-type sum:
$\mathfrak{Z}^{{\mathfrak{K}_0}}_\alpha = \mathfrak{Z}^{\mathfrak{K}}_\alpha$.
\end{description}
\end{prop}

\begin{rem}\label{rem: iterative zero replacement}
\textit{%
Parts (a) and (c) are often applied iteratively: the modified kernels $\widetilde{\mathfrak{K}}$
and ${\mathfrak{K}_0}$ are themselves symmetric kernels, and the above replacement rules
continue to hold for them.
}
\end{rem}

\begin{proof}

Parts (b) and (c) will be obtained as rather straightforward consequences of part (a).
We prove the antisymmetry rule of part (a) by regrouping the terms of the inverse Fomin
type sum~\eqref{eq: inverse Fomin-type sum} into antisymmetric groups.
Recall from the discussion after Lemma~\ref{lem: pairs} that Dyck paths
having a wedge at $j$ appear in the sum \eqref{eq: inverse Fomin-type sum} in pairs: if $\beta$ is a Dyck path 
with a down-wedge, $\downwedgeat{{j}} \in \beta$,
and $\beta \wedgelift{{ j }}$ its wedge-lift, then $\beta \DPgeq \alpha$ is equivalent to $\beta \wedgelift{{ j }} \DPgeq \alpha$.
By Lemma~\ref{lem: numofC}, the coefficients corresponding to $\beta$ and 
$\beta \wedgelift{{ j }}$ in the inverse Fomin type sum~\eqref{eq: inverse Fomin-type sum} are equal: 
$\# \CItilingsof(\alpha / \beta) = 
\# \CItilingsof\big(\alpha / (\beta \wedgelift{{ j }})\big)$. Thus, we split the sum as
\begin{align*}
\mathfrak{Z}^{\mathfrak{K}}_\alpha = 
\sum_{\substack{\beta \DPgeq \alpha \\ \downwedgeat{{ j } } \in \beta }} \# \CItilingsof(\alpha / \beta) \, \Big( \Delta^{\mathfrak{K}}_\beta + \Delta^{\mathfrak{K}}_{\beta \wedgelift{j} } \Big) 
+     \sum_{\substack{\beta \DPgeq \alpha \\ \slopeat{ { j }} \in \beta }} \# \CItilingsof(\alpha / \beta) \, \Delta^{\mathfrak{K}}_\beta .
\end{align*}
The terms in the first sum on the right-hand side are antisymmetric under the exchange of $j$ and $j+1$,
as seen from the expressions in Remark \ref{rem: interpretation before Zero-replacing Rule} given for
$\Delta^{\mathfrak{K}}_\beta$ and $\Delta^{\mathfrak{K}}_{\beta \wedgelift{j} }$.
In the second sum, where $\slopeat{ { j }} \in \beta$, the endpoints $j$ and $j+1$ are
either both exits or both entrances in the 
left-to-right orientation of $\beta$. For definiteness, assume
that they are entrances. Then, in the matrix in
$\Delta^{\mathfrak{K}}_\beta$, both collections $\mathfrak{K}(j, \cdot)$ and $\mathfrak{K}(j + 1, \cdot)$
appear on a row. The determinant $\LPdet{\beta}{\mathfrak{K}}$ changes sign under the
exchange of these rows. 

For part (b), if we have $\mathfrak{K}(j, \cdot) = \mathfrak{K}(j+1, \cdot)$, then $\widetilde{\mathfrak{K}} = \mathfrak{K}$,
and $\mathfrak{Z}^{\mathfrak{K}}_\alpha$ is symmetric under interchange of indices $j$ and $j+1$.
On the other hand, it is antisymmetric by part (a), 
and must therefore vanish.

For part (c), recall first that
the coefficient $[\mathfrak{Z}_\alpha^\mathfrak{K} ]_{j, j+1}$ of $\mathfrak{K}(j, j+1)$ 
is a polynomial in the variables other than
$\mathfrak{K}(j, \cdot)$ and $\mathfrak{K}(j+1, \cdot)$. Let us evaluate the polynomial $\mathfrak{Z}_\alpha^\mathfrak{K}$ at
$\mathfrak{K}(i, j) = \delta_{i, j+1}$ and $\mathfrak{K}(i, j+1) = \delta_{i,j}$, leaving the variables of $[\mathfrak{Z}_\alpha^\mathfrak{K} ]_{j, j+1}$  undetermined. Using part (b), we obtain
\begin{align*}
0 = \mathfrak{Z}_\alpha^\mathfrak{K} = \sum_{\substack{1 \leq i \leq 2N \\ i \neq j+1}} [\mathfrak{Z}_\alpha^\mathfrak{K}]_{i, j+1} \mathfrak{K}(i, j+1) = [\mathfrak{Z}_\alpha^\mathfrak{K}]_{j, j+1}.
\end{align*}  
The property $\mathfrak{Z}_\alpha^ { \mathfrak{K}_0 } =  \mathfrak{Z}_\alpha^\mathfrak{K}$ is clear.
\end{proof}

The above result details the behavior of the inverse Fomin type sum $\mathfrak{Z}^{\mathfrak{K}}_{\alpha}$
when $\upwedgeat{j} \not \in \alpha$. We will also need the complementary case
$\upwedgeat{j} \in \alpha$ where, in terms of the link pattern $\alpha$, 
the indices $j$ and $j+1$ are connected by a link. For this purpose, we define the wedge removal of 
a symmetric kernel as follows:
as a matrix, let the kernel $\mathfrak{K}\removewedge{j} \in \bC^{2(N-1) \times 2(N-1)}$ be obtained
from $\mathfrak{K} \in \bC^{2N \times 2N}$ by removing the rows and columns $j$ and $j+1$.
Let us first calculate the coefficients of $\mathfrak{K}(j, j+1)$
in the determinants $\LPdet{\beta}{\mathfrak{K}}$.
\begin{lem}\label{lem: coefficients in LP determinants}
The coefficient of $\mathfrak{K}(j, j+1)$ in the determinant $\LPdet{\beta}{\mathfrak{K}}$ is given by
\begin{align}\label{eq: coefficients in LP determinants}
\big[ \LPdet{\beta}{\mathfrak{K}} \big]_{j,j+1} = 
\begin{cases}
\LPdet{\beta \removeupwedge{j} }{\mathfrak{K}\removewedge{j} }, & \text{if } \upwedgeat{j} \in \beta \\
- \LPdet{\beta \removedownwedge{j} }{\mathfrak{K}\removewedge{j} }, & \text{if } \downwedgeat{j} \in \beta \\
0, & \text{if } \slopeat{j} \in \beta.
\end{cases}
\end{align}
\end{lem}
\begin{proof}
Assume first that $\upwedgeat{j} \in \beta$, and
let $((a_\ell,b_\ell))_{\ell=1}^N$ be the left-to-right orientation of $\beta$.
Since there is a link $\link{j}{j+1}$ in $\beta$,
we have $j=a_s$ and $j+1 = b_s$ for some $s$. Then, applying the subdeterminant rule
 in the definition of the  determinant $\LPdet{\beta}{\mathfrak{K}}$, we have
\begin{align}\label{eq: case 1 subdeterminant}
\big[ \LPdet{\beta}{\mathfrak{K}} \big]_{j,j+1} =
(-1)^{s + s} \det \Big( \mathfrak{K}(a_k, b_\ell) \Big)_{\substack{k,\ell=1 \\ k, \ell \ne s }}^N 
= \LPdet{\beta \removeupwedge{j} }{\mathfrak{K}\removewedge{j}}.
\end{align}

Assume next that $\downwedgeat{j} \in \beta$. Applying the subdeterminant rule in
the matrices written out in Remark~\ref{rem: interpretation before Zero-replacing Rule},
the case~\eqref{eq: case 1 subdeterminant} above, and 
the fact $( \beta \wedgelift{j} ) \removeupwedge{j} =  \beta  \removedownwedge{j}$,
we observe that 
\begin{align*}
\big[ \LPdet{\beta}{\mathfrak{K}} \big]_{j,j+1} =
- \big[ \LPdet{\beta \wedgelift{j} }{\mathfrak{K}} \big]_{j,j+1} 
= - \LPdet{( \beta \wedgelift{j} ) \removeupwedge{j} }{\mathfrak{K}\removewedge{j} }
= - \LPdet{\beta \removedownwedge{j} }{\mathfrak{K}\removewedge{j} } .
\end{align*}

Finally, if $\slopeat{j} \in \beta$, then
$\mathfrak{K}(j, j+1)$ and $\mathfrak{K}(j +1, j)$ do not appear
as entries in the matrix of $\LPdet{\beta}{\mathfrak{K}}$, because
either both collections $\mathfrak{K}(j, \cdot)$ and $\mathfrak{K}(j + 1, \cdot)$ 
appear on a row, or both appear on a column.
\end{proof}

We now prove a cascade property for the inverse Fomin type sums,
see also Figure~\ref{fig: link removal}.


\begin{prop} \label{prop: inverse Fomin cascade}
Assume that $\upwedgeat{j} \in \alpha$. Then, the coefficient
of $\mathfrak{K}(j, j+1)$ in $\mathfrak{Z}^{\mathfrak{K}}_\alpha$ is
$\mathfrak{Z}^{\mathfrak{K}\removewedge{j} }_{\alpha \removeupwedge{j}}$.
\end{prop}
\begin{proof}
By Example \ref{ex: signed incidence matrix of KWleq}, the coefficients
$\# \CItilingsof (\alpha / \beta)$  in the inverse Fomin type sums~\eqref{eq: inverse Fomin-type sum}
are the entries $\Minv_{\alpha, \beta}$
of the inverse signed incidence matrix of the relation $\KWleq$. Thus, the inverse Fomin
type sums $\mathfrak{Z}_{\gamma}^{\mathfrak{K}}$ are uniquely determined by the system of equations
\begin{align}\label{eq: linear system determining inverse Fomin sum}
\LPdet{\beta}{\mathfrak{K}}  
= \sum_{\substack{\gamma \in \DP_N \\ \beta \KWleq \gamma }} \Mmat_{\beta, \gamma}
\mathfrak{Z}_{\gamma}^{\mathfrak{K}},
\end{align}
indexed by $\beta \in \DP_N$,
where $\Mmat_{\beta, \gamma} = (-1)^{\vert \nestedtilingof ( \beta / \gamma) \vert}$.
Assume that $\upwedgeat{j} \in \beta$, and denote $\hat{\beta} = \beta \removeupwedge{j}$.
Applying Equation~\eqref{eq: linear system determining inverse Fomin sum} with both $\beta$ and $\hat{\beta}$,
and using  Proposition~\ref{prop: Zero-replacing Rule}(c) and
Lemma~\ref{lem: coefficients in LP determinants} gives
\begin{align}\label{eq: manipulated linear system for inverse Fomin sum}
\sum_{\substack{\gamma \in \DP_{N} \\ \upwedgeat{j} \in \gamma \ \& \ \beta \KWleq \gamma }} \Mmat_{\beta, \gamma} \,
\big[ \mathfrak{Z}_{\gamma}^{\mathfrak{K}} \big]_{j,j+1}
= \big[ \LPdet{\beta}{\mathfrak{K}} \big]_{j,j+1} 
= \LPdet{\beta \removeupwedge{j}}{\mathfrak{K}\removewedge{j}}
= \sum_{\substack{ \hat{\gamma} \in \DP_{N-1} \\ \hat{\beta} \KWleq \hat{\gamma} }} \Mmat_{\hat{\beta}, \hat{\gamma}} \,
\mathfrak{Z}_{\hat{\gamma} }^{\mathfrak{K}\removewedge{j}} .
\end{align}

Next, notice that $\gamma \mapsto \gamma \removeupwedge{j} $ gives is a bijection between the Dyck paths $\gamma \in \DP_{N} $ such that $\upwedgeat{j} \in \gamma$ 
and $ \DP_{N-1}$,
and Lemma~\ref{lem: main wedge lemma} furthermore guarantees
that $\beta \KWleq \gamma $ is equivalent to $\hat{\beta} \KWleq \gamma \removeupwedge{j}  $.
Thus, re-indexing the sum on the left-hand side of
Equation~\eqref{eq: manipulated linear system for inverse Fomin sum} by $\hat{\gamma} = \gamma \removeupwedge{j} $
gives
\begin{align*}
\sum_{\substack{ \hat{\gamma} \in \DP_{N-1} \\ \hat{\beta} \KWleq \hat{\gamma} }} \Mmat_{\beta, \gamma} \,
\big[ \mathfrak{Z}_{\gamma}^{\mathfrak{K}} \big]_{j,j+1}
= \sum_{\substack{ \hat{\gamma} \in \DP_{N-1} \\ \hat{\beta} \KWleq \hat{\gamma} }} \Mmat_{\hat{\beta}, \hat{\gamma}} \,
\mathfrak{Z}_{\hat{\gamma} }^{\mathfrak{K}\removewedge{j}} .
\end{align*}
Observe that $\Mmat_{ {\beta}, {\gamma}} = \Mmat_{\hat{\beta}, \hat{\gamma}}$, by 
the Cascade Recursion~\eqref{eq: recurrence} with tile weight $1$.
Since the above equations hold for all $\hat{\beta} \in \DP_{N-1}$, solving the system yields
$[ \mathfrak{Z}_{\gamma}^{\mathfrak{K}} ]_{j,j+1}
    = \mathfrak{Z}_{\hat{\gamma} }^{\mathfrak{K}\removewedge{j}}$
for all $\gamma \in \DP_{N}$ such that $\upwedgeat{j} \in \gamma$.
\end{proof}


\bigskip{}

\section{\label{sec: application to UST}Uniform spanning trees and loop-erased random walks}
\label{sec: applications to UST}

We now consider planar loop-erased random walks (LERW) 
and the planar
uniform spanning tree (UST). The main results of this section are the following. 
First, we derive explicit determinantal formulas for the connectivity 
probabilities of boundary branches in the UST, see 
Theorem~\ref{thm: the formula for UST partition functions}.
These formulas are obtained using the combinatorial results of 
Sections~\ref{subsec: Combinatorial objects and bijections}--\ref{subsec: Dyck tilings}
combined with Fomin's formulas, given in Section~\ref{sec: Fomin}
--- hence the determinantal form.
Using the connectivity probabilities
and Wilson's algorithm, we obtain formulas for boundary visit probabilities of the LERW, see 
Corollary~\ref{cor: LERW boundary visit probabilities with partition functions}.

Second, we establish results concerning scaling limits of these quantities.
In Theorem~\ref{thm: scaling limit of partition functions}, we prove that
the suitably renormalized connectivity probabilities have 
explicit conformally covariant scaling limits.
We will prove in Section~\ref{sec: applications to SLEs} that
these functions satisfy a system of PDEs of second order.
The formula of 
Corollary~\ref{cor: LERW boundary visit probabilities with partition functions}
also enables us to prove a similar convergence result
(Theorem~\ref{thm: scaling limit of LERW bdry visits})
for the LERW boundary visit probabilities --- remarkably, the scaling limit satisfies
a system of PDEs of second and third order. 
The proof of this fact is postponed to
Sections~\ref{sec: determinant Taylor expansions}--\ref{sub: proof of bdry visit thm}. 

This section is organized as follows.
In Section~\ref{sec: UST definitions}, we define the UST with wired boundary conditions
and study connectivity  and boundary visit events for its branches. 
Section~\ref{sec: LERW and UST} contains the definition of the LERW
and relates it to a branch of the UST, via Wilson's algorithm.
In Section~\ref{sec: Fomin}, we recall Fomin's formulas,
which are crucial tools for explicitly solving the probabilities of interest
in Section~\ref{sec: partition functions}.
The scaling limit setup and main scaling limit results
are the topic of Section~\ref{sec: scaling limits}.
Finally, a number of further generalizations is briefly
discussed in Section~\ref{sec: generalizations}, including results
about the uniform spanning tree with free boundary conditions.

\subsection{\label{sec: graphs in planar domain}Graphs embedded in a planar domain}


\begin{figure}
\includegraphics[width=.4\textwidth]{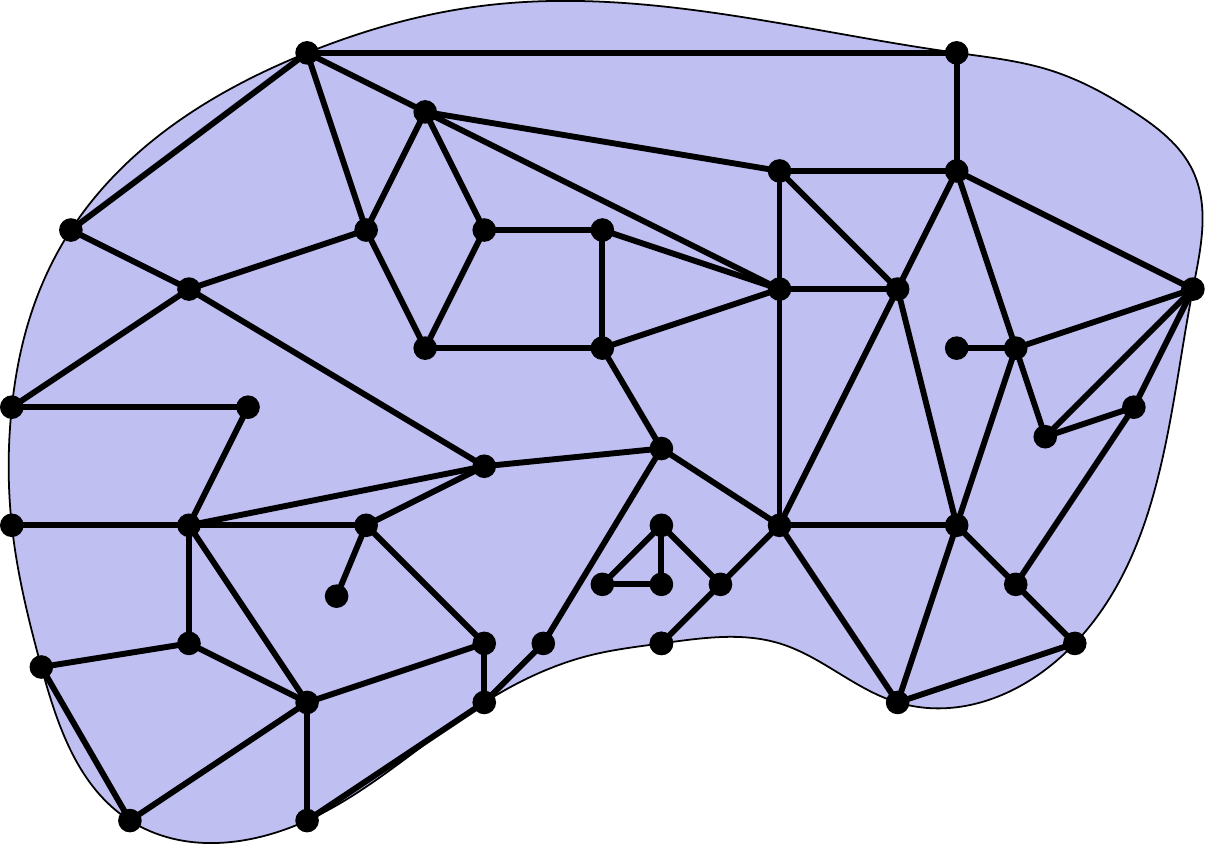}
\hspace{1cm}
\includegraphics[width=.4\textwidth]{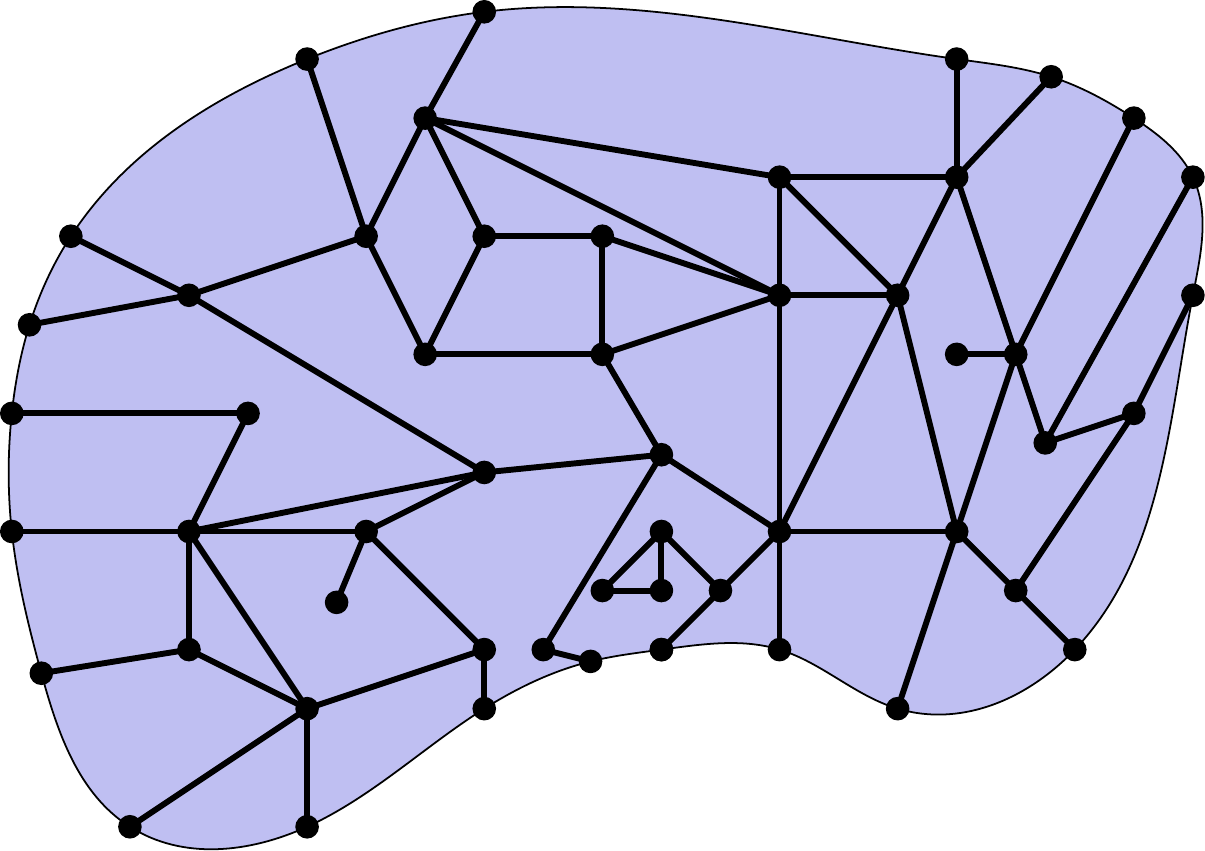}
\caption{\label{fig: graph in a Jordan domain}
Planar graphs $\Gr'$ (left) and $\Gr$ (right) embedded in a Jordan domain. 
The UST and LERW models are the same in $\Gr' / \bdry$ and $\Gr / \bdry$.}
\end{figure}

We consider a finite planar graph $\Gr = (\Vert, \Edg)$ with a non-empty subset 
of vertices $\bdry \Vert \subset \Vert$ declared as \textit{boundary vertices}.
All other vertices are called \textit{interior vertices}, and the set of them is 
denoted by $\Vert^\circ = \Vert \setminus \bdry \Vert$. The set $\bdry \Edg \subset \Edg$ of \textit{boundary edges} consists of those edges 
$e = \edgeof{e^\bdry}{e^\circ}$ which connect a boundary vertex 
$e^\bdry \in \bdry \Vert$ to an interior vertex 
$e^\circ \in \Vert^\circ$. 
We assume that the graph is embedded inside a Jordan domain in the plane in such a way that
the boundary vertices $\bdry \Vert$ are embedded on the boundary of the domain. 
We also assume that every interior vertex is connected to some boundary vertex by a path on $\Gr$, see Figure~\ref{fig: graph in a Jordan domain}.

We denote by $\GrWired$ the (multi-)graph obtained by collapsing all boundary vertices
$\bdry \Vert$ into a single vertex $v_\bdry$. The vertex set of $\GrWired$ 
is $\Vert^\circ \, \cup \, \set{v_\bdry}$,
where the single vertex $v_\bdry$ represents the collapsed boundary, and
the edges of $\GrWired$ are obtained from the edges of $\Gr$ by replacing
any boundary vertex by $v_\bdry$.
As illustrated in Figure~\ref{fig: graph in a Jordan domain}(right),
we can assume that all boundary vertices of $\Gr$ have degree one and that
no edge of $\Gr$ connects two boundary vertices, 
so that the boundary points of $\Gr$ correspond one-to-one
to the edges from $v_\bdry$ and
the (multi-)graph $\GrWired$ has no self-loops at~$v_\bdry$.




\subsection{\label{sec: UST definitions}The planar uniform spanning tree}

%
%
%
%
%

A spanning tree of a connected graph is a subgraph which is
connected and has no cycles (tree) and which contains every vertex (spanning).
Let $\Gr$ be a finite graph embedded in a planar domain as in
Section~\ref{sec: graphs in planar domain} above.
A uniformly randomly chosen spanning tree $\tree$ of the quotient graph $\GrWired$
is called a \textit{uniform spanning tree with wired boundary conditions} on $\Gr$,
below just referred to as a uniform spanning tree (UST). To facilitate the
discussion, we occasionally view $\tree$ as a collection of edges of the original 
graph $\Gr$, by the obvious identification of the edges of $\Gr$ and those of $\GrWired$.



If $v \in \Vert^\circ$ is an interior vertex, then 
there exists a unique path $\gamma_v$ in $\tree$ 
from $v$ to the boundary, i.e., a sequence 
$\gamma_v = (v_0 , v_1 , \ldots, v_\ell)$ of distinct
vertices $v_0 , v_1 , \ldots, v_\ell$ with $v_0 = v$ and $v_\ell = v_\bdry$ and
$\edgeof{v_{j-1}}{v_{j}} \in \tree$ for all $j=1,\ldots,\ell$
--- for an illustration, see Figure~\ref{fig: UST and its boundary branch from bulk point}(right).
We call $\gamma_v$ the \textit{boundary branch} from the vertex $v$,
and say that the branch reaches the boundary via the boundary
edge $\edgeof{v_{\ell-1}}{v_{\ell}} \in \bdry \Edg$.
For a given $\eout \in \bdry \Edg$, we denote by $v \rightsquigarrow \eout$
on the event $\edgeof{v_{\ell-1}}{v_{\ell}} = \eout$ that the branch $\gamma_v$ from $v$
reaches the boundary via $\eout$.

We will mostly be interested in boundary-to-boundary connectivities
of the type illustrated in Figure~\ref{fig: UST boundary branch}, where
the boundary branches from the interior vertices of boundary edges are considered.
It is convenient to label these by the boundary edge rather than its interior vertex.
For a boundary edge $\ein = \edgeof{\ein^\bdry}{\ein^\circ}$, 
we thus write simply $\gamma_{\ein}$ instead of $\gamma_{\ein^\circ}$,
and $\ein \rightsquigarrow \eout$ instead of
$\ein^\circ \rightsquigarrow \eout$. 

\subsubsection{\textbf{The connectivity partition functions}}

Let $N \in \bN$ and consider $2N$ marked distinct boundary edges 
$e_1 , \ldots, e_{2N} \in \bdry \Edg$
appearing in counterclockwise order along the boundary of the domain.
We are interested in the probability that the boundary branches
$\gamma_{e_1} , \ldots, \gamma_{e_{2N}}$ of the UST 
connect the marked boundary edges
$e_1 , \ldots, e_{2N}$ in a particular topological manner, encoded in 
a link pattern $\alpha \in \LP_{N}$ ---
see Figure~\ref{fig: connectivity} for an example.
To make precise sense of this, fix the link pattern $\alpha \in \LP_N$,
and choose an orientation $\big( (a_\ell, b_\ell) \big)_{\ell = 1}^N$ of $\alpha$.
This selects $N$ of the marked edges, $e_{a_1} , \ldots, e_{a_N}$, as
entrances and the other $N$ of the marked edges, $e_{b_1} , \ldots, e_{b_N}$, 
as exits.
By the connectivity $\alpha$ we then mean that, for all $\ell = 1 , \ldots, N$,
the boundary branch $\gamma_{e_{a_\ell}}$ from $e_{a_\ell}^\circ$ connects to 
the wired boundary via the edge $e_{b_\ell}$ (i.e., 
$\pathfromto{e_{a_\ell}}{e_{b_\ell}}$).
The partition function 
for the connectivity $\alpha$ is the probability of this event,
\begin{align}\label{eq: partition function for connectivity alpha}
Z_\alpha(e_1 , \ldots, e_{2N}) := \PR \Big[ \bigcap_{\ell=1}^N \set{\pathfromto{e_{a_\ell}}{e_{b_\ell}}} \Big] .
\end{align}
A priori, the definition of the event depends on our choice of orientation of $\alpha$
which determines the entrance points, but we will show below in
Lemma~\ref{lem: well definedness of the partition function for connectivity alpha}
that the 
probability is the same for any orientation of $\alpha$.
Note that we do not assign any connectivity
unless the boundary branches from some $N$ marked boundary edges 
connect to the boundary via the other $N$ marked boundary edges.
In particular, even the total partition function, defined as the sum of all connectivity probabilities
\[ Z(e_1 , \ldots, e_{2N}) = \sum_{\alpha \in \LP_{N}} Z_\alpha (e_1 , \ldots, e_{2N}) , \]
is typically of small order of magnitude.


\begin{lem}\label{lem: well definedness of the partition function for connectivity alpha}
The connectivity probability
\[ \PR \Big[ \bigcap_{\ell=1}^N \set{\pathfromto{e_{a_\ell}}{e_{b_\ell}}} \Big] \]
does not depend on the choice of orientation 
$\big( (a_\ell, b_\ell) \big)_{\ell = 1}^N$
of the link pattern $\alpha \in \LP_N$.
In particular, the partition function $Z_\alpha(e_1, \ldots,e_{2N})$ 
is well-defined by Equation~\eqref{eq: partition function for connectivity alpha}.
\end{lem}
\begin{figure}
\vspace{-2cm}

\includegraphics[width=.35\textwidth]{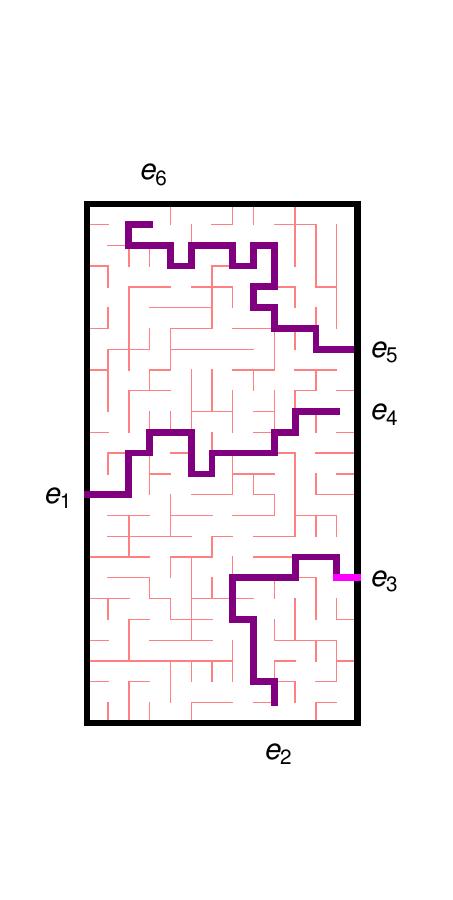}
\hspace{.5cm}
\includegraphics[width=.35\textwidth]{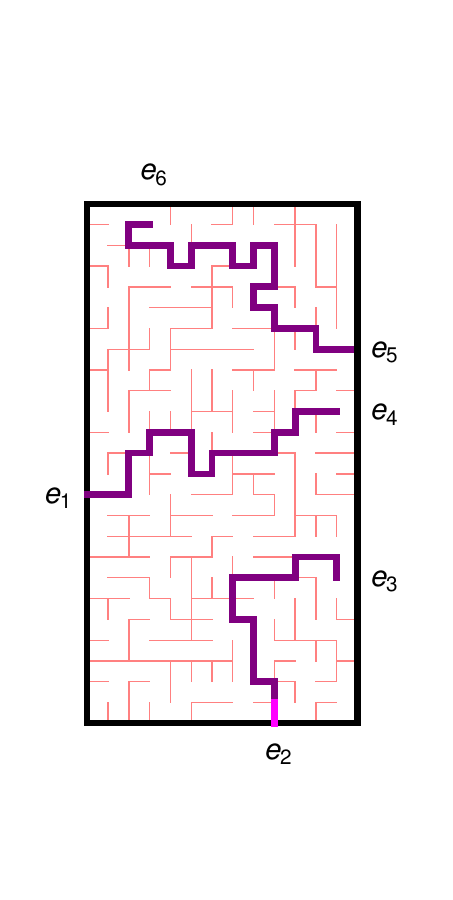}

\vspace{-1.5cm}
\caption{\label{fig: well definedness of connectivity partition functions}
Illustration of the bijection used in the proof of
Lemma~\ref{lem: well definedness of the partition function for connectivity alpha}
to show that UST boundary branch connectivity probabilities do not depend on the
orientation. Both pictures have connectivities
$\pathfromto{e_{4}}{e_{1}}$ and $\pathfromto{e_{6}}{e_{5}}$, but the connectivity
$\pathfromto{e_{2}}{e_{3}}$ on the left is changed to 
the connectivity
$\pathfromto{e_{3}}{e_{2}}$ on the right by deleting one edge and adding one in the tree.}
\end{figure}
\begin{proof}
Fix the link pattern $\alpha \in \LP_N$, and let
$\big( (a_\ell, b_\ell) \big)_{\ell = 1}^N$ and 
$\big( (a'_\ell, b'_\ell) \big)_{\ell = 1}^N$
be two orientations of $\alpha$. The order of the links does not affect 
the definition of the connectivity event, so we may assume that for each 
$\ell$, we have $\set{a_\ell , b_\ell} = \set{a'_\ell , b'_\ell}$
(in either order). Let $R$ be the 
set of those link indices $\ell$ for which the orientation of the link is 
reversed, $a_\ell = b'_\ell$ and $b_\ell = a'_\ell$.
Now, define 
the following bijection between the sets of spanning trees with these connectivities. To a tree $\tree$ for which the connectivity
$\pathfromto{e_{a_\ell}}{e_{b_\ell}}$ holds for each $\ell$, we associate the tree
\[ \Big( \tree \, \cup \, \set{e_{a_r} \, \big| \, r \in R} \Big) \setminus \set{e_{b_r} \, \big| \, r \in R} \]
obtained by deleting the exit edges and adding the entrance edges of the reversed links,
as illustrated in Figure~\ref{fig: well definedness of connectivity partition functions}.
This defines a bijection 
between the two connectivity events
\[ \bigcap_{\ell=1}^N \set{\pathfromto{e_{a_\ell}}{e_{b_\ell}}} \; \longrightarrow \;
\bigcap_{\ell=1}^N \set{\pathfromto{e_{a'_\ell}}{e_{b'_\ell}}} \]
and shows that the probabilities of the connectivity events for the
uniform spanning tree are equal.
\end{proof}

\subsubsection{\label{sec: boundary visits from partition functions}\textbf{Boundary visit probabilities of boundary branches}}

Fix two boundary edges $\ein , \eout \in \bdry \Edg$ of the graph $\Gr$.
Consider the uniform spanning tree with wired boundary conditions
on $\Gr$,
conditioned on the event $\set{\pathfromto{\ein}{\eout}}$ 
that the boundary branch $\gamma = \gamma_{\ein}$ from $\ein^\circ$ connects 
to the boundary via $\eout$.
Note that the probability of the event $\set{\pathfromto{\ein}{\eout}}$ 
on which we condition equals
\[ \PR[\pathfromto{\ein}{\eout}] = Z(\ein,\eout) , \]
the partition function of the unique link pattern with just one link.
To make a distinction, we continue to denote by $\PR$ the uniform measure on spanning trees,
and use $\PR_{\ein,\eout}$ for the conditioned measure.

We are interested in the boundary visit probabilities of $\gamma$,
i.e., the probabilities of the event that $\gamma$ contains
given edges at unit distance from the boundary, as illustrated in Figure~\ref{fig: LERW boundary visit event}.
We say that an edge $\hat{e}_s \in \Edg$ is at unit distance from the boundary
if $\hat{e}_s = \edgeof{\hat{e}_{s;1}^\circ}{\hat{e}_{s;2}^\circ}$ joins the interior vertices
$\hat{e}_{s;1}^\circ, \hat{e}_{s;2}^\circ$ of two boundary edges $\hat{e}_{s;1} , \hat{e}_{s;2} \in \bdry \Edg$,
as in Figure~\ref{fig: edge at unit distance from the boundary}.
Let $N' \in \bN$ and let $\hat{e}_1 , \ldots \hat{e}_{N'} \in \Edg$
be $N'$ edges at unit distance from the boundary,
and assume that the edges $\ein, \eout, \hat{e}_1 , \ldots \hat{e}_{N'}$ do not
have any common vertices.
We will calculate the probability
\[ \PR_{\ein,\eout} \big[ \gamma \ni \hat{e}_1 , \ldots, \hat{e}_{N'} \big] ,\]
and, even more specifically, the probability that the branch $\gamma$ visits the
edges $\hat{e}_1 , \ldots \hat{e}_{N'}$ in this order, as illustrated schematically in
Figure~\ref{fig: alpha omega}(left). 

For each of the edges $\hat{e}_s$, choose two boundary edges
$\hat{e}_{s;1}, \hat{e}_{s;2}$ so that $\hat{e}_{s}$ joins their interior vertices.
Note that due to planarity and the fact that these edges are at unit distance
from the boundary, any simple path from $\ein$ to $\eout$ 
can only traverse $\hat{e}_s$ in one possible direction.
We assume $\hat{e}_{s;1}^\circ, \hat{e}_{s;2}^\circ$ chosen so that
$\hat{e}_{s;1}^\circ$ must be visited before $\hat{e}_{s;2}^\circ$.
If $\hat{e}_s$ is on the counterclockwise boundary arc $\bdry \Gr_+$ 
from $\ein$ to $\eout$,
then the directed edge from $\hat{e}_{s;1}^\circ$ to $\hat{e}_{s;2}^\circ$ is counterclockwise
along the boundary, and if $\hat{e}_s$ is on the clockwise boundary arc $\bdry \Gr_-$ 
from $\ein$ to $\eout$, 
then the directed edge is clockwise along the boundary, see
Figure~\ref{fig: alpha omega}.
The sequence $\omega = (\omega_1 , \ldots , \omega_{N'}) \in \set{+,-}^{N'}$,
with $\omega_s = \pm$ if $\hat{e}_s$ is on $\bdry \Gr_\pm$, 
is called a \textit{boundary visit order}: it
specifies the order in which the (unordered) collection of
$N'$ edges at unit distance from the boundary is to be visited.

We now have $2 N = 2 N' + 2$ marked boundary edges
$\ein, \eout, \hat{e}_{1;1} , \hat{e}_{1;2} , \ldots , \hat{e}_{N';1} , \hat{e}_{N';2}$. 
We order these boundary edges counterclockwise along the boundary,
to obtain a sequence of boundary edges that we denote by $e_1, \ldots, e_{2N} \in \bdry \Edg$
as in the previous section. By convention, we choose to start the labeling from $e_1 = \ein$.
The boundary visit order $\omega$ determines a link pattern $\alpha(\omega) \in \LP_{N}$
as illustrated in Figure~\ref{fig: alpha omega} --- see also
\cite[Section~5.2]{KP-pure_partition_functions_of_multiple_SLEs} for a more formal definition.
The next result relates the probability of the boundary visits in the order 
$\omega$ to the partition function of the connectivity $\alpha(\omega)$.

\begin{figure}
\includegraphics[width=.4\textwidth]{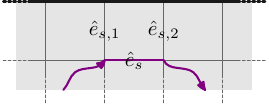}
\caption{\label{fig: edge at unit distance from the boundary}
For an edge $\hat{e}_s \in \Edg$ at unit distance from the boundary,
we associate two boundary edges $\hat{e}_{s;1} , \hat{e}_{s;2} \in \bdry \Edg$.
We choose them in such an order that any simple path from $\ein$ to $\eout$
that uses $\hat{e}_s$ will visit $\hat{e}_{s;1}^\circ$ before $\hat{e}_{s;2}^\circ$.}
\end{figure}
\begin{figure}
\includegraphics[width=.45\textwidth]{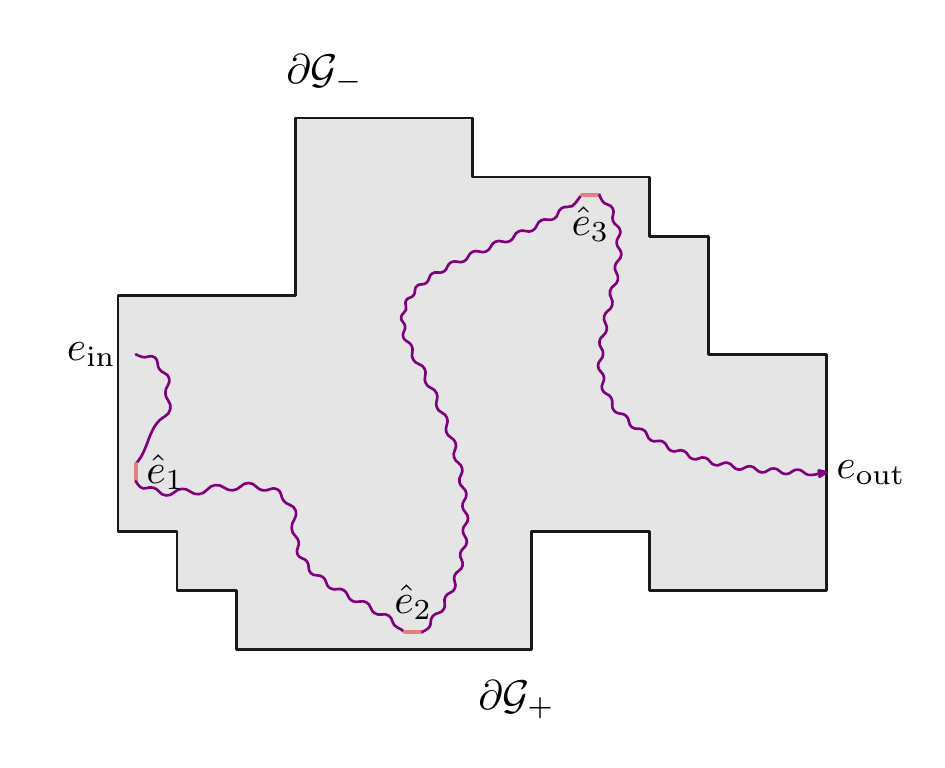}
\hspace{1cm}
\includegraphics[width=.45\textwidth]{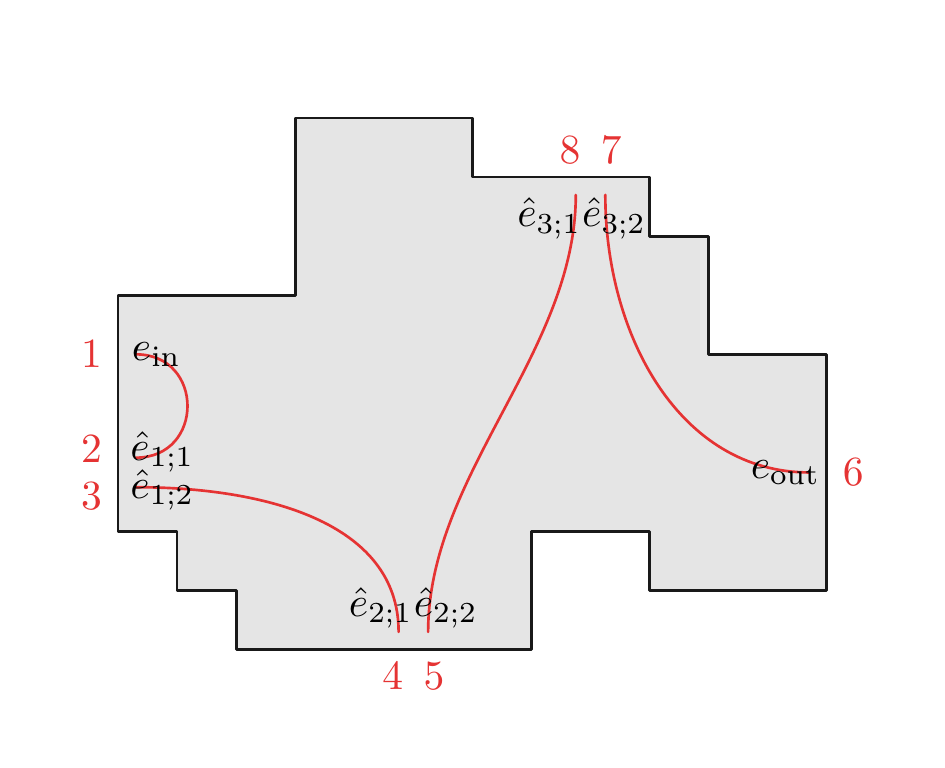}
\caption{\label{fig: alpha omega}
To any boundary visit order $\omega \in \set{+,-}^{N'}$
we associate a corresponding link pattern $\alpha(\omega) \in \LP_N$ with $N=N'+1$.
For the case illustrated in the figure, we have $\omega = (+,+,-)$ and
$\alpha(\omega) = \set{ \{ 1 , 2 \} , \{ 3 , 4 \} , \{ 5 , 8 \} , \{ 6 , 7 \} }$.
}
\end{figure}

\begin{lem}\label{lem: boundary visit probabilities with partition functions}
Let $\ein, \eout \in \bdry \Edg$ be two boundary edges and let
$\hat{e}_1 , \ldots \hat{e}_{N'}$ be edges at unit distance from the boundary,
as above. Associate to them the boundary edges $e_1, \ldots,e_{2N}$ 
and the link pattern $\alpha(\omega)$,
as above. Then, the boundary visit probability for the branch $\gamma$ from 
$\ein$ to $\eout$ is given by the ratio
\[
\PR_{\ein,\eout} \big[ \gamma \text{ uses $\hat{e}_1 , \ldots , \hat{e}_{N'}$ in this order} \big]
= \frac{Z_{\alpha(\omega)}(e_1 , \ldots, e_{2N})}{Z(\ein,\eout)}  \]
of partition functions.
\end{lem}
\begin{proof}
Recall first that the assertion concerns the uniform spanning tree
conditioned on the event ${\pathfromto{\ein}{\eout}}$
of probability $Z(\ein,\eout)$.
We again use a bijection in the uniform spanning tree.
To define the bijection, we use the link pattern $\alpha(\omega)$
and the orientation $\big( (a_\ell , b_\ell) \big)_{\ell=1}^N$ of it
which naturally corresponds to the direction that the path $\gamma$ travels.

Let $\tree$ be a tree such that we have $\pathfromto{\ein}{\eout}$
and the branch $\gamma$ contains the edges $\hat{e}_1 , \ldots \hat{e}_{N'}$.
The bijection associates to this tree the tree obtained by replacing in $\tree$ 
each edge $\hat{e}_s$ at unit distance from the boundary by the boundary edge 
$\hat{e}_{s;1}$,
see Figure~\ref{fig: multiple boundary visit bijection for a branch}.
This transformation is bijective onto the set of
spanning trees which have the connectivity $\alpha(\omega)$
for the edges $e_1 , \ldots, e_{2N}$. Therefore, we have
\[ \PR \big[ \pathfromto{\ein}{\eout} \text{ and } 
\gamma_{\ein} \text{ uses } \hat{e}_1 , \ldots \hat{e}_{N'} \text{ in this order} \big]
= \PR\Big[ \bigcap_{\ell=1}^N \set{ \pathfromto{e_{a_\ell}}{e_{b_\ell}} } \Big]
= Z_{\alpha(\omega)}(e_1 , \ldots, e_{2N}) . \]
With the conditioning, we get the asserted formula
\begin{align*}
\PR_{\ein,\eout} \big[ \gamma \text{ uses } \hat{e}_1 , \ldots \hat{e}_{N'} \text{ in this order} \big]
= \; & \frac{\PR \big[ \pathfromto{\ein}{\eout} \text{ and } \gamma_{\ein} \text{ uses } \hat{e}_1 , \ldots \hat{e}_{N'} \text{ in this order} \big]}{\PR [ \pathfromto{\ein}{\eout} ]} \\
= \; &  \frac{Z_{\alpha(\omega)}(e_1 , \ldots, e_{2N})}{Z(\ein,\eout)} .
\end{align*}
\end{proof}
\begin{figure}
\vspace{-.5cm}

\includegraphics[width=.7\textwidth]{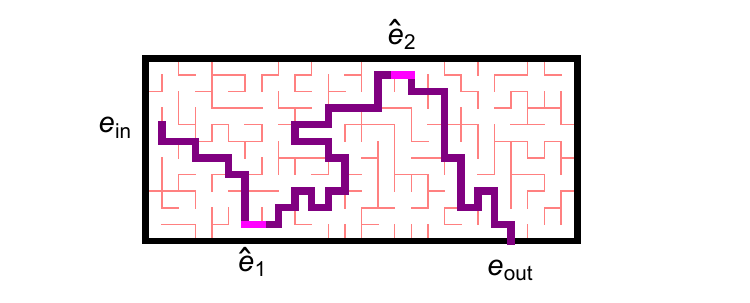}\\
\vspace{.2cm}
\includegraphics[width=.7\textwidth]{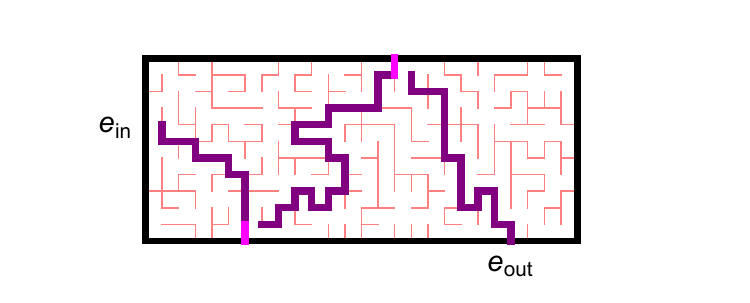}
\caption{\label{fig: multiple boundary visit bijection for a branch}
Illustration of the bijection used in
Lemma~\ref{lem: boundary visit probabilities with partition functions}
to obtain general boundary visit probabilities of a
UST boundary branch from multiple branch connectivity probabilities.}
\end{figure}

\subsection{\label{sec: LERW and UST}Relation between uniform spanning trees and loop-erased random walks}

We next turn to Wilson's algorithm, introduced
in~\cite{Wilson-generating_random_spanning_trees}
as an efficient method of sampling a uniform spanning tree.
For us, the algorithm is mainly important because with it, the relation 
of the uniform spanning tree to loop-erased random walks becomes apparent,
see also~\cite{Pemantle-choosing_a_spanning_tree_for_the_integer_lattice_uniformly}.


\subsubsection{\label{sec: discrete harmonic measures}\textbf{Random walks}}

Let $\Gr = (\Vert , \Edg)$ and $\bdry \Vert \subset \Vert$, $\bdry \Edg \subset \Edg$ be
as in Section~\ref{sec: graphs in planar domain}.
The \textit{symmetric random walk} (SRW) on the graph $\Gr$ is
the Markov process on the vertex set $\Vert$
whose transition probability from $v \in \Vert$ to $w \in \Vert$ is
\begin{align}\label{eq: symmetric random walk transition probabilities}
P_{v,w} = \begin{cases}
\frac{1}{\deg(v)} & \text{ if $\edgeof{v}{w} \in \Edg$} \\
0 & \text{ otherwise.}
\end{cases}
\end{align}
We use in particular the random walk $\eta$ started from an interior vertex
$\eta(0) = v \in \Vert^\circ$, and stopped at the first time $\tau$
at which it is on the boundary $\bdry \Vert$.
The last step of the stopped random walk $\eta = \big(\eta(t)\big)_{t=0}^\tau$
is a boundary edge $\edgeof{\eta(\tau-1)}{\eta(\tau)} \in \bdry \Edg$.
For any given boundary edge $\eout \in \bdry \Edg$, the \textit{harmonic
measure} $\HarmMeas_v(\eout)$ of $\eout$
seen from $v \in \Vert^\circ$ is the probability that 
the random walk $\eta$ started from $v$ exits via the edge $\eout$:
\[ \HarmMeas_v(\eout) = \PR\Big[ \edgeof{\eta(\tau-1)}{\eta(\tau)} = \eout \; \Big| \; \eta(0)=v \Big] . \]

If $\ein = \edgeof{\ein^\bdry}{\ein^\circ} \in \bdry \Edg$ and
$\eout = \edgeof{\eout^\bdry}{\eout^\circ} \in \bdry \Edg$ 
are two boundary edges, then, an easy path reversal argument shows
that the harmonic measure of $\eout$ seen from $\ein^\circ$ and the
harmonic measure of $\ein$ seen from $\eout^\circ$ coincide.
\begin{lem}\label{lem: symmetry of discrete excursion kernel}
For any $\ein,\eout \in \bdry \Edg$, we have 
$\HarmMeas_{\ein^\circ}(\eout) = \HarmMeas_{\eout^\circ}(\ein)$.
\end{lem}
The \textit{random walk excursion kernel} $\ExcK(\ein,\eout)$ between $\ein$ and $\eout$ is
either one of the above harmonic measures
\[ \ExcK(\ein,\eout) := \HarmMeas_{\ein^\circ}(\eout) = \HarmMeas_{\eout^\circ}(\ein) . \]
In particular, the excursion kernel is symmetric, $\ExcK(\ein,\eout) = \ExcK(\eout,\ein)$.

\subsubsection{\label{sec: Wilsons algorithm}\textbf{Loop-erased random walk and Wilson's algorithm}}

If $z = \big( z(t) \big)_{t=0}^n$ is 
a finite sequence of symbols, its \textit{loop-erasure} $\LE(z)$ is defined as the sequence
$\big( \lambda(s) \big)_{s=0}^m$ given recursively by
\begin{align*}
\lambda(0) = \; & z(0), & & \text{and for $s \geq 0$} \\  
\lambda(s+1) = \; & z(t_s + 1),
    & & \text{where } \quad
    t_s = \max \set{t \in \bZnn \; \Big| \; t > t_{s-1} \text{ and } z(t) = \lambda(s)}
\end{align*}
(interpret $t_{-1} = -1$),
and its number of steps is the smallest $m$ such that $t_m = n$.
Note that the loop-erasure 
has the same first and last symbols as the original sequence,
$\lambda(0) = z(0)$ and $\lambda(m) = z(n)$, and it
is self-avoiding in the sense that we have $\lambda(s) \neq \lambda(s')$ 
whenever $s \neq s'$.

A \textit{loop-erased random walk}
is the loop-erasure of some random walk of finitely many steps.
We consider a symmetric random walk $\eta$ as in 
Section~\ref{sec: discrete harmonic measures}.
The random walk $\eta$ is started from an interior vertex $\eta(0) = v \in \Vert$,
its transition probabilities 
\eqref{eq: symmetric random walk transition probabilities}
to all neighboring vertices are equal, 
and the walk is stopped at the (almost surely finite) first time $\tau$ at which 
the walk reaches a given non-empty set $S \subset \Vert$, for example the boundary 
$\bdry \Vert$. We define the loop-erased random walk (LERW) 
$\lambda$ from $v$ to $S$ as the loop-erasure of
$\eta = \big( \eta(t) \big)_{t=0}^\tau$, that is, $\lambda = \LE(\eta)$.
With a slight abuse of notation, we view $\lambda$ alternatively as
the list of its vertices $\big(\lambda(s)\big)_{s=0}^m$, or as the list of edges
$\big(\edgeof{\lambda(s-1)}{\lambda(s)}\big)_{s=1}^m$ it uses, or as the subgraph
formed by these vertices and edges.


The following procedure of constructing a uniform spanning tree is known as Wilson's algorithm.

\begin{thm}{{\cite{Wilson-generating_random_spanning_trees}}}
\label{thm: Wilsons algorithm}
Let $v_0 , \ldots, v_n$ be any enumeration of the set of vertices
of a finite connected graph. Define $\tree_0$ as the subgraph consisting of
only the vertex $v_0$, and recursively
for $k=1,\ldots,n$ define $\tree_k$ 
as the union of $\tree_{k-1}$ and a loop-erased random walk
from $v_k$ to $\tree_{k-1}$, independently for each $k$. Then $\tree = \tree_n$
is a uniform spanning tree of the original 
graph.
\end{thm}
\begin{figure}
\includegraphics[width=.35\textwidth]{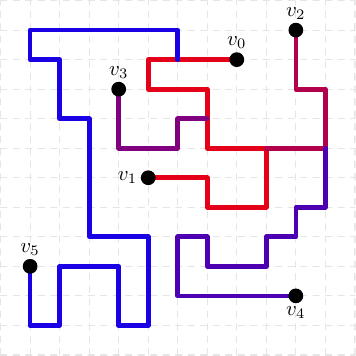}
\caption{\label{fig: Wilsons algorithm}
Wilson's algorithm (Theorem~\ref{thm: Wilsons algorithm})
constructs the uniform spanning tree 
step by step, by adding loop-erased random walks from
vertices $v_1, v_2 , \ldots$.
}
\end{figure}

To get a uniform spanning tree with wired boundary conditions, i.e.,
a uniform spanning tree of $\Gr / \bdry$, we will make
various convenient choices of enumeration of the interior vertices,
but we always use the boundary as the zeroth step in the
construction, $v_0 = v_\bdry$.

Let us now record a few direct consequences of Wilson's algorithm.
\begin{cor}\label{cor: branches of UST are LERWs}
For the uniform spanning tree with wired boundary conditions on $\Gr$, we have:
\begin{description}
\item[(a)]
The boundary branch $\gamma_v$
from $v \in \Vert$ has the law of a loop-erased random walk
from $v$ to $\bdry \Vert$.
\item[(b)]
The probability that the boundary branch $\gamma_v$
from $v \in \Vert$ reaches the boundary via a boundary edge $\eout \in \bdry \Edg$ is
the harmonic measure
\[ \PR \big[ \pathfromto{v}{\eout} \big] = \HarmMeas_v(\eout) .\]
\item[(c)]
Conditionally on the event that the boundary
branch $\gamma_v$ reaches the boundary via an edge $\eout \in \bdry \Edg$,
the law of the branch $\gamma_v$ is the loop-erasure of a symmetric
random walk from $v$ to $\bdry \Vert$ conditioned to reach $\bdry \Vert$
via the edge $\eout$.
\end{description}
\end{cor}
\begin{proof}
Construct the uniform spanning tree $\tree$ by Wilson's algorithm, 
using an enumeration of vertices such that $v_0 = v_\bdry$ and $v_1 = v$.
For part (a), note that the loop-erased random walk from $v_1=v$ to $\tree_0 = \set{v_\bdry}$
is a path in the tree $\tree = \tree_n$ from $v$ to $\bdry \Vert$,
and this unique path is the boundary branch $\gamma_v$.
For parts (b) and (c), note that the last edge used by the random walk and its
loop-erasure is the same. Therefore, the event that the branch reaches the boundary
via $\eout$ coincides with the event that the random walk reaches the
boundary via $\eout$.
\end{proof}

In particular, we can describe the boundary branch $\gamma_{\ein}$ from a
boundary edge $\ein \in \bdry \Edg$, 
by choosing $v = \ein^\circ$.
The partition function for the connectivity of just two boundary edges
$\ein , \eout \in \bdry \Edg$ is the excursion kernel
\begin{align}\label{eq: excursin kernel is the one branch connectivity proba}
Z(\ein , \eout) = \PR \big[ \pathfromto{\ein^\circ}{\eout} \big]
= \HarmMeas_{\ein^\circ}(\eout) = \ExcK(\ein , \eout) .
\end{align}

To consider multiple branches, we will crucially use the following
slight generalization. We use the notation
$\walksfromto{u}{e}$ for the set of all finite walks on the graph $\Gr$
from an interior vertex $u \in \Vert^\circ$ to the boundary via
the boundary edge $e \in \bdry \Edg$, i.e.,
\begin{align}
\label{eq: walksfromto definition}
\begin{split}
\walksfromto{u}{e} := \Big\{ (v_0 , v_1 , \ldots , v_{\ell-1}, v_{\ell}) \; \Big| \; &
    \ell \in \bN , \; v_0 = u , \; \edgeof{v_{\ell-1}}{v_\ell} = e , \\
  & \text{for all $j<\ell$ we have }  v_j \in \Vert^\circ \text{ and }
    \edgeof{v_{j-1}}{v_j} \in \Edg \Big\} .
\end{split}
\end{align}

\begin{lem}\label{lem: multiple branch connectivity probability}
Let $e_1 , \ldots, e_{N} \in \bdry \Edg$ be distinct boundary edges,
and $u_1 , \ldots, u_N \in \Vert^\circ$ distinct interior
vertices. Let also $\eta^1 , \ldots, \eta^N$ be independent symmetric random walks
started from $u_1 , \ldots , u_N$, respectively, and stopped upon
hitting the boundary. Then the probability that, in the uniform spanning tree, 
the $N$ boundary branches from $u_1 , \ldots , u_N$ connect to the boundary via 
the edges $e_1 , \ldots, e_N$, respectively, equals
\begin{align*}
\PR \big[ \bigcap_{j=1}^N \set{\pathfromto{u_j}{e_j}} \big]
= \PR \Big[ 
\walkfromto{\eta^j}{u_j}{e_j} \text{ for all $j$, and }
            \eta^j \cap \LE(\eta^i) = \emptyset \text{ for all $i<j$} \Big] .
\end{align*}
\end{lem}
\begin{proof}
The assertion becomes clear when the uniform spanning tree is
constructed by Wilson's algorithm, 
using an enumeration of vertices such that $v_0 = v_\bdry$,
$v_1 = u_1$, \ldots, $v_N = u_N$.
\end{proof}

\subsection{\label{sec: Fomin}Fomin's formula}

We now recall a formula by Fomin~\cite{Fomin-LERW_and_total_positivity},
which we will use for the calculation
of the connectivity probabilities of UST branches.

To state the original formulation of Fomin's theorem,
consider for a moment a weighted directed graph $\Gr$, with each directed
edge $(v,v')$ assigned a weight $w_{v, v'}$. To a finite
walk $\chi = (v_0 , \ldots, v_\ell)$ on the graph $\Gr$,
assign the weight $w(\chi) = \prod_{s=1}^\ell w_{\edgeof{v_{s-1}}{v_s}}$
given by the product of the weights of the edges used by the walk.
We denote $\walkfromto{\chi}{u}{v}$ if the path $\chi$
starts from $v_0=u$ and ends at $v_\ell = v$ --- this differs from~\eqref{eq: walksfromto definition}
in that the path must end at a given vertex $v$ rather than a given boundary edge $e$.
The generalized Green's
function is defined as the sum of weights of such (finite) walks,
\begin{align*}
G(u,v) = \; & \sum_{\walkfromto{\chi}{u}{v}} w(\chi) .
\end{align*}

Fomin found a formula for a determinant of Green's functions
as a sum over paths, from the starting points $u_1 , \ldots, u_N$ to the end points
$v_1 , \ldots , v_N$ in any order, subject to the requirement that
later paths do not intersect the loop-erasures of the former paths.
His proof was a clever generalization of a path-switching argument of
Karlin and McGregor~\cite{KM-coincidence_probabilities},
see also Lindstr\"om~\cite{Lindstrom-vector_representations_of_induced_matroids}
and Gessel and Viennot~\cite{GV-binomial_determinants_paths_and_hook_length_formulae}.

\begin{thm}{{\cite[Theorem~6.1]{Fomin-LERW_and_total_positivity}}}
\label{thm: Fomins general formula}
Let $u_1, \ldots, u_N$ and $v_1, \ldots, v_N$
be distinct vertices of a weighted directed graph. Then we have
\begin{align*}
\det \big[ G(u_i,v_j) \big]_{i,j=1}^N
= \sum_{\sigma \in \SymmGrp_N} \sgn(\sigma)
    \sum_{\substack{\chi^1 , \ldots , \chi^N \\ \walkfromto{\chi^i}{u_i}{v_{\sigma(i)}} \\ \forall i < j \; : \; \chi^j \cap \LE(\chi^i) = \emptyset}}
            w(\chi^1) \cdots w(\chi^N) .
\end{align*}
\end{thm}


Note that if the weights are the transition probabilities of
a random walk, $w_{v,v'} = P_{v,v'}$, then the generalized Green's function
is the usual probabilistic Green's function
$G(u,v) = \sum_{t \geq 0} \PR \big[ \eta(t)=v \, \big| \, \eta(0)=u \big]$.
We consider the symmetric random walk of 
\eqref{eq: symmetric random walk transition probabilities}
stopped upon reaching the boundary $\bdry \Vert$. For this, set
\begin{align}\label{eq: symmetric random walk weights}
w_{v,v'} = \begin{cases}
\frac{1}{\deg(v)} & \text{ if $v \in \Vert^\circ$ and $\edgeof{v}{v'} \in \Edg$} \\
0 & \text{ otherwise.}
\end{cases}
\end{align}
In particular, the weights from the boundary vertices are set to zero, to correctly
account for the stopping.
Then the harmonic measure of the boundary edge $e \in \bdry \Edg$ seen from 
an interior vertex $u \in \Vert^\circ$ can be written as a Green's function, by summing
over all possible walks:
\[ \HarmMeas_u(e) = \sum_{\walkfromto{\chi}{u}{e^\bdry}} \PR\big[ \big(\eta(t) \big)_{t=0}^\tau = \chi \big]
= \sum_{\walkfromto{\chi}{u}{e^\bdry}} w(\chi) = G(u,e^\bdry). \]
In particular, the random walk excursion kernel can be written as
\[ \ExcK(\ein,\eout) = \HarmMeas_{\ein^\circ}(\eout)
= G(\ein^\circ,\eout^\bdry) . \]

There is more than a superficial resemblance
between Theorem~\ref{thm: Fomins general formula} and
Lemma~\ref{lem: multiple branch connectivity probability}.
In fact, when the endpoints $v_1 , \ldots, v_N$ are boundary vertices
and the weights are chosen as in \eqref{eq: symmetric random walk weights},
the inner sums on the right-hand side in Theorem~\ref{thm: Fomins general formula}
are connectivity probabilities in the uniform spanning tree.
\begin{lem}\label{lem: multiple branch connectivity probability expanded}
Let $e_1 , \ldots, e_{N} \in \bdry \Edg$ be distinct boundary edges,
and $u_1 , \ldots, u_N \in \Vert^\circ$ distinct interior
vertices. Then the probability that in the uniform spanning tree, the $N$ boundary branches
from $u_1 , \ldots , u_N$ connect to the boundary via 
the edges $e_1 , \ldots, e_N$, respectively, equals
\begin{align*}
\PR \big[ \bigcap_{j=1}^N \set{\pathfromto{u_j}{e_j}} \big] =
\sum_{\substack{\chi^1 , \ldots , \chi^N \\ \walkfromto{\chi^j}{u_j}{e^\bdry_{j}} \\ \forall i < j \; : \; \chi^j \cap \LE(\chi^i) = \emptyset}}
            w(\chi^1) \cdots w(\chi^N) ,
\end{align*}
where the weights $w$ are as in \eqref{eq: symmetric random walk weights}.
\end{lem}
\begin{proof}
Let $\eta^j$ be a symmetric random walk started from $u_j$ and stopped upon reaching the boundary.
For any given path $\chi^j$ which starts from $u_j$ and ends on the boundary, we have
$\PR[\eta^j = \chi^j] = w(\chi^j)$.
The assertion now follows from Lemma~\ref{lem: multiple branch connectivity probability},
by writing the probability concerning the random walks $\eta^1 , \ldots, \eta^N$
as the sum of probabilities that the random walks take the specific trajectories $\chi^1 , \ldots, \chi^N$.
\end{proof}
Rewriting the terms on the right-hand side of Fomin's formula as connectivity
probabilities of branches, 
we arrive at the following interpretation for the uniform spanning tree. 
\begin{prop}\label{prop: Fomin determinant and UST}
Let $e_1 , \ldots, e_{N} \in \bdry \Edg$ be distinct boundary edges,
and $u_1 , \ldots, u_N \in \Vert^\circ$ distinct interior
vertices. Then,
for the uniform spanning tree with wired boundary conditions, we have
\begin{align}\label{eq: Fomin determinant and UST}
\sum_{\sigma \in \SymmGrp_N} \sgn(\sigma) \;
    \PR \Big[ \bigcap_{\ell=1}^N \set{\pathfromto{u_\ell}{e_{\sigma(\ell)}}} \Big] 
= \; & \det \Big( \HarmMeas_{u_k}(e_{\ell}) \Big)_{k,\ell=1}^N ,
\end{align}
where $\HarmMeas_{u}(e)$ is the harmonic measure of $e \in \bdry \Edg$ seen from $u \in \Vert$.
\end{prop}
\begin{proof}
Apply Theorem~\ref{thm: Fomins general formula} with
sources $u_k$, for $k = 1, \ldots, N$, and targets 
$v_\ell = e^\bdry_{\ell}$, for $\ell = 1, \ldots, N$.
Simplify the inner summations with Lemma~\ref{lem: multiple branch connectivity probability expanded},
and observe that when the target points are on the boundary, 
$v_\ell = e^\bdry_\ell$,
the Green's functions in the determinant are harmonic measures, 
$G(u_k,e_\ell^\bdry) = \HarmMeas_{u_k}(e_\ell)$.
\end{proof}

In general, the difficulty in applying Fomin's formula 
to the uniform spanning tree connectivity probabilities is that
the formula contains simultaneously connectivities from the starting
points $u_1 , \ldots , u_N$ to the end points $e_1 , \ldots , e_N$
in all possible permutations $\sigma$.
Fomin also noted \cite[Theorem 6.4]{Fomin-LERW_and_total_positivity}, however,
that if the graph and the choice of points is such that
it is only possible to connect the starting points to the end points in
one order without intersections of the trajectories, then the sum over permutations only
contains one term. In this situation, it is possible to compute the 
sum over paths in Lemma~\ref{lem: multiple branch connectivity probability expanded}
as a determinant of the much simpler Green's functions.
This happens most naturally in the planar setup, if the points
$u_1 , \ldots , u_N , v_N, \ldots, v_1$ are ordered counterclockwise
(or clockwise) along the boundary.
This special case has become quite well known in the
two-dimensional statistical physics research. In particular, the following
consequence has been observed by many authors, for instance
\cite{Dubedat-Euler_integrals, LK-configurational_measure}.
We include the proof, because some of our further results then become
transparent generalizations.
\begin{cor}\label{cor: usual Fomin conclusion}
Let $\Gr$ be a graph embedded in a Jordan domain as
in Section~\ref{sec: graphs in planar domain}, and let
$e_1 , \ldots, e_{2N} \in \bdry \Edg$ distinct boundary edges
ordered counterclockwise along the boundary of the domain.
Consider the uniform spanning tree with wired boundary conditions
on $\Gr$. Then the probability of the rainbow connectivity
$\pathfromto{e_1}{e_{2N}}$, $\pathfromto{e_2}{e_{2N-1}} $, \ldots, $\pathfromto{e_N}{e_{N+1}}$
of the $N$ boundary branches is given by the determinant
\begin{align*}
\PR \Big[ \bigcap_{\ell=1}^N \set{ \pathfromto{e_\ell}{e_{2N+1-\ell}} } 
          \Big] = \; & \det \Big( \ExcK(e_k, e_{2N+1-\ell}) \Big)_{k,\ell=1}^N ,
\end{align*}
where $\ExcK$ is the random walk excursion kernel on $\Gr$.
\end{cor}
\begin{proof}
Apply Proposition~\ref{prop: Fomin determinant and UST} with
sources $u_\ell = e_\ell^\circ$, for $\ell = 1, \ldots, N$, and targets 
$v_k = e^\bdry_{2N+1-k}$, for $k = 1, \ldots, N$.
Note that the connectivity probability
\[ \PR \Big[ \bigcap_{\ell=1}^N \set{\pathfromto{e_\ell}{e_{2N+1-\sigma(\ell)}}} \Big] \]
vanishes unless $\sigma$ is the identity permutation, since 
the boundary branches $\gamma_{e_1} , \ldots, \gamma_{e_N}$ cannot cross each other
in the planar domain. The left-hand side of~\eqref{eq: Fomin determinant and UST}
thus contains only one non-vanishing term, 
\[ \PR \Big[ \bigcap_{\ell=1}^N \set{\pathfromto{e_\ell}{e_{2N+1-\ell}}} \Big] 
= \det \Big( \HarmMeas_{e_k^\circ}(e_{2N+1-\ell}) \Big)_{k,\ell=1}^N . \]
It remains to recognize the harmonic measures seen from boundary points
as random walk excursion kernels:
$\HarmMeas_{e_k^\circ}(e_{2N+1-\ell}) = \ExcK(e_k,e_{2N+1-\ell})$.
\end{proof}

\subsection{\label{sec: partition functions}Solution of the connectivity partition functions}

The well-known Corollary~\ref{cor: usual Fomin conclusion}
of Fomin's formula is particular first because of planarity, and second
because of the maximally nested rainbow connectivity
that it describes, encoded in the link pattern $\nested_N$ 
(see Figure~\ref{fig: minimal and maximal}).

Let us keep the planar graph embedded in the Jordan domain
and marked boundary edges $e_1 , \ldots, e_{2N}$ counterclockwise
along the boundary of the domain, but generalize the choice of the $N$ source
points $u_\ell$, for $\ell=1,\ldots N$ among $e_1^\circ , \ldots, e_{2N}^\circ$.
Each term in Fomin's formula can still be interpreted as a connectivity probability,
up to a sign, but the determinant is a sum of various possibilities.

As in Section~\ref{subsec: inverse Fomin sums}, for a link pattern $\alpha \in \LP_N$
with the left-to-right orientation $\big( (a_\ell, b_\ell) \big)_{\ell=1}^N$, denote by
\begin{align}\label{eq: definition of LPdet in the discrete}
\LPdet{\alpha}{\ExcK} (e_1 , \ldots , e_{2N})
    := \det \Big( \ExcK(e_{a_k}, e_{b_\ell}) \Big)_{k,\ell=1}^N 
\end{align}
the determinant of $\alpha$ with the random walk excursion
kernel $\ExcK$. 
Recalling the definition \eqref{eq: partition function for connectivity alpha}
of partition functions as the probabilities of connectivities, the conclusion
of Corollary~\ref{cor: usual Fomin conclusion} can be written as
\[ Z_{\nested_N}(e_1 , \ldots , e_{2N}) = 
\LPdet{\nested_N}{\ExcK} (e_1 , \ldots , e_{2N}) . \]
With the choice of sources $u_\ell = e_{a_\ell}^\circ$, for $\ell=1,\ldots N$,
and targets $v_k = e_{b_k}^\bdry$, for $k=1,\ldots N$, 
the conclusion of the more
general Proposition~\ref{prop: Fomin determinant and UST} becomes
\begin{align}\label{eq: LPdet as a sum of connectivities not simplified}
\LPdet{\alpha}{\ExcK} (e_1 , \ldots , e_{2N})
    = \sum_{\sigma \in \SymmGrp_N} \sgn(\sigma) \;
        \PR \Big[ \bigcap_{\ell=1}^N \set{\pathfromto{e_{a_\ell}}{e_{b_{\sigma(\ell)}}}} \Big] .
\end{align}
Although there is, in general, more than one non-vanishing term in the sum over
permutations $\sigma$, planarity still puts certain constraints on them.
In fact, we have the following.
\begin{prop}\label{prop: general planar Fomins formula for UST}
Let $\alpha \in \LP_N$ be a link pattern. 
We have
\begin{align}\label{eq: LPdet as a sum of connectivities} 
\LPdet{\alpha}{\ExcK} (e_1 , \ldots, e_{2N})
= \sum_{\beta \in \LP_N} \Mmat_{\alpha, \beta} \, Z_\beta(e_1 , \ldots, e_{2N}) ,
\end{align}
where $\Mmat$ is the unit weight incidence matrix~\eqref{eq: basic KW matrix}
of the parenthesis reversal relation $\KWleq$,
explicitly given in Example~\ref{ex: signed incidence matrix of KWleq}.
\end{prop}
\begin{proof}
Keep the link pattern $\alpha$ fixed throughout.
To get the asserted formula~\eqref{eq: LPdet as a sum of connectivities},
we just simplify Equation~\eqref{eq: LPdet as a sum of connectivities not simplified}
for the same determinant $\LPdet{\alpha}{\ExcK}$. Consider a permutation $\sigma \in \SymmGrp_N$
of the exit points. Then, the connectivity probability
$\PR\big[ \bigcap_{\ell=1}^N \set{ \pathfromto{e_{a_\ell}}{e_{b_{\sigma(\ell)}}} } \big]$ vanishes unless 
this connectivity determined by $\sigma$ is planar, i.e.,
$\big( (a_\ell , b_{\sigma(\ell)}) \big)_{\ell=1}^N$ is an orientation of some link pattern $\beta$.
By Example~\ref{ex: signed incidence matrix of KWleq}, for any link pattern $\beta \in \LP_N$, we have
$\Mmat_{\alpha,\beta} = \sgn(\sigma)$ if $\beta$ is obtained from $\alpha$ by a
permutation $\sigma$ of exits, and $\Mmat_{\alpha,\beta} = 0$ if no such permutation exists.
The formula~\eqref{eq: LPdet as a sum of connectivities} follows.
%
%
\end{proof}

The matrix $\Mmat$ is invertible, and a formula for the inverse is explicitly 
given in Example~\ref{ex: signed incidence matrix of KWleq}, as a special case of
Theorem~\ref{thm: weighted KW incidence matrix inversion}.
We can therefore solve~\eqref{eq: LPdet as a sum of connectivities} for the partition functions
of connectivities.
\begin{thm}\label{thm: the formula for UST partition functions}
Let $\alpha \in \LP_N$ be a link pattern. 
We have
\begin{align}\label{eq: discrete partition function as a sum of LPdets} 
Z_\alpha(e_1 , \ldots, e_{2N})
= \sum_{\beta \succeq \alpha} \Minv_{\alpha, \beta} \, \LPdet{\beta}{\ExcK}(e_1 , \ldots, e_{2N}) ,
\end{align}
where $\Minv$ is explicitly given in~\eqref{eq: basic KW inverse matrix} in 
Example~\ref{ex: signed incidence matrix of KWleq}.
\end{thm}
\begin{proof}
Multiply Equation~\eqref{eq: LPdet as a sum of connectivities} by $\Minv_{\gamma,\alpha}$
and sum over $\alpha \in \LP_N$ to get $Z_\gamma(e_1 , \ldots, e_{2N})$.
\end{proof}



Recall from Section~\ref{sec: boundary visits from partition functions} that
the boundary visit probabilities of a UST branch from $\ein$ to $\eout$
can be expressed in terms of the connectivity partition function $Z_\alpha$
by Lemma~\ref{lem: boundary visit probabilities with partition functions},
and therefore, Theorem~\ref{thm: the formula for UST partition functions}
also leads to explicit determinantal formulas for them.
Since the law of the UST boundary branch is that of a loop-erased
random walk, we get also the boundary visit probabilities for a
loop-erased random walk. We state the result in this form.
\begin{cor}
\label{cor: LERW boundary visit probabilities with partition functions}
Let $\ein, \eout \in \bdry \Edg$
be two boundary edges and 
let $\hat{e}_1 , \ldots \hat{e}_{N'}$ 
be edges at unit distance from the boundary.
Associate to them the boundary edges $e_1, \ldots,e_{2N} \in \bdry \Edg$
and the link pattern $\alpha(\omega)$, 
as in Section~\ref{sec: boundary visits from partition functions}.
Then, the boundary visit probability for the loop-erasure $\lambda = \LE(\eta)$
of a random walk $\eta$ on $\Gr$ from $\ein^\circ$, 
conditioned to reach the boundary $\bdry \Vert$ via $\eout$, is given by 
\[
\PR_{\ein,\eout} \big[ \lambda \text{ uses $\hat{e}_1 , \ldots , \hat{e}_{N'}$ in this order} \big]
= \frac{Z_{\alpha(\omega)}(e_1 , \ldots , e_{2N})}{Z(\ein , \eout)}
=  \sum_{\beta \succeq \alpha} \Minv_{\alpha, \beta} \frac{ \LPdet{\beta}{\ExcK}(e_1 , \ldots, e_{2N}) }{\ExcK(\ein,\eout)} . \]
\end{cor}
\begin{proof}
This is a direct consequence of
Corollary~\ref{cor: branches of UST are LERWs},
Lemma~\ref{lem: boundary visit probabilities with partition functions},
and Theorem~\ref{thm: the formula for UST partition functions}.
\end{proof}

\begin{rem}
\label{rem: degenerate boundary visit probabilities}
\textit{
The boundary visit probability is given by the formula
on the right-hand side of Corollary~\ref{cor: LERW boundary visit probabilities with partition functions}
even if the distinct boundary edges $\hat{e}_1 , \ldots , \hat{e}_{N'}$ share some vertices, i.e.,
not all edges $e_1 , \ldots, e_{2N}$ are distinct. The partition function expression for this probability,
however, has no immediate interpretation in such degenerate cases.
}
\end{rem}

\subsection{\label{rem: weighted graph version}Random spanning 
trees on weighted graphs}

For simplicity, we have formulated the results above for a spanning tree chosen uniformly at random.
The arguments however generalize to other natural random spanning trees, where 
the edges $e$ of the graph have weights $c(e) > 0$,
and the natural random walk has the transition probabilities
\[ w_{v,v'} = \frac{c(\edgeof{v}{v'})}{\sum_{u} c(\edgeof{v}{u}) }. \]

Using Wilson's algorithm, the natural random spanning tree on $\Gr / \bdry$ is obtained from the loop-erasures of these weighted random walks. 
The probabilities of the spanning tree are~\cite{Wilson-generating_random_spanning_trees} 
\[ \PR \big[ \set{\tree} \big] \propto \prod_{e \in \tree} c(e) . \]



Let $\alpha \in \LP_N$ be a link pattern with any orientation $((a_\ell, b_\ell))_{\ell = 1}^N$, and
let $e_1, \ldots, e_{2N}$ be distinct boundary edges. Associate to them the partition function
\begin{align*} 
Z_\alpha (e_1, \ldots, e_{2N}) 
= \left( \prod_{\ell = 1}^N c({e_{a_\ell}}) \right) \times \PR \Big[\set{ \pathfromto{e_{a_1}}{e_{b_1}} } \cap \cdots \cap \set{ \pathfromto{e_{a_N}}{e_{b_N}} }\Big],
\end{align*}
which is independent of the chosen orientation by the bijection argument of Lemma~\ref{lem: well definedness of the partition function for connectivity alpha}. 
These connectivity partition functions can be solved by repeating the analysis of Sections~\ref{sec: LERW and UST}--\ref{sec: partition functions} with 
the weighted random walks: let $\HarmMeas$ denote the harmonic measure of the weighted random walk, and define 
\begin{align*}
\ExcK_{ \mathrm{S} } (\ein, \eout ) = c(\ein ) \, \HarmMeas_{\eout} (\ein^\circ)
= \ExcK_{ \mathrm{S} } (\eout , \ein) .
\end{align*}
Then, Theorem~\ref{thm: the formula for UST partition functions} generalizes to
\begin{align*}
Z_\alpha (e_1, \ldots, e_{2N})  =  \sum_{\beta \succeq \alpha} \Minv_{\alpha, \beta}  \LPdet{\beta}{ \ExcK_{ \mathrm{S} } }(e_1 , \ldots, e_{2N}),
\end{align*}
and, analogously, the statements of Corollary~\ref{cor: LERW boundary visit probabilities with partition functions} and 
Remark~\ref{rem: degenerate boundary visit probabilities}  hold with the formula
\begin{align*}
\PR_{\ein,\eout} \big[ \lambda \text{ uses $\hat{e}_1 , \ldots , \hat{e}_{N'}$ in this order} \big]
=  \left( \prod_{s=1}^{N'} \frac{ c ( \hat{e}_s ) }{ c ( \hat{e}_{s;1} ) c ( \hat{e}_{s; 2} ) } \right) \times 
\sum_{\beta \succeq \alpha} \Minv_{\alpha, \beta} \frac{ \LPdet{\beta}{\ExcK_{ \mathrm{S} }}(e_1 , \ldots, e_{2N}) }{\ExcK_{ \mathrm{S} }(\ein,\eout)} . 
\end{align*}

By virtue of the robustness of the scaling limit results for random 
walks and discrete harmonic 
functions~\cite{YY-Loop-erased_random_walk_and_Poisson_kernel_on_planar_graphs},
some scaling limit results of the next section could be shown to be universal
for suitable weighted random spanning trees.


\subsection{\label{sec: scaling limits}Scaling limits}

So far we have considered the discrete model of uniform
spanning tree on a planar graph. We now turn to the question of scaling limits,
where a fixed planar domain $\domain \subset \bC$
is approximated by graphs with increasingly fine mesh and probabilities are renormalized 
by suitable power laws of the mesh size.

\subsubsection{\label{sec: finer mesh graphs}\textbf{Finer mesh graphs}}
For concreteness, when discussing scaling limits, we always take graphs 
$\Gr^\delta$ which are subgraphs of the regular square lattice $\delta \bZ^2$ with
mesh size $\delta>0$, and study the limit $\delta \to 0$.
Our scaling limit results could be extended to more general setups, 
as long as the random walks on the graphs $\Gr^\delta$ tend to the Brownian motion 
on $\domain$ and the boundary approximation near the marked points is regular 
enough so that also suitably renormalized random walk excursion kernels 
tend to the Brownian excursion kernel --- 
see Lemma~\ref{lem: scaling limit of discrete excursion kernels}.

Fix the domain $\domain$, and $2N$ boundary points
$p_1 , \ldots, p_{2N} \in \bdry \domain$ appearing 
in counterclockwise order
along the boundary $\bdry \domain$. Assume throughout that locally near each $p_j$, 
the boundary is a straight horizontal or vertical line segment. This property is assumed in 
order to control the scaling limit behavior of the random walk excursion kernels 
from these marked boundary points.

\begin{figure}
\includegraphics[width=.7\textwidth]{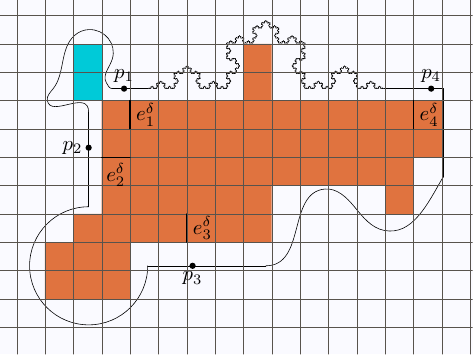}
\caption{\label{fig: square grid approx}
A Jordan domain with marked boundary points and its square grid approximation.
}
\end{figure}

For a given (small) mesh size $\delta>0$, we define the 
\textit{square grid approximation}
of the domain $\domain$ as the following graph
$\Gr^\delta = (\Vert^\delta,\Edg^\delta)$, illustrated in Figure~\ref{fig: square grid approx}.
Consider the closed squares 
$[n \delta , (n+1)\delta] \times [m \delta, (m+1)\delta]$
of the square lattice $\delta \bZ^2$ that are contained in $\domain$. 
Take a connected component $A$ of the interior of their union, with a maximal number of squares. 
The graph $\Gr^\delta$ is taken to have vertices
$\Vert^\delta = \cl{A} \cap \delta \bZ^2$. The boundary vertices are defined as
$\bdry \Vert^\delta = \bdry A \cap \delta \bZ^2$. The set $\Edg^\delta$ of edges
consists of all pairs of vertices at distance $\delta$ from each other, such that
at least one of the vertices is an interior vertex.

For all $j$, we denote by $e_j^\delta \in \bdry \Edg^\delta$
a boundary edge nearest to the marked boundary point $p_j \in \bdry \domain$, 
that is, an edge $e_j = \edgeof{e_j^\bdry}{e_j^\circ}$ that contains the
boundary vertex $e_j^\bdry$ which is the nearest to $p_j$, and its neighbor 
$e_j^\circ$ at distance $\delta$ to the direction of the inwards normal to 
the boundary $\bdry \domain$
(recall that the boundary is assumed locally horizontal or vertical near $p_j$).
The choice of these boundary edges is also illustrated
in Figure~\ref{fig: square grid approx}.

\subsubsection{\textbf{Scaling limits of excursion kernels}}
\label{sec: scaling limit of excursion kernels}

In the scaling limit $\delta \to 0$,
the random walk excursion kernels $\ExcK^\delta$ on $\Gr^\delta$
can be approximated with the Brownian excursion kernel $\ExcKdom$ in 
the domain $\domain$, in the sense of Lemma~\ref{lem: scaling limit of discrete excursion kernels}
below.

We use the simple notation $\ExcKH$ for 
\begin{align}\label{eq: Brownian excursion kernel in H}
\ExcKH(x_1 , x_2) = \frac{1}{(x_2 - x_1)^2} 
 \qquad \text{for $x_1 , x_2 \in \bR$, $x_1 \neq x_2$},
\end{align}
which is a constant multiple of the Brownian excursion kernel 
in the upper half-plane $\bH$.
In any simply connected domain $\domain$, with two boundary points 
$p_1 , p_2 \in \bdry \domain$ on straight boundary segments, 
the Brownian excursion kernel
$\ExcKdom$ is expressed in terms of $\ExcKH$ by conformal covariance:
\begin{align}\label{eq: Brownian excursion kernel in domain}
\ExcKdom(p_1 , p_2) = 
\frac{1}{\pi} |\confmap'(p_1)| \, |\confmap'(p_2)| \; 
\ExcKH \big( \confmap(p_1) , \confmap(p_2) \big) ,
\end{align}
where $\confmap \colon \domain \to \bH$ is any conformal map from 
$\domain$ to 
$\bH$ such that $\confmap(p_1) \neq \infty$ and $\confmap(p_2) \neq \infty$.

Let 
$e_{1+}^\delta$ be the boundary edge one lattice unit
from $e_1^\delta$ to the counterclockwise direction, and define the
discrete tangential derivative of the excursion kernel with respect to the first
variable as 
\[ D^\delta_{\tau;1} \, \ExcK^\delta
    := \frac{\ExcK^\delta (e_{1+}^\delta , e_2^\delta) - \ExcK^\delta (e_1^\delta , e_2^\delta)}{\delta} . \]
The discrete derivatives $D^\delta_{\tau;2} \ExcK^\delta$ and $D^\delta_{\tau;1} D^\delta_{\tau;2} \ExcK^\delta$
are defined similarly in terms of differences.
We also denote by 
$\partial_{\tau;i}$ the usual counterclockwise tangential derivative with respect to 
the $i$:th argument.


\begin{lem}\label{lem: scaling limit of discrete excursion kernels}
As $\delta \to 0$, we have
\[ \ExcK^\delta(e_1^\delta , e_2^\delta)
= \delta^2 \, \ExcKdom(p_1, p_2) + \oo(\delta^2) , \]
and
\begin{align*}
D^\delta_{\tau;i} \, \ExcK^\delta 
    = \; & \delta^2 \; \partial_{\tau;i} \, \ExcKdom(p_1, p_2) + \oo(\delta^2) \qquad \qquad \text{for } i=1,2 \\
D^\delta_{\tau;1} D^\delta_{\tau;2} \, \ExcK^\delta 
    = \; & \delta^2 \; \partial_{\tau;1} \partial_{\tau;2} \, \ExcKdom(p_1, p_2) + \oo(\delta^2) .
\end{align*}
\end{lem}
\begin{proof}
This follows from known convergence results of discrete harmonic functions.
The renormalized harmonic measure 
$v \mapsto \frac{1}{\delta} \, \HarmMeas_{v}^\delta(e_2^\delta)$
is a discrete harmonic function on $\Vert^\delta$. 
It is known to converge to the Poisson
kernel when $e_2^\delta$ is on a straight boundary segment,
see e.g. \cite{CS-discrete_complex_analysis_on_isoradial}.
The convergence of the discrete harmonic function and all its discrete derivatives
is uniform on compact subsets of the domain~$\domain$
\cite{CFL-uber_die_PDE_der_mathphys, CS-discrete_complex_analysis_on_isoradial}. 
But, by Schwarz reflection, the convergence of the function 
and its discrete derivatives also holds when $v$ is taken to some straight part of 
the boundary. It remains to note that $\ExcK^\delta(e_1^\delta,e_2^\delta)$ is $\delta^2$ times the 
discrete normal derivative of 
$v \mapsto \frac{1}{\delta} \, \HarmMeas_{v}^\delta(e_2^\delta)$ at $e_1^\delta$.
\end{proof}

\subsubsection{\label{sec: scaling limits of partition functions}\textbf{Scaling limits of partition functions for connectivities}}


To prepare for the scaling limit statement, we first
give two definitions for the continuum setup.
For a link pattern $\alpha$ with the left-to-right orientation
$\big( (a_\ell, b_\ell) \big)_{\ell=1}^N$,
and for any $x_1 < x_2 < \cdots < x_{2N}$, we set
\begin{align}\label{eq: definition of LPdet in the continuum}
\LPdet{\alpha}{\ExcKH} (x_1 , \ldots , x_{2N})
    := \; & \det \Big( \ExcKH(x_{a_k}, x_{b_\ell}) \Big)_{k,\ell=1}^N
    = \det \Bigg( \frac{1}{ ( x_{b_\ell} - x_{a_k} )^2 } \Bigg)_{k,\ell=1}^N ,
\end{align}
analogously to~\eqref{eq: definition of LPdet in the discrete},
with the 
kernel $\ExcKH(x_1,x_2) = (x_2-x_1)^{-2}$
in the place of the random walk excursion kernel $\ExcK(e_1 , e_2)$.
Analogously to~\eqref{eq: discrete partition function as a sum of LPdets},
we also set
\begin{align}\label{eq: SLE partition function as a sum of LPdets} 
\PartF_\alpha(x_1 , \ldots, x_{2N})
= \sum_{\beta \succeq \alpha} \Minv_{\alpha, \beta} \, \LPdet{\beta}{\ExcKH}(x_1 , \ldots, x_{2N}) .
\end{align}

In the scaling limit setup of Section~\ref{sec: finer mesh graphs},
we get the following limiting formula for the partition functions.
\begin{thm}\label{thm: scaling limit of partition functions}
Let $\alpha \in \LP_N$, and for all $\delta>0$ denote by
$Z^{\Gr^\delta}_\alpha \big( e_1^\delta , \ldots, e_{2N}^\delta \big)$
the corresponding connectivity partition function for the UST on the
square grid approximation $\Gr^\delta$ of the domain $\domain$.
Then in the scaling limit $\delta \to 0$, we have
\begin{align*}
\frac{1}{\delta^{2N}} \, Z^{\Gr^\delta}_\alpha \big( e_1^\delta , \ldots, e_{2N}^\delta \big)
\longrightarrow \frac{1}{\pi^{N}} \times \prod_{j=1}^{2N} |\confmap'(p_j)| \times
     \PartF_\alpha \big( \confmap(p_1) , \ldots, \confmap(p_{2N}) \big) ,
\end{align*}
where $\confmap \colon \domain \to \bH$ is any conformal map
such that $\confmap(p_1) < \cdots < \confmap(p_{2N})$.
\end{thm}
\begin{proof}
This follows by combining Theorem~\ref{thm: the formula for UST partition functions}
with the first statement of Lemma~\ref{lem: scaling limit of discrete excursion kernels}.
\end{proof}
In Section~\ref{sec: applications to SLEs}, we will furthermore prove that 
the limit function $\PartF_\alpha$ above is a positive solution to a system of
second order partial differential equations of conformal field theory,
see Theorem~\ref{thm: pure partition functions at kappa equals 2}.


\subsubsection{\label{sec: boundary visits}\textbf{Scaling limits of boundary visit probabilities}}
The formulas of Corollary~\ref{cor: LERW boundary visit probabilities with partition functions}
for boundary visit probabilities are also amenable to a scaling limit analysis, although
this case is considerably more involved than that of
Section~\ref{sec: scaling limits of partition functions}.
The difficulties arise because among the arguments $e_1 , \ldots, e_{2N}$ of the determinant
expressions, $N'$ pairs of edges $\hat{e}_{s;1} , \hat{e}_{s;2}$ are separated 
by just one lattice unit $\delta$, and we are letting $\delta \to 0$.
We state below the conclusion of the analysis, which will be done in
Sections~\ref{sec: determinant Taylor expansions}--\ref{sub: proof of bdry visit thm}. 
\begin{thm}\label{thm: scaling limit of LERW bdry visits}
Fix a domain $\domain$ and distinct boundary points 
$\pin , \pout, \hat{p}_1 , \ldots , \hat{p}_{N'}$
on horizontal or vertical boundary segments.
Take a square grid approximation $\Gr^\delta$ of $\domain$,
with $\ein^\delta, \eout^\delta \in \bdry \Edg^\delta$
and $\hat{e}_1^\delta , \ldots, \hat{e}_{N'}^\delta$ nearest to
$\pin , \pout$ and $\hat{p}_1 , \ldots , \hat{p}_{N'}$, respectively. Let $\lambda^\delta$ be 
the loop-erasure $\lambda^\delta = \LE(\eta^\delta)$
of a random walk $\eta^\delta$ on $\Gr^\delta$ from $(\ein^\delta)^\circ$,
conditioned to reach the boundary $\bdry \Vert^\delta$ via $\eout^\delta$.
Then, in the scaling limit as $\delta \to 0$, we have
\begin{align*}
& \frac{1}{\delta^{3 N'}} \, \PR_{\ein^\delta,\eout^\delta}
        \Big[ \lambda^\delta \text{ uses $\hat{e}_1^\delta , \ldots , \hat{e}_{N'}^\delta$ in this order} \Big] \\
\longrightarrow \; & \frac{1}{\pi^{N'}} \times \prod_{s=1}^{N'} |\confmap'(\hat{p}_s)|^3 \times
     \frac{\Ampl_\omega \big( \confmap(\pin) ; \confmap(\hat{p}_1) , \ldots, \confmap(\hat{p}_{N'}) ; \confmap(\pout) \big)}{\big( \confmap(\pout) - \confmap(\pin) \big)^{-2}} ,
\end{align*}
where $\confmap \colon \domain \to \bH$ is any conformal map
such that $\confmap(\pin) < \confmap(\pout)$ and 
$\confmap(\pin) < \confmap(\hat{p}_j)$ for all $j$,
and $\Ampl_\omega$~is~a function that satisfies two second order 
PDEs~\eqref{eq: second order PDEs at kappa equals two} 
and $N'$ third order PDEs~\eqref{eq: third order PDEs at kappa equals two}.
\end{thm}
The proof is summarized in Section~\ref{sub: proof of bdry visit thm}.
It will also be shown (Proposition~\ref{prop: scaling limit of boundary visits})
that $\zeta_\omega$ is positive unless the order of visits $\omega$ to the
given edges at unit distance from the boundary is already impossible for
curves on the planar graphs $\Gr^\delta$.






\subsection{\label{sec: generalizations}Generalizations of the main results}


We finish this section by mentioning further results for uniform spanning
trees that generalize or are closely related to the above main results,
and which can still be proved with the same techniques.
Trusting that the reader can modify our arguments appropriately
to cover these generalizations, we choose not to provide full details
of their proofs.

\subsubsection{\label{subsubsec: simultaneous connectivities and visits}\textbf{Joint boundary visit and connectivity probabilities}}

Theorem~\ref{thm: the formula for UST partition functions} gives the
connectivity probability of boundary branches in a wired UST, and
Lemma~\ref{lem: boundary visit probabilities with partition functions}
gives the boundary visit probabilities of one boundary branch.
Theorems~\ref{thm: scaling limit of partition functions}
and~\ref{thm: scaling limit of LERW bdry visits} give the scaling
limits of these respective probabilities. It is straightforward to generalize
Lemma~\ref{lem: boundary visit probabilities with partition functions}
to obtain formulas for the probability that any
given boundary points are connected by UST branches and that these branches
visit any given edges at unit distance from the boundary.
Furthermore, this probability can be properly renormalized to have
a nontrivial scaling limit. More precisely, we have the following.
\begin{thm}
\label{thm: bdary visits and connectivities}
Fix a domain $\domain$ and distinct boundary points 
$p_1 , \ldots, p_{2N}, \hat{p}^{(1)}_1 , \ldots , \hat{p}^{(1)}_{N'_1}, \ldots ,
\hat{p}^{(N)}_1 , \ldots , \hat{p}^{(N)}_{N'_N}$
on horizontal or vertical boundary segments.
Take a square grid approximation $\Gr^\delta$ of $\domain$,
with boundary edges $e_1, \ldots , e_{2N} \in \bdry \Edg^\delta$
nearest to $p_1 , \ldots, p_{2N}$ and,
for $\ell=1,\ldots,N$, 
edges $\hat{e}^{(\ell)}_1 , \ldots, \hat{e}^{(\ell)}_{N'_\ell}$ 
at unit distance from the boundary 
nearest to $\hat{p}^{(\ell)}_1 , \ldots , \hat{p}^{(\ell)}_{N'_\ell}$, respectively. 
Let $\alpha \in \LP_{N}$ be a link pattern with any orientation $(a_\ell , b_{\ell})_{\ell=1}^N$ 
and denote by $\gamma^{(\ell)}$ the boundary branch from $e_{a_\ell}$.
Then, denoting $N' = \sum_{\ell=1}^N N'_\ell$, in the scaling limit as $\delta \to 0$, we have
\begin{align*}
\frac{1}{\delta^{2 N + 3 N'}} \, \PR
        \Big[ \pathfromto{e_{a_\ell}}{e_{b_\ell}} \text{ and }
        \gamma^{(\ell)} \text{ uses $\hat{e}^{(\ell)}_1 , \ldots , \hat{e}^{(\ell)}_{N'_\ell}$ in this order, for each $\ell=1,\ldots,N$} \Big] 
\longrightarrow \; & F^\domain ( \boldsymbol{p} ; \hat{\boldsymbol{p}} ) ,
\end{align*}
where $F^\domain$ is a conformally covariant function of $2N + N'$ boundary
points of $\domain$, which satisfies $2N$ second order and $N'$ third order
PDEs of conformal field theory, of the form given
in Section~\ref{sec: third order PDEs}. 
\end{thm}

\subsubsection{\label{sss: free UST}\textbf{Boundary touching subtrees in a free uniform spanning tree}}
So far we have discussed the uniform spanning tree only with
wired boundary conditions (\textit{wired UST}),
meaning that we collapsed the boundary $\bdry \Vert \subset \Vert$
into a single vertex. In contrast, the uniform spanning tree with free boundary
conditions (\textit{free UST}) is just the uniformly randomly chosen spanning
tree of the given graph, without any collapsing of boundary. We now present some
results for the free UST on square lattice graph approximations of a domain in the plane.

Let $\Gr$ be a square grid approximation of a Jordan domain $\domain$ as
in Section~\ref{sec: finer mesh graphs}, with mesh size $\delta>0$ that we keep implicit
in the notation below.
In order to relate the results to earlier ones in a transparent
manner, for the free UST the graph approximation of $\domain$ is taken to
be the dual graph $\Gr^* = (\Vert^* , \Edg^*)$ of~$\Gr/\bdry$:
the vertices $v^* \in \Vert^*$ are the square faces of $\Gr$ and the
edges $e^* = \edgeof{v_1^*}{v_2^*} \in \Edg^*$ join two square faces $v_1^*, v_2^*$
that share one side with each other. %
There is a natural notion of boundary also in $\Gr^*$.
A face $v^* \in \Vert^*$ is said to be a \textit{boundary face} if 
at least one of the corners of the square face $v^*$ is a boundary vertex of $\Gr$.
A dual edge $e^* \in \Edg^*$ is said to be a \textit{boundary dual edge} if
the edge $e$ which it crosses is a boundary edge, $e \in \bdry \Edg$.

A well-known simple fact is that the free UST $\tree^*$
on $\Gr^*$ can be obtained from the wired UST $\tree$ on $\Gr$ by
the duality
(illustrated in Figure~\ref{fig: dual UST}):
\[ e \in \tree \qquad \Longleftrightarrow \qquad e^* \notin \tree^* \;
\text{ for the dual edge $e^*$ which crosses the edge $e$.} \]
\begin{figure}
\includegraphics[width=.3\textwidth]{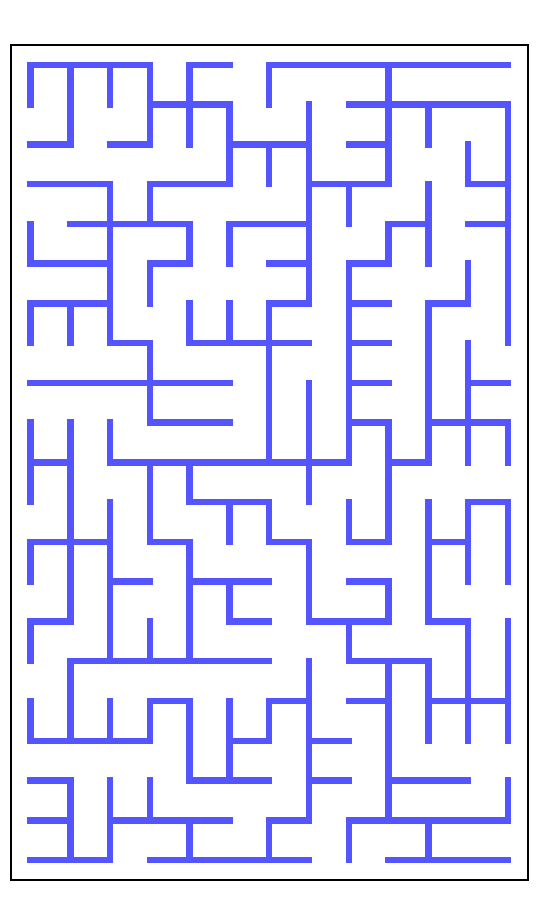} \quad 
\includegraphics[width=.3\textwidth]{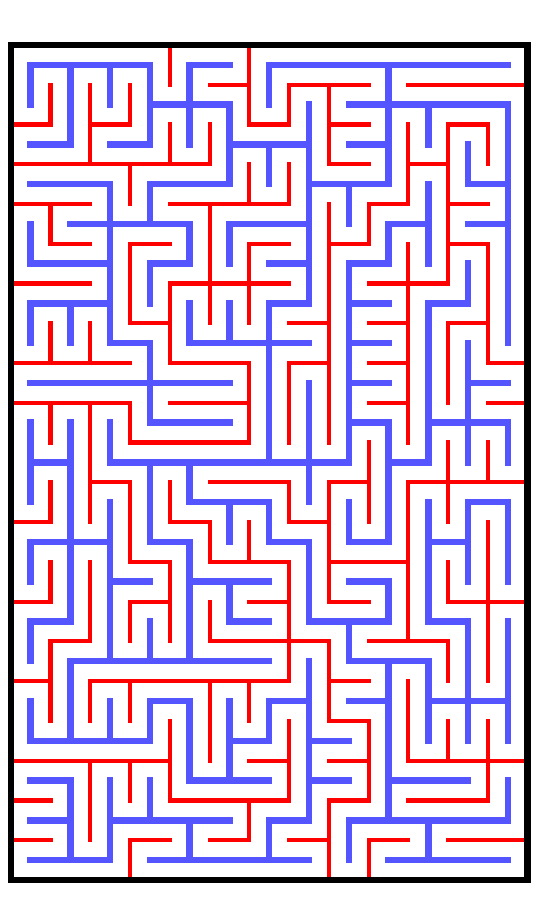} \quad
\includegraphics[width=.3\textwidth]{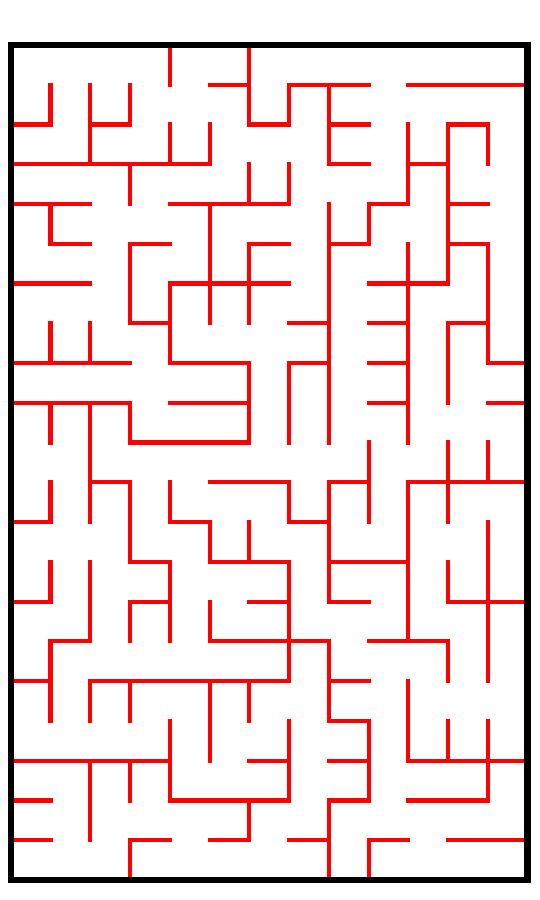}
\caption{\label{fig: dual UST}
The uniform spanning tree with free boundary conditions (left figure) is dual
to the uniform spanning tree with wired boundary conditions (right figure)
as illustrated here (middle figure).
}
\end{figure}

\begin{figure}
\includegraphics[width=.4\textwidth]{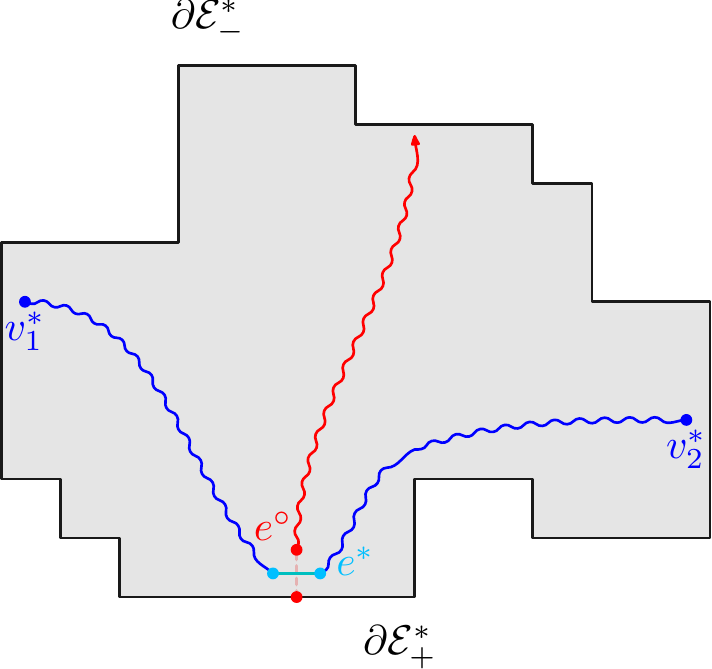} \quad 
\caption{\label{fig: free UST boundary visits}
Schematic illustration of a boundary-to-boundary 
branch of the free~UST $\tree^*$ using the boundary dual edge $e^*$ (blue), and
the boundary branch of the edge $e$ in the wired~UST $\tree $
on the primal graph (red).
}
\end{figure}


In the free UST model, the most naive analogue of our earlier questions would be
the boundary visits of a boundary-to-boundary branch,
illustrated in Figure~\ref{fig: free UST boundary visits}.
Formally, given two boundary faces $v_1^*$ and $v_2^*$, and a boundary dual edge
$e^*$ on the counterclockwise dual boundary segment $\bdry \Edg^*_+$ 
from $v_1^*$ to $v_2^*$: what is the probability that the unique
branch of $\tree^*$ connecting $v_1^*$ to $v_2^*$ passes through $e^*$?
By duality, this occurs if and only if the boundary branch of $e $ in $\tree$
connects  to the primal boundary $\bdry \Vert$ so that it crosses 
the clockwise dual boundary segment $\bdry \Edg^*_-$ 
from $v_1^*$ to $v_2^*$, as depicted in Figure~\ref{fig: free UST boundary visits}.
By Wilson's algorithm, the probability of this event
is given by the harmonic measure in $\Gr$ of the
boundary edges crossing $ \bdry \Edg_-^*$ seen from $e^\circ$. This gives
the discrete boundary visit probability in the free UST.


The scaling limit behavior of the probability of such a boundary visit event
is then an easy consequence of the
convergence of discrete harmonic measures and their derivatives to the corresponding
continuum objects (see, e.g., \cite{CS-discrete_complex_analysis_on_isoradial}).
Assume that $\Gr$, $v_1^*$, $v_2^*$, and $e^*$ are an approximation of
the domain $\domain$ with three boundary points $p _1$, $p_2$, and $\hat{p}$, and
suppose that $\hat{p}$ lies on a horizontal or vertical boundary segment.
The free UST branch from $v_1^*$ to $v_2^*$ visits $e^*$ with probability $\OO(\delta)$,
where $\delta$ is the mesh size of the graph approximation.
Renormalized by $\delta^{-1}$, this probability converges to the normal derivative at
$\hat{p}$ of the continuum harmonic measure of the clockwise arc from $p _1$ to $p_2$
The limit function is known to be conformally covariant and to satisfy a second order PDE 
of conformal field theory. 
This is to be contrasted with the third order PDEs and 
probability $\OO(\delta^3)$ of boundary visits in the wired UST.
In summary, the boundary visit probability
in the free UST model is considerably easier than in the wired UST,
and essentially different in terms of its scaling exponent and PDEs.


The more interesting counterpart is the opposite question: what is the probability that two
boundary faces are connected by a path not visiting the boundary in the free UST?
More precisely,
we consider the following problem of boundary touching subtrees
in the free UST $\tree^*$.
The \textit{interior forest} of $\tree^*$ is the subgraph obtained by
removing all boundary dual edges from $\tree^*$. The connected components
$\tau^*_1 , \ldots , \tau^*_M$ of the interior forest are trees,
and we ask whether there is a component which intersects the boundary exactly
at some given faces $v_1^*, \ldots, v_N^* \in \bdry \Vert^*$.


\begin{thm}\label{thm: boundary touching subtrees in FUST}
Let $v_1^*, \ldots, v_N^* \in \bdry \Vert^*$ be non-neighboring distinct boundary faces
in counterclockwise order along $\bdry \Vert^*$, and let $e_1 , \ldots, e_{2N} \in \bdry \Edg$ be the (unique) $2N$ boundary egdes adjacent to these faces, enumerated counterclockwise starting from $v_1^*$.
Then the probability that some component $\tau^*_m$ of the interior forest of $\tree^*$ intersects
the boundary $\partial \Vert^*$ exactly at the faces $v_1^*, \ldots, v_N^*$ is
given by
\[ \PR \Big[ \exists \, m \text{ such that }
            \tau^*_m \cap \bdry \Vert^* = \{ v_1^* , \ldots, v_N^* \} \Big]
    = 2^N \; Z_{\unnested_N}(e_1 , \ldots, e_{2N}) , \]
where $\unnested_N \in \LP_N$ is the completely unnested link pattern.
\end{thm}


The key to the proof is illustrated in Figure~\ref{fig: key fUST figure}. 
By duality, the components of the interior forest
are separated by boundary-to-boundary branches of the primal wired UST.
The event that one component contains exactly
the boundary faces $v_1^*, \ldots, v_N^*$ is a completely unnested connectivity event of
the boundary edges $e_1 , \ldots, e_{2N} \in \bdry \Edg$ chosen as in the statement of the theorem.
There are $2^N$ possible orientations of the $N$ branches, each contributing equally.

\begin{figure}
\includegraphics[width = 0.5\textwidth]{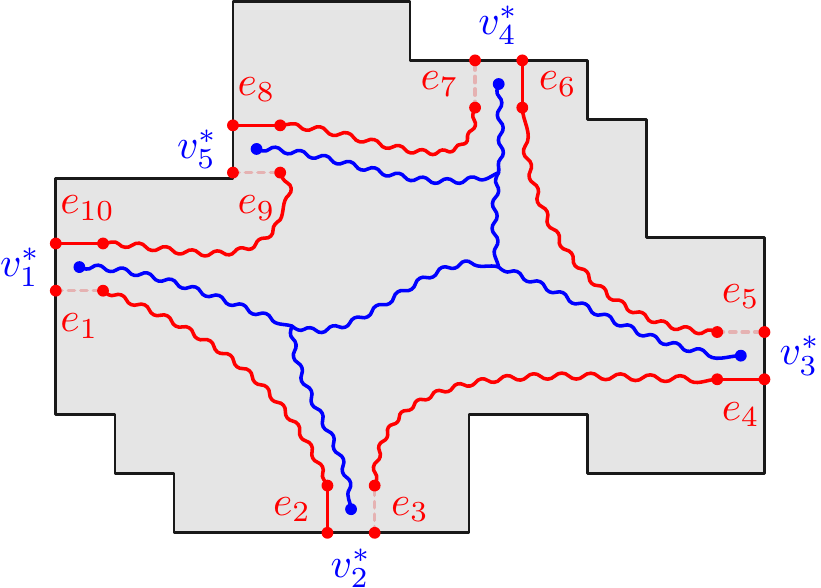}
\caption{\label{fig: key fUST figure}
Schematic illustration of a
boundary touching subtree in the interior forest of the free~UST $\tree^*$ (blue),
and the corresponding completely unnested connectivity in the wired~UST $\tree$
on the primal graph (red).
}
\end{figure}

Among the boundary edges $e_1 , \ldots, e_{2N} \in \bdry \Edg$ as above, there are $N$
pairs separated by just one lattice unit,
and none of these pairs is connected in the completely unnested pattern $\unnested_N$. 
The scaling limit of connectivity probabilities in such a setup is
treated in Section~\ref{appendix: boundary visits}. The following
scaling limit result for boundary touching subtrees of the free UST can be straightforwardly
inferred.

\begin{thm} \label{thm: scaling limit of boundary touching subtrees in FUST}
Fix a domain $\domain$, and let
$\hat{p}_1, \ldots, \hat{p}_N \in \bdry \domain$ be distinct boundary points on
horizontal or vertical boundary segments. Let
$v_1^*, \ldots, v_N^* \in \bdry \Vert^*$ be boundary faces closest to
$\hat{p}_1, \ldots, \hat{p}_N$, respectively. Then in the scaling limit as $\delta \to 0$,
the probability that some component $\tau^*_m$ of the interior forest of $\tree^*$ intersects
the boundary $\partial \Vert^*$ exactly at the faces $v_1^*, \ldots, v_N^*$ is
given by
\[ \frac{1}{\delta^{3N}} \;\PR \Big[ \exists \, m \text{ such that }
            \tau^*_m \cap \bdry \Vert^* = \{ v_1^* , \ldots, v_N^* \} \Big]
    \longrightarrow F^\domain (\hat{\boldsymbol{p}}) , \]
where $F^\domain$ is a conformally covariant function of $N$ boundary points of $\domain$,
which satisfies $N$ third order PDEs of conformal field theory, of the form given
in Section~\ref{sec: third order PDEs}.
\end{thm}


\bigskip{}

\section{\label{sec: applications to SLEs}Relation to multiple $\SLE$s}


Schramm-Loewner Evolutions ($\SLE$)
are random curves in planar domains, whose laws in any two conformally equivalent domains are
related to each other via a push-forward by a conformal map.
$\SLE$ type random curves were originally introduced in \cite{Schramm-LERW_and_UST},
motivated in particular by the scaling limits of loop-erased random walks and uniform spanning trees.
A number of variants of $\SLE$s exists, each relevant for a slightly different setup,
but the most important characteristics of any $\SLE$ type curve is
captured by one parameter, $\kappa>0$. For instance, the
scaling limits of LERWs and branches in the UST are $\SLE$s with $\kappa=2$, see
\cite{LSW-LERW_and_UST, Zhan-scaling_limits_of_planar_LERW, 
YY-Loop-erased_random_walk_and_Poisson_kernel_on_planar_graphs}.
The main new results in this section pertain to that particular value, $\kappa=2$,
and variants of $\SLE$s known as multiple $\SLE$s.

For the purposes of this section, 
we assume at least superficial familiarity with
the most standard $\SLE$ variant, the chordal $\SLEk$, which is a random curve in a simply
connected domain $\domain \subset \bC$ between two boundary points $\pin, \pout \in \bdry \domain$.
The reader can find the definition, basic properties, and applications of chordal $\SLEk$
in, e.g., \cite{KN-guide_to_SLE, RS-basic_properties_of_SLE,
Lawler-conformally_invariant_processes_in_the_plane}.
We briefly describe the definition of multiple $\SLE$s relying on the chordal $\SLE$,
but for the details we again
refer to the literature
\cite{BBK-multiple_SLEs, Dubedat-Euler_integrals, Dubedat-commutation, 
KP-pure_partition_functions_of_multiple_SLEs}.

The description of multiple $\SLE$s is given in
Section~\ref{sec: multiple SLEs}, with particular emphasis on
their local definition using partition functions. The main result
of this section, Theorem~\ref{thm: pure partition functions at kappa equals 2},
states that the scaling limits of UST 
connectivity probabilities given in Theorem~\ref{thm: scaling limit of partition functions}
are the so-called multiple $\SLE$ pure partition functions at $\kappa=2$.
As a consequence, we obtain 
in Theorem~\ref{thm: existence of local multiple SLE2s} the existence and extremality
of the corresponding local multiple $\SLE$ processes at $\kappa=2$,
which for $N$ curves are indexed by link patterns $\alpha \in \LP_N$ of $N$ links.
The key ingredients are second order partial differential
equations~\eqref{eq: PDE for multiple SLEs at kappa equals 2},
M\"obius covariance~\eqref{eq: COV for multiple SLEs at kappa equals 2}, 
and boundary conditions~\eqref{eq: ASY for multiple SLEs at kappa equals 2}
for the functions $\PartF_\alpha$,
whose derivations are given in Section~\ref{sec: proof of pure partition function properties}.
We remark that the second order PDEs 
will also be needed as an intermediate step in 
the derivation of the third order PDEs in Section~\ref{appendix: boundary visits}.

%

For notational consistency, we introduce the following parameters depending on $\kappa$:
\begin{align*}
h_{1,2} = & \; h_{1,2}(\kappa) = \frac{6-\kappa}{2 \kappa} &
\Delta =  & \; \Delta(\kappa) = - 2 h_{1,2}(\kappa) = 1-\frac{6}{\kappa} \\
h_{1,3} = & \; h_{1,3}(\kappa) = \frac{8-\kappa}{ \kappa} &
\Delta' = & \; \Delta'(\kappa) = h_{1,3}(\kappa) - 2 h_{1,2}(\kappa) = \frac{2}{\kappa} .
\end{align*}
For the case of our primary interest, $\kappa=2$,
these parameters are just the following constants:
\begin{align*}
h_{1,2} = 1 , \quad
h_{1,3} = 3 , \quad
\Delta =  -2 , \quad
\Delta' =  1 .
\end{align*}

\subsection{\label{sec: multiple SLEs}Multiple SLEs}

Multiple $\SLE$s are processes of several interacting random curves.
A multiple $\SLEk$ in a simply connected planar domain 
consists of $N$ random curves connecting $2N$ distinct points 
$p_1 , \ldots, p_{2N}$ on the boundary 
pairwise without crossing. For example, the joint law of several boundary touching branches 
of the UST (with wired boundary conditions) should converge in the scaling limit
to such  a process with $\kappa = 2$, as stated in more detail below in
Conjecture~\ref{conj: multiple branches of UST converge to multiple SLE2}.

The success of the $\SLE$ theory relies largely on the growth process description
of curves that employs the Loewner chain technique from complex analysis.
The growth process description does
not directly give a global definition of the curves, but provides a construction
of their initial segments. For this
reason, we will consider so-called local multiple $\SLE$s in the sense of
\cite{BBK-multiple_SLEs, Dubedat-Euler_integrals, Dubedat-commutation, 
KP-pure_partition_functions_of_multiple_SLEs}.
A local multiple $\SLE$ is constructed by 
a growth process encoded in a Loewner chain, and the construction relies on a partition function.
The definition of local multiple $\SLE$s is given in Section~\ref{sub: local multiple SLEs},
after some relevant preliminaries about the partition functions
in Section~\ref{sub: multiple SLE partition functions}.

\subsubsection{\label{sub: multiple SLE partition functions}\textbf{Partition functions of multiple SLEs}}

%
The construction of a local multiple $\SLEk$ describing $N$ curves
from $2N$ marked boundary points, 
employs a \textit{multiple SLE partition function}
\begin{align*}
\PartF \colon \chamber_{2N} \to \bRpos,
\end{align*}
a positive function defined on the chamber
\begin{align*}
\chamber_{2N} = \set{ (x_1 , \ldots, x_{2N}) \; \Big| \; x_1 < \cdots < x_{2N} } ,
\end{align*}
satisfying the following $2N$ partial differential equations of second order:
\begin{align}
\label{eq: PDE for multiple SLEs} \tag{PDE} 
& \left[ \frac{\kappa}{2} \pdder{x_j}
    + \sum_{i \neq j} \Big( \frac{2}{x_i-x_j} \pder{x_i} - \frac{2h_{1,2}}{(x_i-x_j)^2} \Big) \right] \PartF (x_1 , \ldots, x_{2N}) = 0 \qquad \text{for all } j=1,\ldots,2N,
\end{align}
and the covariance under M\"obius transformations:
\begin{align}
\label{eq: COV for multiple SLEs} \tag{COV} 
& \PartF(x_1 , \ldots, x_{2N}) = 
    \prod_{j=1}^{2 N} \Mob'(x_j)^{h_{1,2}} \times \PartF(\Mob(x_1) , \ldots, \Mob(x_{2N}))  \\
\nonumber
& \text{for all } \Mob(z) = \frac{a z + b}{c z + d}, \; \text{ with } a,b,c,d \in \bR, \; ad-bc > 0, 
 \text{ such that } \Mob(x_1) < \cdots < \Mob(x_{2N}).
\end{align}


The solution space of the system
\eqref{eq: PDE for multiple SLEs}--\eqref{eq: COV for multiple SLEs} has dimension $\Catalan_N = \frac{1}{N + 1} \binom{2N}{N}$
(when solutions with at most power-law growth are considered),
by the work in~\cite{FK-solution_space_for_a_system_of_null_state_PDEs_1,
FK-solution_space_for_a_system_of_null_state_PDEs_2,
FK-solution_space_for_a_system_of_null_state_PDEs_3}.
In \cite{KP-pure_partition_functions_of_multiple_SLEs}, 
a set of $\Catalan_N$ distinguished linearly independent solutions, the pure partition functions
$(\PartF_\alpha^{(\kappa)})_{\alpha \in \LP_N}$, were found for generic $\kappa$.
These were argued to correspond to the extremal multiple $\SLEk$ probability measures that
cannot be written as non-trivial convex combinations in the set of all 
multiple $\SLEk$ probability measures.
Besides \eqref{eq: PDE for multiple SLEs} and \eqref{eq: COV for multiple SLEs},
these distinguished functions satisfy the following specific asymptotics properties 
on the pairwise diagonals 
(the codimension one boundary of the chamber domain $\chamber_{2N}$):
\begin{align}
\label{eq: ASY for multiple SLEs} \tag{ASY} 
& \lim_{x_j , x_{j+1} \to \xi} 
\frac{\PartF_\alpha^{(\kappa)} (x_1 , \ldots, x_{2N})}{(x_{j+1} - x_j)^{\Delta}}
= \begin{cases}
    \PartF_{\alpha \removeupwedge{j} }^{(\kappa)} (x_1, \ldots, x_{j-1} , x_{j+2} , \ldots, x_{2N}) & \text{ if } \upwedgeat{j} \in \alpha \\
    0 & \text{ if } \upwedgeat{j} \notin \alpha
    \end{cases} 
\end{align}
for all $j \in \set{1, \ldots, 2N - 1}$
and any $\xi \in (x_{j-1} , x_{j+2})$, where 
$\PartF_\emptyset \equiv 1$ by convention,
and the combinatorial notations are as defined in 
Section~\ref{subsec: Wedges, slopes, and link removals}.
This asymptotic boundary condition is illustrated in 
Figure~\ref{fig: cascade property for SLE}.

\begin{figure}
\begin{displaymath}
\xymatrixcolsep{3.5pc}
\xymatrixrowsep{2.5pc}
\xymatrix{
        & \begin{minipage}{6cm} \begin{center} \includegraphics[width=\textwidth]{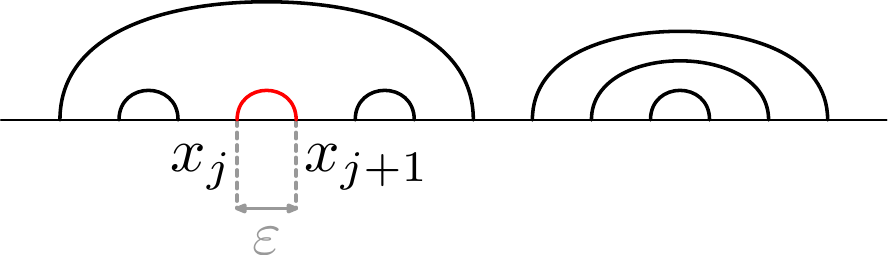} \end{center} \end{minipage}\ar[d] \\
        & \hspace{-0mm}\begin{minipage}{6cm} \begin{center}  \includegraphics[width=\textwidth]{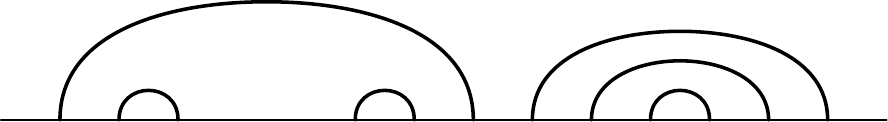} \end{center} \end{minipage}  & 
        }
\end{displaymath}
\caption{\label{fig: cascade property for SLE}
Schematic illustration of the link patterns in the
cascade property \eqref{eq: ASY for multiple SLEs},
where the distance $\varepsilon$ of the two boundary points $x_j$ and $x_{j+1}$ is taken to zero.
The corresponding combinatorial operation is depicted in Figure~\ref{fig: link removal}.}
\end{figure}

In~\cite{KP-pure_partition_functions_of_multiple_SLEs},
the solutions $\PartF_\alpha^{(\kappa)}$ were constructed for 
$\kappa \in (0,8) \setminus \bQ$.
In the present article, 
we show that the scaling limits $\PartF_\alpha$ of
uniform spanning tree connectivity probabilities~\eqref{eq: partition function for connectivity alpha}
satisfy the defining requirements \eqref{eq: PDE for multiple SLEs},
\eqref{eq: COV for multiple SLEs} and \eqref{eq: ASY for multiple SLEs}
for pure partition functions $\PartF_\alpha^{(\kappa)}$ 
at $\kappa = 2$. At $\kappa=2$
we can also prove their positivity 
and conclude the existence of the corresponding local multiple $\SLE$ processes,
whereas for generic $\kappa$, the positivity of the pure partition functions $\PartF^{(\kappa)}_\alpha$
has remained conjectural\footnote{The recent
works~\cite{PH-Global_multiple_SLEs_and_pure_partition_functions, Wu-HyperSLE} 
show the positivity for $\kappa \leq 6$.}.

\begin{thm}\label{thm: pure partition functions at kappa equals 2}
For any link pattern $\alpha \in \LP_N$, the function
$\PartF_\alpha \colon \chamber_{2N} \to \bR$ 
given by~\eqref{eq: SLE partition function as a sum of LPdets},
\begin{align*}
\PartF_\alpha (x_1 , \ldots , x_{2N})
= \; & \sum_{\beta \succeq \alpha} \Minv_{\alpha,\beta} \, \LPdet{\beta}{\ExcKH} (x_1 , \ldots, x_{2N}) ,
\end{align*}
satisfies \eqref{eq: PDE for multiple SLEs},
\eqref{eq: COV for multiple SLEs} and \eqref{eq: ASY for multiple SLEs}
at $\kappa = 2$, i.e., we have
\begin{align}
\label{eq: PDE for multiple SLEs at kappa equals 2} \tag{PDE2} 
& \left[ \pdder{x_j}
    + \sum_{i \neq j} \Big( \frac{2}{x_i-x_j} \pder{x_i} - \frac{2}{(x_i-x_j)^2} \Big) \right]
  \PartF_\alpha (x_1 , \ldots , x_{2N}) = 0 \qquad \text{for all } j=1,\ldots,2N \quad \\
\label{eq: COV for multiple SLEs at kappa equals 2} \tag{COV2} 
& \PartF_\alpha(x_1 , \ldots, x_{2N}) = 
    \prod_{j=1}^{2 N} \Mob'(x_j) \times \PartF_\alpha(\Mob(x_1) , \ldots, \Mob(x_{2N}))  \\
\nonumber
& \text{for all } \Mob(z) = \frac{a z + b}{c z + d}, \; \text{ with } a,b,c,d \in \bR, \; ad-bc > 0, 
 \text{ such that } \Mob(x_1) < \cdots < \Mob(x_{2N})  \\
 \label{eq: ASY for multiple SLEs at kappa equals 2} \tag{ASY2} 
& \lim_{x_j , x_{j+1} \to \xi} \frac{\PartF_\alpha (x_1 , \ldots, x_{2N})}{(x_{j+1} - x_j)^{-2}}
= \begin{cases}
    \PartF_{\alpha \removeupwedge{j} } (x_1, \ldots, x_{j-1} , x_{j+2} , \ldots, x_{2N}) & \text{ if } \upwedgeat{j} \in \alpha \\
    0 & \text{ if } \upwedgeat{j} \notin \alpha
    \end{cases} \\    
& 
\text{for all $j \in \set{1, \ldots, 2N - 1}$ and any $\xi \in (x_{j-1} , x_{j+2})$} .
\nonumber
\end{align}
Moreover, 
the collection $( \PartF_\alpha )_{\alpha \in \LP_N}$
of functions is linearly independent, and each function $\PartF_\alpha$ is positive:
we have $\PartF_\alpha (x_1 , \ldots , x_{2N}) > 0$
for all $(x_1,\ldots, x_{2N}) \in \chamber_{2N}$.
\end{thm}

The proof of this theorem
will be given in Section~\ref{sec: proof of pure partition function properties}.

We remark that in \cite{Dubedat-Euler_integrals}, Dub\'edat
also stated without explicit proof that the determinants $\LPdet{\beta}{\ExcKH}$
satisfy the partial differential equations~\eqref{eq: PDE for multiple SLEs at kappa equals 2},
and suggested the problem of finding the appropriate linear combination
of the determinants which satisfies the specific boundary
conditions~\eqref{eq: ASY for multiple SLEs at kappa equals 2}.
Theorem~\ref{thm: pure partition functions at kappa equals 2} above
settles this problem.


\subsubsection{\label{sub: local multiple SLEs}\textbf{Local multiple SLEs}}

In general, for $\kappa>0$, a local multiple $\SLEk$ is defined in
\cite[Appendix~A]{KP-pure_partition_functions_of_multiple_SLEs} by requiring
conformal invariance, a domain Markov property,
and absolute continuity of initial segments with respect to the chordal $\SLEk$.
By \cite{Dubedat-commutation} and
\cite[Theorem~A.4]{KP-pure_partition_functions_of_multiple_SLEs}, 
such local $N$-$\SLEk$ processes are classified by multiple $\SLE$
partition functions $\PartF \colon \chamber_{2N} \to \bRpos$ (modulo multiplicative constant),
that is, positive solutions to \eqref{eq: PDE for multiple SLEs}--\eqref{eq: COV for multiple SLEs}.
For concreteness, we therefore describe below only the local multiple $\SLEk$ associated with
a given partition function $\PartF$.

\begin{figure}
 \includegraphics[width=0.5\textwidth]{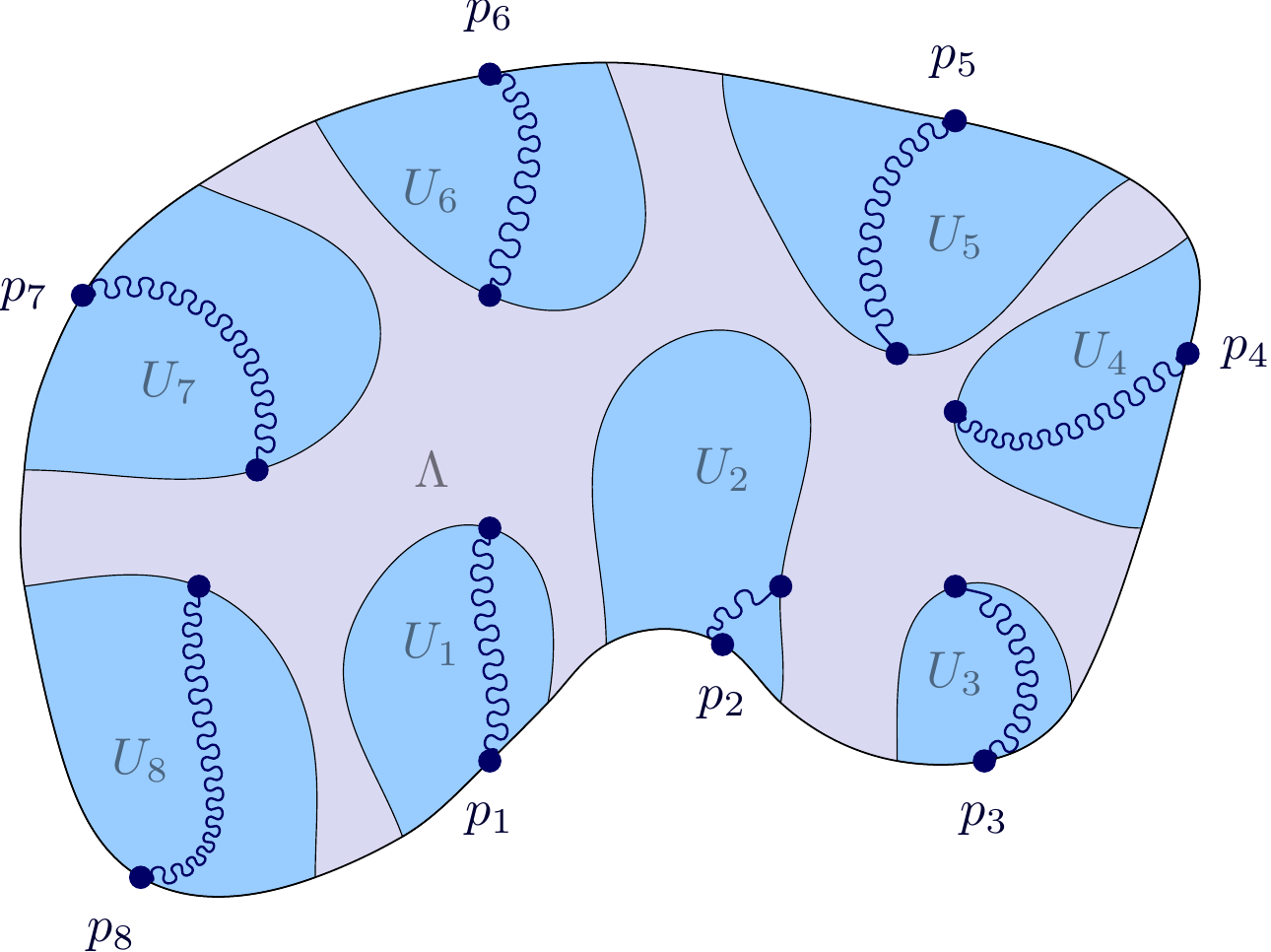} 
\caption{\label{fig: NSLE}
Schematic illustration of a local multiple $\SLE$.
}
\end{figure}

A local $N$-$\SLEk$ associates a probability measure on 
initial segments of $2N$ random curves to any domain and localization data of the following kind.
The domain and localization data consists of a simply connected domain 
(where the curves live) with $2N$ marked boundary points (the starting points 
of the curves), and disjoint localization neighborhoods of 
the marked points (where the initial segments of the curves 
lie) --- see Figure \ref{fig: NSLE} for an illustration.
More precisely, for any $\kappa > 0$, the local $N$-$\SLEk$ is a collection
of probability measures indexed by the domain $\domain$, marked points 
$p_1,\ldots,p_{2N} \in \bdry \domain$, and their closed localization neighborhoods
$U_1,\ldots,U_{2N} \subset \cl{\domain}$, such that $p_j \in U_j$ for all $j$, 
the complements $\domain \setminus U_j$ are simply connected, 
and $U_j \cap U_k = \emptyset$ for $j \neq k$. 
The probability measures are supported on the set of $2N$-tuples 
$(\gamma^{(1)} , \ldots, \gamma^{(2N)})$ of oriented non-self-crossing unparametrized curves $\gamma^{(j)}$, traversing from the marked points 
$p_j$ to the boundary of their localization neighborhoods $U_j$.

For each $j = 1,\ldots,2N$,
the partition function $\PartF$ 
determines the marginal law $\SLEmeasure_{\gamma^{(j)}}$ 
of the $j$:th curve $\gamma^{(j)}$ started from $p_j$, traversing  
the localization neighborhood $U_j$, as follows.
The curve $\gamma^{(j)}$ is absolutely continuous with respect to the law
$\SLEmeasure^{(\domain;p_j,p)}$ of  an initial segment of the single chordal
$\SLEk$ curve in $\domain$ from $p_j$ to any target point $p \in \bdry\domain$ 
outside $U_j$. 
The Radon-Nikodym derivative is 
\begin{align*}
\frac{\ud \SLEmeasure_{\gamma^{(j)}}}{\ud \SLEmeasure^{(\domain;p_j,p)}}  
= \; & \prod_{i\neq j} g'(x_i)^{h_{1,2}} 
\times \frac{\PartF \big(g(x_1),\ldots,g(x_{j-1}),g(\gamma^{(j)}(\tau)),g(x_{j+1}),\ldots, g(x_{2N})\big)}{\PartF(x_1,\ldots,x_{2N})},
\end{align*}
where 
\begin{itemize}
\item $x_j = \confmap(p_j)$ for all $j$, where $\confmap \colon \domain \to \bH$ is a conformal map such that $\confmap(p) = \infty$,
\item $g\colon H_\tau^{(j)} \rightarrow \bH$ is the unique conformal isomorphism from 
the unbounded component $H_\tau^{(j)}$ of the complement 
$\bH \setminus \confmap \left(\gamma^{(j)}[0,\tau]\right)$ 
of the conformal image of the curve $\gamma^{(j)}$
to the upper-half plane $\bH$, normalized
so that $g(z)=z+\oo(1)$ as $z\to\infty$,
\item $\tau$ is the hitting time of the curve
$\gamma^{(j)}$ to the boundary of its localization neighborhood $U_j$.
\end{itemize}
Together with the domain Markov property, these marginals on the individual
curve segments $\gamma^{(j)}$ in fact determine the joint probability measure of all $2N$
curve segments, see \cite[Appendix~A]{KP-pure_partition_functions_of_multiple_SLEs}.

It is also very natural to form convex combinations of probability measures, such as 
the localizations of a local multiple $\SLE$.
Appropriately accounting for conformal transformation
properties, one gets a convex structure on the space of local multiple $\SLE$s,
see~\cite[Theorem~A.4(c)]{KP-pure_partition_functions_of_multiple_SLEs} for details.

In the case $\kappa = 2$, our partition functions $\PartF_\alpha$
can be used to construct local multiple $\SLE_2$ processes:

\begin{thm}\label{thm: existence of local multiple SLE2s}
For each link pattern $\alpha \in \LP_N$, there exists a local $N$-$\SLEk$ 
with $\kappa=2$, associated with the partition function $\PartF_\alpha$ as in
Theorem~\ref{thm: pure partition functions at kappa equals 2}
and Equation~\eqref{eq: SLE partition function as a sum of LPdets}.
Moreover, the convex hull of the local $N$-$\SLE$s corresponding to $\PartF_\alpha$ 
for $\alpha \in \LP_N$ is of dimension $\Catalan_N-1$ and the $\Catalan_N$
local $N$-$\SLE$s corresponding to $\PartF_\alpha$ are the extremal points of this convex set.
\end{thm}
\begin{proof}
The first assertion follows from the properties of the functions $\PartF_\alpha$ established in
Theorem~\ref{thm: pure partition functions at kappa equals 2},
combined with the classification of local multiple $\SLE$s
\cite[Theorem~A.4(a)]{KP-pure_partition_functions_of_multiple_SLEs} at $\kappa = 2$.
By \cite[Theorem~A.4(c)]{KP-pure_partition_functions_of_multiple_SLEs},
convex combinations of local $N$-$\SLE$s correspond to positive linear combinations of their
partition functions modulo multiplicative constants. The second assertion then follows
from the linear independence of $(\PartF_\alpha)_{\alpha \in \LP_N}$ and the fact that
a linear combination $\PartF = \sum_{\alpha} c_\alpha \PartF_\alpha$ is 
non-negative only if $c_\alpha \geq 0$ for all $\alpha$, as can be shown by expressing each $c_\alpha$ 
as a suitable limit of the partition function~$\PartF$
as in~\cite[Proposition~4.2]{KP-pure_partition_functions_of_multiple_SLEs}.
\end{proof}

The partition functions $\PartF_\alpha$ are the scaling limits of connectivity
probabilities of branches in the uniform spanning tree,
by Theorem~\ref{thm: scaling limit of partition functions}.
The local multiple $\SLE$ curves determined by them
should of course be closely related to the branches of the UST.
Let $\Gr^\delta$ be the square grid approximation of the simply connected 
domain $\domain \subset \bC$ with mesh size $\delta > 0$, as in 
Section~\ref{sec: finer mesh graphs}. For $2N$ boundary points
$p_1 , \ldots, p_{2N} \in \bdry \domain$,
let $e_1^\delta,\ldots,e_{2N}^\delta \in \bdry \Edg$
be boundary edges of 
the square grid graph $\Gr^\delta$ nearest to them. Assume that the points appear 
in counterclockwise order along the boundary.
\begin{conj}\label{conj: multiple branches of UST converge to multiple SLE2}
Let $\alpha \in \LP_N$ be a link pattern with any orientation
$\left( (a_\ell,b_\ell) \right)_{\ell=1}^N$.
In the uniform spanning tree on $\Gr^\delta$ with wired boundary conditions,
conditioned on the connectivity event
$\bigcap_{\ell=1}^N \set{e_{a_\ell}^\delta \rightsquigarrow e_{b_\ell}^\delta}$,
the law of the $N$ boundary branches 
$(\gamma_{e_{a_1}^\delta}, \ldots, \gamma_{e_{a_N}^\delta})$
converges in the scaling limit as $\delta \to 0$ to a process of $N$ curves,
whose localization in any non-overlapping closed neighborhoods $(U_1,\ldots,U_{2N})$ of the points 
$(p_1,\ldots,p_{2N})$
is the local $N$-$\SLE_{2}$ determined by the partition function $\PartF_\alpha$.
\end{conj}
\begin{upd*}
\emph{%
In light of some recent developments,
several approaches to establish this conjecture are possible.
For the tightness of the boundary branches, 
see~\cite{Karrila_inprep}. For the identification,
using the convergence of a single UST branch to
chordal $\SLE_2$~\cite{Zhan-scaling_limits_of_planar_LERW} and
recent classification results of global multiple 
$\SLE$s~\cite{MillerSheffieldIG2, BeffaraPeltolaWuUniqueness},
the scaling limit of multiple 
UST branches can be identified as a global multiple SLE,
whose localizations are local multiple 
SLEs~\cite{PH-Global_multiple_SLEs_and_pure_partition_functions}.
Alternatively, requiring 
some more investigation of the discrete model but only basic SLE theory, 
one can use the results of the present article to
generalize the observable approach used 
in the 
single UST branch convergence proofs \cite{LSW-LERW_and_UST, 
Zhan-scaling_limits_of_planar_LERW}; see~\cite{Karrila_inprep}.
}
\end{upd*}

\subsection{\label{sec: proof of pure partition function properties}Proof of Theorem~\ref{thm: pure partition functions at kappa equals 2}}

We prove the asserted properties of the functions $\PartF_\alpha$
separately, and conclude the proof of Theorem~\ref{thm: pure partition functions at kappa equals 2}
in Section~\ref{subsub: proof}.
The M\"obius covariance 
\eqref{eq: COV for multiple SLEs at kappa equals 2}
and the partial differential equations 
\eqref{eq: PDE for multiple SLEs at kappa equals 2}
are linear conditions, so 
it suffices to establish them for the determinant functions
$\LPdet{\alpha}{\ExcKH}$ 
appearing in the formula~\eqref{eq: SLE partition function as a sum of LPdets}
which defines the functions $\PartF_\alpha$.
The proof of property 
\eqref{eq: ASY for multiple SLEs at kappa equals 2} concerning the asymptotics
of the functions requires in addition specific combinatorial tools from
Section~\ref{subsec: inverse Fomin sums}.



\subsubsection{\label{sec: Mobius covariance for determinants}\textbf{M\"obius covariance}}

We first check the M\"obius covariance property
\eqref{eq: COV for multiple SLEs at kappa equals 2} for the determinants
$\LPdet{\alpha}{\ExcKH}$ given by Equation~\eqref{eq: definition of LPdet in the continuum}.
\begin{lem}\label{lem: Mobius covariance for determinants}
Let $\Mob(z) = \frac{a z + b}{c z + d}$ be as 
in \eqref{eq: COV for multiple SLEs at kappa equals 2}. 
Then, for any $(x_1,\ldots, x_{2N}) \in \chamber_{2N}$, we have
\[ \LPdet{\alpha}{\ExcKH} (x_1 , \ldots, x_{2N})
= \prod_{j=1}^{2 N} \Mob'(x_j) \times \LPdet{\alpha}{\ExcKH} (\Mob(x_1) , \ldots, \Mob(x_{2N})) . \]
\end{lem}
\begin{proof}
It is straightforward to check that 
$\frac{\Mob(z)-\Mob(w)}{z-w} = \sqrt{\Mob'(z)} \sqrt{\Mob'(w)}$
for any $z,w \in \mathbb{C}$, see e.g. \cite[Lemma~4.7]{KP-pure_partition_functions_of_multiple_SLEs}.
This identity can be used in the matrix elements of the determinant function $\LPdet{\alpha}{\ExcKH}$,
%
%
\begin{align*}
\frac{1}{(x_{a_k} - x_{b_\ell})^2} 
= \; & \frac{\Mob'(x_{a_k}) \, \Mob'(x_{b_\ell})}{(\Mob(x_{a_k}) - \Mob(x_{b_\ell}))^2} .
\end{align*} 
The asserted covariance property follows by multilinearity of the determinant.
\end{proof}

\subsubsection{\label{sec: PDEs for determinants}\textbf{Partial differential equations}}

We next check the partial differential equations 
\eqref{eq: PDE for multiple SLEs at kappa equals 2} 
for the determinants $\LPdet{\alpha}{\ExcKH}$.

\begin{lem}\label{lem: PDEs for determinants}
The function $\LPdet{\alpha}{\ExcKH} (x_1 , \ldots, x_{2N})$ satisfies the
partial differential 
equations~\eqref{eq: PDE for multiple SLEs at kappa equals 2}.
\end{lem}
\begin{proof}
The proof is a direct computation.
To simplify notation, we denote $x_{a_k} = A_k$ and $x_{b_\ell} = B_\ell$, and write
\begin{align}\label{eq: LP determinant in entering pts a and exiting pts b} 
\LPdet{\alpha}{\ExcKH} (x_1 , \ldots, x_{2N}) 
= \det \left( \frac{1}{(A_k - B_\ell)^2}\right)_{k,\ell=1}^N  
= \sum_{\sigma \in \SymmGrp_N} \sgn(\sigma) 
\prod_{m=1}^N \frac{1}{(A_m - B_{\sigma(m)})^2}.
\end{align}
It suffices to consider 
the partial differential equation in~\eqref{eq: PDE for multiple SLEs at kappa equals 2} with $j = 1$.
We write the differential operator as
\begin{align}
\label{eq: Da1 in entering pts a and exiting pts b}
\sD_1 = \; & \pdder{x_1}
    + \sum_{i \neq 1} \Big( \frac{2}{x_i-x_1} \pder{x_i} - \frac{2}{(x_i-x_1)^2} \Big) \\
    \nonumber
       = \; & \pdder{A_1} + \frac{2}{B_{\sigma(1)}- A_1} \pder{B_{\sigma(1)}} - \frac{2}{(B_{\sigma(1)}-A_1)^2} \\
   & \qquad + \sum_{k = 2}^N \Big( \frac{2}{A_k- A_1} \pder{A_k} - \frac{2}{(A_k-A_1)^2} + \frac{2}{B_{\sigma(k)}- A_1} \pder{B_{\sigma(k)}} - \frac{2}{(B_{\sigma(k)}-A_1)^2} \Big) \nonumber
\end{align}
where we re-arranged the terms containing the variable 
$B_\ell = x_{b_\ell}$ by a permutation $\sigma \in \SymmGrp_N$.

By linearity, let us first study the action of the partial differential 
operator~\eqref{eq: Da1 in entering pts a and exiting pts b} on one term in the 
determinant~\eqref{eq: LP determinant in entering pts a and exiting pts b}. 
Straightforward differentiation yields
\begin{align*}
& \sD_1 \Bigg( \prod_{m=1}^N \frac{1}{(A_m - B_{\sigma(m)})^2} \Bigg) \\
= \; & \Bigg( \prod_{m=1}^N \frac{1}{(A_m - B_{\sigma(m)})^2} \Bigg) \Bigg[ \frac{6}{(A_1 - B_{\sigma(1)})^2} - \frac{4}{(A_1 - B_{\sigma(1)})^2} - \frac{2}{(A_1 - B_{\sigma(1)})^2} \\
& \qquad + \sum_{k=2}^N \left( \frac{-4}{(A_k - A_1)(A_k - B_{\sigma(k)})}- \frac{2}{(A_k - A_1)^2} + \frac{4}{(B_{\sigma(k)} - A_1)(A_k - B_{\sigma(k)})}- \frac{2}{(B_{\sigma(k)} - A_1)^2} \right) \Bigg].
\end{align*}
Using the identities
\begin{align*}
\frac{-4}{(A_k - A_1)(A_k - B_{\sigma(k)})} 
    + \frac{4}{(B_{\sigma(k)} - A_1)(A_k - B_{\sigma(k)})} = \; & \frac{4}{(A_k - A_1) (B_{\sigma(k)} - A_1)}  \\
\text{and } \qquad\qquad
-2 \left(\frac{1}{ (B_{\sigma(k)} - A_1)} - \frac{1}{(A_k - A_1)} \right)^2  
= \; & - \frac{2(B_{\sigma(k)} - A_k)^2}{(A_k - A_1)^2 (B_{\sigma(k)} - A_1)^2} ,
\end{align*}
we simplify
\begin{align*}
 \sD_1 \Bigg( \prod_{m=1}^N \frac{1}{(A_m - B_{\sigma(m)})^2} \Bigg)
&= -2  \sum_{k=2}^N \frac{1}{(A_k - A_1)^2}  \left(\prod_{ m \neq k} \frac{1}{(A_m - B_{\sigma(m)})^2}  \right) \frac{1}{(B_{\sigma(k)} - A_1)^2}.
\end{align*}

By linearity, $\sD_1$ now acts on the determinant 
\eqref{eq: LP determinant in entering pts a and exiting pts b} by
\begin{align*}
& \sD_1 \  \LPdet{\alpha}{\ExcKH} (x_1 , \ldots, x_{2N})  
= \sum_{\sigma \in \SymmGrp_N} \sgn(\sigma) \times 
\sD_1  \Bigg( \prod_{m=1}^N \frac{1}{(A_m - B_{\sigma(m)})^2} \Bigg) \\
&= -2  \sum_{k=2}^N \frac{1}{(A_k - A_1)^2} \left[ \sum_{\sigma \in \SymmGrp_N} \sgn(\sigma) \left(\prod_{m \neq k} \frac{1}{(A_m - B_{\sigma(m)})^2}  \right) \frac{1}{(B_{\sigma(k)} - A_1)^2} \right].
\end{align*}
Now, the sum in the square brackets is the determinant of the matrix
\begin{equation*}
\left( \frac{1}{(\widetilde{A}_m - B_n)^2}  \right)_{m,n=1}^N,
\end{equation*}
where $\widetilde{A}_m = A_m$ if $m \neq k$ and $\widetilde{A}_k = A_1$. The first 
and $k$:th rows of this matrix are hence identical, so the determinant is zero, and we obtain the desired property
$\sD_1 \, \LPdet{\alpha}{\ExcKH} (x_1 , \ldots, x_{2N}) = 0$.
\end{proof}

\subsubsection{\label{subsub: partition function asymptotics}\textbf{Asymptotics}}
To obtain the asymptotics of the partition functions $\PartF_\alpha$,
we rely on results from Section~\ref{subsec: inverse Fomin sums}.
The asymptotics of the partition functions can be derived by considering the asymptotics of
the determinants~$\LPdet{\alpha}{\ExcKH}$, which are somewhat simpler. Moreover,
the asymptotics of~$\LPdet{\alpha}{\ExcKH}$ are closely analogous to the defining
properties of the conformal block functions in CFT,
see~\cite{KKP-companion}. 

\begin{prop}\label{prop: vanishing-link asymptotics of excursion kernel determinants}
%
For all $j=1,\ldots,2N-1$, 
the determinant function 
$\LPdet{\alpha}{\ExcKH} \colon \chamber_{2N} \to \bR$ satisfies
\begin{align} \label{eq: det limits} \tag{$\FominDet$-ASY} 
(x_{j+1} - x_j)^{2} \LPdet{\alpha}{\ExcKH} (x_1 , \ldots, x_{2N}) \longrightarrow \begin{cases}
0 & \text{if } \slopeat{j} \in \alpha \\
\LPdet{ \alpha \removeupwedge{j} }{\ExcKH} (x_1 , \ldots, x_{j-1}, x_{j+2}, \ldots, x_{2N})  & \text{if } \upwedgeat{j} \in \alpha \\
-\LPdet{\alpha \removedownwedge{j} }{\ExcKH} (x_1 , \ldots, x_{j-1}, x_{j+2}, \ldots, x_{2N})  & \text{if } \downwedgeat{j} \in \alpha,
\end{cases}
\end{align}
as $x_j , x_{j+1} \to \xi$, for any $\xi \in (x_{j-1}, x_{j+2})$,
where $\LPdet{\emptyset}{\ExcKH} \equiv 1$ by convention,
and the combinatorial notations are as defined in 
Section \ref{subsec: Wedges, slopes, and link removals}.
\end{prop}
\begin{proof}
Notice that $\LPdet{\alpha}{\ExcKH} $ is an example of a link pattern
determinant $\LPdet{\alpha}{\mathfrak{K}} $ studied in
Section~\ref{subsec: inverse Fomin sums}, with the kernel
$\mathfrak{K}(i, j) = \ExcKH(x_i, x_j) = (x_{i} - x_j)^{-2}$. 
In the limit $x_j , x_{j+1} \to \xi$,
all kernel entries except $\mathfrak{K}(j, j+1)$ remain bounded, so the renormalized
limit~\eqref{eq: det limits} picks the coefficient $[\LPdet{\alpha}{\mathfrak{K}}]_{j,j+1}$
of $\mathfrak{K}(j, j+1)$ in the determinant~$\LPdet{\alpha}{\mathfrak{K}}$.
This coefficient is given in Lemma~\ref{lem: coefficients in LP determinants}.
The asserted limits follow.
\end{proof}

\begin{prop}\label{prop: Vanishing-link asymptotics of the part fcns}
The function $\PartF_\alpha \colon \chamber_{2N} \to \bR$ defined in
\eqref{eq: SLE partition function as a sum of LPdets} satisfies
the asymptotics~\eqref{eq: ASY for multiple SLEs at kappa equals 2}.
\end{prop}
\begin{proof}
The partition function $\PartF_\alpha$ is an inverse Fomin type
sum~$\mathfrak{Z}_{\alpha}^{\mathfrak{K}} $ studied in Section~\ref{subsec: inverse Fomin sums},
with the kernel $\mathfrak{K}(i, j) = \ExcKH(x_i, x_j) = (x_{i} - x_j)^{-2}$. 
The renormalized
limit in~\eqref{eq: ASY for multiple SLEs at kappa equals 2}
picks the coefficient $[\mathfrak{Z}_{\alpha}^{\mathfrak{K}}]_{j,j+1}$
of $\mathfrak{K}(j, j+1)$ in $\mathfrak{Z}_{\alpha}^{\mathfrak{K}}$. This coefficient is given in
Proposition~\ref{prop: Zero-replacing Rule}(c) and Proposition~\ref{prop: inverse Fomin cascade}.
\end{proof}


\subsubsection{\label{subsub: proof}\textbf{Finishing the proof of Theorem~\ref{thm: pure partition functions at kappa equals 2}}}
We now collect the results of
Sections~\ref{sec: Mobius covariance for determinants}--\ref{subsub: partition function asymptotics}
to prove Theorem~\ref{thm: pure partition functions at kappa equals 2}.
By Lemma~\ref{lem: Mobius covariance for determinants} and linearity,
it follows that $\sZ_\alpha$ satisfies the M\"obius covariance
condition~\eqref{eq: COV for multiple SLEs at kappa equals 2}.
Similarly, by Lemma~\ref{lem: PDEs for determinants} and linearity,
it follows that $\sZ_\alpha$ satisfies the partial differential
equations~\eqref{eq: PDE for multiple SLEs at kappa equals 2}.
Proposition~\ref{prop: Vanishing-link asymptotics of the part fcns}
contains the asserted asymptotics~\eqref{eq: ASY for multiple SLEs at kappa equals 2}.
Linear independence of the collection $(\sZ_\alpha)_{\alpha \in \LP_N}$
can be deduced from the asymptotics by arguments presented, e.g., in
\cite[Proposition~4.2]{KP-pure_partition_functions_of_multiple_SLEs}
or \cite{FK-solution_space_for_a_system_of_null_state_PDEs_3}.
It remains only to verify the positivity of the functions~$\PartF_\alpha$.

By Theorem~\ref{thm: scaling limit of partition functions},
$\sZ_\alpha(x_1 , \ldots, x_{2N})$ equals a positive constant times
the limit of the non-negative quantities $\delta^{-2N} \, Z^{\Gr^\delta}_\alpha(e_1^\delta , \ldots, e_{2N}^\delta)$
as $\delta \to 0$, and as such, $\sZ_\alpha(x_1 , \ldots, x_{2N}) \geq 0$.
There are many ways to promote the non-negativity to positivity.
One, purely analytical option, is to use the ellipticity of the PDEs and a maximum principle.
A probabilistic option that uses results from Section~\ref{sec: application to UST}
is the following: we can argue that for fixed $x_1 < \cdots < x_{2N}$ the quantities
$\delta^{-2N} \, Z^{\Gr^\delta}_\alpha(e_1^\delta , \ldots, e_{2N}^\delta)$
are uniformly lower bounded in $\delta$, which implies positivity: 
$\sZ_\alpha(x_1 , \ldots, x_{2N}) > 0$.
The uniform lower bound only essentially relies on the observation that for one UST branch 
(i.e., a LERW), $\delta^{-2}$ times the probability to connect two given boundary
points in a given domain tends to a positive limit as $\delta \to 0$,
see Equation~\eqref{eq: excursin kernel is the one branch connectivity proba}
and Lemma~\ref{lem: scaling limit of discrete excursion kernels}.
Now fix a domain $\domain$, marked points $p_1 , \ldots, p_{2N}$, and a
connectivity $\alpha = \set{\{a_1 , b_1\} , \ldots , \{a_N , b_N\} } \in \LP_{N}$.
Then, for $\ell = 1 , \ldots, N$, choose non-overlapping subdomains 
$\domain_\ell \subset \domain$ such that
$\domain_\ell$ contains neighborhoods of the boundary points $p_{a_\ell}$ and $p_{b_\ell}$ ---
such subdomains, illustrated in Figure~\ref{fig: square grid approx 3}, 
exist since the connectivity $\alpha$ is planar.
Construct the UST branches from $e^\delta_{a_\ell}$ by Wilson's algorithm 
as LERWs. A lower bound for the probability 
of $\bigcap_{\ell=1}^N \set{\pathfromto{e^\delta_{a_\ell}}{e^\delta_{b_\ell}}}$ 
is the product over $\ell$ of the LERW connectivity probabilities in 
the subdomains $\domain_\ell$ from $e^\delta_{a_\ell}$ to $e^\delta_{b_\ell}$. 
These, in turn, when divided by $\delta^2$ each, are 
lower bounded by positive quantities, since they have positive limits as $\delta \to 0$.
This proves positivity: $\sZ_\alpha(x_1 , \ldots, x_{2N}) > 0$.
{ \ }$\hfill \qed$

\begin{figure}
\includegraphics[width=.6\textwidth]{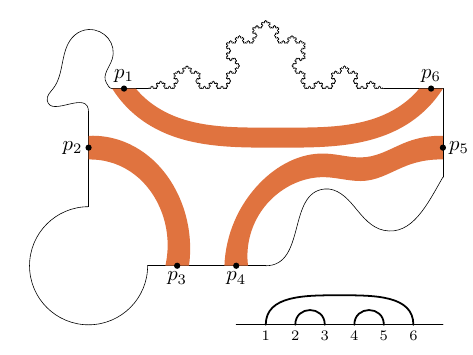}
\caption{\label{fig: square grid approx 3}
Illustration of the choice of non-overlapping subdomains 
$\domain_\ell \subset \domain$ in the proof of positivity of $\PartF_\alpha$.
}
\end{figure}

\bigskip{}

\section{\label{appendix: boundary visits}Boundary visit probabilities and third order PDEs}


This section addresses the convergence in the scaling limit of the loop-erased
random walk boundary visit probabilities, as well as the partial differential equations and
boundary conditions satisfied by the limit functions.
The main purpose is thus to finish the proof of Theorem~\ref{thm: scaling limit of LERW bdry visits}
about the boundary visit probabilities.

As in Section~\ref{sec: boundary visits from partition functions},
let ${ \omega \in \set{+,-}^{N'} }$ be a specification of an order of boundary visits for a chordal curve
and let $\alpha (\omega)$ be the link pattern obtained by 
opening up the curve at each boundary visit,
as 
illustrated in Figure~\ref{fig: alpha omega}.
Recall from Corollary~\ref{cor: LERW boundary visit probabilities with partition functions}
that the boundary visit probability of a LERW in the order~$\omega$
can be expressed in terms of a connectivity probability $Z_{\alpha(\omega)}$
for multiple branches of the UST.
In this section, we study the scaling limit $\delta \to 0$ of this 
probability.
Note that the proof of our previous scaling limit result, 
Theorem~\ref{thm: scaling limit of partition functions}
for the connectivity probabilities 
$Z^{\Gr^\delta}_\alpha (e_1^\delta , \ldots , e_{2N}^\delta)$,
does not directly apply to the case of the boundary visit probability,
since now some pairs of consecutive boundary
edges $e_j^\delta, e_{j+1}^\delta$ tend to the same point in the scaling limit.
Therefore, significantly more care is needed to establish 
Theorem~\ref{thm: scaling limit of LERW bdry visits}.

We begin in Section~\ref{sec: determinant Taylor expansions}
with a refined treatment of the connectivity probabilities 
$Z^{\Gr^\delta}_\alpha(e_1^\delta , \ldots , e_{2N}^\delta)$
in the scaling limit,
applicable also when some boundary edges $e_{j}^\delta, e_{j+1}^\delta$
are at one lattice unit away from each other.
The combinatorial techniques of Section~\ref{subsec: inverse Fomin sums}
are used to derive
Theorems~\ref{thm: continuous boundary visit probabilities}
and~\ref{thm: discrete boundary visit probabilities}, which show that
the renormalized discrete boundary visit probability has scaling limit $\Ampl_\omega$,
which can moreover be expressed as a limit of the continuum partition function $\PartF_{\alpha(\omega)}$.

Section~\ref{sec: third order PDEs} takes care of
the remaining claim in Theorem~\ref{thm: scaling limit of LERW bdry visits}:
the second and third order partial differential equations for $\Ampl_\omega$,
predicted by conformal field theory.
Once we have been able to exchange the
order of the scaling limit with another limit 
by Theorem~\ref{thm: discrete boundary visit probabilities},
these PDEs can be proved by a 
fusion argument parallel to one in~\cite{Dubedat-SLE_and_Virasoro_representations_fusion},
see Lemma~\ref{lem: fusion for third order PDEs}.


The proof of Theorem~\ref{thm: scaling limit of LERW bdry visits}
is then summarized in Section~\ref{sub: proof of bdry visit thm}.

Finally, Section~\ref{sub: asymptotics of boundary visits} is a
supplement to the content of Theorem~\ref{thm: scaling limit of LERW bdry visits}:
in Propositions~\ref{prop: asymptotics for xin} and \ref{prop: asymptotics for visit points}
we establish the predicted asymptotic behaviors for the scaling limit functions, which
together with the system of PDEs has been conjectured to uniquely determine
them~\cite{JJK-SLE_boundary_visits}.

\subsection{\label{sec: determinant Taylor expansions}Scaling limit of boundary visit probabilities}

Throughout this section, we adopt the following assumptions and notation.
A link pattern $\alpha \in \LP_N$ is fixed, and among $2N$ counterclockwise ordered boundary
edges $e_1^\delta , \ldots, e_{2N}^\delta$, 
some consecutive pairs are one lattice unit
apart. Let $N' \leq N$ denote the number of such pairs, and for $1 \leq s \leq N'$,
denote the indices of the pairs by $j_s, j_s + 1$.
The link pattern $\alpha$ must not have any links formed by these pairs, i.e.,
$\upwedgeat{{j_s}} \notin \alpha$.
As in Section~\ref{sec: boundary visits from partition functions},
in this setup there is then an edge $\hat{e}_s^\delta$ at unit distance from the boundary
which joins the interior vertices of the two boundary edges $e_{j_s}^\delta$
and $e_{j_s+1}^\delta$, see Figure~\ref{fig: edge at unit distance from the boundary} 
(the two boundary edges $e_{j_s}^\delta, e_{j_s+1}^\delta$
are then the edges $\hat{e}_{s;1}^\delta , \hat{e}_{s;2}^\delta$ in the figure). 

The graphs $\Gr^\delta$ form a square grid approximation
of a domain $\domain$, 
and the boundary $\bdry \domain$ is locally a vertical or horizontal line near 
the points
\begin{align*}
p_j = \lim_{\delta \to 0} e_j^\delta \qquad \qquad \text{ and } \qquad \qquad
\hat{p}_s = \lim_{\delta \to 0} \hat{e}_s^\delta
    = \lim_{\delta \to 0} e_{j_s}^\delta = \lim_{\delta \to 0} e_{j_s + 1}^\delta.
\end{align*}
The limit points $p_j$, for $j \notin \{ j_1 , j_1 + 1 , \ldots, j_{N'}, j_{N'} +1 \}$,
and $\hat{p}_s$, for $s \in \set{1, 2, \ldots, N'}$, are assumed to be distinct.
In the reference domain $\bH$, 
we denote the boundary points by $x_1 < x_2 < \cdots < x_{2N}$,
as in Section~\ref{sec: scaling limits of partition functions}.

\begin{thm}
\label{thm: continuous boundary visit probabilities}
Let $\alpha \in \LP_N$ and let $\big( (j_s , j_{s}+1) \big)_{s=1}^{N'}$
be a collection of $N'$ disjoint pairs such that $\upwedgeat{ { j_s }} \not \in \alpha$ holds for all $s$.
Let $\PartF_\alpha$ be the function defined in 
Equation~\eqref{eq: SLE partition function as a sum of LPdets}.
Then, the iterated limit 
\begin{align}
\label{eq: continuous boundary visit probability in H}
\lim_{x_{j_1}, x_{j_1 + 1} \to \hat{x}_1 }  \frac{1}{ x_{j_1 + 1} - x_{j_1} } \cdots \lim_{x_{j_{N'}}, x_{j_{N'} + 1} \to \hat{x}_{N'} }
    \frac{1}{ x_{j_{N'} + 1} - x_{j_{N'}}  }  \PartF_\alpha (x_1, \ldots, x_{2N}) 
\end{align}
exists, 
is finite, explicitly given by the
replacing algorithm~\ref{alg: continuous replacing} below, and 
independent of the order of the limits.
\end{thm}

In particular, in the case where $\alpha = \alpha(\omega)$ for 
some boundary visit order $\omega$, this limit function is denoted by 
$\Ampl_\omega(\xin ; \hat{x}_1 , \ldots , \hat{x}_{N'}; \xout )$,
where $\xin = x_i$ and $\xout = x_j$ for the only indices $i,j$ 
not in the pairs $(j_s , j_s +1)$
--- see Equation~\eqref{eq: iterated limit definition of boundary visit amplitude}
and Figure~\ref{fig: boundary visit amplitude as a limit}.
In Proposition~\ref{prop: scaling limit of boundary visits},
the function $\Ampl_\omega$ will be shown to be a positive
solution to the second and third order PDEs predicted for 
$\SLE$ boundary visit amplitudes in \cite{JJK-SLE_boundary_visits}.

\begin{figure}
\bigskip
\bigskip
\begin{displaymath}
\xymatrixcolsep{3.5pc}
\xymatrixrowsep{2.5pc}
\xymatrix{
        & \begin{minipage}{6cm} \begin{center} \includegraphics[width=\textwidth]{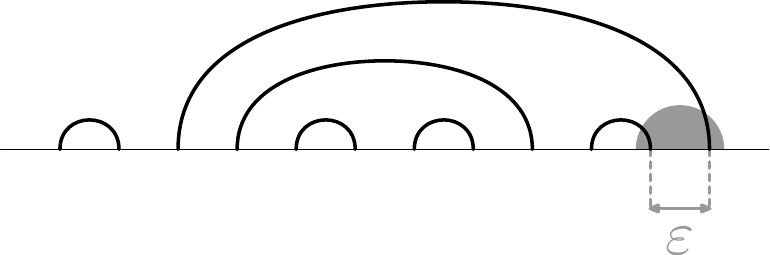} \end{center} \end{minipage}\ar[d] \\
                & \begin{minipage}{6cm} \begin{center} \includegraphics[width=\textwidth]{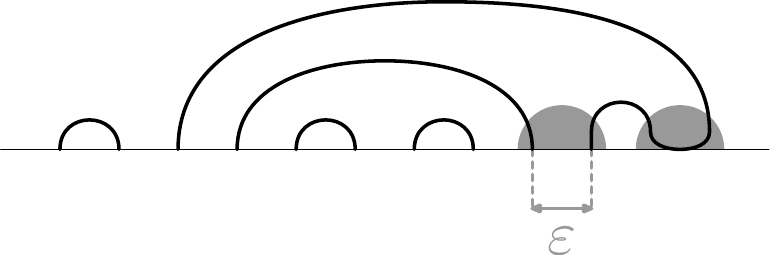} \end{center} \end{minipage}\ar[d] \\
                        & \begin{minipage}{6cm} \begin{center} \includegraphics[width=\textwidth]{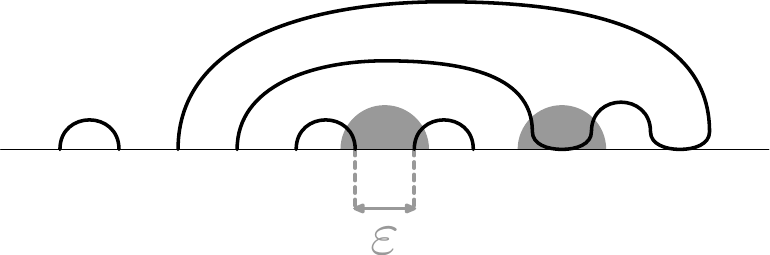} \end{center} \end{minipage}\ar[d] \\
                                & \begin{minipage}{6cm} \begin{center} \includegraphics[width=\textwidth]{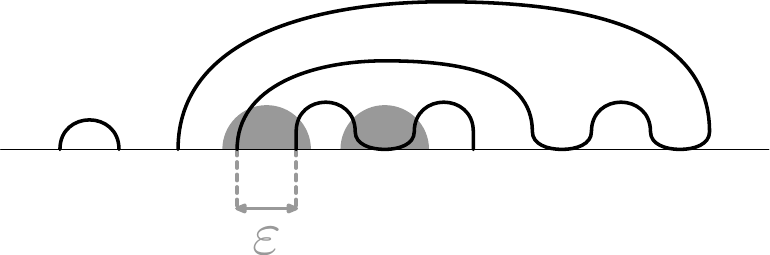} \end{center} \end{minipage}\ar[d] \\
                                        & \begin{minipage}{6cm} \begin{center} \includegraphics[width=\textwidth]{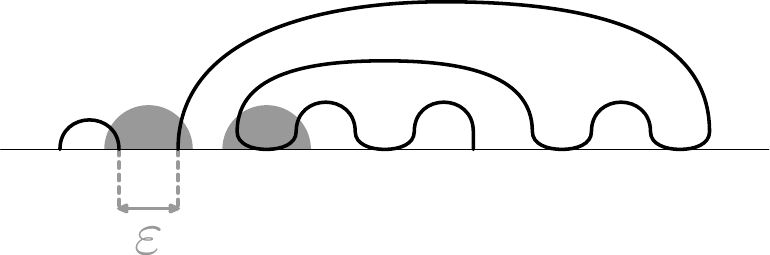} \end{center} \end{minipage}\ar[d] \\
        & \hspace{-0mm}\begin{minipage}{6cm} \begin{center}  \includegraphics[width=\textwidth]{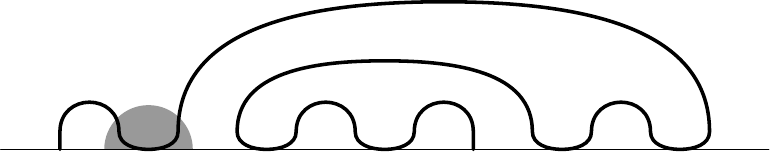} \end{center} \end{minipage}  & 
        }
\end{displaymath}
\bigskip
\bigskip
\caption{\label{fig: boundary visit amplitude as a limit}
Schematic illustration of the iterated limits \eqref{eq: continuous boundary visit probability in H}
in Theorem~\ref{thm: continuous boundary visit probabilities}.
}
\end{figure}

It is in fact natural to associate such functions to each domain $\domain$
by conformal covariance. Let $\confmap \colon \domain \to \bH$ be a conformal map.
We assume below that the boundary is a straight segment locally near all
boundary points where derivatives of $\confmap$ are needed.

First of all, the Brownian excursion kernel $\ExcKdom$ 
in the domain $\domain$ is given by
\begin{align}\label{eq: convariance rule of partition function}
\ExcKdom (p_1 , p_2) = 
     \frac{1}{\pi} \; |\confmap'(p_1)| \; |\confmap'(p_2)| \; 
     \ExcKH (\confmap(p_1) , \confmap(p_2)) ,
\end{align}
where $\ExcKH(x_1,x_2) = (x_2-x_1)^{-2}$. Second,
the limit of the connectivity probabilities of 
Theorem~\ref{thm: scaling limit of partition functions} in the domain $\domain$
is a determinant of the Brownian excursion kernels, given by
\[ \PartF_\alpha^{\domain}(p_1 , \ldots, p_{2N})
    := \frac{1}{\pi^{N}} \times \prod_{j=1}^{2N} |\confmap'(p_j)| \times \PartF_\alpha \big( \confmap(p_1) , \ldots, \confmap(p_{2N}) \big) .\]

Let $S$ be a set of indices $s$ specifying the pairs $(j_s , j_s + 1)$, and
let $J = \set{1, \ldots, 2N} \setminus \bigcup_{s \in S} \set{j_s , j_s +1 }$ 
be the set of all the other indices. Let $\boldsymbol{x} = (x_j)_{j \in J}$ 
and $\hat{\boldsymbol{x}} = (\hat{x}_s)_{s \in S}$ 
denote the corresponding real variables as in 
Theorem~\ref{thm: continuous boundary visit probabilities}.
The limit \eqref{eq: continuous boundary visit probability in H} 
in the statement of the theorem is a function
\[ F(\boldsymbol{x} ; \hat{\boldsymbol{x}}) . \]
In any other domain $\domain$, we can form a similar limit
\begin{align}\label{eq: limit in general domain}
F^\domain (\boldsymbol{p} ; \hat{\boldsymbol{p}}) :=
\lim_{p_{j_1}, p_{j_1 + 1} \to \hat{p}_1 }  \frac{1}{ |p_{j_1 + 1} - p_{j_1} | } \cdots \lim_{p_{j_{N'}}, p_{j_{N'} + 1} \to \hat{p}_{N'} }
    \frac{1}{ | p_{j_{N'} + 1} - p_{j_{N'}} | }  \PartF_\alpha^\domain (p_1, \ldots, p_{2N}) .
\end{align}
This limit can then be related to 
the limit~\eqref{eq: continuous boundary visit probability in H} by
the conformal covariance property
\begin{align}\label{eq: conformal covariance for general limit}
F^\domain (\boldsymbol{p} ; \hat{\boldsymbol{p}}) =
\frac{1}{\pi^{N}} \times \prod_{j \in J} |\confmap'(p_j)| \times \prod_{s \in S} |\confmap'(\hat{p}_s)|^3
    \times F\big( \confmap(\boldsymbol{p}) ; \confmap(\hat{\boldsymbol{p}}) \big) ,
\end{align}
where, for each point $\hat{p}_s$, two powers of $|\confmap'(\hat{p}_s)|$
originate from the factors $|\confmap'(p_{j_s})|$ and $|\confmap'(p_{j_s+1})|$ 
in the conformal covariance 
rule~\eqref{eq: convariance rule of partition function} 
of the functions $\PartF_\alpha^\domain$,
and one additional power comes from the factor 
$\frac{1}{ |p_{j_s + 1} - p_{j_s} |} \approx 
\frac{|\confmap'(\hat{p}_s)|}{| \confmap(p_{j_s + 1}) - \confmap(p_{j_s}) |} $.
The constant factors including $\pi$ are included for convenience in the domain 
specific notations $\ExcKdom$, $\PartF_\alpha^\domain$ and $F^\domain$, 
but are not included for simplicity in the functions $\ExcKH$, $\PartF_\alpha$, 
and $F$, defined on real variables. 

The following theorem is a key ingredient in the proof of 
Theorem~\ref{thm: scaling limit of LERW bdry visits}.
Informally, the content is that we can exchange the order of the scaling
limit $\delta \to 0$ with the limit~\eqref{eq: limit in general domain}
that collapses pairs of link endpoints to boundary visit locations as in
Figure~\ref{fig: boundary visit amplitude as a limit}.
The importance of this is twofold: first of all, it gives an explicit formula for the
scaling limit of boundary visit probabilities, and secondly,
for limits of type~\eqref{eq: limit in general domain} higher order partial differential
equations can be proved by a fusion argument.
\begin{thm}
\label{thm: discrete boundary visit probabilities}
Let $\alpha\in \LP_N$, and let $\big( (j_s , j_{s}+1) \big)_{s \in S}$ 
be a collection of disjoint pairs such that $\upwedgeat{j_s} \not \in \alpha$ 
holds for all $s \in S$. Let the boundary edges  
$e_1^\delta, \ldots , e_{2N}^\delta \in \bdry \Edg^\delta$ be as above.
Then, the UST connectivity probability 
\begin{align*}
Z^{\Gr^\delta}_{\alpha} (e_1^\delta, \ldots , e_{2N}^\delta)
\end{align*}
is, for each fixed $\delta$, explicitly given by the replacing 
algorithm \ref{alg: discrete replacing} below.
Its 
scaling limit as $\delta \to 0$ is 
given by Equation~\eqref{eq: limit in general domain}:
\begin{align*}
\lim_{\delta \to 0} \frac{1}{\delta^{N' + 2N}} 
Z^{\Gr^\delta}_{\alpha} (e_1^\delta, \ldots , e_{2N}^\delta)
= F^\domain (\boldsymbol{p} ; \hat{\boldsymbol{p}}) .
\end{align*}
\end{thm}

The proofs the Theorems~\ref{thm: continuous boundary visit probabilities} and 
\ref{thm: discrete boundary visit probabilities} are based on the 
replacing algorithms~\ref{alg: continuous replacing} 
and \ref{alg: discrete replacing} below. 
Recall from Section \ref{sec: applications to UST} 
that we denote by $D^\delta_{\tau; i}$ the discrete tangential derivative
with respect to the $i$:th variable, and by $\partial_{ \tau; i}$ 
the usual tangential derivative with respect to the $i$:th variable, 
taken in the counterclockwise direction.
For a function $f \colon \Vert^\delta \to \bC$, 
the discrete tangential derivative is defined 
at $\hat{e}_s = \edgeof{e_{j_s}^\circ}{e_{j_s + 1}^\circ}$ by
\begin{align*}
( D^{\delta}_{\tau} f ) \big(\hat{e}_s\big) = 
( D^{\delta}_{\tau} f ) \big(\edgeof{e_{j_s}^\circ}{e_{j_s + 1}^\circ}\big)
:= \frac{f (e_{j_s + 1}^\circ) - f (e_{j_s }^\circ) }{\delta}.
\end{align*}

\begin{algorithm}[Replacing algorithm for 
Theorem~\ref{thm: continuous boundary visit probabilities}]
Under the notation and assumptions of 
Theorem~\ref{thm: continuous boundary visit probabilities},
the following procedure of replacements gives an explicit expression for 
the limit~\eqref{eq: continuous boundary visit probability in H}.
\label{alg: continuous replacing}
\begin{itemize}
\item[0.] Start from the expression
\eqref{eq: SLE partition function as a sum of LPdets} 
of the partition function,
\[ \PartF_\alpha(x_1 , \ldots, x_{2N}) 
= \sum_{\beta \DPgeq \alpha} \CItilingsof(\alpha / \beta) \, 
\LPdet{\beta}{\ExcKH} (x_1 , \ldots , x_{2N}) .\]
\item[1.] 
Replace all the kernel entries
$\ExcKH (x_{j_s}, x_{j_s + 1})$ and $\ExcKH (x_{j_s +1}, x_{j_s})$,
for $s = 1, \ldots, N'$, 
by zero in all the determinants 
$\LPdet{\beta}{\ExcKH}$ in the above sum.
\item[2.] For all $s = N',N'-1, \ldots, 2, 1$, in all the (now modified) determinants,
replace the (unique) row or column that is a function of $x_{j_s + 1}$ by its 
derivative at $\hat{x}_{s}$, and in the (unique) row or column that is a function
of $x_{j_s}$, replace all appearances of $x_{j_s}$ by $\hat{x}_{s}$.
\end{itemize}
\end{algorithm}
\begin{rem}\label{rem: replacing algorithm in general domain}
\textit{
A procedure analogous to the replacing algorithm~\ref{alg: continuous replacing}
can also be applied to the limit~\eqref{eq: limit in general domain}
in any domain $\domain$ with boundary points on 
straight boundary segments so that the conformal covariance 
formula~\eqref{eq: conformal covariance for general limit} is valid.
The only changes are replacing $\ExcKH$ by $\ExcKdom$, and
derivatives in step 2 by counterclockwise tangential derivatives.
}
\end{rem}

\begin{algorithm}[Replacing algorithm for 
Theorem~\ref{thm: discrete boundary visit probabilities}]
\label{alg: discrete replacing}
Under the notation and assumptions of 
Theorem~\ref{thm: discrete boundary visit probabilities},
the following procedure of replacements gives an explicit expression for
the UST connectivity probability $Z_{\alpha}^{\Gr^\delta}$. 
\begin{itemize}
\item[0.] Start from the expression
\eqref{eq: discrete partition function as a sum of LPdets} 
of the connectivity probability,
\[ Z_{\alpha}^{\Gr^\delta} (e_1^\delta, \ldots, e_{2N}^\delta)
= \sum_{\beta \DPgeq \alpha} \CItilingsof(\alpha / \beta) \, 
\LPdet{\beta}{\ExcK} (e_1^\delta , \ldots , e_{2N}^\delta) .\]
\item[1.] 
Replace all the kernel entries
$\ExcK (e_{j_s}^\delta, e_{j_s + 1}^\delta)$ 
and $\ExcK (e_{j_s + 1}^\delta, e_{j_s}^\delta)$
by zero in all the determinants 
$\LPdet{\beta}{\ExcK}$ in the above sum.
%

\item[2.] For all $s = N',N'-1, \ldots, 2, 1$, in all the (now modified) determinants, replace the (unique) row or column that is a function of 
$e_{j_s + 1}^\delta$ by its discrete tangential derivative at $\hat{e}_s^\delta$ 
and add a factor $\delta$ in front of the modified 
connectivity probability.  
\end{itemize}
\end{algorithm}

\begin{proof}[Proof for Algorithm~\ref{alg: continuous replacing}]
The proof is based on repeated application of the computation rules for inverse Fomin type sums 
listed in Proposition~\ref{prop: Zero-replacing Rule} of Section~\ref{subsec: inverse Fomin sums}.

Step~0 starts from the inverse Fomin type sum $\PartF_\alpha(x_1 , \ldots, x_{2N}) = \sum_{\beta \DPgeq \alpha}
\CItilingsof(\alpha / \beta) \, \LPdet{\beta}{\ExcKH} (x_1 , \ldots , x_{2N})$.

Step~1 applies 
Proposition~\ref{prop: Zero-replacing Rule}(c) to obtain another inverse Fomin type sum with the same value.

In Step~2, each value of $s$ takes care of one renormalized limit
of Equation \eqref{eq: continuous boundary visit probability in H}. For example, 
in the first step $s = N'$, we expand the kernel entries with index $j_{N'}+1$ as 
\begin{align}\label{eq: Taylor expansion}
\ExcKH(x_{j_{N'} + 1}, x_i )
= \ExcKH(x_{j_{N'}}, x_i ) + (x_{j_{N'} + 1} - x_{j_{N'}}) \,
\partial_{\tau; 1}  \ExcKH(x_{j_{N'}}, x_i )
+ \oo(\vert x_{j_{N'} + 1} - x_{j_{N'}} \vert) 
\end{align}
for $i \neq j_{N'}, j_{N'}+1$.
(Other values of $s$ are handled identically, except that derivatives of $\ExcKH$ might appear here.)
Then, by the multilinearity of determinants, we can split the inverse Fomin type sum into 
three different inverse Fomin type sums, with the kernel entries $\mathfrak{K}(j_{N' + 1}, i )$ 
given by one term of the above expansion. The first one, corresponding to the first term 
in~\eqref{eq: Taylor expansion}, 
has a kernel satisfying the property
denoted $\mathfrak{K}(j_{N'} , \; \cdot \,) = \mathfrak{K}(j_{N'}+1 , \; \cdot \,)$, i.e.,
\begin{align*}
\mathfrak{K}(j_{N'}, i ) = \ExcKH(x_{j_{N'}}, x_{i}) = \mathfrak{K} (j_{N'} +1 , i ) .
\end{align*}
This inverse Fomin type sum is zero by Proposition~\ref{prop: Zero-replacing Rule}(b). 
The inverse Fomin type sum corresponding to the third term of \eqref{eq: Taylor expansion} vanishes 
in the renormalized limit. The second term yields the 
asserted replacing algorithm.
\end{proof}

\begin{proof}[Proof for Algorithm \ref{alg: discrete replacing}]
The proof 
is practically identical to the previous one. The only difference between
the proofs is that when writing the discrete Taylor expansion, there is no
error term since the discrete tangential derivative is a difference:
for any function $f \colon \Vert^\delta \to \bC$, we have
\begin{align*}
f(e_{j_s + 1}^\circ) = f(e_{j_s}^\circ) + 
\delta \, ( D^\delta_{\tau} f ) \big(\hat{e}_s \big) .
\end{align*}
By multilinearity arguments identical to
the case of algorithm~\ref{alg: continuous replacing},
we cancel the first term of this expansion, 
and extract the coefficient $\delta$ from the determinant.
This leaves a determinant with the discrete tangential
derivatives, as desired.
\end{proof}
Using the replacing algorithms~\ref{alg: continuous replacing} and 
\ref{alg: discrete replacing}, we can now prove 
Theorems~\ref{thm: continuous boundary visit probabilities} and
\ref{thm: discrete boundary visit probabilities}.

\begin{proof}[Proof of Theorem \ref{thm: continuous boundary visit probabilities}]
Above we saw that the replacing algorithm~\ref{alg: continuous replacing} gives 
the limit~\eqref{eq: continuous boundary visit probability in H}.
The multiple limit is independent of the order of limits, 
since the replacing algorithm is. 
\end{proof}


\begin{proof}[Proof of Theorem~\ref{thm: discrete boundary visit probabilities}]
The replacing algorithm \ref{alg: discrete replacing} guarantees that 
we can omit the exponent $N'$
in the expression 
$\frac{1}{\delta^{N' + 2N}} Z_{\alpha}^{\Gr^\delta} (e_1^\delta, \ldots, e_{2N}^\delta)$,
and replace in the UST connectivity probability
\[ Z_{\alpha}^{\Gr^\delta} (e_1^\delta, \ldots, e_{2N}^\delta)
    = \sum_{\beta \DPgeq \alpha} \# \CItilingsof(\alpha / \beta) \, \LPdet{\beta}{\ExcK} (e_1^\delta , \ldots , e_{2N}^\delta) \]
any kernel entries
involving $e_{j_s + 1}^\delta$ 
by the discrete tangential derivatives at $\hat{e}_s^\delta$.
By Lemma~\ref{lem: scaling limit of discrete excursion kernels},
the random walk excursion kernels and their discrete derivatives, 
normalized by $\delta^{-2}$, converge to the Brownian
excursion kernels and their derivatives. Hence, by the similarity of the replacing
algorithms \ref{alg: continuous replacing} and \ref{alg: discrete replacing} 
and Remark~\ref{rem: replacing algorithm in general domain},
$\frac{1}{\delta^{N' + 2N}} Z_{\alpha} (e_1^\delta, \ldots, e_{2N}^\delta)$ tends
to the limit given by Equation~\eqref{eq: limit in general domain}. 
\end{proof}

\subsection{\label{sec: third order PDEs}Third order partial differential equations through fusion}

We now prove the third order PDEs predicted by CFT for the functions
$\Ampl_\omega$
of Theorem~\ref{thm: scaling limit of LERW bdry visits}.
The proof needs two crucial inputs, discussed below.

Fix 
a boundary visit order $\omega \in \set{+,-}^{N'}$, and
let $\alpha = \alpha(\omega)$ be the link pattern associated to this boundary
visit order as in Section~\ref{sec: boundary visits from partition functions}
and in Figure~\ref{fig: alpha omega}. In this case, we have
$N=N'+1$. As the first input to the proof, we use the iterated limit expression
for $\Ampl_\omega$ obtained from 
Theorems~\ref{thm: continuous boundary visit probabilities} 
and~\ref{thm: discrete boundary visit probabilities}:
\begin{align}\label{eq: iterated limit definition of boundary visit amplitude}
\Ampl_\omega(\xin ; \hat{x}_1 , \ldots , \hat{x}_{N'} ; \xout)
= \; &  \lim_{\substack{x_{j_1}, x_{j_1 + 1} \to \hat{x}_1 \\ \vdots \\ x_{j_{N'}}, x_{j_{N'} + 1} \to \hat{x}_{N'}} }
    \prod_{s=1}^{N'} \frac{1}{ x_{j_s + 1} - x_{j_s} } \times \PartF_{\alpha(\omega)} (x_1, \ldots, x_{2N}),
\end{align}
where the limits are taken in a specific order: 
first $x_{j_1}, x_{j_1 + 1} \to \hat{x}_1$,
then $x_{j_2}, x_{j_2 + 1} \to \hat{x}_2$, and so on. 
The order of the limits does not matter for the result, but we rely on taking 
the limits one at a time. 

As the second input, we use the $2N$ second order 
partial differential equations~\eqref{eq: PDE for multiple SLEs at kappa equals 2} 
for the partition functions $\PartF_\alpha$, which
were verified by an explicit calculation in Lemma~\ref{lem: PDEs for determinants}
of Section~\ref{sec: applications to SLEs}.

The strategy of the proof is as follows. We take the limits 
in~\eqref{eq: iterated limit definition of boundary visit amplitude}
one at a time, and check recursively that after taking
$r$ limits, the resulting function of $2N-r$ variables satisfies $2N-2r$
PDEs of second order and $r$ PDEs of third order. This recursion step is
proven with a fusion procedure 
described by Dub\'edat in \cite{Dubedat-SLE_and_Virasoro_representations_fusion}.
A version of this fusion procedure sufficient for our purposes is 
Lemma~\ref{lem: fusion for third order PDEs} below. Technical assumptions
in the lemma could be relaxed (see~\cite{Dubedat-SLE_and_Virasoro_representations_fusion}),
but since our explicit functions clearly admit well-behaved
Frobenius expansions, we prefer the approach that only invokes direct calculations.

To make a clear distinction,
the variables that have not yet been involved in limits will be denoted $x_j$,
and the variables that have been involved in some limit will be denoted $\hat{x}_s$,
with indices $j$ and $s$ in index sets that are appropriate subsets
$J \subset \set{1,\ldots,2N}$ and $S \subset \set{1, \ldots, N'}$ depending on how
many limits have been taken. We will thus consider functions
\[ F(\boldsymbol{x};\hat{\boldsymbol{x}}) \]
of variables $\boldsymbol{x} = (x_j)_{j \in J}$ and $\hat{\boldsymbol{x}} = (\hat{x}_s)_{s \in S}$.

Let us now write down the second and third order PDEs in the general form,
which applies likewise to the input function $\PartF_\alpha$, the final answer
$\Ampl_\omega$, and all the intermediate steps.
Without additional difficulty we first work with general $\kappa>0$,
and afterwards specialize to $\kappa=2$ to derive the main results of this section.
We thus use the conformal weight parameters $h = h_{1,2}(\kappa) = \frac{6 - \kappa}{2 \kappa}$ and
$\hat{h} = h_{1,3}(\kappa) = \frac{8-\kappa}{\kappa}$, which at $\kappa=2$ become
$h=1$ and $\hat{h} = 3$.
We express the differential equations in terms of the following
first order differential operators:
\begin{align*}
\sL_{-n}^{(j)} =
    \sum_{k \in J \setminus \set{j}} \Big( \frac{(n-1) h}{(x_k-x_j)^{n}} - \frac{1}{(x_k-x_j)^{n-1}} \pder{x_k} \Big)
    + \sum_{t \in S} \Big( \frac{(n-1) \hat{h}}{(\hat{x}_t-x_j)^{n}} - \frac{1}{(\hat{x}_t-x_j)^{n-1}} \pder{\hat{x}_t} \Big)
\end{align*}
for $n \in \bZpos$ and $j \in J$, and 
\begin{align*}
\widehat{\sL}_{-n}^{(s)} =
    \sum_{k \in J } \Big( \frac{(n-1) h}{(x_k-\hat{x}_s)^{n}} - \frac{1}{(x_k-\hat{x}_s)^{n-1}} \pder{x_k} \Big)
    + \sum_{t \in S \setminus \set{s}} \Big( \frac{(n-1) \hat{h}}{(\hat{x}_t-\hat{x}_s)^{n}} - \frac{1}{(\hat{x}_t-\hat{x}_s)^{n-1}} \pder{\hat{x}_t} \Big)
\end{align*}
for $n \in \bZpos$ and $s \in S$. Although we keep it implicit in the notation, note that
$\sL_{-n}^{(j)}$ and $\widehat{\sL}_{-n}^{(s)}$ depend on $J$ and $S$,
and they thus denote different operators at different intermediate stages.
The second order PDEs will always be of the form
\begin{align}\label{eq: general second order PDEs at generic kappa}
\sD^{(j)} F(\boldsymbol{x};\hat{\boldsymbol{x}}) = 0
\qquad \text{ where $j \in J$ and } \qquad
\sD^{(j)} = \pdder{x_j} - \frac{4}{\kappa} \sL^{(j)}_{-2}
\end{align}
and the third order PDEs will always be of the form
\begin{align}\label{eq: general third order PDEs at generic kappa}
\widehat{\sD}^{(s)} F(\boldsymbol{x};\hat{\boldsymbol{x}}) = 0
\qquad \text{ where $s \in S$ and } \qquad
\widehat{\sD}^{(s)} = \pddder{\hat{x}_s}
    - \frac{16}{\kappa} \widehat{\sL}^{(s)}_{-2} \pder{\hat{x}_s}
    + \frac{8(8-\kappa)}{\kappa^2} \widehat{\sL}^{(s)}_{-3} .
\end{align}
Both are special cases of the partial differential equations of conformal field theory introduced by 
Belavin, Polyakov, and Zamolodchikov
in~\cite{BPZ-infinite_conformal_symmetry_in_2D_QFT, BPZ-infinite_conformal_symmetry_of_critical_fluctiations},
and specifically within the family for which explicit formulas were given by
Benoit and Saint-Aubin in~\cite{BSA-degenerate_CFTs_and_explicit_expressions}.
Note that at $\kappa=2$ and $J=\set{1,\ldots,2N}$ and $S=\emptyset$,
Equation~\eqref{eq: general second order PDEs at generic kappa} is nothing
but the system~\eqref{eq: PDE for multiple SLEs at kappa equals 2}.


The next key result pertains to fusion in conformal field theory,
and it achieves the $r$:th recursive step of our procedure.
In the limit $x_{j_r}, x_{j_r + 1} \to \hat{x}_r$ of interest,
solutions to the PDEs~\eqref{eq: general second order PDEs at generic kappa}
above can have power law behavior with two possible exponents: the roots
$\Delta = -2h = \frac{\kappa-6}{\kappa}$ and $\Delta' = \hat{h} - 2 h = \frac{2}{\kappa}$
of the indicial equation of a Frobenius series
(see the proof of Lemma~\ref{lem: fusion for third order PDEs} below). 
The result specifically addresses the coefficient of the subleading power law behavior, with
exponent\footnote{For $\kappa<8$ we have $\Delta < \Delta'$, whence the terminology
leading and subleading for $\Delta$ and $\Delta'$, respectively.}
$\Delta'$.
The argument follows ideas in~\cite{Dubedat-SLE_and_Virasoro_representations_fusion}.
We include the proof, because this is a key step towards our main results.

\begin{lem}
\label{lem: fusion for third order PDEs}
Let $J \subset \set{1,\ldots, 2N}$ be a subset of size $\# J = 2N - 2 (r-1)$
and let $S = \set{1, \ldots, r-1}$. 
Let also $j_{r} , j_{r}+1 \in J$ be two consecutive indices,
and $J' = J \setminus \set{j_{r} , j_{r}+1}$ and $S' = \set{1, \ldots, r}$.
Denote
\[ \hat{x}_r = \frac{1}{2} \big( x_{j_r+1} + x_{j_r} \big) \qquad
\text{ and } \qquad
\eps = x_{j_r+1} - x_{j_r} .\]
Let $F(\boldsymbol{x};\hat{\boldsymbol{x}})$ be a function which satisfies
Equations~\eqref{eq: general second order PDEs at generic kappa}
for all $j \in J$ and Equations~\eqref{eq: general third order PDEs at generic kappa}
for all $s \in S$, and which for any $\boldsymbol{x}' = (x_j)_{j \in J'}$ and $\hat{\boldsymbol{x}}' = (\hat{x}_s)_{s \in S'}$
and small positive $\eps$ has the Frobenius series expansion
\begin{align}\label{eq: Frobenius series}
F(\boldsymbol{x};\hat{\boldsymbol{x}}) = \sum_{m=0}^\infty \eps^{\Delta'+m} \, F^{(m)} (\boldsymbol{x}';\hat{\boldsymbol{x}}') 
\end{align}
with exponent $\Delta' = \frac{2}{\kappa}$.
Assume that the coefficients $F^{(m)}$ are smooth functions of $(\boldsymbol{x}';\hat{\boldsymbol{x}}')$,
and that for any compact subset $K$ of the 
domain of their definition and for any multi-index~$\alpha$, 
there exist positive constants $r_{K;\alpha}, C_{K;\alpha}$ such that the 
bound on partial derivatives,
\[ \Big| \partial^\alpha F^{(m)} (\boldsymbol{x}';\hat{\boldsymbol{x}}') \Big| \leq C_{K;\alpha} \, r_{K;\alpha}^{-m} , \]
holds for all $m \in \bN$ and $(\boldsymbol{x}';\hat{\boldsymbol{x}}') \in K$.
Then the limit
\[ 
F_{\lim} (\boldsymbol{x}';\hat{\boldsymbol{x}}') := \lim_{x_{j_r}, x_{j_r + 1} \to \hat{x}_r}
    \frac{1}{(x_{j_r + 1} - x_{j_r})^{\Delta'}} F(\boldsymbol{x};\hat{\boldsymbol{x}})
\]
defines a function $F_{\lim} = F^{(0)}$ of
$\boldsymbol{x}' = (x_j)_{j \in J'}$ and 
$\hat{\boldsymbol{x}}' = (\hat{x}_s)_{s \in S'}$,
which satisfies all the second and third order
PDEs~\eqref{eq: general second order PDEs at generic kappa modified}
and~\eqref{eq: general third order PDEs at generic kappa modified}
associated with the
index sets $J'$ and $S'$ and the corresponding variables.
\end{lem}
\begin{proof}
Paying attention to the dependence of the differential operators on the index sets,
the precise claim is that the limit function 
$F_{\lim} (\boldsymbol{x}';\hat{\boldsymbol{x}}') = F^{(0)}(\boldsymbol{x}';\hat{\boldsymbol{x}}')$
satisfies the differential equations
\begin{align}\label{eq: general second order PDEs at generic kappa modified}
\sD'^{(j)} F^{(0)}(\boldsymbol{x}';\hat{\boldsymbol{x}}') = 0
\qquad \text{ for all $j \in J'$ and } \qquad
\sD'^{(j)} = \pdder{x_j} - \frac{4}{\kappa} \sL'^{(j)}_{-2} ,
\end{align}
where
\begin{align*}
\sL_{-n}'^{(j)} =
    \sum_{k \in J' \setminus \set{j}} \Big( \frac{(n-1) h}{(x_k-x_j)^{n}} - \frac{1}{(x_k-x_j)^{n-1}} \pder{x_k} \Big)
    + \sum_{t \in S'} \Big( \frac{(n-1) \hat{h}}{(\hat{x}_t-x_j)^{n}} - \frac{1}{(\hat{x}_t-x_j)^{n-1}} \pder{\hat{x}_t} \Big)
\end{align*}
and
\begin{align}\label{eq: general third order PDEs at generic kappa modified}
\widehat{\sD}'^{(s)} F^{(0)}(\boldsymbol{x}';\hat{\boldsymbol{x}}') = 0
\qquad \text{ for all $s \in S'$ and } \qquad
\widehat{\sD}'^{(s)} = \pddder{\hat{x}_s} - \frac{16}{\kappa} \widehat{\sL}'^{(s)}_{-2} \pder{\hat{x}_s} + \frac{8(8-\kappa)}{\kappa^2} \widehat{\sL}'^{(s)}_{-3} ,
\end{align}
where
\begin{align*}
\widehat{\sL}'^{(s)}_{-n} =
    \sum_{k \in J'} \Big( \frac{(n-1) h}{(x_k-\hat{x}_s)^{n}} - \frac{1}{(x_k-\hat{x}_s)^{n-1}} \pder{x_k} \Big)
    + \sum_{t \in S' \setminus \set{s}} \Big( \frac{(n-1) \hat{h}}{(\hat{x}_t-\hat{x}_s)^{n}} - \frac{1}{(\hat{x}_t-\hat{x}_s)^{n-1}} \pder{\hat{x}_t} \Big) .
\end{align*}
The proof is divided to three separate cases, each establishing some of the asserted PDEs for $F^{(0)}$:
\begin{itemize}
\item[(i)] $\sD'^{(j)} F^{(0)}(\boldsymbol{x}';\hat{\boldsymbol{x}}') = 0$ for $j \in J' = J \setminus \set{j_r , j_r +1}$
\item[(ii)] $\widehat{\sD}'^{(s)} F^{(0)}(\boldsymbol{x};\hat{\boldsymbol{x}}) = 0$ for $s \neq r$
\item[(iii)] $\widehat{\sD}'^{(r)} F^{(0)}(\boldsymbol{x};\hat{\boldsymbol{x}}) = 0$.
\end{itemize}
The verifications of the PDEs (i) and (ii)
are in principle straightforward, although the number of terms renders the
calculations somewhat lengthy. The truly interesting part is (iii), where
we have the appearance of the new third order PDE
$\widehat{\sD}'^{(r)} F^{(0)}(\boldsymbol{x};\hat{\boldsymbol{x}}) = 0$ via fusion from
two second order PDEs.

{\bf Case (i), second order PDEs in the limit:}
We begin by verifying the straightforward second order PDEs 
$\sD'^{(j)} F^{(0)}(\boldsymbol{x}';\hat{\boldsymbol{x}}') = 0$ for $j \in J'$.
From the original variables $(\boldsymbol{x};\hat{\boldsymbol{x}})$
we change to new variables $(\boldsymbol{x}';\hat{\boldsymbol{x}}';\eps)$
as in the statement of the lemma.
The assumed original PDE
$\sD^{(j)} F(\boldsymbol{x};\hat{\boldsymbol{x}}) = 0$ becomes, after changing
the order of differentiation and summation in the Frobenius series~\eqref{eq: Frobenius series},
\begin{align} \label{eq: trivial second order PDE expansion}
\sum_{m=0}^\infty \sD^{(j)} \Big( \eps^{\Delta'+m} \, F^{(m)} (\boldsymbol{x}';\hat{\boldsymbol{x}}') \Big) = 0 .
\end{align}
The exchange of the order of differentiation and summation is justified by the
assumed locally uniform bounds on the series coefficients and their partial derivatives.

Let us consider the various terms appearing in
$\sD^{(j)} \big( \eps^{\Delta'+m} \, F^{(m)} (\boldsymbol{x}';\hat{\boldsymbol{x}}') \big)$.
We split the differential operator $\sD^{(j)}$ to the following parts:

\begin{description}
\item[(a)] $\pdder{x_j}$
\item[(b)] $-\frac{4}{\kappa} \frac{h}{(x_k-x_j)^{2}} + \frac{4}{\kappa} \frac{1}{x_k-x_j} \pder{x_k}$ for $k \in J' = J \setminus \set{j_r , j_r +1}$
\item[(c)] $-\frac{4}{\kappa} \frac{\hat{h}}{(\hat{x}_t-x_j)^{2}} + \frac{4}{\kappa} \frac{1}{\hat{x}_t-x_j} \pder{\hat{x}_t}$ for $t \in S$
\item[(d)] $-\frac{4}{\kappa} \frac{h}{(x_{j_r}-x_j)^{2}} + \frac{4}{\kappa} \frac{1}{x_{j_r}-x_j} \pder{x_{j_r}}$
\item[(e)] $-\frac{4}{\kappa} \frac{h}{(x_{j_r+1}-x_j)^{2}} + \frac{4}{\kappa} \frac{1}{x_{j_r+1}-x_j} \pder{x_{j_r+1}}$.
\end{description}

The chain rule expresses the derivatives in the old variables appearing in $\sD^{(j)}$
in terms of the new variables as
$ \pder{x_{j_r}} = -\pder{\eps} + \frac{1}{2} \pder{\hat{x}_{r}} $
and $\pder{x_{j_r+1}} = +\pder{\eps} + \frac{1}{2} \pder{\hat{x}_{r}}$.
The variables $x_j$, $x_k$ for $k \in J'$, and $\hat{x}_t$ for $t \in S$ are
not affected by the change of variables, so the terms (a), (b), and (c) simply become
\begin{align}
\label{eq: trivial term 1 in trivial second order PDE}
\tag{\ref{eq: trivial second order PDE expansion}a}
& \eps^{\Delta'+m} \; \pdder{x_j} F^{(m)} (\boldsymbol{x}';\hat{\boldsymbol{x}}') , \\
\label{eq: trivial term 2 in trivial second order PDE}
\tag{\ref{eq: trivial second order PDE expansion}b}
& \eps^{\Delta'+m} \; \frac{4}{\kappa} \,
      \bigg( \frac{1}{x_k-x_j} \pder{x_k} F^{(m)} (\boldsymbol{x}';\hat{\boldsymbol{x}}')
              - \frac{h}{(x_k-x_j)^2} F^{(m)} (\boldsymbol{x}';\hat{\boldsymbol{x}}') \bigg) , \\
\label{eq: trivial term 3 in trivial second order PDE}
\tag{\ref{eq: trivial second order PDE expansion}c}
& \eps^{\Delta'+m} \; \frac{4}{\kappa} \,
      \bigg( \frac{1}{\hat{x}_t-x_j} \pder{\hat{x}_t} F^{(m)} (\boldsymbol{x}';\hat{\boldsymbol{x}}')
              - \frac{\hat{h}}{(\hat{x}_t-x_j)^2} F^{(m)} (\boldsymbol{x}';\hat{\boldsymbol{x}}') \bigg) .
\end{align}

The remaining terms (d) and (e) involve one of the original variables $x_{j_r}$ or $x_{j_r+1}$, which
we express in terms of $\hat{x}_r$ and $\eps$. Noticing
$x_{j_r} - x_j = \hat{x}_{r} - x_j - \frac{\eps}{2}$, we get the expansions
\[ \frac{1}{x_{j_r} - x_j} 
= \frac{1}{\hat{x}_{r} - x_j} \; \sum_{\ell=0}^\infty \frac{\eps^\ell}{2^\ell \, (\hat{x}_{r} - x_j)^\ell}
\qquad \text{ and } \qquad 
\frac{h}{(x_{j_r} - x_j)^2} 
= \frac{h}{(\hat{x}_{r} - x_j)^2} \; \sum_{\ell=0}^\infty \frac{(\ell+1) \; \eps^\ell}{2^\ell \, (\hat{x}_{r} - x_j)^\ell} . \]
Now use the chain rule to write the terms in (d) as
\begin{align*}\label{eq: nontrivial term d in trivial second order PDE}
\tag{\ref{eq: trivial second order PDE expansion}d}
\eps^{\Delta'+m} \; \frac{4}{\kappa} \, 
      \Bigg( \Big( \frac{1}{2} \pder{\hat{x}_r} F^{(m)} (\boldsymbol{x}';\hat{\boldsymbol{x}}')
           - \frac{\Delta'+m}{\eps} \, F^{(m)} (\boldsymbol{x}';\hat{\boldsymbol{x}}') \Big) \; \frac{1}{\hat{x}_r-x_j} \sum_{\ell=0}^\infty \frac{\eps^\ell}{2^\ell \, (\hat{x}_{r} - x_j)^\ell} &  \\
\nonumber
      - F^{(m)} (\boldsymbol{x}';\hat{\boldsymbol{x}}') \; \frac{h}{(\hat{x}_r-x_j)^2} \sum_{\ell=0}^\infty \frac{(\ell+1) \; \eps^\ell}{2^\ell \, (\hat{x}_{r} - x_j)^\ell} \Bigg) & \; .
\end{align*}
Proceeding similarly, the terms in (e) are written as
\begin{align*}\label{eq: nontrivial term e in trivial second order PDE}
\tag{\ref{eq: trivial second order PDE expansion}e}
\eps^{\Delta'+m} \; \frac{4}{\kappa} \, 
      \Bigg( \Big( \frac{1}{2} \pder{\hat{x}_r} F^{(m)} (\boldsymbol{x}';\hat{\boldsymbol{x}}')
           + \frac{\Delta'+m}{\eps} \, F^{(m)} (\boldsymbol{x}';\hat{\boldsymbol{x}}') \Big) \; \frac{1}{\hat{x}_r-x_j} \sum_{\ell=0}^\infty \frac{(-\eps)^\ell}{2^\ell \, (\hat{x}_{r} - x_j)^\ell} & \\
\nonumber
      - F^{(m)} (\boldsymbol{x}';\hat{\boldsymbol{x}}') \; \frac{h}{(\hat{x}_r-x_j)^2} \sum_{\ell=0}^\infty \frac{(\ell+1) \; (-\eps)^\ell}{2^\ell \, (\hat{x}_{r} - x_j)^\ell} \Bigg) & \; .
\end{align*}
From expressions \eqref{eq: trivial term 1 in trivial second order PDE},
\eqref{eq: trivial term 2 in trivial second order PDE}, 
\eqref{eq: trivial term 3 in trivial second order PDE},
\eqref{eq: nontrivial term d in trivial second order PDE}, and
\eqref{eq: nontrivial term e in trivial second order PDE},
we can read the $\eps$-expansion of
$\sD^{(j)} F(\boldsymbol{x};\hat{\boldsymbol{x}})$. Note that terms
of order $\eps^{\Delta'-1}$ cancel in \eqref{eq: nontrivial term d in trivial second order PDE} and
\eqref{eq: nontrivial term e in trivial second order PDE}, and the leading order 
of $\sD^{(j)} F(\boldsymbol{x};\hat{\boldsymbol{x}})$ is
\begin{align*}
\eps^{\Delta'} \; \Bigg( \pdder{x_j} F^{(0)}(\boldsymbol{x}';\hat{\boldsymbol{x}}') & \,
    + \frac{4}{\kappa} \sum_{k \in J'} \frac{1}{x_k-x_j} \pder{x_k} F^{(0)} (\boldsymbol{x}';\hat{\boldsymbol{x}}')
    - \frac{4}{\kappa} \sum_{k \in J'} \frac{h}{(x_k-x_j)^2} F^{(0)} (\boldsymbol{x}';\hat{\boldsymbol{x}}') \\
& \, + \frac{4}{\kappa} \sum_{t \in S} \frac{1}{\hat{x}_t-x_j} \pder{\hat{x}_t} F^{(0)} (\boldsymbol{x}';\hat{\boldsymbol{x}}')
    - \frac{4}{\kappa} \sum_{t \in S} \frac{\hat{h}}{(\hat{x}_t-x_j)^2} F^{(0)} (\boldsymbol{x}';\hat{\boldsymbol{x}}') \\
& \, + \frac{4}{\kappa} \frac{1}{\hat{x}_r-x_j} \pder{\hat{x}_r} F^{(0)} (\boldsymbol{x}';\hat{\boldsymbol{x}}')
    - \frac{4}{\kappa} \frac{\Delta' + 2 h}{(\hat{x}_r-x_j)^2} F^{(0)} (\boldsymbol{x}';\hat{\boldsymbol{x}}') 
    \Bigg) .
\end{align*}
By the assumed original PDE $\sD^{(j)} F(\boldsymbol{x};\hat{\boldsymbol{x}}) = 0$,
the above expression must vanish. Taking into account the relation $\Delta'+2h = \hat{h}$,
this implies the asserted second order PDE for the limit function,
\[ \sD'^{(j)} F^{(0)}(\boldsymbol{x}';\hat{\boldsymbol{x}}') =
\Big( \pdder{x_j} - \frac{4}{\kappa} \sL'^{(j)}_{-2} \Big) F^{(0)}(\boldsymbol{x}';\hat{\boldsymbol{x}}') = 0 .
\]

{\bf Case (ii), third order PDEs in the limit:}
We leave it for the reader to verify these third order PDEs
for the limit function. No new ideas are needed, and the calculations are similar 
to case (i).

{\bf Case (iii), new third order PDEs via fusion:}
Let us finally turn to how the recursion makes new third order PDEs appear.
The starting point is either of the two second order PDEs
$\sD^{(j_r)} F(\boldsymbol{x};\hat{\boldsymbol{x}}) = 0$
or $\sD^{(j_r+1)} F(\boldsymbol{x};\hat{\boldsymbol{x}}) = 0$
for the original function.
We again differentiate the Frobenius series~\eqref{eq: Frobenius series} term by term,
and after some straightforward calculations, we obtain
\begin{align}
\nonumber
\sD^{(j_r)} F(\boldsymbol{x};\hat{\boldsymbol{x}}) 
= \quad \; & \eps^{\Delta'-2} \Big( \Delta' (\Delta' - 1) + \frac{4}{\kappa} (\Delta' - h) \Big)
        \; F^{(0)}(\boldsymbol{x}';\hat{\boldsymbol{x}}') \\
\label{eq: Djr expansion}
+ \; & \eps^{\Delta'-1} \Big( \big(\frac{2}{\kappa}-\Delta'\big) \pder{\hat{x}_r} F^{(0)}(\boldsymbol{x}';\hat{\boldsymbol{x}}')
     + \big( (\Delta' + 1) \Delta' + \frac{4}{\kappa} (\Delta'+1-h) \big) F^{(1)}(\boldsymbol{x}';\hat{\boldsymbol{x}}') \Big) \\
\nonumber
+ \; & \mathcal{O}(\eps^{\Delta'}) .
\end{align}
We will in fact below need to calculate two further coefficients of this expansion, but they can be
significantly simplified by conclusions drawn from the above.

The coefficient of $\eps^{\Delta'-2}$ in~\eqref{eq: Djr expansion} 
vanishes simply because $\Delta' = \hat{h} - 2 h = \frac{2}{\kappa}$ is a
solution of the indicial equation
\[ \Delta' (\Delta' - 1) + \frac{4}{\kappa} (\Delta' - \frac{6-\kappa}{2 \kappa}) = 0 \]
of the Frobenius series (the other solution is $\Delta = - 2h = \frac{\kappa-6}{\kappa}$).
In the higher order terms we can also perform the related simplification
$(\Delta'+m) (\Delta' + m - 1) + \frac{4}{\kappa} (\Delta' + m - \frac{6-\kappa}{2 \kappa})
=  m^2 - m + 2 m \Delta' + m \frac{4}{\kappa} 
=  m^2 + m \frac{8-\kappa}{\kappa} $.

Next consider the coefficient of $\eps^{\Delta'-1}$ in~\eqref{eq: Djr expansion},
which must also vanish because of the assumed original partial differential equation
$\sD^{(j_r)} F(\boldsymbol{x};\hat{\boldsymbol{x}})=0$.
Since $\frac{2}{\kappa}-\Delta' = 0$ and
$(\Delta'+1) \Delta'  + \frac{4}{\kappa} (\Delta' + 1 - \frac{6-\kappa}{2 \kappa}) \neq 0$,
the vanishing of this coefficient 
is equivalent to $F^{(1)}(\boldsymbol{x}';\hat{\boldsymbol{x}}') = 0$.
In other words, the next-to-leading order coefficient in the Frobenius series~\eqref{eq: Frobenius series} is necessarily zero.

These observations can be used to simplify the result of the calculation of
$\sD^{(j_r)} F(\boldsymbol{x};\hat{\boldsymbol{x}})$ up to order $\eps^{\Delta'}$ to the following form:
\begin{align}
\nonumber
\sD^{(j_r)} F(\boldsymbol{x};\hat{\boldsymbol{x}}) 
= \quad \; & 0 \times \eps^{\Delta'-2} + 0 \times \eps^{\Delta'-1} \\
\nonumber
+ \; & \eps^{\Delta'} \Bigg( \frac{1}{4} \pdder{\hat{x}_r} F^{(0)}(\boldsymbol{x}';\hat{\boldsymbol{x}}')
            + \big( 4 + 2 \frac{8-\kappa}{\kappa} \big) F^{(2)}(\boldsymbol{x}';\hat{\boldsymbol{x}}') \\
\tag{\ref{eq: Djr expansion}'}\label{eq: Djr expansion 2}
& \qquad    + \frac{4}{\kappa} \sum_{k \in J'} \frac{1}{x_k - \hat{x}_r} \pder{x_k} F^{(0)}(\boldsymbol{x}';\hat{\boldsymbol{x}}')
            + \frac{4}{\kappa} \sum_{t \in S} \frac{1}{\hat{x}_t - \hat{x}_r} \pder{\hat{x}_t} F^{(0)}(\boldsymbol{x}';\hat{\boldsymbol{x}}') \\
\nonumber
& \qquad    - \frac{4}{\kappa} \sum_{k \in J'} \frac{h}{(x_k-\hat{x}_r)^2} \, F^{(0)}(\boldsymbol{x}';\hat{\boldsymbol{x}}') 
            - \frac{4}{\kappa} \sum_{t \in S} \frac{\hat{h}}{(\hat{x}_t-\hat{x}_r)^2} \, F^{(0)}(\boldsymbol{x}';\hat{\boldsymbol{x}}') \Bigg) \\
\nonumber
+ \; & \mathcal{O}(\eps^{\Delta'+1}) .
\end{align}

Using the assumption $\sD^{(j_r)} F(\boldsymbol{x};\hat{\boldsymbol{x}}) = 0$,
the coefficient of $\eps^{\Delta'}$ in \eqref{eq: Djr expansion 2} must vanish, and we
can solve for the next term $F^{(2)}(\boldsymbol{x}';\hat{\boldsymbol{x}}')$
in the Frobenius series~\eqref{eq: Frobenius series} in terms of the leading term
$F^{(0)}(\boldsymbol{x}';\hat{\boldsymbol{x}}')$:
\begin{align}
\label{eq: F2 in terms of F0}
F^{(2)}(\boldsymbol{x}';\hat{\boldsymbol{x}}') = \frac{\kappa}{8 + \kappa}
    \Big( \frac{2}{\kappa} \widehat{\sL}_{-2}^{(r)} F^{(0)}(\boldsymbol{x}';\hat{\boldsymbol{x}}')
          - \frac{1}{8} \pdder{\hat{x}_r} F^{(0)}(\boldsymbol{x}';\hat{\boldsymbol{x}}') \Big) .
\end{align}

It remains to inspect
$\sD^{(j_r)} F(\boldsymbol{x};\hat{\boldsymbol{x}})$ up to order $\eps^{\Delta'+1}$, 
\begin{align}
\nonumber
\sD^{(j_r)} F(\boldsymbol{x};\hat{\boldsymbol{x}}) 
= \quad \; & 0 \times \eps^{\Delta'-2} + 0 \times \eps^{\Delta'-1} + 0 \times \eps^{\Delta'} \\
\nonumber
+ \; & \eps^{\Delta'+1} \Bigg( -2 \pder{\hat{x}_r} F^{(2)}(\boldsymbol{x}';\hat{\boldsymbol{x}}') \\
\tag{\ref{eq: Djr expansion}''}\label{eq: Djr expansion 3}
& \qquad    - \frac{2}{\kappa} \sum_{k \in J'} \frac{1}{(x_k - \hat{x}_r)^2} \pder{x_k} F^{(0)}(\boldsymbol{x}';\hat{\boldsymbol{x}}')
            - \frac{2}{\kappa} \sum_{t \in S} \frac{1}{(\hat{x}_t - \hat{x}_r)^2} \pder{\hat{x}_t} F^{(0)}(\boldsymbol{x}';\hat{\boldsymbol{x}}') \\
\nonumber
& \qquad    + \frac{4}{\kappa} \sum_{k \in J'} \frac{h}{(x_k-\hat{x}_r)^3} \, F^{(0)}(\boldsymbol{x}';\hat{\boldsymbol{x}}') 
            + \frac{4}{\kappa} \sum_{t \in S} \frac{\hat{h}}{(\hat{x}_t-\hat{x}_r)^3} \, F^{(0)}(\boldsymbol{x}';\hat{\boldsymbol{x}}') \Bigg) \\
\nonumber
+ \; & \mathcal{O}(\eps^{\Delta'+2}) .
\end{align}

From the vanishing of the coefficient of $\eps^{\Delta'+1}$ in \eqref{eq: Djr expansion 3},
we solve
\[ \pder{\hat{x}_r} F^{(2)}(\boldsymbol{x}';\hat{\boldsymbol{x}}')
    = \frac{1}{\kappa} \, \widehat{\sL}_{-3}^{(r)} F^{(0)}(\boldsymbol{x}';\hat{\boldsymbol{x}}') . \]
Substituting here the expression~\eqref{eq: F2 in terms of F0}
for $F^{(2)}(\boldsymbol{x}';\hat{\boldsymbol{x}}')$,
we arrive at the third order PDE 
\[ - \frac{\kappa}{8(8+\kappa)} \pddder{\hat{x}_r} F^{(0)}(\boldsymbol{x}';\hat{\boldsymbol{x}}')
    + \frac{2}{8+\kappa} \pder{\hat{x}_r} \widehat{\sL}_{-2}^{(r)} F^{(0)}(\boldsymbol{x}';\hat{\boldsymbol{x}}')
= \frac{1}{\kappa} \, \widehat{\sL}_{-3}^{(r)} F^{(0)}(\boldsymbol{x}';\hat{\boldsymbol{x}}') \]
for the leading coefficient $F^{(0)}(\boldsymbol{x}';\hat{\boldsymbol{x}}')$ of
the Frobenius series~\eqref{eq: Frobenius series}.
Finally, using the commutation relation
$\pder{\hat{x}_r} \widehat{\sL}_{-2}^{(r)} - \widehat{\sL}_{-2}^{(r)} \pder{\hat{x}_r} = \widehat{\sL}_{-3}^{(r)}$,
we get the nontrivial third order PDE that we wanted to establish:
\[ 0 = \pddder{\hat{x}_r} F^{(0)}(\boldsymbol{x}';\hat{\boldsymbol{x}}')
       - \frac{16}{\kappa} \widehat{\sL}_{-2}^{(r)} \pder{\hat{x}_r} F^{(0)}(\boldsymbol{x}';\hat{\boldsymbol{x}}')
       + \frac{8(8-\kappa)}{\kappa^2} \, \widehat{\sL}_{-3}^{(r)} F^{(0)}(\boldsymbol{x}';\hat{\boldsymbol{x}}') .\]
\end{proof}

Using Lemma~\ref{lem: fusion for third order PDEs} 
recursively, we now prove the third order partial differential equations for 
the function $\Ampl_\omega$ given by the iterated
limit~\eqref{eq: iterated limit definition of boundary visit amplitude}.
\begin{prop}\label{prop: scaling limit of boundary visits}
The function $\Ampl_\omega(x_1 ; \hat{x}_1 , \ldots , \hat{x}_{N'} ; x_2)$ satisfies
the two second order PDEs, for ${j \in J = \set{1,2}}$,
\begin{align}\label{eq: second order PDEs at kappa equals two}
\sD^{(j)} \Ampl_\omega(x_1 ; \hat{x}_1 , \ldots , \hat{x}_{N'} ; x_2) = 0 ,
\qquad \text{ where } \qquad
\sD^{(j)} = \pdder{x_j} - 2 \, \sL^{(j)}_{-2}
\end{align}
and the $N'$ third order PDEs, for $s \in S = \set{1, \ldots, N'}$,
\begin{align}\label{eq: third order PDEs at kappa equals two}
\widehat{\sD}^{(s)} \Ampl_\omega(x_1 ; \hat{x}_1 , \ldots , \hat{x}_{N'} ; x_2) = 0 ,
\qquad \text{ where } \qquad
\widehat{\sD}^{(s)} = \pddder{\hat{x}_s}
    - 8 \, \widehat{\sL}^{(s)}_{-2} \pder{\hat{x}_s}
    + 12 \, \widehat{\sL}^{(s)}_{-3} .
\end{align}
Moreover, the function $\Ampl_\omega(x_1 ; \hat{x}_1 , \ldots , \hat{x}_{N'} ; x_2)$ is positive.
\end{prop}
\begin{proof}
Start from the function
\[ F_0(x_1 , \ldots , x_{2N}) := \PartF_\alpha (x_1 , \ldots , x_{2N}) \]
with $\alpha = \alpha(\omega)$. At $\kappa=2$,
the system of second order PDEs \eqref{eq: general second order PDEs at generic kappa}
for $F_0$ is nothing but the system \eqref{eq: PDE for multiple SLEs at kappa equals 2} for $\PartF_\alpha$, which is
satisfied by virtue of Lemma~\ref{lem: PDEs for determinants} (and at this stage there are no third order PDEs).
By~iterated limits, recursively define new functions
for $r=1,2,\ldots,N'$,
\[ F_r (\boldsymbol{x}';\hat{\boldsymbol{x}}')  
:= \lim_{x_{j_r}, x_{j_r + 1} \to \hat{x}_r}
    \frac{1}{(x_{j_r + 1} - x_{j_r})^{\Delta'}} F_{r-1}(\boldsymbol{x};\hat{\boldsymbol{x}}) .
\]
By the replacing algorithm~\ref{alg: continuous replacing},
each function $F_{r-1}$ is a linear combination of certain determinants, whose Taylor
series expansion is a Frobenius series of the required form.
Therefore, by Lemma~\ref{lem: fusion for third order PDEs}, each $F_r$ satisfies
the corresponding joint system of second and third order PDEs.
Finally, with $r= N'$ we obtain the function
\[ F_{N'}(x_1 , \hat{x}_1 , \ldots , \hat{x}_{N'} , x_2) 
= \Ampl_\omega(x_1 , \hat{x}_1 , \ldots , \hat{x}_{N'} , x_2) , \] 
according to the expression
\eqref{eq: iterated limit definition of boundary visit amplitude}. 
The asserted partial differential equations thus follow.

Positivity is shown by the same recursive argument.
First of all, $F_0 = \PartF_\alpha$ is positive by
Theorem~\ref{thm: pure partition functions at kappa equals 2}.
Each iterated limit $F_r$ is therefore non-negative.
It remains to show that $F_r$ is pointwise non-vanishing.
The value of $F_r(\boldsymbol{x}';\hat{\boldsymbol{x}}')$ is
the coefficient of $(x_{j_r +1} - x_{j_r})^{-\Delta'}$ in the
Frobenius series expansion of $F_{r-1}(\boldsymbol{x};\hat{\boldsymbol{x}})$.
The coefficient of the leading term $(x_{j_r +1} - x_{j_r})^{-\Delta}$ is
vanishing, so if also $F_r(\boldsymbol{x}';\hat{\boldsymbol{x}}')$ were vanishing,
the function $F_{r-1}(\boldsymbol{x};\hat{\boldsymbol{x}})$ would have to be zero.
Non-vanishingness then follows recursively.
\end{proof}

\subsection{\label{sub: proof of bdry visit thm}Proof of Theorem~\ref{thm: scaling limit of LERW bdry visits}}

Let us now conclude the proof of one of our main results,
Theorem~\ref{thm: scaling limit of LERW bdry visits},
about the scaling limit of boundary visit probabilities of the
loop-erased random walk. The starting point is the formula of
Corollary~\ref{cor: LERW boundary visit probabilities with partition functions}
for the probability that a LERW $\lambda^\delta$ on $\Gr^\delta$
from $\ein^\delta$ to $\eout^\delta$ uses edges $\hat{e}_1^\delta , \ldots , \hat{e}_{N'}^\delta$
at unit distance from the boundary in an order specified by $\omega$,
\[
\PR_{\ein^\delta,\eout^\delta} \big[ \lambda^\delta \text{ uses $\hat{e}^\delta_1 , \ldots , \hat{e}^\delta_{N'}$ in the order $\omega$} \big]
= \frac{Z^{\Gr^\delta}_{\alpha(\omega)}(e_1^\delta , \ldots , e_{2N}^\delta)}{Z^{\Gr^\delta}(\ein^\delta , \eout^\delta)} .
\]
In the setup of Theorem~\ref{thm: scaling limit of LERW bdry visits},
$\Gr^\delta$ is a square grid approximation of a domain $\domain$,
and the edges $\ein^\delta , \eout^\delta , \hat{e}_1^\delta , \ldots , \hat{e}_{N'}^\delta$
are the nearest to the
marked points $\pin, \pout, \hat{p}_1 , \ldots, \hat{p}_{N'}$,
which lie on horizontal or vertical boundary segments.
Then Theorem~\ref{thm: discrete boundary visit probabilities} applied to the numerator
and Lemma~\ref{lem: scaling limit of discrete excursion kernels} to the denominator
give the existence of the limit
\[
\lim_{\delta \to 0} \frac{\PR_{\ein^\delta,\eout^\delta} \big[ \lambda^\delta
        \text{ uses $\hat{e}^\delta_1 , \ldots , \hat{e}^\delta_{N'}$ in the order $\omega$} \big]
    }{\delta^{3 N'}} \; .
\]
Moreover, with
Equation~\eqref{eq: conformal covariance for general limit}, 
this limit is explicitly expressed as
\[ \pi^{-N'} \times \prod_{s=1}^{N'} |\confmap'(\hat{p}_s)|^3 \times
    \frac{ \Ampl_\omega \big( \confmap(\pin),\confmap(\pout);\confmap(\hat{\boldsymbol{p}}) \big)}{\ExcKH \big( \confmap(\pin),\confmap(\pout) \big)},
\]
where $\Ampl_\omega$ is given by
Equation~\eqref{eq: iterated limit definition of boundary visit amplitude}
and $\ExcKH(\xin,\xout) = \big( \xout - \xin \big)^{-2}$.
Finally, in Proposition~\ref{prop: scaling limit of boundary visits}, we 
verified the asserted
two second order and $N'$ third order partial differential equations for $\Ampl_\omega$. 
This concludes the proof.
$\hfill \qed$

\subsection{\label{sub: asymptotics of boundary visits}Asymptotics of the scaling limits of the boundary visit probabilities}
To finish this section, we prove
asymptotics properties for the boundary visit probabilities,
predicted in~\cite{JJK-SLE_boundary_visits}.
\begin{prop}\label{prop: asymptotics for xin}
Let $\omega = (\omega_1 , \ldots , \omega_{N'}) \in \set{\pm 1}^{N'}$ be an order of boundary visits,
and let $\hat{x}_s$ be the boundary visit point closest to $\xin$ on either side of $\xin$.
Then, as $\hat{x}_s , \xin \to \xin'$, we have the following asymptotics of $\Ampl_\omega$:
\begin{align}
\label{eq: bdary visit amplitude cascade for xin}
& \lim_{\hat{x}_s , \xin \to \xin'} | \hat{x}_s - \xin |^3 \; \Ampl_\omega(\xin ; \hat{x}_1 , \hat{x}_2 , \ldots , \hat{x}_{N'} ; \xout ) \\
\nonumber
= \; & \begin{cases} 0 \qquad
        & \text{if $\hat{x}_s$ is not the first boundary visit, i.e., $\omega_1 \neq \sgn(\hat{x}_s - \xin)$} \\
 2 \, \Ampl_{\omega'}(\xin' ; \hat{x}_2 , \ldots , \hat{x}_{N'} ; \xout ) \;
         & \text{if $\hat{x}_s$ is the first boundary visit, i.e., $\omega_1 = \sgn(\hat{x}_s - \xin)$,} 
 \end{cases}
\end{align}
where $\omega' = (\omega_2 , \ldots , \omega_{N'}) \in \set{\pm 1}^{N'-1}$ is obtained
from $\omega$ by omitting the  first boundary visit, as illustrated in
Figure~\ref{fig: merging 1st bdary visit to xin}.
\end{prop}
\begin{figure}
\begin{displaymath}
\xymatrixcolsep{1.5pc}
\xymatrixrowsep{2.5pc}
\xymatrix{
        \begin{minipage}{6cm} \begin{center} \includegraphics[width=\textwidth]{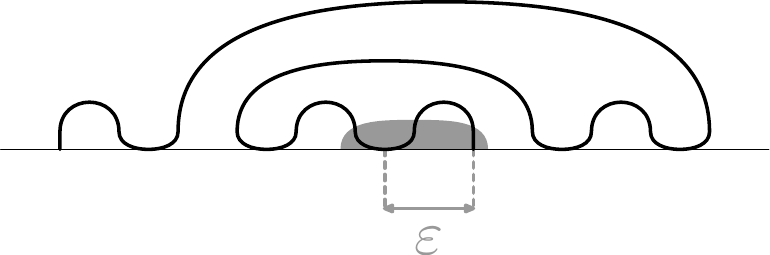} \end{center} \end{minipage} \ar[d]  & 
        \begin{minipage}{6cm} \begin{center} \includegraphics[width=\textwidth]{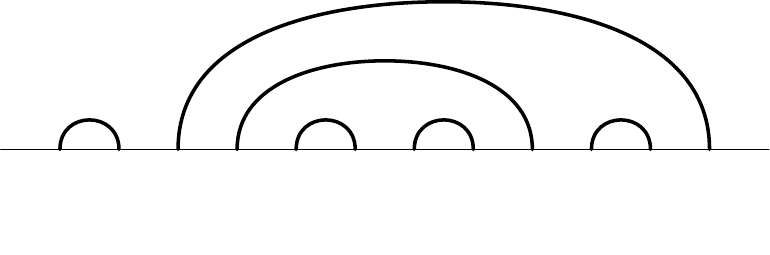} \end{center} \end{minipage}\\
         \hspace{-0mm}\begin{minipage}{6cm} \begin{center} \includegraphics[width=\textwidth]{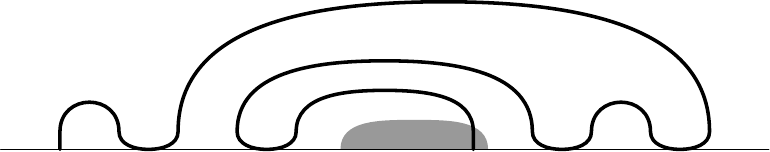} \end{center} \end{minipage}  &  \begin{minipage}{6cm} \begin{center} \includegraphics[width=\textwidth]{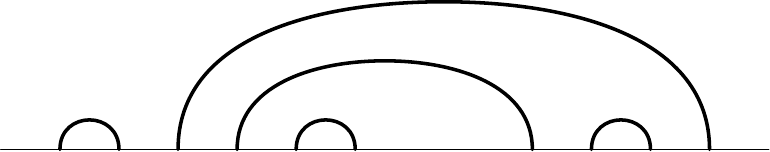} \end{center} \end{minipage}
        }
\end{displaymath}
\caption{\label{fig: merging 1st bdary visit to xin}
Illustration of the asymptotics property in Proposition~\ref{prop: asymptotics for xin}:
collapsing the first visited point with the starting point. The boundary visit orders $\omega$
and $\omega'$ are shown on the top left and bottom left, and
the corresponding link patterns $\alpha(\omega)$ and $\alpha(\omega')$ on the top right and bottom right, respectively.
}
\end{figure}

\begin{prop}\label{prop: asymptotics for visit points}
Let $\omega \in \set{\pm 1}^{N'}$ be an order of boundary visits,
and let $\hat{x}_s$ and $\hat{x}_{s+1}$ be two consecutive boundary visit points.
Then, as $\hat{x}_s, \hat{x}_{s+1} \to \hat{x}'$, we have the following asymptotics of $\Ampl_\omega$:
\begin{align}
\label{eq: bdary visit amplitude cascade for two visits}
& \lim_{ \hat{x}_s, \hat{x}_{s+1} \to \hat{x}'} | \hat{x}_{s+1} - \hat{x}_{s} |^3
    \; \Ampl_\omega(\xin ; \ldots , \hat{x}_{s-1}, \hat{x}_{s} , \hat{x}_{s+1} , \hat{x}_{s+2}, \ldots ; \xout ) \\
\nonumber
= \; & \begin{cases} 0 \qquad
    & \text{if $\hat{x}_s$ and $\hat{x}_{s+1}$ are not successively visited,} \\
10 \, \Ampl_{\omega'}(\xin ; \ldots , \hat{x}_{s-1}, \hat{x}' , \hat{x}_{s+2}, \ldots ; \xout )  \;
    & \text{if $\hat{x}_s$ and $\hat{x}_{s+1}$ are successively visited,}
\end{cases}
\end{align}
where $\omega' \in \set{\pm 1}^{N'-1}$ is obtained from $\omega \in \set{\pm 1}^{N'}$
by omitting the boundary visits corresponding to $\hat{x}_{s+1}$,
as illustrated in Figure~\ref{fig: merging subsequent bdary visits together}.
\end{prop}
\begin{figure}
\begin{displaymath}
\xymatrixcolsep{1.5pc}
\xymatrixrowsep{2.5pc}
\xymatrix{
        \begin{minipage}{6cm} \begin{center} \includegraphics[width=\textwidth]{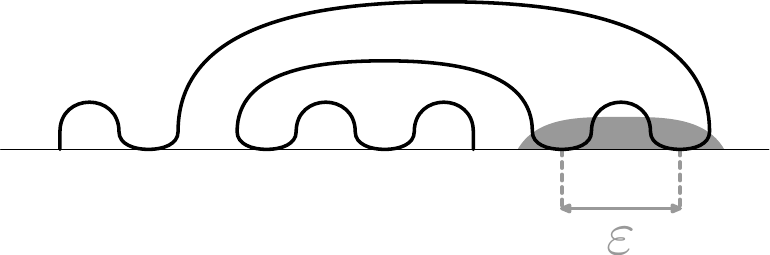} \end{center} \end{minipage} \ar[d] & 
        \begin{minipage}{6cm} \begin{center} \includegraphics[width=\textwidth]{bdarycascade_w_nbhds-3.pdf} \end{center} \end{minipage}\\
         \hspace{-0mm}\begin{minipage}{6cm} \begin{center} \includegraphics[width=\textwidth]{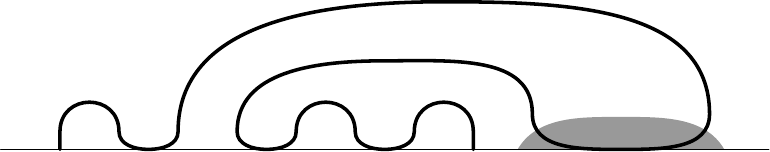} \end{center} \end{minipage}  &  \begin{minipage}{6cm} \begin{center} \includegraphics[width=\textwidth]{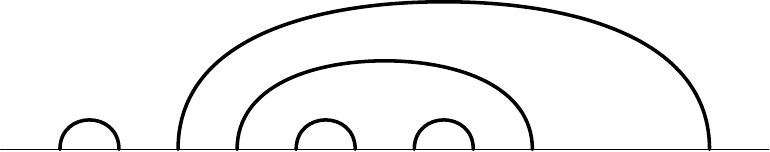} \end{center} \end{minipage}
        }
\end{displaymath}
\caption{\label{fig: merging subsequent bdary visits together}
Illustration of the asymptotics property in Proposition~\ref{prop: asymptotics for visit points}:
collapsing two successively visited points.
The boundary visit orders $\omega$ and $\omega'$ are shown on the top left and bottom left, and
the corresponding link patterns $\alpha(\omega)$ and $\alpha(\omega')$ on the top right and bottom right, respectively.
}
\end{figure}

\begin{rem}
\textit{
The above asymptotics of $\Ampl_\omega$, including their multiplicative constants,
coincide with what one gets from the case of generic $\kappa$ treated in \cite{KP-pure_partition_functions_of_multiple_SLEs}
by letting $\kappa \to 2$.
Namely, if $\zeta_\omega^{(\kappa)}$ is defined as the limit of
$\prod_{s=1}^{N'} |x_{j_s+1}-x_{j_s}|^{-2/\kappa} \times \PartF^{(\kappa)}_{\alpha(\omega)}$,
generalizing our definition~\eqref{eq: iterated limit definition of boundary visit amplitude},
then we have the following analogues of the above propositions.
The non-zero asymptotics of Proposition~\ref{prop: asymptotics for xin}
gets replaced by $\Ampl_\omega^{(\kappa)} \sim C_1 \, |\hat{x}_s - \xin|^{\frac{\kappa-8}{\kappa}} \, \Ampl_{\omega'}^{(\kappa)}$
with $C_1 = \frac{4 \cos^2(\frac{4}{\kappa}\pi)}{1 + 2 \cos(\frac{8}{\kappa}\pi)} \, \frac{\Gamma(1-\frac{8}{\kappa}) \, \Gamma(2-\frac{8}{\kappa})}{\Gamma(1-\frac{4}{\kappa}) \, \Gamma(2-\frac{12}{\kappa})}$.
In the limit $\kappa \to 2$, we have $C_1 \to 2$, in accordance with Equation~\eqref{eq: bdary visit amplitude cascade for xin}. 
Likewise, the non-zero asymptotics of Proposition~\ref{prop: asymptotics for visit points} gets replaced by $\Ampl_\omega^{(\kappa)} \sim C_2 \, |\hat{x}_{s+1} - \hat{x}_s|^{\frac{\kappa-8}{\kappa}} \, \Ampl_{\omega'}^{(\kappa)}$
with $C_2 = \frac{2 \cos^2(\frac{4}{\kappa}\pi)}{\cos(\frac{8}{\kappa}\pi)} \, \frac{\Gamma(1-\frac{8}{\kappa})^2 \, \Gamma(2-\frac{8}{\kappa})}{\Gamma(1-\frac{4}{\kappa})^2 \, \Gamma(2-\frac{16}{\kappa})}$. In the limit $\kappa \to 2$, we have $C_2 \to 10$, in accordance
with Equation~\eqref{eq: bdary visit amplitude cascade for two visits}. 
}
\end{rem}

In proving the above propositions, we make use of the combinatorial
properties of inverse Fomin type sums listed
in Propositions~\ref{prop: Zero-replacing Rule} and~\ref{prop: inverse Fomin cascade}
in Section~\ref{subsec: inverse Fomin sums}.
Note that the replacing algorithm \ref{alg: continuous replacing} yields an expression
for $\Ampl_\omega$ as an inverse Fomin type sum.
\begin{proof}[Proof of Proposition~\ref{prop: asymptotics for xin}]
Assume for definiteness that $\hat{x}_s$ is on the left side of $\xin$. 
Now, 
in the link pattern $\alpha(\omega)$ associated to the boundary visit order $\omega$,
the point $\hat{x}_s$ corresponds to some consecutive indices $j-1, j$  and $\xin$ to $j+1$.
As illustrated in Figure~\ref{fig: wedge in first bdary visit},
$\hat{x}_s$ is the first boundary visit of $\omega$ if and only if there is a
link $\link{j}{j+1}$ in the link pattern $\alpha(\omega)$, that is,
$\upwedgeat{j} \in \alpha(\omega)$.
By the replacing
algorithm~\ref{alg: continuous replacing}, the kernel entries of the inverse
Fomin type sum giving $\Ampl_\omega = \mathfrak{Z}_{\alpha(\omega)}$ are for these indices
\begin{align*}
\mathfrak{K} (j-1, j) = 0 ,
\qquad
\mathfrak{K} (j-1, j + 1) = \frac{1}{\vert \xin - \hat{x}_s \vert^2} ,
\qquad
\text{ and } \qquad
\mathfrak{K} (j, j + 1) = \frac{2}{\vert \xin - \hat{x}_s \vert^3} .
\end{align*}
All other kernel entries remain bounded in the limit $|\xin - \hat{x}_s| \to 0$.
Hence, the desired asymptotics~\eqref{eq: bdary visit amplitude cascade for xin}
is twice the coefficient of $\mathfrak{K}(j, j + 1)$ in the inverse Fomin
type sum $\Ampl_\omega = \mathfrak{Z}_{\alpha(\omega)}$, because the product
$\mathfrak{K} (j, j + 1) \, \mathfrak{K} (j-1, j + 1)$ does not
appear. 
By Proposition~\ref{prop: Zero-replacing Rule}(c), this coefficient is zero if
$\upwedgeat{j} \not \in \alpha(\omega)$ and
by Proposition~\ref{prop: inverse Fomin cascade}, it
equals $\mathfrak{Z}^{\mathfrak{K}\removewedge{j} }_{\alpha(\omega) \removeupwedge{j}}$ if $\upwedgeat{j} \in \alpha(\omega)$.
These two possibilities correspond to the two cases in the
assertion~\eqref{eq: bdary visit amplitude cascade for xin}.
It remains to observe that $\mathfrak{Z}^{\mathfrak{K}\removewedge{j} }_{\alpha(\omega) \removeupwedge{j}} = \Ampl_{\omega'}$,
by the replacing algorithm~\ref{alg: continuous replacing}.
\end{proof}
\begin{figure}
\includegraphics[scale=1.0]{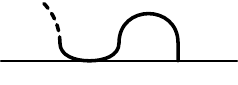} \qquad \qquad
\includegraphics[scale=1.0]{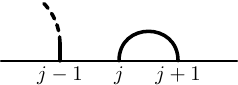}
\caption{\label{fig: wedge in first bdary visit}
In the setup of Proposition~\ref{prop: asymptotics for xin},
the first boundary visit is on the left of $\xin$
if and only if there is a link $\link{j}{j+1}$ in the link pattern $\alpha(\omega)$
and $j+1$ is the index corresponding to the point $\xin$.
}
\end{figure}
\begin{figure}
\includegraphics[scale=1.0]{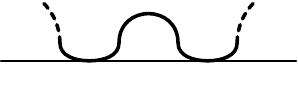} \qquad \qquad
\includegraphics[scale=1.0]{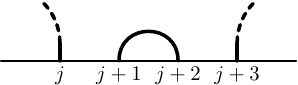}
\caption{\label{fig: wedge in consecutive bdary visits} 
In the setup of Proposition~\ref{prop: asymptotics for visit points},
the visits to $\hat{x}_s$ and $\hat{x}_{s+1}$ are successive
if and only if the link pattern $\alpha(\omega)$ contains the link
$\link{j+1}{j+2}$.
}
\end{figure}
\begin{figure}
\includegraphics[scale=1.0]{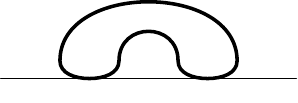} \qquad \qquad
\includegraphics[scale=1.0]{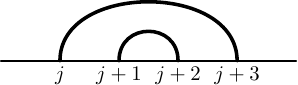}
\caption{\label{fig: bonus figure for careful reader} 
This cannot happen in the setup of Proposition~\ref{prop: asymptotics for visit points}.
}
\end{figure}
\begin{proof}[Proof of Proposition~\ref{prop: asymptotics for visit points}]
In the link pattern $\alpha(\omega)$ associated to the boundary visit order $\omega$,
the consecutive points $\hat{x}_s$ and $\hat{x}_{s+1}$ correspond to some pairs of consecutive indices
$j, j + 1$ and $j + 2, j + 3$, respectively.
As illustrated in Figure \ref{fig: wedge in consecutive bdary visits},
the boundary visits are successive if and only if there is a
link $\link{j+1}{j+2}$ in $\alpha(\omega)$, that is, $\upwedgeat{{j + 1}} \in \alpha$.
We again list the kernel entries at the indices of interest for the inverse Fomin type
sum $\Ampl_\omega = \mathfrak{Z}_{\alpha(\omega)}$:
\begin{align}
\nonumber
\mathfrak{K} (j , j+1) = \; & 0     &
\mathfrak{K} (j+1 , j+2) = \; & \frac{2}{\vert \hat{x}_{s+1} - \hat{x}_s \vert^3} \\
\label{eq: a list of kernels}
\mathfrak{K} (j , j+2) = \; & \frac{1}{\vert \hat{x}_{s+1} - \hat{x}_s \vert^2}     &
\mathfrak{K} (j+1 , j+3) = \; & \frac{-6}{\vert \hat{x}_{s+1} - \hat{x}_s \vert^4} \\
\nonumber
\mathfrak{K} (j , j+3) = \; & \frac{-2}{\vert \hat{x}_{s+1} - \hat{x}_s \vert^3}     &
\mathfrak{K} (j+2 , j+3) = \; & 0.
\end{align}
All other kernel entries 
remain bounded in the limit $|\hat{x}_{s+1}-\hat{x}_{s}| \to 0$.
Hence, as for 
Proposition~\ref{prop: asymptotics for xin},
the asymptotics~\eqref{eq: bdary visit amplitude cascade for two visits} can be read off from
the coefficients of the four diverging kernels above.

Assume first that the boundary visits are not successive,
or equivalently, $\upwedgeat{{j+1}} \not \in \alpha$.
Then, by Proposition~\ref{prop: Zero-replacing Rule}(c), the coefficient of
$\mathfrak{K} (j+1, j+2) $ is zero in $\Ampl_\omega = \mathfrak{Z}_{\alpha(\omega)}$.
By construction, in the link pattern $\alpha(\omega)$
there are no links $\link{j}{j+1}$ or $\link{j+2}{j+3}$,
that is, $\upwedgeat{j}, \upwedgeat{{j + 2}} \not \in \alpha$.
In such a situation, we can apply the interchange operation 
of Proposition~\ref{prop: Zero-replacing Rule}(a)
to verify that the coefficients of $\mathfrak{K}(j, j+2) $, $\mathfrak{K} (j+1, j+3)$,
and $\mathfrak{K} (j, j+3)$ are also zero.
Indeed, the first of these coefficients is obtained from the coefficient of $\mathfrak{K} (j+1, j+2) $
by interchanging the pair of indices $j, j+1$, the second by interchanging the pair
$j+2, j+3$, and the third by interchanging both pairs.
We conclude that all the kernel entries in~\eqref{eq: a list of kernels} that diverge in the
limit $|\hat{x}_{s+1}-\hat{x}_{s}| \to 0$ cancel out in the inverse Fomin type
sum $\Ampl_\omega = \mathfrak{Z}_{\alpha(\omega)}$.
This proves the first case in Equation~\eqref{eq: bdary visit amplitude cascade for two visits}.

Assume then that the boundary visits are successive, or equivalently,
$\upwedgeat{{j + 1}} \in \alpha$. We first claim that no product of the non-zero kernel
entries listed in~\eqref{eq: a list of kernels} appears in the inverse Fomin type
sum $\Ampl_\omega =  \mathfrak{Z}_{\alpha(\omega)}$. To see this, by
Proposition~\ref{prop: inverse Fomin cascade}, the coefficient of $\mathfrak{K}(j+1,j+2)$
is $\mathfrak{Z}^{\mathfrak{K}\removewedge{{j+1}} }_{\alpha(\omega) \removeupwedge{ {j+1} }}$.
Next, as illustrated in Figure \ref{fig: bonus figure for careful reader},
the link $\link{j}{j+3}$ cannot belong to $\alpha(\omega)$, which implies
$\upwedgeat{j} \not \in \alpha(\omega) \removeupwedge{{j+1}}$.
Then, Proposition~\ref{prop: Zero-replacing Rule}(c) implies that the
coefficient of $\mathfrak{K}(j,j+3)$ is zero in $\mathfrak{Z}^{\mathfrak{K}\removewedge{{j+1}} }_{\alpha(\omega) \removeupwedge{{j+1}}}$.
Thus, the product $\mathfrak{K}(j+1, j+2) \, \mathfrak{K} (j,j+3)$ does not appear in
the inverse Fomin type sum $\Ampl_\omega =  \mathfrak{Z}_{\alpha(\omega)}$.
Other similar products are then excluded by repeated application of the
interchange operation of Proposition~\ref{prop: Zero-replacing Rule}(a)
with the pairs of indices $j, j+1$ or $j+2, j+3$.

Now we can compute the limit of interest~\eqref{eq: bdary visit amplitude cascade for two visits}
simply by adding up the contributions of the individual divergent kernel
entries in Equation~\eqref{eq: a list of kernels}.
The entry $\mathfrak{K}(j, j+2)$ has too mild a divergence to contribute,
so there are in fact only three contributions to consider.
We already saw that the coefficient of ${\mathfrak{K}(j+1, j+2)}$
is $\mathfrak{Z}^{\mathfrak{K}\removewedge{{j+1}} }_{\alpha(\omega) \removeupwedge{{j+1}} } $, and by
the replacing algorithm~\ref{alg: continuous replacing}, we have
$\mathfrak{Z}^{\mathfrak{K}\removewedge{{j+1}} }_{ \alpha(\omega) \removeupwedge{{j+1}} } = \Ampl_{\omega'}$.
The entry $\mathfrak{K}(j+1, j+2)$ thus contributes
\begin{align}
\label{eq: first kernel contribution}
2 \; \Ampl_{\omega'} (\xin ; \ldots , \hat{x}_{s-1}, \hat{x}', \hat{x}_{s+2}, \ldots ; \xout ) .
\end{align}

Next, applying twice the antisymmetry under the interchange
operation of Proposition~\ref{prop: Zero-replacing Rule}(a) gives the coefficient
of $\mathfrak{K}(j,j+3)$ in terms of the previous one. 
Note, in addition, that comparing the
coefficient of $\mathfrak{K}(j,j+3)$ to the replacing algorithm~\ref{alg: continuous replacing}
to yield $\Ampl_{\omega'}$, we have instead of $x_{j+2}$
differentiated with respect to $x_{j + 1}$, i.e., the 
smaller of the two real points $x_{j + 1}$ and $x_{j + 2}$ that
now correspond to the collapsed boundary visit in $\omega'$.
Thus, we use Proposition~\ref{prop: Zero-replacing Rule}(a) a third time to interchange
${j + 1}$ and ${j + 2}$, and obtain $\Ampl_{\omega'}$.
Thus, the entry $\mathfrak{K}(j, j+3)$ contributes
\begin{align}
\label{eq: second kernel contribution}
(-1)^3 \times (-2) \; \Ampl_{\omega'} (\xin ; \ldots , \hat{x}_{s-1}, \hat{x}', \hat{x}_{s+2}, \ldots ; \xout ) .
\end{align}

We are only left with finding the coefficient of $\mathfrak{K}(j+1, j+3)$.
Proposition~\ref{prop: Zero-replacing Rule}(a) relates this to the coefficient
of $\mathfrak{K}(j+1, j+2)$, which is $\Ampl_{\omega'}$, and we see
that the coefficient of $\mathfrak{K}(j+1, j+3)$ is an inverse Fomin type
sum otherwise similar to $\Ampl_{\omega'}$, except that we have not differentiated
with respect to either one of the boundary points $x_j$ and $x_{j+2}$ related
to the collapsed boundary visit in $\omega'$. An argument identical to the proof of
the replacing algorithm~\ref{alg: continuous replacing} allows us to differentiate
with respect to $x_{j+2}$ and cancel one power of $\vert \hat{x}_{s+1} - \hat{x}_s \vert$.
The contribution of $\mathfrak{K}(j+1, j+3)$ to the limit~\eqref{eq: bdary visit amplitude cascade for two visits}
is thus
\begin{align}
\label{eq: third kernel contribution}
-1 \times (-6) \; \Ampl_{\omega'} (\xin ; \ldots , \hat{x}_{s-1}, \hat{x}', \hat{x}_{s+2}, \ldots ; \xout ) .
\end{align}
Summing the contributions 
given in Equations~\eqref{eq: first kernel contribution}, \eqref{eq: second kernel contribution},
and \eqref{eq: third kernel contribution}, 
we obtain the asserted
limit~\eqref{eq: bdary visit amplitude cascade for two visits}. This concludes the proof.
\end{proof}

\appendix

\section{Example boundary visit formulas}

In this appendix, we provide a few examples of the boundary visit amplitudes $\Ampl$ and 
the order-refined amplitudes $\Ampl_\omega$
with small numbers of marked points.

\subsection{One boundary visit}

For one boundary visit, $N' = 1$, there is only one possible order (so we omit $\omega$ from the notation).
The corresponding amplitude is obtained using replacing algorithm~\ref{alg: continuous replacing}:
\begin{align*}
\Ampl(\xin ; \hat{x}; \xout ) 
= \; & \quad \vcenter{\hbox{\includegraphics[scale=0.5]{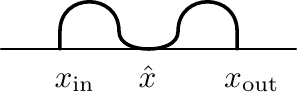}}} \quad
= \; \frac{2 (\xout-\xin)}{(\hat{x}-\xin )^3 (\xout-\hat{x})^3} .
\end{align*}
One can directly check that this function solves the PDE system in Theorem~\ref{thm: scaling limit of LERW bdry visits} 
and asymptotics properties~\eqref{eq: bdary visit amplitude cascade for xin}.

%

\subsection{Two boundary visits}

For two boundary visits, $N' = 2$, there are two essentially different cases:
boundary visits on the same side or boundary visits on the opposite sides.
As representative examples, we consider the orders $\omega = (-,-)$ and $\omega = (-,+)$.

For the visits on the same side, using replacing algorithm~\ref{alg: continuous replacing}, we find the solution
\begin{align*}
\Ampl_{- -}(\xin ; \hat{x}_1, \hat{x}_2 ; \xout )
= \; & \qquad \qquad \quad  \vcenter{\hbox{\includegraphics[scale=0.5]{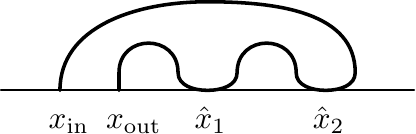}}} \\[1em]
= \; & \hphantom{+} \det \left( \begin{array}{c c c}
\frac{-2}{(\hat{x}_2-\xin)^3} & \frac{-2}{(\hat{x}_2-\xout)^3} & \frac{-6}{(\hat{x}_2-\hat{x}_1)^4} \\[.225cm]
\frac{1}{(\hat{x}_1-\xin)^2} & \frac{1}{(\hat{x}_1-\xout)^2} & 0 \\[.225cm]
\frac{1}{(\hat{x}_2-\xin)^2} & \frac{1}{(\hat{x}_2-\xout)^2} & \frac{2}{(\hat{x}_2-\hat{x}_1)^3}
\end{array}
 \right) \\[.5cm]
  \; &  + \det \left( \begin{array}{c c c}
\frac{-2}{(\hat{x}_2-\xin)^3} & \frac{-2}{(\hat{x}_2-\xout)^3} & \frac{-2}{(\hat{x}_2-\hat{x}_1)^3} \\[.225cm]
\frac{1}{(\hat{x}_2-\xin)^2} & \frac{1}{(\hat{x}_2-\xout)^2} & \frac{1}{(\hat{x}_2-\hat{x}_1)^2} \\[.225cm]
\frac{-2}{(\hat{x}_1-\xin)^3} & \frac{-2}{(\hat{x}_1-\xout)^3} & 0
\end{array}
 \right) \\[.5cm]
  = \; & \frac{
   4\hat{x}_1+2\hat{x}_2-6
   \xin}{(\hat{x}_1-\hat{x}_2)^4(\hat{x}_1-\xout)^2
   (\hat{x}_2-\xin)^3}-
   \frac{2\hat{x}_1+4
   \hat{x}_2-6
   \xin}{(\hat{x}_1-\hat{x}_2)^4(\hat{x}_1-\xin)^3
   (\hat{x}_2-\xout)^2} \\
   \; & +
   \frac{4
   \hat{x}_1-4\xin}{(\hat{x}_1-\hat{x}_2)^3(\xout-\hat{x}_1)^3
   (\hat{x}_2-\xin)^3}+
   \frac{4
   \hat{x}_2-4\xin}{(\hat{x}_1-\hat{x}_2)^3(\hat{x}_1-\xin)^3
   (\xout-\hat{x}_2)^3} .
\end{align*}
Up to a coordinate change and an overall multiplicative constant, this formula coincides with the specialization to $\kappa=2$
of the Schramm-Zhou two-point boundary proximity function of the chordal $\SLEk$~\cite{Schramm-Zhou:Boundary_proximity_of_SLE},
as well as with an $\SLEk$ boundary zig-zag amplitude in~\cite{JJK-SLE_boundary_visits}.

For the visits on opposite sides, using replacing algorithm~\ref{alg: continuous replacing}, we find the solution
\begin{align*}
\Ampl_{- +}(\xin ; \hat{x}_1, \hat{x}_2 ; \xout )
= \; & \qquad \qquad \quad  \vcenter{\hbox{\includegraphics[scale=0.5]{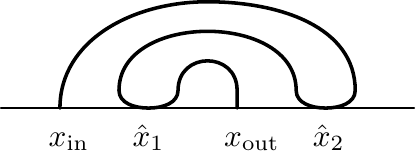}}} \\[1em]
= \; &\det \left( \begin{array}{c c c}
\frac{-2}{(\hat{x}_2-\xin)^3} & \frac{-2}{(\hat{x}_2-\hat{x}_1)^3} & \frac{-6}{(\hat{x}_2-\hat{x}_1)^4}  \\[.225cm]
\frac{1}{(\hat{x}_2-\xin)^2} & \frac{1}{(\hat{x}_2-\hat{x}_1)^2} & \frac{2}{(\hat{x}_2-\hat{x}_1)^3}  \\[.225cm]
\frac{1}{(\xout-\xin)^2} & \frac{1}{(\xout-\hat{x}_1)^2} & \frac{2}{(\xout-\hat{x}_1)^3}
\end{array}
 \right) \\[.5cm]
 = \; & 
  \frac{2
   (\hat{x}_1-\xin)^2
   (\hat{x}_2-\xout)^2
   ( \hat{x}_1\hat{x}_2-3\hat{x}_1
   \xin+2 \hat{x}_1\xout+2
   \hat{x}_2 \xin-3
   \hat{x}_2
   \xout+\xin
   \xout)}{(\hat{x}_1-\hat{x}_2)^6
   (\hat{x}_1-\xout)^3
   (\hat{x}_2-\xin)^3
   (\xin-\xout)^2} .
\end{align*}
Up to a coordinate change and an overall multiplicative constant, this formula coincides with a specialization to $\kappa=2$
of an $\SLEk$ boundary zig-zag amplitude in~\cite{JJK-SLE_boundary_visits}.
%
%

Again, one can check that these functions solve the PDE system in Theorem~\ref{thm: scaling limit of LERW bdry visits} 
and asymptotics properties~\eqref{eq: bdary visit amplitude cascade for xin} hold.



%


\bibliographystyle{annotate}

\newcommand{\etalchar}[1]{$^{#1}$}

\end{document}